\documentclass[reqno,11pt]{amsart}
\usepackage{geometry}
\usepackage[numbers,sort&compress]{natbib}
\usepackage{mathtools,amssymb,amsthm,mathrsfs,color,lineno,paralist,graphicx,float}
\usepackage[colorlinks,
linkcolor=red,
anchorcolor=green,
citecolor=blue,
]{hyperref}

\usepackage[T1]{fontenc}
\usepackage[utf8]{inputenc}

\usepackage{tikz}
\usetikzlibrary{positioning}
\usepackage{calc}
\definecolor{bleu1}{RGB}{0,57,128}
\def\bleu1{\color{bleu1}}

\usepackage{etoolbox}
\patchcmd{\section}{\normalfont}{\normalfont \bleu1}{}{}
\patchcmd{\subsection}{\normalfont}{\normalfont \bleu1}{}{}
\patchcmd{\subsubsection}{\normalfont}{\normalfont \bleu1}{}{}
\renewcommand{\proofname}{\it \bleu1 Proof}

\def\a{\alpha}

\def\e{\varepsilon}

\let\newpf\proof \let\proof\relax 
\newenvironment{pf}{\newpf[\proofname]}{\qed\endtrivlist}

\newcommand{\ba}{\overline{A}}

\def\be{\begin{equation}}
\def\ee{\end{equation}}

\def\ba{{\begin{align}}}
\def\ea{{\end{align}}}

\def\bm{\begin{matrix}}
\def\em{\end{matrix}}

\def\a{{\alpha}}

\def\l{{\bf l}}

\def\0{{\mathbf 0}}

\newtheorem{Theorem}{Theorem}[section]
\newtheorem{Lemma}{Lemma}[section]
\newtheorem{Proposition}{Proposition}[section]
\newtheorem{Corollary}{Corollary}[section]
\newtheorem{Remark}{Remark}[section]
\newtheorem{Example}{Example}[section]
\newtheorem{Definition}{Definition}[section]

\numberwithin{equation}{section}

\theoremstyle{definition}

\def\tr{{\text{tr}}}

\newcommand{\Ker}{\operatorname{Ker}}

\newcommand{\C}{{\mathbb C}}

\newcommand{\N}{{\mathbb N}}
\newcommand{\Q}{{\mathbb Q}}
\newcommand{\R}{{\mathbb R}}
\newcommand{\T}{{\mathbb T}}

\newcommand{\Z}{{\mathbb Z}}

\newcommand{\h}{{\bf h}}

\def\B0{{\bold{0}}}

\catcode`\@=12

\def\Empty{}
\newcommand\oplabel[1]{
  \def\OpArg{#1} \ifx \OpArg\Empty {} \else
    \label{#1}
  \fi}

\newcommand{\comm}[1]{}
\newcommand{\comment}[1]{}

\begin{document}
\title[]{The Robust Ten Martini Problem}

\author {Lingrui Ge}
\address{Beijing International Center for Mathematical Research, Peking University, Beijing, China
	} \email{gelingrui@bicmr.pku.edu.cn}
\author{Svetlana Jitomirskaya}
\address{
	Department of Mathematics, University of California, Berkeley CA, 94703} \email{sjitomi@berkeley.edu}
\author{Jiangong You}
\address{
	Chern Institute of Mathematics and LPMC, Nankai University, Tianjin 300071, China} \email{jyou@nankai.edu.cn}

      \begin{abstract}
  We solve the robust ten martini problem: Cantor spectrum with no arithmetic restriction on the irrational frequency for Type I analytic one-frequency quasiperiodic Schrödinger operators. We introduce Type I through a new acceleration-type quantity with stability under analytic perturbations. This yields a global open class of analytic potentials whose associated operators have Cantor spectrum for every irrational frequency.  

\end{abstract}
\footnotetext{This paper supercedes our earlier preprint entitled “Kotani theory, Puig's argument, and stability of The Ten Martini Problem
”, where we proved the result for even trigonometric polynomials. The
latter preprint is not intended for publication.} 
\maketitle
\tableofcontents
\section{Introduction}

The Hofstadter butterfly \cite{hof}, a plot of the  band spectra of almost
Mathieu operators
\begin{align}\label{amo}
(H_{\lambda,\alpha,x}u)_n=u_{n+1}+u_{n-1}+2\lambda\cos2\pi(x+n\alpha)u_n,
\end{align}
at rational frequencies $\alpha$, has become a pictorial symbol
of the field of quasiperiodic operators. It is visually clear
from this plot that for {\it all irrational frequencies} the spectrum must be
a Cantor set, a statement that has been dubbed the ten martini problem
by Barry Simon \cite{barry} after an 1981 offer of Mark Kac \cite{kac}. The problem itself is
considered iconic in the field of quasiperiodic operators, its final
solution \cite{aj} requiring a combination of many ideas and techniques and
significant ingenuity. The proof in \cite{aj} used the specific nature of the
almost Mathieu operator in several key ways and was based on different
approaches at the Diophantine and Liouville sides that miraculously
met at the middle \cite{solving}, enough so that in a field driven by
bold conjectures (e.g. \cite{sim15,simXXI}) a conjecture that the same
statement could hold for other operators of the form \begin{align}\label{sch}
(H_{v,\alpha,x}u)_n=u_{n+1}+u_{n-1}+ v(x+n\alpha)u_n,\ \ n\in\Z,
\end{align} with continuous $v,$ has never even
been explicitly made.

Indeed, while there were several Cantor spectrum results for
operators \eqref{sch} with analytic $v$, other than for the almost
Mathieu family, all required various (often implicit) conditions on frequencies,
among other unnatural restrictions. At the same time, the physics
nature and relevance of the almost Mathieu family strongly
suggest that the ten martini problem has to be robust and hold at least in the
entire analytic neighborhood of \eqref{amo}.  Here we develop a method that does not rely neither on the arithmetics,
nor on the almost Mathieu specifics to prove for the first time the {\it robust ten
  martini problem}:
Cantor spectrum for all $\alpha\in \R\backslash\Q$
  holds for operators \eqref{sch} for a global open set of analytic potentials.
  
The set is explicitly described in our main Theorem \ref{main11}, and is
conjectured to be not only open but also dense in every $C^\omega_h$ \footnote{Let $F$ be a bounded analytic (possibly matrix valued) function defined on $ \{ x |  | \Im x |< h \}$,
$\| F\| _h= \sup_{ | \Im x |< h } \| F(x)\|<\infty $. $C^\omega_{h}(\T,*)$ denotes the
set of all these $*$-valued functions. $C^{\omega}(\T,*)$ is the union $\cup_{h>0}C_h^{\omega}(\T,*)$.},
by e.g. a very compelling analogy with the case of one-dimensional cocycles
\cite{gj}. Note that by a recent counterexample by Argentieri-Avila  \cite{AA},
there do exist analytic potentials without all-frequency Cantor spectrum.

One-frequency analytic quasiperiodic Schr\"odinger operators on
$\ell^2(\Z)$ are given by \eqref{sch} where $v\in C^\omega(\T,\R)$ is
a 1-periodic non-constant real analytic function;
 $\alpha\in\R\backslash\Q$ and $x\in\R$ are parameters (called
the {\it frequency} and the {\it phase} respectively). Their theory has been developed extensively (see \cite{bbook,DF1, DF2,youcongr,jitcongr} for more recent surveys).  
The almost Mathieu operator \eqref{amo}  (AMO) is the
central/prototypical model, lying both at the physics origin and the
center of current physics interest of the
field \cite{Peierls, harper,R,aos,oaag,h,ntw,nichen}, as well as driving many
of the mathematical developments. The latter is, at least historically
due to Barry Simon's problems \cite{sim15,simXXI} prominently featuring several
almost Mathieu questions. This has remarkably led to all of them being solved,
and then many new ones appearing.

The study of general operators \eqref{sch} with analytic
$v$ has long been developed in the perturbative regime, with the key
highlights in \cite{ds,fsw,sin} and especially Eliasson \cite{Eli92,Eli97}. 
The nonperturbative analysis has taken off after the work of Bourgain and collaborators
(see \cite{bbook}), who significantly developed theory of operators
\eqref{sch}, especially in
the regime of positive Lyapunov exponents, an important catalyst to
these developments being again the almost Mathieu result \cite{j}. The
development of nonperturbative/Liouvillean KAM \cite{hy,afk} and quantitative
reducibility (see \cite{youcongr} and references therein) has led to
many strong results in the (almost) reducibility  regime.

Avila's global theory \cite{avila0} of operators \eqref{sch},  based on the analysis of
complexified Lyapunov exponents, has brought new vision and
understanding, in particular, introducing a simple yet fundamental concept of acceleration,
as an important feature that allows to divide the spectrum into more
manageable subsets. With a rough division of the spectrum  into subcritical, critical,
and supercritical energies, Avila showed that critical ones are very
rare in a strong sense \cite{avila0}, while almost reducibility  becomes a
corollary of subcriticality \cite{arc1,arc2}, see also \cite{ge} for a
different proof for the Diophantine Schr\"odinger case.

Despite all these remarkable advances, many major results that do not
require unnecessary and/or non-explicit parameter exclusion
e.g. \cite{jliu1,jliu2,liuresonant,aj,ak, alsz, last, dryAYZ, gk}, still exist only for the almost Mathieu operators and have heavily used
several different almost Mathieu specifics. Of those, the ten martini
problem particularly stands out 
Indeed, while the statement is, by
design, about  {\it all} irrational $\alpha,$ historically, the proofs 
developed very different arguments depending on the arithmetic
properties of $\alpha,$
all using the specific features of \eqref{amo}, with, particularly,
the Liouville side of the argument \cite{cey} based on the almost
Mathieu specific algebraic miracle allowing the rational gap estimate.  


Here we present a method of proof of Cantor spectrum that does not depend on
the almost Mathieu symmetry, self-duality, low degree of the
potential, or any sort of perturbation, treats all
irrational frequencies by the same mechanism. Its natural domain is described by a {\it stable} acceleration-type quantity, which we now introduce.

  Lyapunov exponent of complexified  Schr\"odinger cocycles are defined as
\begin{align}\label{multiergodicsch}
L_\e(E)=\lim\limits_{n\rightarrow\infty}\frac{1}{n}\int_\T\ln \|S_E^v(x+i\e+(n-1)\alpha)\cdots S_E^v (x+i\e)\|dx
\end{align}
where
\begin{equation}\label{S}
S_E^v (x)=\begin{pmatrix}E-v(x)&-1\\ 1&0\end{pmatrix}.
\end{equation}
See Section \ref{pre} for details. Avila showed \cite{avila0} that for fixed $E\in \R$, $L_\e(E)$ (as a function of $\e$) is an even convex piecewise
affine function with integer slopes, in particular  
 integer-valued  {\it acceleration}
$$
\omega(E)=\lim\limits_{\e\rightarrow
  0^+}\frac{L_\e(E)-L_0(E)}{2\pi\e}
$$
that allows to divide the spectrum into analytic strata. Acceleration allows to isolate
a key feature of the {\it supercritical} almost Mathieu. 
It does not, however, distinguish the {\it subcritical} almost Mathieu operator from a general subcritical operator, and, moreover,  acceleration  is not stable under analytic perturbations
\footnote{See e.g. Figure \ref{fig}, illustrating the three almost Mathieu regimes, with acceleration on the spectrum changing from $0$ to $1$ at
$\lambda=1.$ However, these
pictures also illustrate that the almost Mathieu acceleration is {\it always} bounded by $1,$ and it is this
feature that turns out to be both very robust and important for our proof of Cantor spectrum. What is stable is the slope immediately after the first turning point of the complexified Lyapunov exponent.}. 
Here  we introduce a new (and stable) concept that achieves that, allowing to
divide {\it both sub-} and (super-)critical parts of the spectrum into more
manageable sets. 


\begin{figure}[htbp]\label{fig}
\centering
\begin{minipage}[t]{0.3\linewidth}
\centering
\includegraphics[width=1.2\linewidth]{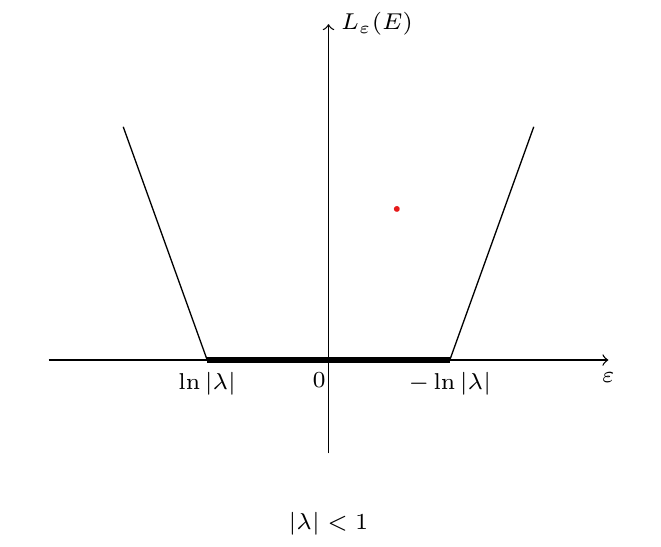}
\end{minipage}
\begin{minipage}[t]{0.3\linewidth}
\centering
\includegraphics[width=1.2\linewidth]{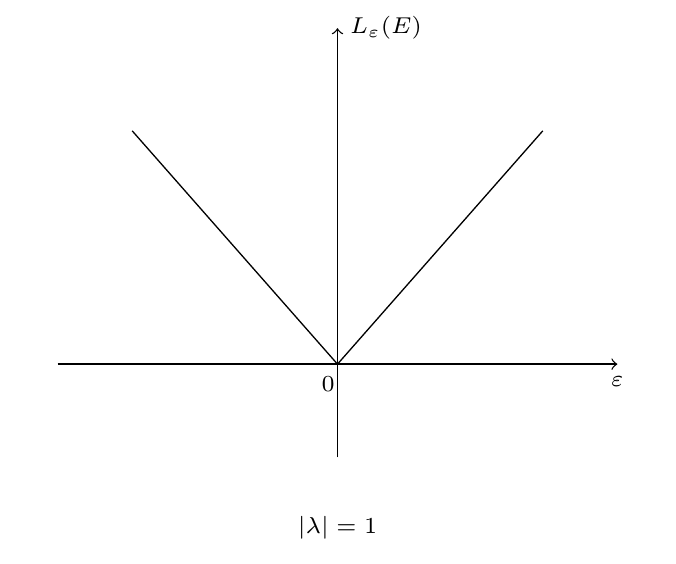}
\end{minipage}
\begin{minipage}[t]{0.3\linewidth}
\centering
\includegraphics[width=1.2\linewidth]{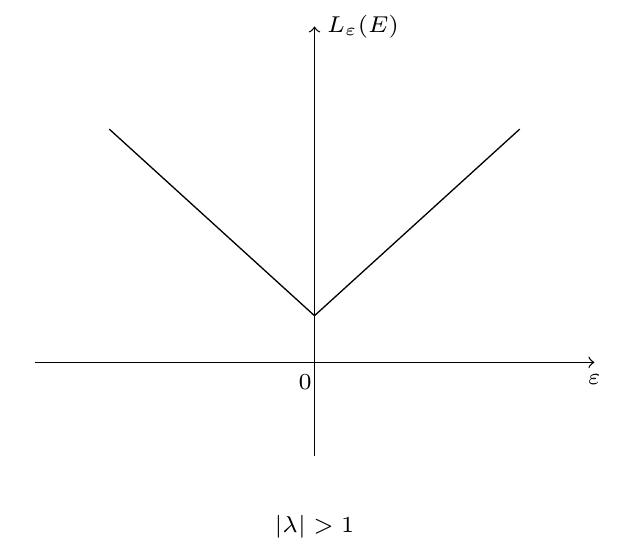}
\end{minipage}
\end{figure}


\begin{Definition}[T-acceleration]\label{gege1}
{\rm The {\it T-acceleration} is defined by
$$
\bar{\omega}(E)=\lim\limits_{\e\rightarrow \e_1^+}\frac{L_\e(E)-L_{\e_1}(E)}{2\pi(\e-\e_1)}
$$
where $0\le\e_1<\infty$ is the first turning point \footnote{It is an
  easy corollary of  the results of \cite{gjyz} that for
non-constant trigonometric polynomial $v$ we always have $\e_1<\infty$.} of the piecewise affine
function $L_\e(E)$.  If there is no
turning point, we set $\bar{\omega}(E)=0$.}
\end{Definition}
\begin{Remark}
{\rm
\begin{enumerate}
    \item $\bar{\omega}(\alpha,A)$ can be defined for any analytic
      one-frequency $SL(2,\C)$-cocycle $(\alpha,A)$, see Section 4.1.
      \item
Obviously, $\omega(E)\leq \bar{\omega}(E)$ for any $E\in\R$ and
  the equality holds if and only if $\omega(E)>0$. In particular, for
  the almost Mathieu operator, $\bar{\omega}(E)=1$ for all $E$ in the
  spectrum, in all three regimes. 
  \item  Lemma \ref{con} shows that the condition \(\overline\omega=1\) is locally stable under simultaneous perturbations of the frequency and the analytic cocycle in every analytic
topology whose strip extends strictly beyond the first turning point.
  \end{enumerate}
  }
\end{Remark}
This allows us to introduce an important class of operators \eqref{sch}
\begin{Definition}[Type I]
  {\rm We say $E$ is a {\it type I energy} for operator
    $H_{v,\alpha,x}$ if
  $\bar{\omega}(E)=1$. We say $H_{v,\alpha,x}$ is a {\it type I operator,}
  if every $E$ in the spectrum of
  $H_{v,\alpha,x}$ is type I. We say that $v\in C^\omega(\T,\R)$ is of uniform type I if $H_{v,\alpha,x}$ is type I for every $\alpha\in \R\backslash\Q.$ 
  }
\end{Definition}

While Type I operators generally have none of the other nice almost
Mathieu features such as symmetry, self-duality, or low-degree, we prove
\begin{Theorem}\label{main11}
For any $\alpha\in \R\backslash\Q$ and any Type I operator \eqref{sch}  its
spectrum is a Cantor set.
\end{Theorem}
\begin{Remark}
{\rm\begin{enumerate}\item
T-acceleration and the Type I class were originally introduced in our earlier work \cite{gjy}, that the current paper supercedes. There, building on the duality approach to the global  theory of \cite{gjyz}, we  proved Cantor spectrum for even trigonometric-polynomial Type I potentials at every irrational frequency. The subsequent work \cite{gj}, developed the intrinsic structure of the two-dimensional dual center, including its hidden subcriticality and rotation-number structure.

Building, in part, on this Type I structural and quantitative framework, Li--Xu--Zhou \cite[v1]{lxz} proved the stronger dry Ten Martini statement in the positive-Lyapunov Type I regime for trigonometric-polynomial potentials. In the upcoming new version of \cite{lxz} the dry statement will be extended to analytic potentials. Thus, on the corresponding supercritical subclasses, these works obtain the stronger conclusion that every labelled gap is open.

The subcritical Type I regime is, however, in several respects more delicate. The present theorem treats the full general-analytic Type I class, including these zero-Lyapunov regimes. Its proof is self-contained once the duality framework of \cite{gjyz} is taken as input; in particular, it does not use the Almost Reducibility Conjecture \cite{arc1,arc2,ge} whereas the dry gap-opening arguments above additionally use almost-reducibility/reducibility input.
  \item  By upper-semicontinuity, for fixed $\alpha\in\R\backslash\Q$, the Type I condition is
open in $C_h^\omega(\T,\R)$ whenever
$v\in C_h^\omega(\T,\R)$ and
$$
\sup_{E\in\Sigma_{v,\alpha}}
\e_1(\alpha,S_E^v)<h;
$$
see Corollary~\ref{con}.   
  \item Unlike the almost Mathieu operator, $L_\e(E)$ of general  type I
    operators may have many turning points, see Fig \ref{fig2}.
    \begin{figure}[htbp]\label{fig2}
\centering
\begin{minipage}[t]{0.3\linewidth}
\centering
\includegraphics[width=0.944\linewidth]{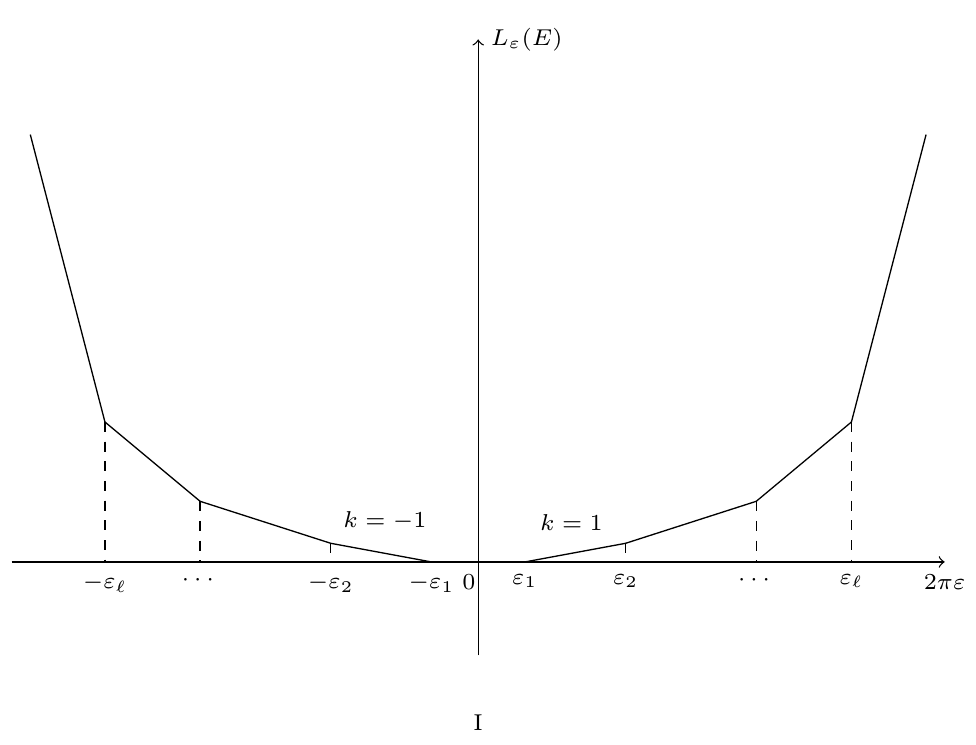}
\end{minipage}
\begin{minipage}[t]{0.3\linewidth}
\centering
\includegraphics[width=0.944\linewidth]{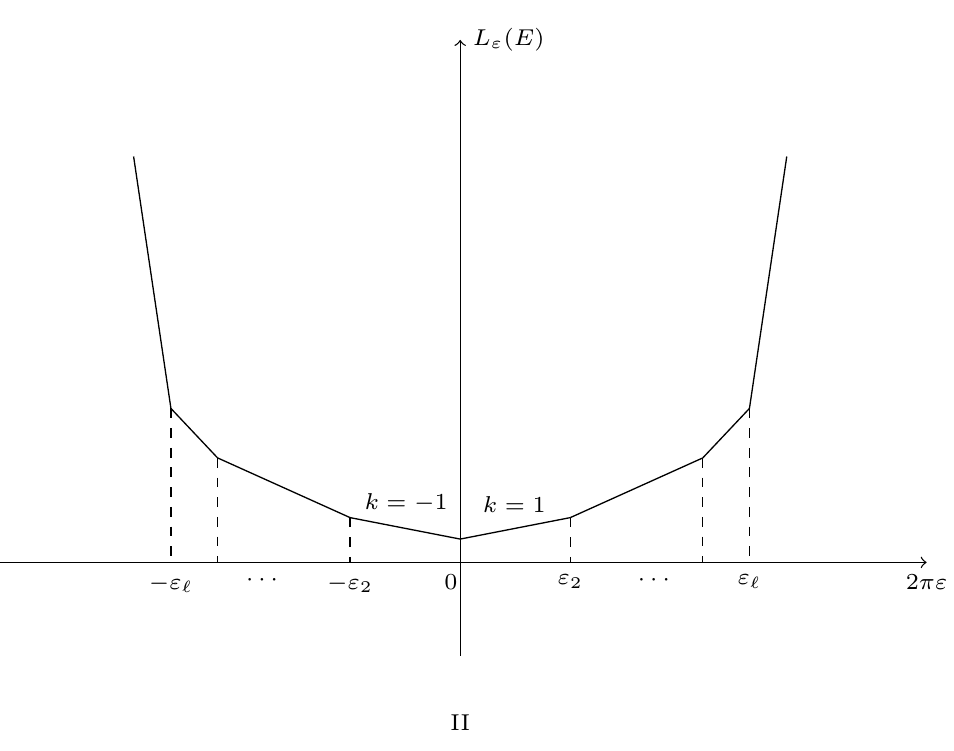}
\end{minipage}
\begin{minipage}[t]{0.3\linewidth}
\centering
\includegraphics[width=0.944\linewidth]{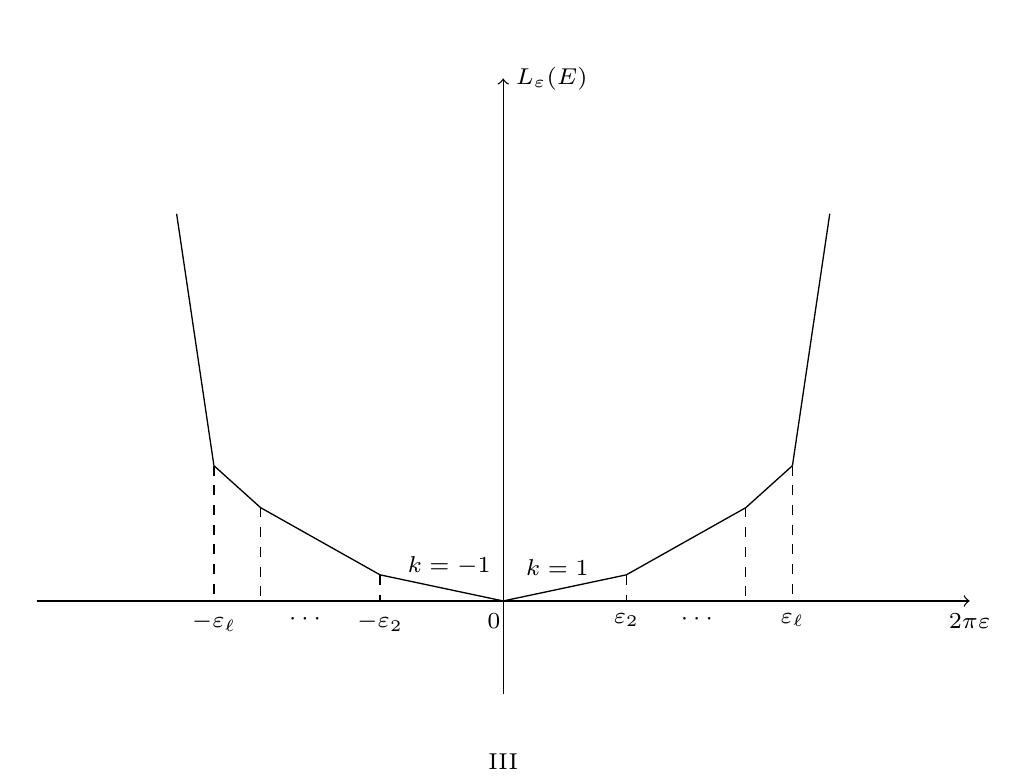}
\end{minipage}
\end{figure}
    \end{enumerate}}
 \end{Remark}
 

Moreover, the result can actually be localized to the set of energies with
$\bar{\omega}(E)=1.$ Let $\Sigma^1_{v,\alpha}=\{E\in
\Sigma_{v,\alpha}: \bar{\omega}(E)=1\}$. Theorem \ref{main11} is a
direct corollary of the local version
\begin{Theorem}\label{main12}
For any $\alpha\in\R\backslash\Q$, if $v$ is real analytic, $\overline{\Sigma_{v,\alpha}^{1}}$ is either empty or a Cantor set.
\end{Theorem}

\begin{Remark}
{\rm\begin{enumerate}\item
 In the subcritical Type I regime, the methods of this paper yield a stronger conclusion. Indeed, combining the all-frequency Puig argument developed here with quantitative almost reducibility developed in \cite{dryAYZ} gives the dry Ten Martini statement for subcritical Type I operators. The argument follows the route of \cite{gjy}, where the corresponding subcritical dry statement was proved for Type I operators with trigonometric polynomial potentials and Diophantine frequencies. To avoid lengthening the present paper, we defer the details to a separate note.
 \item Cantor spectrum results for the almost Mathieu operators have had
  a long history \cite{bs,sin,hs,cey,last,ak,Puig} prior to the celebrated proof of the ten martini
  problem  in \cite{aj} that handled the measure zero set of remaining arithmetically intermediate $\alpha$ left after the Diophantine \cite{Puig} and Liouville \cite{cey} proofs.  There were many remarkable Cantor spectrum results for non-almost Mathieu operators \eqref{sch}
  with analytic $v,$
but they were all either for (unspecified) typical  $v$ and Diophantine $\alpha$
  in the zero Lyapunov exponents regime \cite{avilajams} \footnote{Or
  \cite{aj1} that provides the dry
  version but, again, only for typical $v$ and with a further
  Diophantine/smallness restriction.} or (unspecified) typical $\alpha$ in
  the regime of positive $L(E)$, proved in a combination of very
  technically complicated 
  \cite{gs1,gs2}, or in the perturbative regime, and also Diophantine $\alpha$ \footnote{So, again,
    not for all irrational $\alpha$ for any given $v$.}, \cite{sin,wz,gwx, hhsy, hsy, hz}  and 
  special geometry of $v.$
  \end{enumerate}}
\end{Remark}
The Type I condition is initially a property of the pair $(\alpha,v)$. For a fixed irrational $\alpha$, Corollary~\ref{sgao} gives local
stability in every admissible analytic topology: the
$C_h^\omega(\T,\R)$ topology is admissible at $(\alpha,v)$ when
$v\in C_h^\omega(\T,\R)$ and
$$
\sup_{E\in\Sigma_{v,\alpha}}
\e_1(\alpha,S_E^v)<h.
$$
Thus the topology must control a strip extending uniformly beyond
the first turning points of all spectral energies; it cannot in
general be replaced by an arbitrary narrower $C_h^\omega$ topology.

For $h>0$, let
 $\mathcal U_h\subset C_h^\omega$ be the set of uniform Type I potentials $v$
for which the Type I property persists on a common
$C_h^\omega$-neighborhood, uniformly over all irrational frequencies. $\mathcal U_h$ is open by design and the uniform perturbation results in
Section~4.2 show that $\mathcal U_h$ is nonempty for $h>0$ and, for large $h,$ contains the
neighborhoods of potentials for the almost Mathieu and several other popular models.
Theorem~\ref{main11} therefore implies that
for
$v\in\mathcal U_h$, 
$
\Sigma_{v,\alpha}\ \text{is a Cantor set for every }
\alpha\in\R\backslash\Q,
$ or, in other words the ten martini conjecture holds.

It is natural to expect that Type I behavior is generic in the strong
sense of being open and dense. Here and below, ``generic'' means open
and dense in the indicated topology; for subsets of a spectrum,
openness and density are understood in the relative topology. The
following conjectures express increasing degrees of uniformity and
ambition:

\vskip .1in
\noindent
{\bf Conjecture 1.}\label{conj}
\begin{enumerate}
\item For every $\alpha\in\R\backslash\Q$ and every $h>0$, for an open
and dense set of potentials $v\in C_h^\omega(\T,\R)$, the set
$\Sigma^1_{v,\alpha}$ of Type I energies is open and dense in
$\Sigma_{v,\alpha}$; see also Conjecture 3.2 in \cite{yicmp}.

\item For every $\alpha\in\R\backslash\Q$ and every $h>0$, the set of
potentials $v\in C_h^\omega(\T,\R)$ for which $H_{v,\alpha,x}$ is
Type I and this property persists throughout a
$C_h^\omega(\T,\R)$-neighborhood of $v$ is dense in
$C_h^\omega(\T,\R)$. Since this set is open by definition, it is
generic in the above sense.

\item For some $h>0$, the open set $\mathcal U_h$ is dense in
$C_h^\omega(\T,\R)$, and is therefore generic in the above sense.
\end{enumerate}

\vskip .1in

Indeed, for trigonometric polynomial potentials, quantitative duality identifies the Type I condition with simplicity of the lowest dual Lyapunov exponent, thus 
Conjecture~\ref{conj} may be viewed as a 
generic-simplicity conjecture for the lowest dual Lyapunov exponent, with genericity understood in the open-and-dense sense above. See also \cite{gj} for a compelling 1D analogy.

The introduction of the Type I condition is one of the conceptual points of this work. It captures a feature shared by all spectral energies of the almost Mathieu operator but, unlike ordinary acceleration, is stable under analytic perturbations.  This condition can be recast by \cite{gjyz} as partial hyperbolicity of the dual cocycle with a {\it two-dimensional center}. This is what replaces the special self-dual \(SL(2,\mathbb R)\) dynamics of the almost Mathieu operator in our argument.

Our method builds on two recent developments. The first is the
quantitative global theory developed in \cite{gjyz}, an
Aubry-duality-based approach to Avila's global theory \cite{avila0}
that relates the dynamics of dual cocycles to complexified Lyapunov
exponents and spectral properties of $H_{v,\alpha,x}$. A second conceptual foundation is the intrinsic symplectic structure of
the limiting center dynamics developed in \cite{gj}. This structure survives the passage from
finite-range dual cocycles to the analytic limit and makes it possible
to extend finite-range arguments from trigonometric polynomial
potentials to general analytic potentials. The general
theory there applies to arbitrary center dimension and includes the
associated rotation theory. For completeness, the appendix contains
self-contained proofs of the specialized two-dimensional facts needed
in the present paper.

On this foundation, we establish three general results that are not specific to the ten martini application: a generalized Kotani theory for partially hyperbolic analytic cocycles with two-dimensional center,  simplicity of point spectrum for minimal, uniquely ergodic finite-range  operators with such cocycles, and an all-frequency version of Puig’s argument.

A very special property of AMO, not shared by {\it any} other operators
\eqref{sch}, is that its
  Aubry dual 
   is also of the form \eqref{sch}, allowing to deal only with
$SL(2,\R)$ cocycles on both sides. For  trigonometric
polynomial $v$  of degree $d$, the dual is of the form

\begin{align}\label{fi}
(L_{v,\alpha,\theta}^wu)_n=\sum\limits_{k=-d}^{d} \hat{v}_k u_{n+k}+2w(-(\theta+n\alpha))u_n,\ \ n\in\Z,
\end{align}
where $w=\cos 2\pi x$ and $\hat{v}_k$ is the $k$-th Fourier
coefficient of $v$. 
This requires dealing with $Sp(2d, \C)$ transfer-matrix cocycles, in
place of $SL(2,\R)$ for discrete Schr\"odinger operators \eqref{sch}. 

We will define the key new  concept that allows to define the ``good'' cocycles, in
 a more general setting.  Let $\Omega$ be a compact metric space, and $T:\Omega\rightarrow\Omega$
 be a minimal homeomorphism \footnote{That is
   $\overline{\{T^n\omega\}_{n\in\Z}}=\Omega$ for any
   $\omega\in\Omega$.}.  For continuous $f:\Omega\rightarrow\R,$
 {minimal} finite-range operators on $\ell^2(\Z)$ are defined by
\begin{equation}\label{finite operator11}
(L_{f,\omega}u)_n=\sum\limits_{k=-d}^d a_ku_{n+k}+f(T^n\omega)u_n,\ \ n\in\Z,
\end{equation}
where $a_k=\overline{a_{-k}}$ is a complex sequence. As usual, the eigenvalue
equation $L_{f,\omega}u=Eu$ defines a  complex symplectic cocycle
$(T,L_E^f)$ \footnote{with respect to $S$ defined in \eqref{sdef1}.} (see \eqref{1111}) and we denote its non-negative
Lyapunov exponents by
$\{L_i^f(E)\}_{i=1}^d$ (see Section \ref{3.2} for the definitions). Let
$\Sigma_f$ be the spectrum of $L_{f,\omega}$.

\begin{Definition}\label{defph2}
{\rm We  say that a complex symplectic {\it cocycle} is {\it $PH2$} if it is
{\it partially hyperbolic with two-dimensional center}. We say that an {\it energy $E$} is {\it $PH2$} for operator $L_{f,\omega}$ if the cocycle
$(T,L_E^f)$ is $PH2$. We
say that {\it operator $L_{f,\omega}$}  is {\it $PH2$} if every $E\in
\Sigma_f$ is $PH2$.}
\end{Definition}

In other words, $L_{f,\omega}$ is $PH2$ if

\begin{enumerate}
\item{$L_1^f(E)\geq \cdots\geq L_{d-1}^f(E)>L^f_d(E)$ for all $E\in\Sigma_f$;}
\item $(T,L_{E}^f)$ is $(d-1)$ and $(d+1)$-dominated for all
  $E\in\Sigma_f$ \footnote{See Section \ref{3.3} for the definition of domination.}.
\end{enumerate}

For trigonometric polynomial potentials, the dual operator is
finite-range, and, as shown in Section~4, the Type I condition implies
that its transfer cocycle is $PH_2$. For a general analytic potential,
the dual operator is infinite-range and there is no associated cocycle, so the term $PH_2$ does not apply
literally. Instead, the $PH_2$ splittings of the finite-range
truncations possess an intrinsic limit: their two-dimensional
symplectic center dynamics converge, while the stable and unstable
directions satisfy estimates uniform in the truncation. This limiting
center structure is the infinite-range substitute for the $PH_2$ property
used in Sections~10 and~11.

A key technique that has been of profound importance for the study of
ergodic Schr\"odinger operators is  celebrated Kotani theory
\cite{kot,sim83} that proves, in particular,  that
operators with Lyapunov exponents vanishing on a
set of positive measure are deterministic, allowing, 
the holomorphic extension of the
$m$-function through the interval with zero Lyapunov exponents. The
latter has been crucially used in the ten martini proof of
\cite{aj}.

Kotani theory has been extended by Kotani-Simon \cite{ks} to {\it real}
Jacobi matrices on the strip, for {\it real} matrix-valued Schr\"odinger operators, with corresponding  $Sp(2d,\R)$
transfer-matrix cocycles. However, the key result on the
reflectionness of $M$-matrix  \footnote{See Section \ref{6.2} for the
  definition.}, enabling, in particular, the holomorphic extension
described above, was proved in \cite{ks} only under the condition that {\it all
 } Lyapunov exponents are zero. At the same time, it was conjectured
 in \cite{ks} that under certain additional assumptions at least a
 partial version of the result
 should only require vanishing of {\it some
 } Lyapunov exponent \footnote{Additional assumptions/partial form are needed as
   demonstrated by an example of a model with decoupled potentials in
   \cite{ks}.}. Over the years, there have been some extensions of
 \cite{ks}, e.g.\cite{damanik1,oliviera}, and, most notably, \cite{xu},
 but the above problem stubbornly defied progress.

We first resolve this problem for analytic one-frequency \(PH_2\) operators. More precisely, we prove partial reflectionlessness when only the center Lyapunov exponent vanishes, while the remaining Lyapunov exponents are positive, and we allow complex hopping coefficients. This yields a complete Kotani theory for the \(PH_2\) class. We further establish a Kotani theory for the infinite-range duals of Type I operators—the first, to our knowledge, for infinite-range  operators. In the infinite-range setting where no   cocycle exists, the theory instead uses the intrinsic two-dimensional symplectic center dynamics developed in \cite{gj}, together with the uniform stable/unstable estimates for the finite-range truncations established here.
  This Kotani theory extension forms the most technically
 involved part of the paper.
 
More precisely, the results of \cite{ks} imply that the $M$ matrix (defined in Section
\ref{6.2}) is reflectionless on $\{E: L^f_1(E)=\cdots=L^f_1(E)=0\}$ provided $a_k=a_{-k}$ is a real sequence. A crucial part of our proof of the Cantor spectrum is the following
general result on
{\it partial reflectionless} that allows  for 
 positive Lyapunov exponents and complex sequence $\{a_k\}_{k\in\Z}$.


Let
\begin{equation}\label{lanair2}
Sp_{2d\times 2}(\C)=\{F\in M_{2d\times 2}(\C): F^*S F=J\}, \ \ J=\begin{pmatrix}0&1\\ -1&0\end{pmatrix}.
\end{equation}
To state the reducibility result, we also need the higher
accelerations. For the analytic one-frequency cocycle $(T,L_{E}^f)$, let
\[
L^k(T,L_{E}^f)=\sum_{j=1}^k L_j(T,L_{E}^f)
\]
and define
\[
\omega^k(T,L_{E}^f)
=
\lim_{\varepsilon\to0^+}
\frac{
L^k\bigl(T,L_{E}^f(\cdot+i\varepsilon)\bigr) -L^k(T,L_{E}^f)
}{
2\pi\varepsilon
}.
\]
Thus $\omega^k$ is the acceleration of the sum of the $k$ largest
Lyapunov exponents; see Section~3.4 for further details. We will use
only $\omega^{d-1}$ below.
\begin{Theorem}\label{L2 reducibility}
For analytic one-frequency $PH2$ cocycles $(T,L_{E}^f)$, for almost every $E$ in $\{E: L^f_d(E)=\omega^{d-1}(T,L_E^f)=0\}$, there exist $H_E\in L^2(\Omega,Sp_{2d\times2}(\C))$, $\phi_E\in C^0(\Omega,\R)$ and $R_E:\Omega\rightarrow SO(2,\R)$ such that
$$
L_E^f(\omega)H_E(\omega)=H_E(T\omega)e^{2\pi i\phi_E(\omega)}R_E(\omega).
$$
\end{Theorem}

\begin{Remark}
{\rm \begin{enumerate}
\item The theorem as stated is fully sufficient for our current
  purposes, however the simplicity of $L^f_d(E)$ (or, equivalently, the
  two-dimensionality of the center), while used substantially in the proof, is not essential, and this
  condition will be removed in the   upcoming work of the first author with D.Xu
  \cite{gx}.  What we currently see as crucially important is that the
  cocycle $(T,L_E^f)$ is partially hyperbolic. It is an interesting
  question whether this assumption is necessary for the result of Theorem \ref{L2 reducibility}.
\item $\phi_E\equiv0$ if  $L_{f,\omega}$ is real-valued. 
  Theorem \ref{L2 reducibility} is also the first result on
  $L^2$-reducibility for {\it complex}  finite-range operators 
\end{enumerate}
}
\end{Remark}
We will actually use Theorem \ref{L2 reducibility} through the
following corollary, also of independent interest
\begin{Corollary}\label{C0reducibility}
If $(T,L_{E}^f)$ is $PH2$ and $L_d^f(E)=\omega^{d-1}(T,L_E^f)=0$
 on an interval $I\subset \R$, there exist  $H_E\in C^\omega(\Omega,Sp_{2d\times2}(\C))$, $\phi_E\in C^\omega(\Omega,\R)$ and $R_E\in C^\omega(\Omega,SO(2,\R))$, depending analytically on $E\in I'\subset I,$ such that
$$
L_E^f(\omega)H_E(\omega)=H_E(T\omega)e^{2\pi i\phi_E(\omega)}R_E(\omega).
$$
\end{Corollary}

\begin{Remark}
{\rm Corollary \ref{C0reducibility} implies in a standard way that if
  $L_d^f(E)=0$ on an open interval $I$, then $L_{f,\omega}$ has
   purely absolutely continuous spectrum on $I$ for any $\omega\in\Omega$.}
\end{Remark}

Another general  ingredient is a criterion of simplicity of point
spectrum for all minimal,  uniquely ergodic finite-range $PH2$ operators. It is well known, by an easy
Wronskian argument, that point spectra of second-difference operators
are always simple. Certainly, the argument breaks down for
higher-difference operators, only implying absence of point spectra of
correspondingly high multiplicity. It turns out, however, that for
 $PH2$
operators
$L_{f,\omega},$ given by \eqref{finite operator11} with minimal, uniquely ergodic $T,$
point spectrum is always simple, {\it for any $d$}, thus $PH2$ property implies certain {\it essential
  second-differenceness} of these higher-difference operators.
 \begin{Theorem}\label{tsimp}
 
 $PH2$ operators $L_{f,\omega}$ with minimal, uniquely ergodic $T$ and continuous $f:\Omega\rightarrow \R$ have simple point spectrum for any $\omega\in\Omega$.
 \end{Theorem}

This allows an extension of Puig's argument originally designed for
the almost Mathieu operators. Puig showed \cite{puig} that Schr\"odinger cocycles
associated with almost Mathieu operators cannot be reduced to the
identity. This was a key element in his proof of the Cantor spectrum
for Diophantine $\alpha$. The argument itself is almost Mathieu
specific precisely because it is based on simplicity of point spectrum of
the dual operator, which, for the almost Mathieu, is again in the
almost Mathieu family.  \footnote{This issue is relevant, for example, to the dry Ten Martini claim in
\cite{hs1}. There Puig's argument is invoked for trigonometric-polynomial
potentials, whose Aubry duals are higher-range operators rather than
second-difference operators. Simplicity of the dual point spectrum is
therefore no longer automatic, and no argument establishing the
required simplicity is provided in \cite{hs1}. Thus, as written, the
Puig's argument step in \cite{hs1} contains a genuine gap and does not establish
the claimed dry conclusion without an additional argument such as the
simplicity theorem proved here.} For general operators \eqref{sch}, a dual
operator is  
defined by 
\begin{align}\label{fi}
(L_{v,\alpha,\theta}u)_n=\sum\limits_{k=-\infty}^{\infty} \hat{v}_k u_{n+k}+2\cos2\pi(\theta+n\alpha)u_n,\ \ n\in\Z,
\end{align}
In
particular, if $v$ is a
trigonometric polynomial potential of degree $d,$ operator
$L_{v,\alpha,\theta}$ is finite-range  of the form \eqref{finite
  operator11}, so for every $E$ we can define a 
complex-symplectic transfer matrix cocycle. Slightly abusing the language, we will
call such cocycle the {\it dual cocycle} of $(\alpha,S_E^v)$ and/or of
the operator \eqref{sch}.

Theorem \ref{tsimp} immediately leads to
the following generalized Puig's argument

\begin{Theorem}\label{gjy-puig}

A Schr\"odinger cocycle $(\alpha,S_E^v)$ whose dual is $PH2$,
cannot be (analytically) reduced to the identity, i.e. there does not exist $B\in C^{\omega}(\T,SL(2,\R))$ such that
$$
B^{-1}(x+\alpha)S_E^v(x)B(x)=Id.
$$

\end{Theorem}

The reducibility-based arguments, however, require Diophantine
conditions, thus cannot work within an {\it all $\alpha$} proof. Here we
develop a uniform in $\alpha$ scheme, by replacing localized
eigenfunctions and reducibility to the identity in
a Puig-type argument with {\it almost localized eigenfunctions} and
{\it rotations reducibility}. We prove

\begin{Theorem}\label{rot1}
A Schr\"odinger cocycle $(\alpha,S_E^v)$ whose dual is $PH2,$
cannot be (analytically) reduced to  a rotation with zero
rotation number, i.e. there does not exist $B\in C^{\omega}(\T,SL(2,\R))$ and $\psi\in C^\omega(\T,\R)$ with $\int_\T \psi(x)dx=0$ such that
$$
B^{-1}(x+\alpha)S_E^v(x)B(x)=R_{\psi(x)}.
$$
\end{Theorem}

 Quantitative global theory \cite{gjyz} implies that, at every Type I energy, the dual
finite-range cocycle is $PH_2$.
From this, using the limiting
two-dimensional symplectic center dynamics of the finite-range
truncations, we prove
the corresponding infinite-range Puig-type obstruction 

\begin{Theorem}\label{rot2analytic}
For $\alpha\in\mathbb R\backslash\mathbb Q$,  real
analytic $v$ and Type I $E\in\Sigma_{v,\alpha}$, 
there does not exist $B\in C^{\omega}(\T,SL(2,\R))$ and $\psi\in
C^\omega(\T,\R)$ with $\int_\T \psi(x)dx=0$ such that
$$
B^{-1}(x+\alpha)S_E^v(x)B(x)=R_{\psi(x)}.
$$
\end{Theorem}
This also allows us to establish  a Kotani theory for quasiperiodic {\it infinite-range} operators, in particular,  the rotations-reducibility for the two dimensional center of the infinite-range dual of Type I. 
\begin{Theorem}\label{long-range-kotanimain}
Let $\alpha\in\mathbb R\backslash\mathbb Q$ and $v$ be real
analytic. If $\omega(E)=1$ on an interval $I$. There exist  $O^j_E\in C^\omega(\T,\C^{\Z})$ for $j=1,2$, $\psi_E,\phi_E\in C^\omega(\T,\R)$, depending analytically on $E\in I'\subset I,$ such that
$$
L_{v,\alpha,\theta}O_E^j(E)=EO_E^j(\theta),
$$
$$
T(O^1_E(\theta),O^2_E(\theta))=(O^1_E(\theta+\alpha),O^2_E(\theta+\alpha))e^{2\pi i\phi_E(\theta)}R_{\psi_E(\theta)}.
$$
\end{Theorem}

\subsection{Structure of the rest of the paper.} 
In Section \ref{start} we discuss the ideas and strategy of the proof in more
detail. Section \ref{pre} contains the preliminaries, and in
Section \ref{tio}, we present basic properties and some typical
examples of type I operators.

In Section \ref{puig1} we involve the symplectic orthogonality
property of different eigenfunctions to  prove the simplicity of point
spectrum. In particular, the proof of Theorem \ref{tsimp} is given in Section 5.2, the proof of Theorem \ref{gjy-puig} is given at the beginning of Section \ref{puig1} (see Theorem \ref{contra1}). Section \ref{puig2}  is devoted to our quantitative
and all-frequency version of Puig's argument, proving Theorem \ref{rot1} (see Theorem \ref{contra2}).

In Section \ref{kotanis1}, we establish Kotani theory and thus
$L^2$-reducibility for one-frequency analytic operators whose cocycles are
partially hyperbolic with two-dimensional center, proving Theorem \ref{L2 reducibility} (see Theorem \ref{L2 reducibility1}). As a consequence, we
prove that absence of Cantor spectrum implies improved
$C^\omega$-rotations reducibility for the associated quasiperiodic
finite-range cocycle in an interval, see Theorem \ref{C0reducibility1}, proving Corollary \ref{C0reducibility}.

Sections \ref{puig2} and  \ref{kotanis1}, especially, Section \ref{kotanis1}
  represent the main hard analysis 
and contain the key contributions of this paper beyond the main
result. Both are of
independent interest.

Section \ref{limit} develops the parameter-dependent limiting center construction
and the uniform stable/unstable estimates needed in the analytic
setting. All required specialized symplectic, factorization, signature, and
rotation facts used in Sections~7--10 are proved in
Appendix~\ref{app:center-facts}; adapting their general formulation 
in \cite{gj}.

This allows to bootstrap Kotani theory for the infinite-range dual of Type I operators, and all-frequency Puig's argument  to the entire Type I class, thus extending all previous results from trigonometric polynomials to general real analytic functions.

Sections \ref{slongkotani} and \ref{slongpuig}  are devoted to the proofs of Theorem
\ref{main11} and Theorem \ref{main12}, based on a combination of
generalized Kotani theory (developed in Section \ref{kotanis1}), and
all-frequency Puig's argument (developed in Sections \ref{puig1} and \ref{puig2}) and its limiting version (developed in Section \ref{limit}, \ref{slongkotani} and \ref{slongpuig}).

Finally, the proof of Theorem \ref{rot2analytic} can be found in Section 11.2 (see Theorem \ref{subkey}) and the proof of Theorem \ref{long-range-kotanimain} can be found in Section 10.1 (see Corollary \ref{long-range-kotani1}).

\section{The strategy}\label{start}
The self-duality of the almost Mathieu family plays a key role in
the ten martini proof of \cite{aj} in several aspects, the most trivial  of which is that the ($x$-independent)
spectra of $H_{\lambda,\alpha,x}$ and $H_{\frac 1{\lambda},\alpha,x},$
given by \eqref{amo}, coincide up to scaling by $\lambda.$
Therefore, it is sufficient to work in the (sub)critical regime
$|\lambda|\leq 1.$

For general type I operators \eqref{sch}, self-duality is, of course, lost, so
we have to develop different arguments for sub and super critical
regions. However, Aubry duality (see Section \ref{aubry}) remains
a crucial tool. 

We first explain the case in which $v$ is a trigonometric polynomial
of degree $d$. Then the dual operator \eqref{fi} is finite-range and its
eigenvalue equation defines a $2d$-dimensional complex symplectic
cocycle $(\alpha,L_{E,v})$. We denote its nonnegative Lyapunov
exponents by
 $\gamma_d(E)\geq \cdots\geq \gamma_1(E)\geq 0$.

The quantitative global theory developed in \cite{gjyz} identifies
the $T$-acceleration with the multiplicity of the smallest dual
Lyapunov exponent. In particular,
\[
\overline\omega(E)=1
\quad\Longleftrightarrow\quad
\gamma_1(E)\text{ is simple}.
\]
This will be repeated as Proposition \ref{simple le}
which will be given a more detailed but still a one-line proof.
Moreover, Theorem 4.1 gives
$\omega^{d-1}(\alpha,L_{E,v})=0$
at every Type I energy, and shows that $(\alpha,L_{E,v})$ is
$(d-1)$- and $(d+1)$-dominated. Thus its dual dynamics are $PH_2$,
with a two-dimensional center. The vanishing of
$\omega^{d-1}$ implies that the determinant of the induced center
cocycle has zero winding and hence that the center dynamics are
projectively real; see Appendix~\ref{app:center-facts}.
Those  are  precisely the features that allow to extend many of the
techniques developed for Schr\"odinger operators.

 More specifically, as will be shown in Section \ref{secdual} the results of
\cite{gjyz} immediately imply that for operators \eqref{sch} of type I
we have
\begin{center}{\large
\begin{tabular}{c|c|c}
\hline\hline
Regime& {$H_{v,\alpha,x}$}& $L_{v,\alpha,\theta}$\\
\hline
subcritical&$L(E)=\omega(E)=0$& $L(E)=0$ and $\gamma_1(E)>0$ is simple  \\
\hline
critical &$L(E)=0,\omega(E)=1$&$L(E)=0$ and $\gamma_1(E)=0$ is simple \\
\hline
supercritical &$L(E)>0,\omega(E)=1$ & $L(E)>0$ and $\gamma_1(E)=0$ is simple\\
\hline
\end{tabular}}
\end{center}

Thus the Type I condition has two consequences that will be used
throughout the proof: simplicity of the lowest dual Lyapunov exponent,
which gives the $PH_2$ splitting, and vanishing of the exterior
acceleration $\omega^{d-1}$, which is needed in the Kotani reduction
of the two-dimensional center.

The original proof of the almost Mathieu ten martini \cite{aj} combines different mechanisms on
the Diophantine and Liouville sides. As mentioned in
\cite{aj}, the two earlier breakthroughs, \cite{cey} and
\cite{puig}, already led to the proof of Cantor spectrum for
\eqref{amo} for an explicit set of a.e. $\alpha$, covering
correspondingly the Liouville and Diophantine regimes. The Diophantine
approach of
Puig was based on localization for completely resonant phases
\cite{j,jks}, duality, and Moser-P\"oshel argument. It was conjectured however in \cite{aj,solving} (and very recently proved \cite{Liuresonant1}) that  localization
for completely resonant phases does not hold, and thus this approach
cannot work at all, for the most difficult arithmetic
mid-range of parameters\footnote{see also \cite{liuresonant} for another recent
development}, so an a lot more elaborate approach was developed in \cite{aj}, still requiring dual localization and  a modification of the Puig's argument. However, both
the required localization result and Puig's impossibility of reducibilty to the identity
were based on the almost Mathieu specifics. On the Liouville side, the approach of \cite{aj} was based on a technical bootstrap of the special rational-gap
estimate of Choi-Eliott-Yui \cite{cey}, and the latter has not yet been extended even in a weak way to {\it any} other model.



The fact that the Diophantine and Liouville approaches did meet in the middle has been
viewed as a miracle by the authors of \cite{aj}, with no
rational explanation \footnote{Pun accidental}. At the time, it has been unclear whether the
arithmetic dependence of the proof of \cite{aj} is something
intrinsically required. 

An important development has been the non-critical dry ten martini proof for the
almost Mathieu operator in
\cite{dryAYZ}. While the proof of
\cite{dryAYZ}, requiring delicate estimates, focuses only on the arithmetic
range not previously covered by \cite{aj1} and therefore does not fully
bypass the algebraic argument of \cite{cey},  the key
  idea works for all frequencies, as the authors found a way to run a Moser-P\"oschel
type argument based on quantitative {\it almost} reducibility to identity, rather
than reducibility to the identity for which there are Diophantine
obstructions. 

While inspired by this idea of \cite{dryAYZ}, we instead replace an argument through
 reducibility to the identity by the one through {\it reducibility to rotations with zero rotation number}. The latter holds for all irrational
$\alpha$ (\cite{afk,hy}).  We do it not with Moser-P\"oschel but
with Puig's duality approach itself (Theorem \ref{rot1}), replacing
the localized eigenfunctions in his argument by the  {\it almost
  localized} ones. Most importantly, our argument works for
  higher-dimensional symplectic cocycles, thus allows to use duality while going beyond
  the almost Mathieu family. 
Our method  works for all irrational
frequencies, and for both sub and (super)critical regimes.

Overall, our proof does not attempt to extend the localization or rational-gap
parts of the almost Mathieu argument. Instead, it develops robust high-dimensional versions of  two
ingredients: Kotani
theory and Puig's duality argument. We now discuss the overall strategy once those are obtained.


{\bf Trigonometric polynomial potentials: zero Lyapunov
exponent.}

Suppose that an interval is contained in the Type I part of the
spectrum and that the  Lyapunov exponent vanishes there.
Classical Kotani theory gives analytic reducibility of the
Schr\"odinger cocycle to rotations on a smaller interval. On the
other hand, the dual finite-range cocycle is $PH_2$. The simplicity theorem of
Section~5 replaces the usual Wronskian argument: every square-summable
solution of a $PH_2$ operator lies in its two-dimensional center, and
symplectic orthogonality then forces the point spectrum to be simple.
In Section~6 we combine this  with almost localized dual
states to obtain a version of Puig's argument valid for every
irrational frequency, thus ruling out the above reduction  and
yielding a contradiction. This treats the subcritical and critical
regimes.

{\bf Trigonometric polynomial potentials: positive Lyapunov
exponent.}

In the supercritical regime the original Schr\"odinger cocycle cannot
be reducible to rotations. We therefore use Aubry duality and work
with the finite-range dual operator. Quantitative global theory gives
a simple zero center exponent, while
Theorem 4.1 gives both the $PH_2$ splitting and $\omega^{d-1}(\alpha,L_{E,v})=0.$ Our generalized Kotani theory of Section~7 then gives analytic
reducibility of the two-dimensional center to a scalar phase times
rotation on any hypothetical spectral interval. The higher-dimensional
all-frequency Puig argument rules this out and again gives a
contradiction.

{\bf General analytic potentials.}

For a general analytic potential, the dual operator is infinite-range and
has no finite-dimensional transfer cocycle. We approximate
$v$ by its trigonometric polynomial truncations. Their dual cocycles are
$PH_2$, and their two-dimensional symplectic centers converge to a
limiting center dynamics. The general construction of  the limiting intrinsic symplectic structure 
 is developed in \cite{gj}; its holomorphic parameter-dependent
version needed for the Kotani extension is presented in 
Section~9. The  specialized auxiliary facts from \cite{gj} needed here are proved in Appendix~\ref{app:center-facts}.

Convergence of the center structures alone is not sufficient for the infinite-range
Puig argument. The exponentially decaying solutions obtained through
Aubry duality induce approximate solutions for the finite-range
truncations. We prove a key fact that their stable and unstable components are
exponentially small, uniformly in the truncation dimension, while their
center projections remain quantitatively nondegenerate and form a basis
of the two-dimensional center. This provides the infinite-range substitute
for the finite-range fact that square-summable solutions lie in the
center, and allows the symplectic simplicity argument to survive the
truncation limit.

Other than the facts that are considered classics, the principal external input is the multiplicative Jensen formula and
quantitative global theory developed in \cite{gjyz}. The crucial for us
fact that dual cocycles of type I operators are $PH2$ is essentially
 contained in \cite{gjyz}, but we also give a proof of this statement here, for completeness. Apart from this
input, the proof is self-contained. The broader intrinsic symplectic
and rotation theory motivating the limiting center construction is
developed in \cite{gj}; however all specialized two-dimensional facts from
that work needed here are proved in Appendix~\ref{app:center-facts}, for the reader's convenience.
 

\section{Preliminaries}\label{pre}

Let  $F$ be a bounded analytic (possibly matrix valued) function defined on $ \{ \theta |  | \Im \theta |< h \}$,
$|F| _h= \sup_{ | \Im \theta |< h } \| F(\theta)\| $. $C^\omega_{h}(\T,*)$ denotes the
set of all these $*$-valued functions ($*$ will usually denote $\R$, $SL(2,\R)$, $\C$, $gl(2,\C)$, $Sp(2.\C)$). Denote by $C^{\omega}(\T,*)$  the union $\cup_{h>0}C_h^{\omega}(\T,*)$. 

Similarly,   let $G$ be a bounded analytic (possibly matrix valued) function defined on $ \{z:\Re z\in I, |\Im z|\leq r\}\times \{ \theta :  | \Im \theta |< h \}$ where $I$ is an interval on the real line. We define
$$
|G| _{r,h}= \sup_{\Re z\in I,|\Im z|\leq r, | \Im \theta |< h } \| G(z,\theta)\|.
$$ 
$C^\omega_{r,h}(I\times \T,*)$ denotes the
set of all these $*$-valued functions ($*$ will usually denote $\R$, $SL(2,\R)$, $\C$, $gl(2,\C)$, $Sp(2.\C)$). Denote by $C^{\omega}(I\times \T,*)$  the union $\cup_{r>0,h>0}C_h^{\omega}(I\times \T,*)$.

\subsection{Continued fraction expansion} Let $\alpha\in (0,1)\backslash\Q$, $a_0:=0$ and $\alpha_0:=\alpha$. Inductively, for $k\geq 1$, we define
$$
a_k:=[\alpha_{k-1}^{-1}], \  \ \alpha_k=\alpha_{k-1}^{-1}-a_k.
$$
Let $p_0:=0$, $p_1:=1$, $q_0:=1$, $q_1:=a_1$. Again inductively, set
$p_k:=a_kp_{k-1}+p_{k-2}$, $q_k:=a_kq_{k-1}+q_{k-2}$. Then $q_n$ are
the  denominators of the best rational approximamts of $\alpha,$ since
we have $\|k\alpha\|_{\R/\Z}\geq \|q_{n-1}\alpha\|_{\R/\Z}$ for all
$k$ satisfying $\forall 1\leq k< q_n$. We also have 
$$
\frac{1}{2q_{n+1}}\leq \|q_n\alpha\|_{\R/\Z}\leq \frac{1}{q_{n+1}}.
$$

\subsection{Cocycles and Lyapunov exponents}\label{3.2}
Let ${\rm M}(m,\C)$ be the set of all $m\times m$ matrices, $T:\Omega\rightarrow \Omega$ be a minimal homeomorphisim and $(\Omega,T,\mu)$ be ergodic. Given $A \in C^0(\Omega,{\rm M}(m,\C))$, we define the complex minimal cocycle $(T,A)$ by:
$$
(T,A)\colon \left\{
\begin{array}{rcl}
\Omega \times \C^{m} &\to& \Omega \times \C^{m}\\[1mm]
(\omega,v) &\mapsto& (T\omega,A(\omega)\cdot v)
\end{array}
\right.  .
$$
The iterates of $(T,A)$ are of the form $(T,A)^n=(T^n,A_n)$, where
$$
A_n(\omega):=
\left\{\begin{array}{l l}
A(T^{n-1}\omega) \cdots A(T\omega) A(\omega),  & n\geq 0\\[1mm]
A^{-1}(T^n\omega) A^{-1}(T^{n+1}\omega) \cdots A^{-1}(T^{-1}\omega), & n <0
\end{array}\right.    .
$$
We denote by $L_1(T, A)\geq L_2(T,A)\geq...\geq L_m(T,A)$ the Lyapunov exponents of $(T,A)$ repeatedly according to their multiplicities, i.e.,
$$
L_k(T,A)=\lim\limits_{n\rightarrow\infty}\frac{1}{n}\int_{\Omega}\ln\sigma_k(A_n(\omega))d\mu,
$$
where  $\sigma_1(A_n)\geq...\geq \sigma_m(A_n)$ denote the singular values (eigenvalues of $\sqrt{A_n^*A_n}$). Since the k-th exterior product $\Lambda^kB$ of any $B\in M(m,\C)$ satisfies $\sigma_1(\Lambda^kB)=\|\Lambda^kB\|$, $L^k(T, A)=\sum\limits_{j=1}^kL_j(T,A)$ satisfies

$$
L^k(T,A)=\lim\limits_{n\rightarrow \infty}\frac{1}{n}\int_{\Omega}\ln\|\Lambda^kA_n(\omega)\|d\mu.
$$
\begin{Remark}
{\rm For $A\in C^0(\Omega,Sp(2d,\C)),$ where $Sp(2d,\C)$ is the set of $2d\times 2d$ complex symplectic matrices,  the Lyapunov exponents of $(T,A)$ come in pairs $\{\pm L_i(T,A)\}_{i=1}^d$.}
\end{Remark}

An important for us example is the minimal finite-range cocycle $(T,L_E^f)$ with
\begin{align}\small\label{1111}
L_{E}^{f}(\omega)=\frac{1}{a_d}
\begin{pmatrix}
-a_{d-1}&\cdots&-a_1&E-f(\omega)-a_0&-a_{-1}&\cdots&-a_{-d+1}&-a_{-d}\\
a_d& \\
& &  \\
& & & \\
\\
\\
& & &\ddots&\\
\\
\\
& & & & \\
& & & & & \\
& & & & & &a_{d}&
\end{pmatrix}.
\end{align}
Let
\begin{equation}\label{sdef1}
C=\begin{pmatrix}
a_d&\cdots&a_1\\
0&\ddots&\vdots\\
0&0&a_d
\end{pmatrix},\ \ S=\begin{pmatrix}0&-C^*\\
C&0\end{pmatrix}
\end{equation}
One can check that $L^f_E(\omega)$ is complex symplectic with respect
to $S,$ that is 
$$
(L^f_E(\omega))^*S L^f_E(\omega)=S,
$$ 
when $E\in\R$. We will denote its non-negative Lyapunov exponents by $L^f_i(E)=L_i(T,L_{E}^{f})$ for $1\leq i\leq d$ for short.

\subsection{Uniform hyperbolicity and dominated splitting}\label{3.3}
For  $A\in C^0(\Omega,Sp(2d,\C))$,  we say the  cocycle $(T, A)$ is {\it uniformly hyperbolic} if for every $\omega \in \Omega$, there exists a continuous splitting $\C^{2d}=E^s(\omega)\oplus E^u(\omega)$ such that for some constants $C>0,c>0$, and for every $n\geqslant 0$,
$$
\begin{aligned}
\lvert A_n(\omega)v\rvert \leqslant Ce^{-cn}\lvert v\rvert, \quad & v\in E^s(\omega),\\
\lvert A_n(\omega)^{-1}v\rvert \leqslant Ce^{-cn}\lvert v\rvert,  \quad & v\in E^u(T^n\omega).
\end{aligned}
$$
This splitting is invariant by the dynamics, which means that for every $\omega \in \Omega$,
$$
A(\omega)E^{\ast}(\omega)=E^{\ast}(T\omega),
$$
for $\ast=s,u$.   The set of uniformly hyperbolic cocycles is open in the $C^0$-topology.

For complex minimal cocycle $(T,A)\in C^0(\Omega,{\rm M}(m,\C))$, a
  related property is called {\it dominated splitting}. Recall that  Oseledets theorem provides us with  strictly decreasing
sequence of Lyapunov exponents $L_j(T,A) \in [-\infty,\infty)$ of multiplicity $m_j\in\N$, $1\leq j \leq \ell$ with $\sum_{j}m_j=m$, and for $\mu$ a.e. $\omega$, there exists
a measurable invariant decomposition $$\C^m=E_\omega^1\oplus E_\omega^2\oplus\cdots\oplus E_\omega^\ell$$ with $\dim E_\omega^j=m_j$ for $1\leq j\leq \ell$ such that $$
\lim\limits_{n\rightarrow\infty}\frac{1}{n}\ln\|A_n(\omega)v\|=L_j(T,A),\ \  \forall v\in E_\omega^j\backslash\{0\}.
$$
An invariant decomposition $\C^m=E_\omega^1\oplus
E_\omega^2\oplus\cdots\oplus E_\omega^\ell$  is  {\it dominated} if there exists $n$ such that  for any unit vector $v_j\in E_\omega^j\backslash \{0\}$, we have $$\|A_n(\omega)v_j\|>\|A_n(\omega)v_{j+1}\|.$$
Recall that Oseledets decomposition is a priori only measurable, however if an invariant decomposition $\C^m=E_\omega^1\oplus E_\omega^2\oplus\cdots\oplus E_\omega^\ell$  is  {\it dominated}, then $E_\omega^j$ depends  continuously on $\omega$ \cite{bdv}.

We also recall that  $(T,A)$     is called $k$-dominated (for some $1\leq k\leq m-1$) if there exists a dominated decomposition $\C^m=E^+_\omega \oplus E_\omega^- $    with $\dim E^+_\omega = k.$  It follows from the definitions that the Oseledets splitting is dominated if and only if $(T,A)$ is $k$-dominated for each $k$ such that  $L_k(T,A)> L_{k+1}(T,A)$.

\subsection{Global theory of one-frequency quasiperiodic cocycles} \label{regge}For
$\Omega=\T:=\R/\Z,$  $T:x\rightarrow x+\alpha,$ where
$\alpha\in\R\backslash\Q$, and $A:\T\to {\rm M}(m,\C),$ we call $(\alpha,A)$  a one-frequency quasiperiodic cocycle.
Global  theory of analytic one-frequency quasiperiodic cocycles was first developed for $SL(2,\C)$-cocycles \cite{avila0}, and
later  generalized to any ${\rm M}(m,\C)$-cocycles \cite{ajs}.  The key  concept for the global theory is the acceleration.  If $A\in C^{\omega}(\T,{\rm M}(m,\C))$ admits a holomorphic extension to $|\Im z|<\delta$, then for $|\e|<\delta$ we can define $A_\e\in C^{\omega}(\T,M(m,\C))$ by $A_\e(x)=A(x+i\e)$.
The accelerations of $(\alpha,A)$  are defined as
$$
\omega^k(\alpha,A)=\lim\limits_{\e\rightarrow 0^+}\frac{1}{2\pi\e}(L^k(\alpha,A_\e)-L^k(\alpha,A)), \qquad  \omega_k(\alpha,A)= \omega^k(\alpha,A)-\omega^{k-1}(\alpha,A).
$$
The key ingredient of the global theory is that the acceleration  is quantized.

\begin{Theorem}[\cite{avila0,ajs}]\label{ace}
There exists $1\leq l\leq m$, $l\in\N$, such that $l\omega^k$ and $l \omega_k$ are integers. In particular, if $A\in C^\omega(\T, SL(2,\C))$, then $\omega^1(\alpha,A)$ is an integer.
\end{Theorem}
\begin{Remark}\label{rem1}
{\rm If $L_j(\alpha,A)>L_{j+1}(\alpha,A)$, then $\omega^j(\alpha,A)$
  is an integer, as follows from the proof of Theorem 1.4 in \cite{ajs}, see also footnote 17 in \cite{ajs}.}
\end{Remark}

By subharmonicity, we know $L^k(\alpha,A(\cdot+i\e))$ is a convex function of $\e $  in a neighborhood of $0$,  unless it is identically equal to $-\infty$.
We say that $(\alpha,A)$ is {\it $k$-regular} if $\e\rightarrow L^k(\alpha,A(\cdot+i\e))$ is an affine function of $\e$ in a neighborhood of $0$.  In general, one can relate regularity and dominated splitting as follows.

\begin{Theorem}[\cite{avila0,ajs}]\label{t2.1}
Let $\alpha\in\R\backslash\Q$ and $A\in C^\omega(\T,M(m,\C))$. If $1\leq j\leq m-1$ is such that $L_j(\alpha,A)>L_{j+1}(\alpha,A)$, then $(\alpha,A)$ is $j$-regular if and only if $(\alpha,A)$ is $j$-dominated. In particular, if  $A\in C^\omega(\T, SL(2,\C))$ with $L(\alpha,A)>0$, then $(\alpha,A)$ is $1$-regular (or regular) if and only if $(\alpha,A)$ is uniformly hyperbolic.
\end{Theorem}

\subsection{One-frequency quasiperiodic $SL(2,\R)$-cocycles: rotation
  number and IDS}\label{secrot}
For  one-frequency quasiperiodic $SL(2,\R)$-cocycle $(\alpha,A)$ with $A \in C^0(\T, {\rm SL}(2, \R))$,
assume that $A$ is homotopic to the identity. Then $(\alpha, A)$ induces the projective skew-product $F_A\colon \T \times \mathbb{S}^1 \to \T \times \mathbb{S}^1$
$$
F_A(x,w):=\left(x+\a,\, \frac{A(x) \cdot w}{|A(x) \cdot w|}\right),
$$
which is also homotopic to the identity. Lift $F_A$ to a map $\widetilde{F}_A\colon \T \times \R \to \T \times \R$ of the form $\widetilde{F}_A(x,y)=(x+\alpha,y+\psi_x(y))$, where for every $x \in \T$, $\psi_x$ is $\Z$-periodic.
Map $\psi\colon\T \times \R  \to \R$ is called a {\it lift} of $A$. Let $\mu$ be any probability measure on $\T \times \mathbb{S}^1$ which is invariant by $F_A$, and whose projection on the first coordinate is given by Lebesgue measure.
The number
$$
\rho(\alpha,A):=\int_{\T \times \mathbb{S}^1} \psi_x(y)\ d\mu(x,y) \ {\rm mod} \ \Z
$$
 depends  neither on the lift $\psi$ nor on the measure $\mu$, and is called the \textit{fibered rotation number} of $(\alpha,A)$ (see \cite{H,johonson and moser} for more details).

Given $\theta\in\T$, let $
R_\theta:=
\begin{pmatrix}
\cos2 \pi\theta & -\sin2\pi\theta\\
\sin2\pi\theta & \cos2\pi\theta
\end{pmatrix}$.
If $A\colon \T\to{\rm PSL}(2,\R)$ is homotopic to $x \mapsto R_{nx/2}$ for some $n\in\Z$,
then we call $n$ the {\it degree} of $A$ and denote it by $\deg A$.
The fibered rotation number is invariant under real conjugacies which are homotopic to the identity. More generally, if $(\alpha,A_1)$ is conjugated to $(\alpha, A_2)$, i.e., $B(x+\alpha)^{-1}A_1(x)B(x)=A_2(x)$, for some $B \colon \T\to{\rm PSL}(2,\R)$ with $\deg{B}=n$, then
\begin{equation}\label{rotation number}
\rho(\alpha, A_1)= \rho(\alpha, A_2)+ \frac{n\alpha}{2}.
\end{equation}

In particular, for quasiperiodic Schr\"odinger cocycle
$(\alpha,S_E^v)$ where $S_E^v$ is given by \eqref{S},
we denote the rotation number $\rho(E):=\rho(\alpha,S_E^v)$.

The  {\it integrated density of states} (IDS)
$N_{v,\alpha}:\R\rightarrow [0,1]$ of $H_{v,\alpha,x}$ is defined by
$$
N_{v,\alpha}(E):=\int_{\T}\mu_{v,\alpha,x}(-\infty,E]dx,
$$
where $\mu_{v,\alpha,x}$ is the spectral measure of $H_{v,\alpha,x}$
and vector $\delta_0.$

It is well known that $\rho(E)\in[0,\frac{1}{2}]$ is related to the integrated density of states $N=N_{v,\alpha}$ as follows:
\begin{equation}\label{relation}
N(E)=1-2\rho(E).
\end{equation}

\subsection{Aubry duality}\label{aubry}

Consider the fiber direct integral,
$$
\mathcal{H}:=\int_{\T}^{\bigoplus}\ell^2(\Z)dx,
$$
which, as usual, is defined as the space of $\ell^2(\Z)$-valued, $L^2$-functions over the measure space $(\T,dx)$.  The extensions of the
Sch\"odinger operators  and their long-range duals to  $\mathcal{H}$ are given in terms of their direct integrals, which we now define.
Let $\alpha\in\T$ be fixed. Interpreting $H_{v,\alpha,x}$ as fibers of the decomposable operator,
$$
H_{v,\alpha}:=\int_{\T}^{\bigoplus}H_{v,\alpha,x}dx,
$$
the family $\{H_{v,\alpha,x}\}_{x\in\T}$ naturally induces an operator on the space $\mathcal{H}$, i.e.,
$$
(H_{v,\alpha} \Psi)(x,n)= \Psi(x,n+1)+ \Psi(x,n-1) +  v(x+n\alpha) \Psi(x,n).
$$

Similarly,  the direct integral of long-range operators
$L_{v,\alpha,\theta}$ given by \eqref{fi},
denoted as $L_{v,\alpha}$, is given by
$$
(L_{v,\alpha}  \Psi)(\theta,n)=  \sum\limits_{k\in\Z} \hat{v}_k \Psi(\theta,n+k)+  2\cos2\pi (\theta+n\alpha) \Psi(\theta,n),
$$
where $\hat{v}_k$ is the $k$-th Fourier coefficient of $v(x)$.

Let  $U$ be the following operator on $\mathcal{H}:$
\begin{equation}\label{dual map}
(U\phi)(\eta,m)=\sum_{n\in\Z}\int_{\T}e^{2\pi imx}e^{2\pi in(m\alpha+\eta)}\phi(x,n)dx.
\end{equation}
Then direct computations show that $U$ is unitary and satisfies
$$U H_{v,\alpha} U^{-1}=L_{v,\alpha}.$$ $U$ represents the so-called
{\it Aubry duality} transformation.
The quasiperiodic  long-range operator  $L_{v,\alpha,\theta}$ is called the dual operator of $H_{v,\alpha,x}$ \cite{gjls}.
More generally, let
$$
v(x)=\sum_{k\in\Z}\hat{v}_ke^{2\pi ikx},\ \ w(x)=\sum_{k\in\Z} \hat{w}_k e^{2\pi i kx}
$$
be two 1-periodic real valued functions.  We define a quasiperiodic long-range operator on $\ell^2(\Z)$ by
\begin{equation}\label{lvw}
(L^w_{v,\alpha,x}u)_n=\sum\limits_{k\in\Z} \hat{v}_ku_{n+k}+w(x+n\alpha)u_n, \ \ n\in\Z,
\end{equation}
so that we have $L_{v,\alpha,x}=L^{2\cos}_{v,\alpha,x}.$

As above, we can interpret $L^w_{v,\alpha,x}$ as fibers of the decomposable operator,
$$
L^w_{v,\alpha}:=\int_{\T}^{\bigoplus}L^w_{v,\alpha,x}dx,
$$
thus $L^w_{v,\alpha}$ acts on the space $\mathcal{H}$, by
$$
(L^w_{v,\alpha} \Psi)(x,n)= \sum\limits_{k\in\Z} \hat{v}_k\Psi(x,n+k) +  w(x+n\alpha) \Psi(x,n).
$$

With the Fourier convention in \eqref{dual map}, direct computation gives \cite{haro}

\[
U L^w_{v,\alpha} U^{-1}=L^{v(-\cdot)}_{w,\alpha}.
\]

\subsection{Quantitative global theory}

\begin{Theorem}[Multiplicative Jensen formula \cite{gjyz}]
\label{thm:multiplicative-Jensen}
Let $\alpha\in\mathbb R\backslash\mathbb Q$ and
$v\in C_h^\omega(\mathbb T,\mathbb R)$. There exist nonnegative
dual Lyapunov exponents
\[
0\leq\gamma_1(E)\leq\cdots\leq\gamma_m(E)
\]
obtained as limits of the Lyapunov exponents of the dual finite-range
truncations, such that
\[
L_\varepsilon(E)
=
L_0(E)
-
\sum_{\gamma_j(E)<2\pi|\varepsilon|}\gamma_j(E)
+
2\pi
\#\{j:\gamma_j(E)<2\pi|\varepsilon|\}
|\varepsilon|
\]
for $|\varepsilon|<h$.
\end{Theorem}

Consequently, the $T$-acceleration is the multiplicity of the
smallest dual Lyapunov exponent. In particular,
$
\overline\omega(E)=1
$ implies  $\gamma_1(E)$  is simple.

\section{Type I operators}\label{tio}
 In this section we introduce Type I cocycles and operators and prove
their local stability under analytic perturbations. In particular, for
each fixed irrational frequency, Type I operators form an open set in
admissible analytic topology, in the sense
made precise in Corollary~\ref{sgao}. We also discuss several natural
examples and establish a {\it frequency-independent} neighborhood result for
two popular families, including the almost Mathieu. Finally, for trigonometric polynomial
potentials, quantitative global theory \cite{gjyz} identifies the Type
I condition with simplicity of the smallest nonnegative dual Lyapunov
exponent. We further show that the dual cocycle is $PH_2$ and that
$
\omega^{d-1}(\alpha,L_{E,v})=0,
$
two facts that play a central role in the subsequent arguments.
\subsection{Type I cocycles}
Given $A \in C^\omega(\T,{\rm SL}(2,\C))$, one can extend it to the
band $\left\{z:|\Im z|<h\right\}.$ For $\alpha\in\R\backslash\Q$ and
$|\e|<h,$ we can define the Lyapunov exponent of complexified cocycle
$(\alpha,A_\e)$ just like we did for Schr\"odinger cocycles:
$$
L_\e(\alpha, A):=L(\alpha,A_\e)=\lim\limits_{n\rightarrow\infty}\frac{1}{n}\int_{\T}\ln\|A(x+i\e+(n-1)\alpha)\cdots A(x+i\e+\alpha)A(x+i\e)\|dx.
$$

T-acceleration can also be defined for general analytic cocycles in the
same way as for Schr\"odinger cocycles, that is
$$
\bar{\omega}(\alpha,A):=\lim\limits_{\e\rightarrow \e_1^+}\frac{L_\e(\alpha,A)-L_{\e_1}(\alpha,A)}{2\pi(\e-\e_1)}
$$
Here $\e_1$ is the first turning point encountered when one follows
the graph of $L_\e(\alpha,A)$ from its minimizing set toward increasing
$\e$; thus $\bar{\omega}(\alpha,A)$ is the normalized slope immediately
after that corner. If the infimum is approached only at the boundary
of the natural band of analyticity, or if no such corner occurs before
that boundary---for example, if the profile remains affine or constant
throughout the band---we set
$\bar{\omega}(\alpha,A)$ to be the normalized slope of that branch.

The natural band of analyticity is meant here. A smaller strip may fail to reveal the turning point; in that
case  that
strip is too narrow for the stability statement below. For {\it real
cocycles}, evenness implies that $0$ belongs to the minimizing set, so $\e_1$ is the first turning point at or after $0.$

When $\alpha$ is fixed through the argument, we will often write, for convenience,
$
L_\e(A):=L_\e(\alpha, A)$ and $\bar{\omega}(A):=\bar{\omega}(\alpha,A).$

An important fact is that  cocycles with T-acceleration one are stable under suitable  analytic perturbations.
\begin{Lemma}\label{con}
For $(\alpha,A)\in \R\backslash\Q\times C^\omega_h(\T,SL(2,\R))$ with $\bar{\omega}(\alpha,A)=1$. Then, for every
$h>\e_1(\alpha,A)$ such that $A\in C_h^\omega(\T,SL(2,\R))$,
the condition $\bar{\omega}=1$ is stable under sufficiently small
perturbations in
$(\R\backslash\Q)\times C_h^\omega(\T,SL(2,\R))$.
\end{Lemma}

\begin{Remark}
{\rm This result is not true if one replaces  T-acceleration one by acceleration one.}
\end{Remark}
\begin{pf} Choose $\e_1<a_1<a_2<a_3<a_4<h$ in an interval where the
normalized slope of $L_\e(\alpha,A)$ is one.
By continuity \cite{bj,ajs}, 
convexity and quantization \cite{avila0,ajs} we obtain
normalized slope one on $[a_2,a_3]$.
Since $A'$ is real, its Lyapunov profile is even and convex,
so $\bar{\omega}=1.$ \end{pf}



\subsection{Type I operators}
Recall that analytic one-frequency quasiperiodic  Schr\"odinger
operators are given by \eqref{sch}, and
the corresponding Schr\"odinger cocycles are ($\alpha,S_E^v$)  where
$S_E^v(x)$ is given by \eqref{S}.
The ($x$-independent \cite{as}) spectrum of $H_{v,\alpha,x}$ is
denoted by $\Sigma_{v,\alpha}$. Recall that operators \eqref{sch} are
called Type I if $\bar{\omega}(E)=\bar{\omega}(S_E^v)=1$ for all $E\in \Sigma_{v,\alpha}$.

The property of being type I is stable in suitable analytic topology,
more precisely
\begin{Corollary}\label{sgao}
Let $\alpha\in\R\backslash\Q$ and
$v\in C_h^\omega(\T,\R)$, and suppose that
$\{H_{v,\alpha,x}\}_{x\in\T}$ is Type I. Assume moreover that
$
\sup_{E\in\Sigma_{v,\alpha}}
\e_1(\alpha,S_E^v)<h.
$
Then there is
$\delta_0=\delta_0(\alpha,v,h)>0$ such that, whenever
$v'\in C_h^\omega(\T,\R)$ and
$\alpha'\in\R\backslash\Q$ satisfy
$
\max\left\{
|v'-v|_h,\,
|\alpha'-\alpha|
\right\}<\delta_0,
$
the family
$\{H_{v',\alpha',x}\}_{x\in\T}$ is also Type I.
\end{Corollary}

\begin{pf}
For every $E\in\Sigma_{v,\alpha}$, Lemma~\ref{con}, applied to
$(\alpha,S_E^v)$, gives an interval $I_E$ containing $E$ and a
neighborhood of $(\alpha,v)$ such that
$$
\bar{\omega}(\alpha',S_{E'}^{v'})=1
$$
for every $E'\in I_E$ and every $(\alpha',v')$ in that neighborhood.
The result follows by compactness argument.
\end{pf}


We note, that, by the proof of Lemma \ref{con} and compactness, every Type I pair $(\alpha,v)$ admits at least one
strip satisfying the hypothesis of Corollary~\ref{sgao} below the
natural boundary of analyticity of $v$. 

The class of type I operators also contains several prominent models.
\begin{Example}\label{exa1}{\rm The almost Mathieu operator and its
    analytic perturbations, i.e., $v=2\lambda(\cos2\pi(x)+\delta
    f(x))$ where $|\delta|<\delta(\alpha,\lambda, \|f\|_h)$ and
    \begin{enumerate} 
    \item $h>0$ if $|\lambda|\geq 1$;
    \item $h>-\frac{\ln|\lambda|}{2\pi}$ if $|\lambda|<1$.
    \end{enumerate} }
\end{Example}
\begin{Remark}\begin{itemize}\item
{\rm One should be careful that when $|\lambda|<1$ and $0<h<-\frac{\ln|\lambda|}{2\pi}$, $H_{\alpha,v,x}$ is not necessarily type I even if $\delta$ is sufficiently small.
\item As we will see, for large $h,$ $\delta=\delta(\lambda, \|f\|_h),$ that is does not depend on $\alpha.$}
\end{itemize}
\end{Remark}
\begin{Example}\label{exa2}{\rm The supercritical GPS model in \cite{gps} and its
    analytic perturbations, i.e.,
    $v(x)=\frac{2a \cos2\pi(\theta)}{1-b\cos2\pi(\theta)}+\delta f(x)$
    with $b\in(-1,1)$, with $|\delta|<\delta(\alpha,a,b, \|f\|_h)$, restricted to $E\in \{E:b E>2(1-|a|)\}$.
      }
  \begin{Remark}
  {\rm It is well known that the GPS model has a mobility edge. One can check that it may also have subcritical and critical type I energies. To preserve all type I energies, the perturbations should be chosen to depend locally on $E$.}
  \end{Remark}
\end{Example}
\begin{Example}\label{exa3}{\rm The supercritical generalized Harper's
    model of \cite{hk,se} and its analytic perturbations, i.e.,
    $v(x)=2a\cos2\pi(x)+2b\cos4\pi(x)+\delta f(x)$ with $b\in(-1,1)$
    restricted to the positive Lyapunov exponent regime where
    $|\delta|<\delta(\alpha,a,b, \|f\|_h)$
    is sufficiently small.
  }
\end{Example}

 \begin{Example}\label{exa4}
 {\rm The analytic cosine type quasiperiodic operator, where $v=\lambda f$ and f is a 1-periodic real analytic function satisfying the cosine type condition introduced in \cite{sin} and $\lambda >\lambda_0$ for some $\lambda_0(v)>0$.}
 \end{Example}

The proofs for Examples \ref{exa1}, \ref{exa2}, \ref{exa3} follow from the fact  that the almost Mathieu operator,
supercritical GPS model and supercritical generalized Harper's model are of  type I
\cite{avila0,wwxyz,gjyz}. Moreover, as we will prove, Examples \ref{exa1} and \ref{exa4} are uniform Type I operators.

The analytic perturbations of almost Mathieu operators, as in Example 1, were called PAMO in \cite{gjz}. 
We have the following
\begin{Lemma}\label{lemma-AMO}
Let $\lambda\neq0$, and suppose that for some $0<\eta<h$,
\[
2|\lambda|\sinh(2\pi\eta)>4+2|\lambda|.
\]
Then there exists $r=r(\lambda,\eta)>0$ such that, if
$v\in C_h^\omega(\mathbb T,\mathbb R)$ and
\[
|v-2\lambda\cos 2\pi x|_h<r,
\]
then $H_{v,\alpha,x}$ is Type I for every
$\alpha\in\mathbb R\backslash\mathbb Q$.
\end{Lemma}

\begin{pf}
Choose $r>0$ sufficiently small so that
\[
2|\lambda|\sinh(2\pi\eta)>4+2|\lambda|+2r.
\]
Write
\[
v(x)=2\lambda\cos 2\pi x+w(x),\qquad |w|_h<r.
\]
For any $\alpha\in\mathbb R\backslash\mathbb Q$ and
$E\in\Sigma_{v,\alpha}$,
\[
|E|\le 2+|v|_0\le 2+2|\lambda|+r.
\]
Since
\[
E-v(x+i\eta)
=
-\lambda e^{2\pi\eta}e^{-2\pi ix}
+
E-\lambda e^{-2\pi\eta}e^{2\pi ix}-w(x+i\eta),
\]
we obtain
\[
\inf_{x\in\mathbb T}|E-v(x+i\eta)|
\ge
2|\lambda|\sinh(2\pi\eta)-2-2|\lambda|-2r>2.
\]
Moreover, by Rouch\'e's theorem,
\[
\operatorname{wind}\bigl(E-v(\mathbb T+i\eta),0\bigr)=-1.
\]
The conefield criterion therefore gives uniform hyperbolicity of
$(\alpha,S_E^v(\cdot+i\eta))$, and the argument principle gives
$
\omega(\alpha,S_E^v(\cdot+i\eta))=1.
$
Since $\varepsilon\mapsto L_\varepsilon(E)$ is even and convex, with
nonnegative integer slopes on $\varepsilon>0$, its first positive slope
must also be one. Hence
$
\overline\omega(\alpha,S_E^v)=1.
$
This holds for every $E\in\Sigma_{v,\alpha}$ and every irrational
$\alpha$, proving the result.
\end{pf}
This immediately implies
\begin{Corollary}\label{deigen_pamo}
For $\lambda\neq 0$ and $f\in
C_h^\omega(\T,\R)$, there is $h_0(\lambda)<\infty$ and for $h>h_0,$ $\delta_0(\lambda, \|f\|_h)>0$ such that for $f\in C^\omega_h(\T,\R), |\delta|\leq \delta_0$, operator $\left\{H^\delta_{\lambda,\alpha,x}\right\}_{x\in\T}$  is of type I.
\end{Corollary}


To elaborate on The Example \ref{exa4},  the cosine type condition was introduced in \cite{sin}:
\begin{itemize}
\item $\frac{dv}{dx}=0$ at exactly two points, one is the minimum and the other is the maximum, which is denoted by $x_1$ and $x_2$.
\item These two extrema are non-degenerate, that is, $\frac{d^2v}{dx^2}(x_j)\neq0$ for $j=1,2.$
\end{itemize}
While for some results this condition requires only $f\in C^2(\T,\R),$ here we need stronger regularity
\begin{Proposition}
Let $\alpha\in\R\backslash\Q$ and $f\in C^\omega(\T,\R)$ be of cosine type. Then there is  $\lambda_0(f)$ such that if $\lambda>\lambda_0$, $H_{\lambda f,\alpha,x}$ is a Type I operator, in other words, $\lambda f$ is a uniform Type I potential.
\end{Proposition}
\begin{pf}

Since $f$ is of cosine type, there exists $\rho>0$ such that,
for every $t\in\R$, the function $f-t$ has at most two
zeros in $\Im z\leq 2\rho$, counting multiplicities. Indeed, away from the two critical points this follows from the implicit function theorem, while near each nondegenerate extremum it follows from the Weierstrass preparation theorem; compactness then gives a common strip.
It then follows from \cite[Corollary~1.1]{hm} that, for
$\lambda\geq\lambda_0(f)$,
$
0\leq \omega(\alpha,S_E^{\lambda f})\leq1
$
for every $E\in\R$ and every irrational $\alpha$.

Increasing $\lambda_0(f)$ if necessary, the Sorets--Spencer
estimate \cite{ss,bourgainjanal} gives
$
L(\alpha,S_E^{\lambda f})>0
$
for every $E\in\R$. If $E\in\Sigma_{\lambda f,\alpha}$,
then $\omega(\alpha,S_E^{\lambda f})\neq0$: otherwise the cocycle
would be regular, thus
uniformly hyperbolic, contradicting
$E\in\Sigma_{\lambda f,\alpha}$. Therefore
$
\omega(\alpha,S_E^{\lambda f})=1$
for every  $ E\in\Sigma_{\lambda f,\alpha}.
$
Since the acceleration is positive,
$\overline\omega(\alpha,S_E^{\lambda f})
=\omega(\alpha,S_E^{\lambda f})=1$.
Thus $H_{\lambda f,\alpha,x}$ is Type I.

\end{pf}

\subsection{The duality characterization}\label{secdual}
Throughout this subsection, let
$$
v(x)=\sum\limits_{k=-d}^d \hat{v}_k e^{2\pi ikx}
$$
be a trigonometric polynomial of degree $d$. We will involve the Aubry
duality to study  type I operators \eqref{sch} with   trigonometric
polynomial potentials $v.$ The dual operator $L_{v,\alpha,\theta}$ is then defined by  the \eqref{fi} (see Section \ref{aubry} for details),
We denote the associated cocycle of the eigenequation $L_{v,\alpha,\theta}u=Eu$ by $(\alpha,L_{E,v})$. Its Lyapunov exponents are denoted by $\pm\gamma_{1}(E), \cdots, \pm\gamma_{d}(E)$. We assume that
$$
0\leq \gamma_{1}(E)\leq \gamma_{2}(E)\leq \cdots\leq \gamma_{d}(E).
$$

An important basis of our proof is the following duality characterization of  type I operators

\begin{Proposition}\label{simple le}
For $\alpha\in \R\backslash\Q$ and $E\in\R$, $\bar{\omega}(E)=1$ if and only if $\gamma_1(E)$ is simple.
\end{Proposition}
\begin{pf}
By Theorem 1 in \cite{gjyz},
$
\bar{\omega}(E)$ is equal to the multiplicity of $\gamma_1(E).$
Thus $\bar{\omega}(E)=1$ if and only if $\gamma_1(E)$ is simple.
\end{pf}
More importantly, we have
\begin{Theorem}\label{dominate1}
Let $v$ be a real trigonometric polynomial of degree $d$, $\alpha\in \R\backslash\Q$, and let
 $E$ be a Type I energy. Then
\[
\omega^{d-1}(\alpha,L_{E,v})=0,
\]
and $(\alpha,L_{E,v})$ is both $(d-1)$- and
$(d+1)$-dominated.

\end{Theorem}
\begin{Remark}
The same result holds for complex energies $E$ without change of the proof.
\end{Remark}
\begin{pf}
We let
$$
C=\begin{pmatrix}
\hat{v}_d&\cdots&\hat{v}_1\\
0&\ddots&\vdots\\
0&0&\hat{v}_d
\end{pmatrix},\ \ B(\theta)=\begin{pmatrix}
2\cos2\pi(\theta_{d-1})+\hat{v}_0&\hat{v}_{-1}&\cdots&\hat{v}_{-d+1}\\
\hat{v}_1&\ddots&\ddots&\vdots\\
\vdots&\ddots&2\cos2\pi(\theta_1)+\hat{v}_0&\hat{v}_{-1}\\
\hat{v}_{d-1}&\cdots&\hat{v}_1&2\cos2\pi(\theta)+\hat{v}_0
\end{pmatrix}
$$
where $\theta_j=\theta+j\alpha$.  Then one can check that
\begin{equation}\label{strip2}
L_{E,v}(\theta+(d-1)\alpha)\cdots L_{E,v}(\theta)=:L_{d,E,v}(\theta)=\begin{pmatrix}
C^{-1}(EI-B(\theta))& -C^{-1}C^*\\
I_d&O_d
\end{pmatrix}
\end{equation}
where $I_d$ and $O_d$ are the $d$-dimensional identity and zero matrices, respectively.

Notice that \eqref{strip2} implies that we always have
$$
dL^{d-1}(\alpha,L_{E,v})=L^{d-1}(d\alpha,L_{d,E,v}).
$$
Thus by the definition of regularity, $(\alpha,L_{E,v})$ is
($d-1$)-regular if and only if  $(d\alpha,L_{d,E,v})$ is
($d-1$)-regular. Let $\left(\ell_{ij}\right)_{1\leq i,j\leq
  2d}:=(L_{d,E,v})_n(\theta)$. It is easy to check that each
$\ell_{ij}$ is a Laurent polynomial  with degree $\leq
n$. Similarly,  let $L_{ij}$ be the $ij$-th entry of
$\Lambda^{d-1}(L_{d,E,v})_n(\theta).$ By the definition of wedge
product, each $L_{ij}$ is a Laurent polynomial  of degree $\leq n(d-1)$. Hence
\begin{align*}
&|\omega^{d-1}(d\alpha,L_{d,E,v})|=\left|\lim\limits_{\e\rightarrow 0^+}\frac{1}{2\pi\e}\left(L^{d-1}(d\alpha,L_{d,E,v}(\cdot+i\e))-L^{d-1}(d\alpha,L_{d,E,v})\right)\right|\\
=&\frac{1}{2\pi}\left|\lim\limits_{n\rightarrow \infty}\frac{1}{n}\int_{\T}\ln\|\Lambda^{d-1}\left(L_{d,E,v}\right)_n(\theta+i\e)\|d\theta-\lim\limits_{n\rightarrow \infty}\frac{1}{n}\int_{\T}\ln\|\Lambda^{d-1}\left(L_{d,E,v}\right)_n(\theta)\|d\theta\right|\\
\leq& d-1.
\end{align*}
It follows that
$$
|\omega^{d-1}(\alpha,L_{E,v})|=\left|\frac{\omega^{d-1}(d\alpha,L_{d,E,v})}{d}\right|\leq \frac{d-1}{d}<1.
$$
By Proposition \ref{simple le}, $\gamma_{1}(E)<\gamma_{2}(E)$, thus by
Remark  \ref{rem1}, $\omega^{d-1}(\alpha,L_{E,v})$ is an
integer. Thus, since $|\omega^{d-1}(\alpha,L_{E,v})|$ is strictly smaller than $1$, we have $\omega^{d-1}(\alpha,L_{E,v})=0$.  This implies that
$$
L^{d-1}(\alpha,L_{E,v}(\cdot+i\e))=L^{d-1}(\alpha,L_{E,v})
$$
for sufficiently small $\e>0$. A similar  argument  works for  $\e<0$.
This means  $(\alpha,L_{E,v})$ is ($d-1$)-regular, hence, by Theorem
\ref{t2.1}, $(\alpha,L_{E,v})$ is ($d-1$)-dominated. Since $(\alpha,L_{E,v})$
is complex symplectic, we have $(\alpha,L_{E,v})$ is ($d+1$)-dominated.
\end{pf}
The next corollary follows directly from the definition of dominated splitting,
\begin{Corollary}\label{partial}
For $\alpha\in \R\backslash\Q$, and $E\in\R$ with $\bar{\omega}(E)=1$,
there exists a continuous invariant decomposition
$$
\C^{2d}=E_s(\theta)\oplus E_c(\theta)\oplus E_u(\theta).
$$
Moreover,  for any $\theta\in\T$, we have
\begin{equation}\label{eq10}
\limsup\limits_{n\rightarrow\infty}\frac{1}{n}\ln\|(L_{E,v})_n(\theta)v\|\leq -\gamma_2(E),\ \  \forall v\in E_s(\theta)\backslash\{0\},
\end{equation}
\begin{equation}\label{eq11}
\limsup\limits_{n\rightarrow\infty}\frac{1}{n}\ln\|(L_{E,v})_{-n}(\theta)v\|\leq  -\gamma_2(E),\ \  \forall v\in E_u(\theta)\backslash\{0\},
\end{equation}
\begin{equation}\label{eq11nea}
\limsup\limits_{n\rightarrow\infty}\frac{1}{|n|}\ln\|(L_{E,v})_n(\theta)v\|\leq  \gamma_1(E),\ \  \forall v\in E_c(\theta)\backslash\{0\},
\end{equation}
\begin{equation}\label{eq12}
{\rm dim} E^c(\theta)=2.
\end{equation}
\end{Corollary}
\begin{pf}
It follows from Theorem \ref{dominate1} and the definition of dominated splitting.
\end{pf}
We therefore have
\begin{Corollary} \label{Iph2}
At every Type I energy, the finite-range dual cocycle is $PH_2$ and
satisfies
\[
\omega^{d-1}(\alpha,L_{E,v})=0.
\]
\end{Corollary}

\section{Simplicity of point spectra of minimal $PH2$
    operators. Proof of Theorems \ref{tsimp} and \ref{gjy-puig}}\label{puig1}

In this section, we assume
$$
v(x)=\sum\limits_{k=-d}^d \hat{v}_k e^{2\pi ikx}, \ \ w(x)=\sum\limits_{k=-l}^l \hat{w}_k e^{2\pi i kx}.
$$

Let finite-range operator $L^w_{v,\alpha,x}$ be given by
\eqref{lvw}. We denote  the cocycle induced by the eigenequation
$L^w_{v,\alpha,x}u=Eu$ by $(\alpha,L_{E,v}^w),$ so
\begin{align*}
L_{E,v}^{w}(x)=\frac{1}{\hat{v}_{d}}
\begin{pmatrix}
-\hat{v}_{d-1}&\cdots&-\hat{v}_{1}&E-w(x)-\hat{v}_0&-\hat{v}_{-1}&\cdots&-\hat{v}_{-d+1}&-\hat{v}_{-d}\\
\hat{v}_{d}& \\
& &  \\
& & & \\
\\
\\
& & &\ddots&\\
\\
\\
& & & & \\
& & & & & \\
& & & & & &\hat{v}_{d}&
\end{pmatrix}
\end{align*}
with
$$
S_v=\begin{pmatrix}
0&-C_v^*\\
C_v&0
\end{pmatrix},\ \  C_v=\begin{pmatrix}
\hat{v}_d&\cdots&\hat{v}_1\\
&\ddots&\vdots\\
&&\hat{v}_d
\end{pmatrix}.
$$
We have, by the discussion in Section \ref{3.2}, that $L_{E,v}^w(x)$  is complex symplectic with respect
to $S_v$ for any $E\in\R$.


With $PH2$ property as in Definition \ref{defph2}, Theorem \ref{gjy-puig} follows directly from the following slightly
more general version

\begin{Theorem}\label{contra1}
For $\alpha\in\R\backslash\Q$, if there exist $H\in
C^\omega(\T,Sp_{2l\times 2}(\C))$ and $\phi_0\in\R$ such that
\begin{equation}\label{spe}
L^v_{E,w}(x)H(x)=e^{2\pi i\phi_0}H(x+\alpha),
\end{equation}
then the dual cocycle $(\alpha,L_{E,v}^w)$ is not $PH2.$
\end{Theorem}
Proof of Theorem \ref{contra1} will be split into the following three
subsections. But first we list two important corollaries.
\begin{Corollary}\label{fco1}
If $H_{v,\alpha,x}$ is a type I operator with trigonometric polynomial potential $v$, there does not exist $B\in C^{\omega}(\T,SL(2,\R))$ such that
$$
B^{-1}(x+\alpha)S_E^v(x)B(x)=Id.
$$
\end{Corollary}
\begin{pf}
By Theorem \ref{dominate1}, cocycles corresponding to the  duals of  type I operators are $PH2$.
\end{pf}

\begin{Corollary}
Let $\alpha\in\R\backslash\Q$ and $v$ be a trigonometric polynomial of degree $d$. Then there does not exist $F\in C^\omega(T,Sp_{2d\times 2}(\C))$ and $\phi_o\in\R$ such that
\begin{equation}\label{spe}
L^{2\cos}_{E,v}(x)F(x)=e^{2\pi i\phi_0}F(x+\alpha).
\end{equation}
\end{Corollary}
\begin{pf}
The dual of $L_{v,\alpha,\theta}^{2\cos}$ is  Schr\"odinger operator
$H_{v,\alpha,x}$, and corresponding $SL(2,\R)$ Schr\"odinger cocycles are automatically $PH2.$
\end{pf}

\subsection{Symplectic orthogonality of the eigenpairs}
It is well known that for second-difference operators
$$
 (Hu)_n=u_{n-1}+u_{n+1}+V(n)u(n),\ \ n\in\Z
 $$
 point spectrum is simple. Indeed,
 If $u,v\in\ell^2(\Z)$ satisfy $Hu=Eu$ and $Hv=Ev$, then, by the
 constancy of Wronskian,
 \begin{equation}\label{lg}
 u(n+1)v(n)-u(n)v(n+1)=0, \ \ \forall n\in\Z.
 \end{equation}
For the   finite-range  operators
\begin{equation}\label{finite_difference}
(Lu)(n)=\sum\limits_{k=-d}^{d} a_k u_{n+k}+b(n)u_n, \ \ n\in\Z,
\end{equation}
where $a_{-k}=\overline{a_k}$ and $\{b(n)\}_{n\in\Z}\subset \R^{\Z}$
is a  bounded sequence of real numbers, and $d>1$ this of course no
longer works, as  Wronskians of pairs of eigenfunctions are no longer well defined.

However, one can rewrite \eqref{lg} as
\begin{equation}\label{so}
\left\langle\begin{pmatrix} u(n+1)\\ u(n)\end{pmatrix}, \begin{pmatrix} 0&-1\\ 1&0\end{pmatrix}\begin{pmatrix} v(n+1)\\ v(n)\end{pmatrix}\right\rangle=0, \ \ \forall n\in\Z
\end{equation}
and therefore view the simplicity of the point spectrum for
Schr\"odinger operator as a corollary of  {\it symplectic
  orthogonality} \eqref{so} of  eigenfunctions. It turns out
symplectic orthogonality still holds in the finite-range case.

Let $S,C$ be defined by \eqref{sdef} with $\hat{v}_k$ replaced by $a_k$, and let
$$\vec{u}(n)=\begin{pmatrix} u(nd+d-1) \\ \vdots \\ u(nd)\end{pmatrix},\ \  \vec{v}(n)=\begin{pmatrix}v(nd+d-1) \\ \vdots \\ v(nd)\end{pmatrix}.
$$

\begin{Lemma}\label{symplectic}
For any  two eigenfunctions $u,v$ of $L$, corresponding to the same eigenvalue
$E,$ vectors $\begin{pmatrix}\vec{u}(1)\\ \vec{u}(0)\end{pmatrix}$ and $\begin{pmatrix}\vec{v}(1)\\ \vec{v}(0)\end{pmatrix}$ are symplectic orthogonal with respect to $S.$
\end{Lemma}


\begin{pf}
It is easy to check that $(\vec{u}(n))_{n\in\Z}$ satisfies
\begin{equation}\label{ef}
C\vec{u}(n+1)+T(n)\vec{u}(n)+C^*\vec{u}(n-1)=E\vec{u}(n),
\end{equation}
where $T(n)$ is the Hermitian matrix
$$
T(n)=\begin{pmatrix}
b(nd+d-1)+a_0&a_{-1}&\cdots&a_{-d+1}\\
a_1&\ddots&\ddots&\vdots\\
\vdots&\ddots&b(nd+1)+a_0&a_{-1}\\
a_{d-1}&\cdots&a_1&b(nd)+a_0
\end{pmatrix}.
$$
Note that equation \eqref{ef} is an eigenequation of the following vector-valued Schr\"odinger operator
$$
(L_d\vec{u})(n)=C\vec{u}(n+1)+T(n)\vec{u}(n)+C^*\vec{u}(n-1),
$$
acting on $\ell^2(\Z,\C^d)$.

To obtain a first-order system and the corresponding cocycle we use the fact that $C$ is invertible (since $a_d\neq 0$) and write
$$
\begin{pmatrix}
\vec{u}(n+1)\\
\vec{u}(n)
\end{pmatrix}
=\begin{pmatrix}
C^{-1}(EI_d-T(n))& -C^{-1}C^*\\
I_d&O_d
\end{pmatrix}
\begin{pmatrix}
\vec{u}(n)\\
\vec{u}(n-1)
\end{pmatrix},
$$
where $I_d$ and $O_d$ are the d-dimensional identity and zero matrices, respectively.
Set
$$
L_{d,E}(n)=\begin{pmatrix}
C^{-1}(EI_d-T(n))& -C^{-1}C^*\\
I_d&O_d
\end{pmatrix}
$$
Then for real $E$, matrix $L_{d,E}(n)$ is  complex-symplectic with respect to
$S=\begin{pmatrix} 0&-C^*\\
C&0
\end{pmatrix}$,
\begin{equation}\label{s_eq}
(L_{d,E}(n))^{*}S L_{d,E}(n)=S, \  \ n\in\Z.
\end{equation}

Since $u,v\in \ell^2(\Z),$ 
we have
\begin{equation}\label{limit1}
\lim\limits_{n\rightarrow \infty}\left\langle \begin{pmatrix}\vec{u}(n+1)\\ \vec{u}(n)\end{pmatrix}, \begin{pmatrix}0&-C^*\\
C&0\end{pmatrix}\begin{pmatrix}\vec{v}(n+1)\\ \vec{v}(n)\end{pmatrix}\right\rangle=0.
\end{equation}
On the other hand,
$$
\begin{pmatrix}\vec{u}(n+1)\\ \vec{u}(n)\end{pmatrix}=L_{d,E}(n)\cdots L_{d,E}(1)\begin{pmatrix}\vec{u}(1)\\ \vec{u}(0)\end{pmatrix}.
$$
By \eqref{s_eq}, for any $n\in\Z$,
\begin{align*}
&\left\langle \begin{pmatrix}\vec{u}(n+1)\\ \vec{u}(n)\end{pmatrix}, \begin{pmatrix}0&-C^*\\
C&0\end{pmatrix}\begin{pmatrix}\vec{v}(n+1)\\ \vec{v}(n)\end{pmatrix}\right\rangle\\
=&\left\langle L_{d,E}(n)\cdots L_{d,E}(1)\begin{pmatrix}\vec{u}(1)\\ \vec{u}(0)\end{pmatrix}, \begin{pmatrix}0&-C^*\\
C&0\end{pmatrix}L_{d,E}(n)\cdots L_{d,E}(1)\begin{pmatrix}\vec{v}(1)\\ \vec{v}(0)\end{pmatrix}\right\rangle\\
=&\left\langle \begin{pmatrix}\vec{u}(1)\\ \vec{u}(0)\end{pmatrix}, (L_{d,E}(n)\cdots L_{d,E}(1))^*\begin{pmatrix}0&-C^*\\
C&0\end{pmatrix}L_{d,E}(n)\cdots L_{d,E}(1)\begin{pmatrix}\vec{v}(1)\\ \vec{v}(0)\end{pmatrix}\right\rangle\\
=&\left\langle \begin{pmatrix}\vec{u}(1)\\ \vec{u}(0)\end{pmatrix}, \begin{pmatrix}0&-C^*\\
C&0\end{pmatrix}\begin{pmatrix}\vec{v}(1)\\ \vec{v}(0)\end{pmatrix}\right\rangle.
\end{align*}
By \eqref{limit1}, we obtain
$$
\left\langle \begin{pmatrix}\vec{u}(1)\\ \vec{u}(0)\end{pmatrix}, \begin{pmatrix}0&-C^*\\
C&0\end{pmatrix}\begin{pmatrix}\vec{v}(1)\\ \vec{v}(0)\end{pmatrix}\right\rangle=0.
$$
\end{pf}
\subsection{Simplicity of point spectrum. Proof of Theorem \ref{tsimp}}
Let $\Sigma_f$ be the ($\omega$-independent spectrum of a minimal $PH2$
operator $L_{f,\omega}$ given by  \eqref{finite operator11}. Fix $E\in\Sigma_f$.
By the definitions of $PH2$, dominated splitting and unique ergodicity, there exist continuous invariant decompositions
$$
\C^{2d}=E_s(\omega)\oplus E_c(\omega)\oplus E_u(\omega).
$$
and $C(E),\delta(E)>\delta'(E)>0$, such that for any $\omega\in\Omega$ and $n\geq 1$, we have
\begin{align}\label{eq10}
\left\|(L_{E}^{f})_{-n}(\omega)v\right\|>C^{-1}e^{\delta n},\ \  \forall v\in E_s(\omega)\backslash\{0\},\ \ \|v\|=1,
\end{align}
\begin{align}\label{eq11}
\left\|(L_{E}^{f})_{n}(\omega)u\right\|>C^{-1}e^{\delta n},\ \  \forall u\in E_u(\omega)\backslash\{0\},\ \ \|u\|=1,
\end{align}
\begin{align}\label{eq13}
\left\|(L_{E}^{f})_{\pm n}(\omega)w\right\|<Ce^{\delta' n},\ \  \forall w\in E_c(\omega)\backslash\{0\},\ \ \|w\|=1,
\end{align}
\begin{align}\label{eq12}
{\rm dim} E_c(\omega)=2.
\end{align}

Clearly, by \eqref{eq10}-\eqref{eq13} and invariance,  if $u(\omega)$ is an
$\ell^2$ eigenfunction, vector $\begin{pmatrix}\vec{u}(1,\omega)\\
  \vec{u}(0,\omega)\end{pmatrix}$ cannot have nonzero components in
either $E_s(\omega)$ or $E_u(\omega),$ for otherwise there would be
exponential growth at either $-\infty$ or $\infty.$

We now prove by contradiction. Assume $L_{f,\omega}$ has two linearly
independent eigenfunctions $u(\omega),v(\omega)$ corresponding to the
same eigenvalue $E.$ We then have
\begin{equation}\label{key}
\begin{pmatrix}\vec{u}(1,\omega)\\ \vec{u}(0,\omega)\end{pmatrix},\begin{pmatrix}\vec{v}(1,\omega)\\ \vec{v}(0,\omega)\end{pmatrix}\in E_c(\omega).
\end{equation}

On the other hand, by \eqref{eq12}, for any $\begin{pmatrix}\vec{x}(1,\omega)\\ \vec{x}(0,\omega)\end{pmatrix}\in E_c(\omega)$, we have
$$
\begin{pmatrix}\vec{x}(1,\omega)\\ \vec{x}(0,\omega)\end{pmatrix}=c_1(\omega)\begin{pmatrix}\vec{u}(1,\omega)\\ \vec{u}(0,\omega)\end{pmatrix}+c_2(\omega)\begin{pmatrix}\vec{v}(1,\omega)\\ \vec{v}(0,\omega)\end{pmatrix}.
$$
By Lemma \ref{symplectic}, we have
\begin{equation}\label{c22}
\left\langle \begin{pmatrix}\vec{u}(1,\omega)\\ \vec{u}(0,\omega)\end{pmatrix}, \begin{pmatrix}0&-C^*\\
C&0\end{pmatrix}\begin{pmatrix}\vec{v}(1,\omega)\\ \vec{v}(0,\omega)\end{pmatrix}\right\rangle=\left\langle \begin{pmatrix}\vec{u}(1,\omega)\\ \vec{u}(0,\omega)\end{pmatrix}, \begin{pmatrix}0&-C^*\\
C&0\end{pmatrix}\begin{pmatrix}\vec{u}(1,\omega)\\ \vec{u}(0,\omega)\end{pmatrix}\right\rangle=0.
\end{equation}
It follows that for any $\begin{pmatrix}\vec{x}(1,\omega)\\ \vec{x}(0,\omega)\end{pmatrix}\in E_c(\omega)$,
$$
\left\langle \begin{pmatrix}\vec{u}(1,\omega)\\ \vec{u}(0,\omega)\end{pmatrix}, \begin{pmatrix}0&-C^*\\
C&0\end{pmatrix}\begin{pmatrix}\vec{x}(1,\omega)\\ \vec{x}(0,\omega)\end{pmatrix}\right\rangle=0.
$$
This contradicts the non-degeneracy of the symplectic form.
\qed
\subsection{Proof of Theorem \ref{contra1}} We proceed by contradiction. Assume $(\alpha,L_{E,v}^w)$ is $PH2$ and  there exists $H\in C^\omega(\T,Sp_{2l\times 2}(\C))$ such that
\begin{equation}\label{form11}
L^v_{E,w}(x)H(x)=e^{2\pi i\phi_0}H(x+\alpha).
\end{equation}
Let
$$
H=\begin{pmatrix}
h_{1,1}&h_{1,2}\\
h_{2,1}&h_{2,2}\\
\vdots&\vdots\\
h_{2l,1}&h_{2l,2}
\end{pmatrix}\in C^\omega(\T,Sp_{2l\times 2}(\C)).
$$
By the definition of $L^v_{E,w}(x)$ and \eqref{form11}, one has for $i=1,2$,
\begin{equation}\label{e2e1}
-\frac{1}{\hat{w}_{l}}\left(\sum\limits_{k=1}^{2l} \hat{w}_{l-k}h_{k,i}(x)+(v(x)-E)h_{l,i}(x)\right)-e^{2\pi i\phi_0}h_{1,i}(x+\alpha)=0,
\end{equation}
\begin{equation}\label{e3e1}
h_{k,i}(x)=e^{2\pi i\phi_0}h_{k+1,i}(x+\alpha), \ \ \forall 1\leq k\leq 2l-1.
\end{equation}
It follows from \eqref{e2e1} and \eqref{e3e1} that
\begin{equation}\label{e41}
\sum\limits_{k=-l}^l \hat{w}_{k}e^{2\pi ik\phi_0}h_{l,i}(x+k\alpha)+(v(x)-E)h_{l,i}(x)=0.
\end{equation}
Let $h_{l,i}(x)=\sum_k\hat{h}_i(k)e^{2\pi i kx}$ be the Fourier
expansion. Taking the Fourier transform of \eqref{e41}, we get
\begin{equation*}
\sum\limits_{k=-l}^l \hat{w}_{k}e^{2\pi ik(\phi_0+n\alpha)}\hat{h}_i(n)+\sum\limits_{k=-d}^d\hat{v}_k\hat{h}_i(n-k)-E\hat{h}_i(n)=0.
\end{equation*}
  
Thus $\left\{\overline{\hat{h}_1(n)}\right\}_{n\in\Z}$ and $\left\{\overline{\hat{h}_2(n)}\right\}_{n\in\Z}$ are two linearly independent  eigenfunctions of $L^w_{v,\alpha,\phi_0}$ corresponding to the same eigenvalue $E$.

On the other hand, since $(\alpha,L_{E,v}^w)$ is $PH2$, by Theorem
\ref{tsimp}, $L^w_{v,\alpha,\phi_0}$ has simple point spectrum, a
contradiction. \qed

\section{An all-frequency Puig's argument. Proof of Theorem \ref{rot1}}\label{puig2}
In this section, let $v$ and $w$ be as in Section \ref{puig1}. 
Theorem \ref{rot1} can be equivalently reformulated as
\begin{Theorem}\label{contra2}
For any $\alpha\in\R\backslash\Q$, if there exist $F\in C^\omega(\T,Sp_{2l\times 2}(\C))$ and $\phi,\psi\in C^\omega(\T,\R)$ with $\int_\T\psi(x)dx=0$ such that
\begin{equation}\label{spe}
L^v_{E,w}(x)F(x)=e^{2\pi i\phi(x)}F(x+\alpha)R_{\psi(x)},
\end{equation}
then the dual cocycle $(\alpha,L_{E,v}^w)$ is not $PH2$.
\end{Theorem}
\begin{Remark}
{\rm Roughly speaking, while Theorem \ref{contra1} can be used to prove Cantor
  spectrum for type I operators with  {\it Diophantine}
  frequencies, Theorem \ref{contra2} will be used to prove Cantor
  spectrum for type I operators with {\it all}  irrational frequencies.}
\end{Remark}
 We first list two important corollaries. Recall that $v$ is a trigonometric polynomial.
\begin{Corollary}\label{fco2}
For any $\alpha\in\R\backslash \Q$, if operator $H_{v,\alpha,x}$ is of Type I, then there does not exist $B\in C^{\omega}(\T,SL(2,\R))$  and $\psi\in C^\omega(\T,\R)$ with $\int_\T \psi(x)dx=0$ such that
$$
B^{-1}(x+\alpha)S_E^v(x)B(x)=R_{\psi(x)}.
$$
\end{Corollary}

\begin{Corollary}
For any $\alpha\in\R\backslash\Q$, there does not exist $F\in C^\omega(\T,Sp_{2d\times 2}(\C))$  and $\phi,\psi\in C^\omega(\T,\R)$ with $\int_\T \psi(x)dx=0$ such that
\begin{equation}\label{spe}
L^{2\cos}_{E,v}(x)F(x)=e^{2\pi i\phi(x)}F(x+\alpha)R_{\psi(x)}.
\end{equation}
\end{Corollary}

Proof of Theorem \ref{contra2} will be split into the following two
subsections.
\subsection{Quantitative almost reducibility via rotations
  reducibility} In this subsection, we derive quantitative almost
reducibility from rotations reducibility for quasiperiodic
finite-range operators. Let $p_n/q_n$ be the approximants of the continued fraction expansion of $\alpha$. By the definition of $\beta(\alpha)$, i.e.,
$$
\beta(\alpha)=\limsup\limits_{n\rightarrow\infty}\frac{\ln q_{n+1}}{q_n},
$$
for any $0<\e<\frac{\beta}{100}$, there is a subsequence $q_{n_k}$ of $q_n$ such that
\begin{equation}\label{naaa}
q_{n_{k}+1}>e^{(\beta-\e)q_{n_k}}. \footnote{If $\beta(\alpha)=\infty$, we choose a subsequence such that $q_{n_{k}+1}>e^{100h q_{n_k}}$.}
\end{equation}
The key technical fact is
\begin{Theorem}\label{th1}
For all $\alpha\in\R\backslash\Q$, if there exist $F\in C_h^\omega(\T,Sp_{2l\times 2}(\C))$, and $\phi,\psi\in C_h^\omega(\T,\R)$ with $\int_\T \psi(x)dx=0$ such that
\begin{equation}\label{spe}
L^v_{E,w}(x)F(x)=e^{2\pi i\phi(x)}F(x+\alpha)R_{\psi(x)},
\end{equation}
then we have
\begin{enumerate}
\item if $\beta(\alpha)<h$, then there exists $H\in C^\omega(\T,Sp_{2l\times 2}(\C))$ such that
$$
L^v_{E,w}(x)H(x)=e^{2\pi i\phi_0}H(x+\alpha),
$$
where $\phi_0=\int_\T \phi(x)dx$;
\item if $\beta(\alpha)\geq h$, then there exist $H^k\in C_{h/2}^\omega(\T,Sp_{2l\times 2}(\C))$ and $\e_\phi^k, \e_\psi^k\in C_{h/2}^\omega(\T,\R)$ such that
\begin{equation}\label{al}
L^v_{E,w}(x)H^k(x)=e^{2\pi i\phi_0}H^k(x+\alpha)e^{2\pi i\e^k_\phi}R_{\e_\psi^k(x)},
\end{equation}
with
$$
|H^k|_{\frac{h}{2}}\leq |F|_he^{16\pi(q_{n_k}+e^{-\frac{h}{2}q_{n_k}}q_{n_{k}+1})(|\phi|_h+|\psi|_h)},
$$
$$
|\e_\phi^k|_\frac{h}{2},|\e_\psi^k|_{\frac{h}{2}}\leq e^{-\frac{1}{20}q_{n_k+1}h}(|\phi|_h+|\psi|_h).
$$
\end{enumerate}
\end{Theorem}
\begin{pf}
Case (1): If $\beta(\alpha)<h$, the argument is standard. Define
$$
\Psi(x)=\sum\limits_{k\in\Z\backslash\{0\}}\frac{1}{e^{2\pi ik\alpha}-1}\hat{\psi}(k)e^{2\pi i kx},
$$
$$
\Phi(x)=\sum\limits_{k\in\Z\backslash\{0\}}\frac{1}{e^{2\pi ik\alpha}-1}\hat{\phi}(k)e^{2\pi i kx}.
$$
Since $h>\beta(\alpha)$, we have that $\Phi,\Psi\in C^\omega(\T,\R)$ and
$$
\Psi(x+\alpha)-\Psi(x)=\psi(x),\ \ \Phi(x+\alpha)-\Phi(x)=\phi(x)-\phi_0.
$$
For
$$
H(x)=e^{2\pi i\Phi(x)}F(x)R_{\Psi(x)},
$$
since $F\in C^\omega(\T,Sp_{2l\times 2}(\C))$, one can check that $H\in C^\omega(\T,Sp_{2l\times 2}(\C))$ and
$$
L^v_{E,w}(x)H(x)=e^{2\pi i\phi_0}H(x+\alpha).
$$

Case (2): If $\beta(\alpha)\geq h$, we need the following lemma.
\begin{Lemma}\label{small}
For $\alpha\in\R\backslash\Q$ with $\beta(\alpha)\geq h$ and $f\in C_h^\omega(\T,\R)$, there exist sequences of $g_k\in C_{h/2}^\omega(\T,\R)$ such that
$$
|g_k|_{\frac{h}{2}}\leq 8(q_{n_k}+e^{-\frac{h}{2}q_{n_k}}q_{n_{k}+1})|f|_h,
$$
$$
|f(x)-\int_\T f(x)dx-(g_k(x+\alpha)-g_k(x))|_{\frac{h}{2}}\leq e^{-\frac{1}{20}q_{n_k+1}h}|f|_h.
$$
\end{Lemma}
\begin{pf}
First we observe that
\begin{Proposition}\label{naaa1}
For any $0<|k|<\frac{q_{n+1}}{6}$, if $k\notin R_n=\{\ell q_n: \ell\in\Z\}$ and $q_{n+1}>100q_n$, then
$$
\|k\alpha\|_{\R/\Z}\geq \frac{1}{4q_n}.
$$
\end{Proposition}
\begin{pf}
Since $k\notin R_n=\{\ell q_n: \ell\in\Z\}$, we have that
$$
k=\ell_0 q_n+r, \ \ 0< r\leq q_n-1.
$$
On the other hand,
$$
|\ell_0|=\left|\frac{k-r}{q_n}\right|<\frac{q_{n+1}}{6q_n}+1.
$$
Thus
\begin{align*}
\|k\alpha\|_{\R/\Z}&\geq \|r\alpha\|_{\R/\Z}-|\ell_0|\|q_n\alpha\|_{\R/\Z}\\
&\geq \frac{1}{2q_n}-\frac{1}{6q_n}-\frac{1}{q_{n+1}}\geq \frac{1}{4q_n}.
\end{align*}
\end{pf}
Define $N_k=\left[\frac{q_{n_{k}+1}}{6}\right]$ and
$$
g_k(x)=\sum\limits_{j=-N_k}^{-1}\frac{\hat{f}(j)}{e^{2\pi i j\alpha}-1}e^{2\pi ijx}+\sum\limits_{j=1}^{N_k}\frac{\hat{f}(j)}{e^{2\pi i j\alpha}-1}e^{2\pi ijx}.
$$
In view of \eqref{naaa} and Proposition \ref{naaa1}, for $k$ large enough, we distinguish two cases:

Case 1: $0< |j|< \frac{q_{n_k+1}}{6}$, $j\notin R_{n_k}$, then $|e^{2\pi i j\alpha}-1|\geq \frac{1}{4q_{n_k}}$,

Case 2: $q_{n_k}\leq |j|< \frac{q_{n_k+1}}{6}$, $j\in R_{n_k}$, then $|e^{2\pi i j\alpha}-1|\geq \frac{1}{2q_{n_k+1}}$.\\
It follows that
$$
|g_k|_{\frac{h}{2}}\leq 8q_{n_k}|f|_h+8e^{-\frac{h}{2}q_{n_k}}q_{n_{k}+1}|f|_h.
$$
Moreover, we have
$$
f(x)-\int_\T f(x)dx-(g_k(x+\alpha)-g_k(x))=\sum\limits_{|j|\geq q_{n_k+1}/6}\hat{f}(j)e^{2\pi ijx},
$$
which implies that
$$
|f(\cdot)-\int_\T f(x)dx-(g_k(\cdot+\alpha)-g_k(\cdot))|_{\frac{h}{2}}\leq e^{-\frac{1}{20}q_{n_k+1}h}|f|_h.
$$

\end{pf}

Thus, by Lemma \ref{small}, there are  $\Phi^k,\Psi^k\in C_{h/2}^\omega(\T,\R)$ such that
\begin{equation}\label{sd1}
|\Psi^k|_{\frac{h}{2}}+|\Phi^k|_{\frac{h}{2}}\leq 8(q_{n_k}+e^{-\frac{h}{2}q_{n_k}}q_{n_{k}+1})(|\psi|_h+|\phi|_h),
\end{equation}
\begin{equation}\label{sd2}
|\psi-(\Psi^k(\cdot+\alpha)-\Psi^k(\cdot))|_{\frac{h}{2}}\leq e^{-\frac{1}{20}q_{n_k+1}h}|\psi|_h,
\end{equation}
\begin{equation}\label{sd2'}
|\phi-\phi_0-(\Phi^k(\cdot+\alpha)-\Phi^k(\cdot))|_{\frac{h}{2}}\leq e^{-\frac{1}{20}q_{n_k+1}h}|\phi|_h.
\end{equation}
Define
\begin{equation}\label{gyz20241}
H^k(x)=e^{2\pi i\Phi^k(x)}F(x)R_{\Psi^k(x)}.
\end{equation}
We have
\begin{equation}\label{al}
L^v_{E,w}(x)H^k(x)=e^{2\pi i\phi_0}H^k(x+\alpha)e^{2\pi i\e_\phi^k(x)}R_{\e_\psi^k(x)},
\end{equation}
where $\e_\phi^k(x)=\phi(x)-\phi_0-(\Phi^k(x+\alpha)-\Phi^k(x))$, $\e_\psi^k(x)=\psi(x)-(\Psi^k(x+\alpha)-\Psi^k(x))$ and $H^k\in C^\omega(\T,Sp_{2l\times 2}(\C))$. Moreover, by \eqref{sd1}-\eqref{gyz20241},
$$
|H^k|_{\frac{h}{2}}\leq |F|_he^{16\pi(q_{n_k}+e^{-\frac{h}{2}q_{n_k}}q_{n_{k}+1})(|\phi|_h+|\psi|_h)},
$$
$$
|\e_\phi^k|_{\frac{h}{2}}\leq |\phi(\cdot)-\phi_0-(\Phi^k(\cdot+\alpha)-\Phi^k(\cdot))|_{\frac{h}{2}}\leq e^{-\frac{1}{20}q_{n_k+1}h}|\phi|_h,
$$
$$
|\e_\psi^k|_{\frac{h}{2}}\leq |\psi(\cdot)-(\Psi^k(\cdot+\alpha)-\Psi^k(\cdot))|_{\frac{h}{2}}\leq e^{-\frac{1}{20}q_{n_k+1}h}|\psi|_h.
$$

\end{pf}

\subsection{Proof of Theorem \ref{contra2}}
In case  $\beta(\alpha)<h$, Theorem \ref{contra2} follows
immediately from (1) of Theorem \ref{th1} and Theorem \ref{contra1}. In case  $\beta(\alpha)\geq
h$, we prove Theorem \ref{contra2} via quantitative almost
reducibility and quantitative Aubry duality. Essentially, we need to
establish a {\it quantitative version} of Puig's argument for
finite-range operators. The proof proceeds by contradiction. Given $\alpha\in\R\backslash\Q$, we assume
\begin{enumerate}
\item There exist $F\in C_h^\omega(\T,Sp_{2l\times 2}(\C))$ and $\phi,\psi\in C_h^\omega(\T,\R)$ with $\int_\T \psi(x)dx=0$ such that
\begin{equation}\label{spe1}
L^v_{E,w}(x)F(x)=e^{2\pi i\phi(x)}F(x+\alpha)R_{\psi(x)},
\end{equation}
\item $(\alpha,L_{E,v}^w)$ is $PH2.$
\end{enumerate}

By Theorem \ref{th1} and assumption (1),  there exist  $H^k\in C_{h/2}^\omega(\T,Sp_{2l\times 2}(\C))$ and $\e_\phi^k,\e_\psi^k\in C_{h/2}^\omega(\T,\R)$ such that
\begin{equation}\label{al}
L^v_{E,w}(x)H^k(x)=e^{2\pi i\phi_0}H^k(x+\alpha)e^{2\pi i\e_\phi^k(x)}R_{\e_\psi^k(x)},
\end{equation}
with
\begin{align}\label{ese1}
|H^k|_{\frac{h}{2}}\leq |F|_he^{16\pi(q_{n_k}+e^{-\frac{h}{2}q_{n_k}}q_{n_{k}+1})(|\phi|_h+|\psi|_h)},
\end{align}
\begin{align}\label{ese2}
|\e_\phi^k|_{\frac{h}{2}},|\e_\psi^k|_{\frac{h}{2}}\leq e^{-\frac{1}{20}q_{n_k+1}h}(|\phi|_h+|\psi|_h).
\end{align}
Let
\begin{equation}\label{hk}
\begin{pmatrix}
h_{1,1}&h_{1,2}\\
h_{2,1}&h_{2,2}\\
\vdots&\vdots\\
h_{2l,1}&h_{2l,2}
\end{pmatrix}:= H^k\in C^\omega(\T,Sp_{2l\times 2}(\C)).
\end{equation}
Involving the form of $L^v_{E,w}(x)$ and \eqref{al}, one has for $j=1,2$,
\begin{equation}\label{e22}
-\frac{1}{\hat{w}_{l}}\left(\sum\limits_{n=1}^{2l} \hat{w}_{l-n}h_{n,j}(x)+(v(x)-E)h_{l,j}(x)\right)-e^{2\pi i\phi_0}h_{1,j}(x+\alpha)=e^{2\pi i\phi_0}g_{1,j}(x),
\end{equation}
\begin{equation}\label{e33}
h_{m,j}(x)=e^{2\pi i\phi_0}(h_{m+1,j}(x+\alpha)+g_{m+1,j}(x)), \ \ \forall 1\leq m\leq 2l-1,
\end{equation}
where
$$
g_{m,1}(x)=(e^{2\pi i\e_\phi^k(x)}\cos2\pi(\e_\psi^k(x))-1)h_{m,1}(x+\alpha)+e^{2\pi i\e_\phi^k(x)}\sin2\pi(\e_\psi^k(x))h_{m,2}(x+\alpha),
$$
$$
g_{m,2}(x)=(e^{2\pi i\e_\phi^k(x)}\cos2\pi(\e_\psi^k(x))-1)h_{m,2}(x+\alpha)-e^{2\pi i\e_\phi^k(x)}\sin2\pi(\e_\psi^k(x))h_{m,1}(x+\alpha).
$$
It follows from \eqref{e22} and \eqref{e33} that

\begin{equation}\label{e466}
\sum\limits_{n=-l}^l \hat{w}_{n}e^{2\pi in\phi_0}h_{l,j}(x+n\alpha)+(v(x)-E)h_{l,j}(x)=e_j(x),
\end{equation}   
where $e_j$ is a linear combination of $\{g_{m,j}\}_{m=1}^{2l}$ of at most $4l^2$ terms. Hence by \eqref{ese1} and \eqref{ese2}, for $k$ sufficiently large depending on $E,v,w$, we have
\begin{equation}\label{ee1}
|e_j|_{\frac{h}{2}}\leq C(E,v,w)l^2|F|_he^{16\pi(q_{n_k}+e^{-\frac{h}{2}q_{n_k}}q_{n_{k}+1})(|\phi|_h+|\psi|_h)}(|\e_\phi^k|_{\frac{h}{2}}+|\e_\psi^k|_{\frac{h}{2}})\leq e^{-\frac{1}{40}q_{n_k+1}h}.
\end{equation}

Let $h_{l,j}(x)=\sum_n\hat{h}_j(n)e^{2\pi i nx}$ be the Fourier
expansion. By \eqref{e466} and \eqref{lvw}, since $v$ and $w$ are real-valued,  $\left\{\overline{\hat{h}_1(n)}\right\}_{n\in\Z}$ and $\left\{\overline{\hat{h}_2(n)}\right\}_{n\in\Z}$ are two approximate solutions of $L^w_{v,\alpha,\phi_0}u=Eu$, i.e. they satisfy
\begin{equation}\label{new4}
\left((L^w_{v,\alpha,\phi_0}-E)\overline{\hat{h}_1}\right)(n)=\overline{\hat{e}_1(n)},
\end{equation}
\begin{equation}\label{new4-1}
\left((L^w_{v,\alpha,\phi_0}-E)\overline{\hat{h}_2}\right)(n)=\overline{\hat{e}_2(n)},
\end{equation}
where $\{\hat{e}_j(n)\}_{n\in\Z}$ are the Fourier coefficients of $e_j(x)$.

We denote
$$
I_{k}=\left[-e^{-\frac{h}{8}q_{n_k}}q_{n_k+1},e^{-\frac{h}{8}q_{n_k}}q_{n_{k}+1}\right].
$$
Then by \eqref{ese1}, for $j=1,2$, we have
\begin{align}\label{eeee1}
\sum_{n\notin I_k}|\hat{h}_{j}(n)|\leq |F|_he^{16\pi(q_{n_k}+e^{-\frac{h}{2}q_{n_k}}q_{n_{k}+1})(|\phi|_h+|\psi|_h)} e^{-\frac{h}{4}e^{-\frac{h}{8}q_{n_k}}q_{n_k+1}}
\leq e^{-\frac{h}{4}e^{-\frac{h}{4}q_{n_k}}q_{n_k+1}}.
\end{align}
On the other hand, we have
\begin{Lemma}\label{initial}
There exists $n_0\in I_k$, such that
\begin{equation}\label{z1-estimate-21}
|\hat{h}_{1}(n_0)|\geq e^{\frac{h}{16}q_{n_k}}q^{-1}_{n_k+1}
\end{equation}
provided $k\geq K_0(E,\alpha,v,w)$.
\end{Lemma}
\begin{pf}
Denote
\begin{equation}\label{lanair3}
\overrightarrow{u}=\begin{pmatrix}
h_{1,1}\\
h_{2,1}\\
\vdots\\
h_{2l,1}
\end{pmatrix},\ \ \overrightarrow{v}=\begin{pmatrix}
h_{1,2}\\
h_{2,2}\\
\vdots\\
h_{2l,2}
\end{pmatrix}
\end{equation}
By \eqref{hk}, \eqref{lanair2} and the fact that $H^k\in C^\omega(\T,Sp_{2l\times 2}(\C))$, we have that
$$
\overrightarrow{u}^*S_{w}\overrightarrow{v}=1. \footnote{Note that \begin{equation}\label{sdef}
C_w=\begin{pmatrix}
\hat{w}_l&\cdots&\hat{w}_1\\
0&\ddots&\vdots\\
0&0&\hat{w}_l
\end{pmatrix},\ \ S=\begin{pmatrix}0&-C_w^*\\
C_w&0\end{pmatrix}
\end{equation}
}
$$
Thus
$$
\|\overrightarrow{u}\|_{L^2}\|S_w\overrightarrow{v}\|_{L^2}\geq 1
$$
which implies that
$$\|\overrightarrow{u}\|_{L^2}\geq \frac{1}{\|S_w\overrightarrow{v}\|_{L^2}}>\frac{1}{C(w)|H^k|_{C^0}}.$$
By \eqref{ese1}, \eqref{ese2} and \eqref{e33}, one has
\begin{align}\label{b111}
2l\|\hat{h}_{1}\|_{\ell^2}\geq \|\overrightarrow{u}\|_{L^2}-C(v,w)l^2\sup_{1\leq m\leq 2l}\|g_{m,1}\|_{L^2}\geq (2C|H^k|_{C^0})^{-1} \geq (2C|F|_{C^0})^{-1}\geq c.
\end{align}
By \eqref{eeee1}, we have that
\begin{align*}
\sum\limits_{n\notin I_k}|\hat{h}_{1}(n)|^2
\leq 2e^{-\frac{h}{4}e^{-\frac{h}{4}q_{n_k}}q_{n_k+1}}.
\end{align*}
By \eqref{b111} and the fact that $|I_k|\leq 2e^{-\frac{h}{8}q_{n_k}}q_{n_k+1}+1$, it follows that there exists $n_0\in I_k$, such that
\begin{align*}
(2e^{-\frac{h}{8}q_{n_k}}q_{n_k+1}+1)|\hat{h}_{1}(n_0)|^2 \geq \sum\limits_{n\in I_k}|\hat{h}_{1}(n)|^2 &=\|\hat{h}_{1}\|_{\ell^2}^2-\sum\limits_{n\notin I_k}|\hat{h}_{1}(n)|^2\\
 &\geq  c(2l)^{-1}-2e^{-\frac{h}{4}e^{-\frac{h}{4}q_{n_k}}q_{n_k+1}} .
\end{align*}
Hence  there exists $K_0>0$, such that
$$
|\hat{h}_{1}(n_0)|\geq  e^{\frac{h}{16}q_{n_k}}q^{-1}_{n_k+1},
$$
provided $k>K_0$.
\end{pf}

For $j=1,2$, we define

$$
\vec{h}_j(n)=\begin{pmatrix}
\hat{h}_{j}(n+d-1)\\
\hat{h}_{j}(n+d-2)\\
\vdots\\
\hat{h}_{j}(n-d)
\end{pmatrix},\ \ \vec{p}_j(n)=\frac{1}{\overline{\hat{v}_d}}\begin{pmatrix}\hat{e}_j(n)\\0\\ \vdots\\ 0\end{pmatrix}.
$$
   
Notice that by \eqref{new4} and \eqref{new4-1}, we have

$$
\overline{\vec{h}_j(n+1)}=L^w_{E,v}(\phi_0+n\alpha)\overline{\vec{h}_j(n)}+\overline{\vec{p}_j(n)},
$$  
so identically we obtain,
\begin{align}\label{vech}
\overline{\vec{h}_j(n)}=(L^w_{E,v})_{n-n_0}(\phi_0+n_0\alpha)\overline{\vec{h}_j(n_0)}+\sum\limits_{k=n_0+1}^n (L_{E,v}^w)_{n-k}(\phi_0+k\alpha)\overline{\vec{p}_j(k-1)}.
\end{align}
where $(n\alpha,(L_{E,v}^w)_n):= (\alpha, L_{E,v}^w)^n$, are the iterates of the dual cocycle.
Here and below, variation-of-constants sums are oriented: if the upper
endpoint is smaller than the lower one, then
$\sum_{k=a+1}^{b}F(k)=-\sum_{k=b+1}^{a}F(k)$ for $b<a$.
This convention makes each displayed identity valid on both sides of
the initial site.

Let $\vec{u}_j(n)$ be solutions of $L^w_{v,\alpha,\phi_0}u=Eu$ with initial data $\overline{\vec{h}_j(n_0)}$,
\begin{align}\label{vecu}
\vec{u}_j(n)=(L^w_{E,v})_{n-n_0}(\phi_0+n_0\alpha)\overline{\vec{h}_j(n_0)},
\end{align}

Involving the $PH2$ property, i.e.,  that there exists a continuous invariant decomposition
$$
\C^{2d}=E_s(x)\oplus E_c(x)\oplus E_u(x).
$$
we have
\begin{Lemma}\label{key1}
For $j=1,2$, there exist
$$
\vec{u}^s_j(n_0)\in E_s(\phi_0+n_0\alpha), \ \ \vec{u}^c_j(n_0)\in E_c(\phi_0+n_0\alpha),\ \ \vec{u}^u_j(n_0)\in E_u(\phi_0+n_0\alpha)
$$
such that
\begin{align*}
\vec{u}_j(n_0)=\vec{u}^s_j(n_0)+\vec{u}^c_j(n_0)+\vec{u}^u_j(n_0).
\end{align*}
Moreover,  there are $C>0$ and $\delta>\delta'>0$, such that
$$
\left\|\vec{u}^s_j(n_0)\right\|,\ \ \left\|\vec{u}^u_j(n_0)\right\|\leq Ce^{-(\delta-\delta') e^{-\frac{h}{100}q_{n_k}}q_{n_k+1}}.
$$
\end{Lemma}
\begin{pf}
Note that by \eqref{vech}, \eqref{vecu} and \eqref{ee1},  for $j=1,2$ and $|n|\leq e^{-\frac{h}{1000}q_{n_k}}q_{n_k+1}$, we have
\begin{align}\label{eeee22}
|u_j(n)-\overline{\hat{h}_j(n)}|\leq q_{n_k+1}e^{Ce^{-\frac{h}{1000}q_{n_k}}q_{n_k+1}}e^{-\frac{1}{40}q_{n_k+1}h}\leq e^{-\frac{1}{80}q_{n_k+1}h}.
\end{align}
Combining \eqref{eeee1} and \eqref{eeee22}, we have for $e^{-\frac{h}{20}q_{n_k}}q_{n_k+1}\leq|n|\leq e^{-\frac{h}{100}q_{n_k}}q_{n_k+1}$,
\begin{equation}\label{esu}
|u_j(n)|\leq |\overline{\hat{h}_j(n)}|+e^{-\frac{1}{80}q_{n_k+1}h}\leq 2e^{-\frac{h}{4}e^{-\frac{h}{4}q_{n_k}}q_{n_k+1}}.
\end{equation}
Since $(\alpha,L_{E,v}^w)$ is one-frequency $PH2$, there exist continuous invariant decompositions
$$
\C^{2d}=E_s(x)\oplus E_c(x)\oplus E_u(x).
$$
Moreover, there are $C(E),\delta(E)>\delta'(E)>0$, such that for any $x\in\T$ and $n\geq 1$, we have
\begin{align}\label{eq10a1}
\left\|(L_{E,v}^{w})_{-n}(x)v\right\|>C^{-1}e^{\delta n},\ \  \forall v\in E_s(x)\backslash\{0\},\ \ \|v\|=1,
\end{align}
\begin{align}\label{eq11a1}
\left\|(L_{E,v}^{w})_{n}(x)u\right\|>C^{-1}e^{\delta n},\ \  \forall v\in E_u(x)\backslash\{0\},\ \ \|u\|=1,
\end{align}
\begin{align}\label{eq13a1}
\left\|(L_{E,v}^{w})_{\pm n}(x)w\right\|<Ce^{\delta' n},\ \  \forall v\in E_c(x)\backslash\{0\},\ \ \|w\|=1,
\end{align}
\begin{align}\label{eq12a1}
{\rm dim} E_c(x)=2.
\end{align}
Thus, for $j=1,2$, there exist
$$
\vec{u}^s_j(n_0)\in E_s(\phi_0+n_0\alpha), \ \ \vec{u}^c_j(n_0)\in E_c(\phi_0+n_0\alpha),\ \ \vec{u}^u_j(n_0)\in E_u(\phi_0+n_0\alpha)
$$
such that
\begin{align*}
\vec{u}_j(n_0)=\vec{u}^s_j(n_0)+\vec{u}^c_j(n_0)+\vec{u}^u_j(n_0).
\end{align*}
By \eqref{esu} and Lemma \ref{lee2}, we have
\begin{equation}\label{err}
\left\|\vec{u}^s_j(n_0)\right\|,\ \ \left\|\vec{u}^u_j(n_0)\right\|\leq Ce^{-(\delta-\delta') e^{-\frac{h}{100}q_{n_k}}q_{n_k+1}}.
\end{equation}
\end{pf}
Let $n_k=[(\delta-\delta')e^{-\frac{h}{50}q_{n_k}}q_{n_{k+1}}]$. 
The following proposition is crucial
\begin{Proposition}\label{ff5}
There exist $K_1(E,v,w,\alpha)$ such that if $k>K_1$, then
\begin{equation}\label{fwn}
\det{\begin{pmatrix}\vec{u}_1^c(n_0)&\vec{u}_2^c(n_0)\end{pmatrix}^*\begin{pmatrix}\vec{u}_1^c(n_0)&\vec{u}_2^c(n_0)\end{pmatrix}} \geq e^{-n_k}.
\end{equation}
\end{Proposition}
\begin{pf}
By \eqref{eq12a1}, there is a unit vector $\vec{v}_1^c(n_0)\in E_c(\phi_0+n\alpha)$ which is orthogonal to $\vec{u}_1^c(n_0)$ such that
\begin{equation}\label{gjnew2}
\vec{u}_2^c(n_0)=\frac{\langle \vec{u}_2^c(n_0), \vec{u}_1^c(n_0)\rangle}{\|\vec{u}_1^c(n_0)\|^2}\vec{u}_1^c(n_0)+\langle \vec{u}_2^c(n_0), \vec{v}_1^c(n_0)\rangle  \vec{v}_1^c(n_0).
\end{equation}
By direct calculation, we have
\begin{equation}\label{gjnew1}
\det{\begin{pmatrix}\vec{u}_1^c(n_0)&\vec{u}_2^c(n_0)\end{pmatrix}^*\begin{pmatrix}\vec{u}_1^c(n_0)&\vec{u}_2^c(n_0)\end{pmatrix}}=|\langle \vec{u}_2^c(n_0),\vec{v}_1^c(n_0)\rangle|^2 \|\vec{u}_1^c(n_0)\|^2.
\end{equation}

We now prove \eqref{fwn} by contradiction. Assume
$$
\det{\begin{pmatrix}\vec{u}_1^c(n_0)&\vec{u}_2^c(n_0)\end{pmatrix}^*\begin{pmatrix}\vec{u}_1^c(n_0)&\vec{u}_2^c(n_0)\end{pmatrix}}< e^{-n_k},
$$
By Lemma \ref{initial} and the fact that $\|H^k\|_{C^0}=\|F\|_{C^0}\leq C$, for $j=1,2$, one has
\begin{equation}\label{ff4}
e^{\frac{h}{16}q_{n_k}}q^{-1}_{n_k+1}\leq \|\vec{u}_{j}(n_0)\| \leq 2Cd.
\end{equation}
Thus by \eqref{gjnew1} and  \eqref{ff4}, we further have
\begin{equation}\label{error}
|\langle \vec{u}_2^c(n_0),\vec{v}_1^c(n_0)\rangle|\leq e^{-\frac{h}{16}q_{n_k}}q_{n_k+1}e^{-\frac{n_k}{3}}.
\end{equation}
By \eqref{err}, we have
\begin{equation}\label{gj1}
\|\vec{u}_j^c(n_0)-\vec{u}_j(n_0)\| \leq Ce^{-(\delta-\delta') e^{-\frac{h}{100}q_{n_k}}q_{n_k+1}}\ll e^{-n_k}.
\end{equation} 
By \eqref{gjnew2}, \eqref{error}, \eqref{ff4} and \eqref{gj1}, we have that
\begin{equation}\label{errornew}
\left\|u_2(n_0)-\frac{\langle u_2(n_0),u_1(n_0)\rangle}{\|u_1(n_0)\|^2}u_1(n_0)\right\|\leq e^{-\frac{1}{4}n_k}
\end{equation}
which means the orthogonal projection of $\vec{u}_2(n_0)$ on the vector  $\vec{u}_1(n_0)$ is large.  

In the following, we consider
\begin{align*}
b(x)=\overrightarrow{u}(x)^*S_w(\overrightarrow{v}(x)-c_1\overrightarrow{u}(x))=\overrightarrow{u}^*(x)S_w\overrightarrow{v}(x)\end{align*}
where $c_1=\overline{\frac{\langle u_2(n_0),u_1(n_0)\rangle}{\|u_1(n_0)\|^2}}$, $\overrightarrow{u}(x), \overrightarrow{v}(x)$ are defined by \eqref{lanair3}, and aim to estimate $b(x)$.   First, we will show as a consequence of
\eqref{errornew}, that
$
\overrightarrow{v}(x)-c_1\overrightarrow{u}(x)
$
is small. To this end, we only need to estimate its
Fourier coefficients
$$
\hat{h}_2(n)-c_1\hat{h}_1(n)=\int_\T (h_{l,2}(x)-c_1h_{l,1}(x))e^{-2\pi inx}dx.
$$  
Note that by \eqref{ff4}, we have
\begin{equation}\label{gyz20242}
|c_1|\leq C(v,w)e^{-\frac{h}{16}q_{n_k}}q_{n_k+1}.
\end{equation}

If $|n|\leq e^{-\frac{h}{10}q_{n_k}}q_{n_k+1}$, we set
$$
\overline{\tilde{p}_n}=\frac{1}{\overline{\hat{v}_d}}\begin{pmatrix}\hat{e}_2(n)-c_1\hat{e}_1(n)\\0\\ \vdots\\ 0\end{pmatrix},
$$%
$$
\overline{\tilde{y}_n}=\begin{pmatrix}\hat{h}_2(n+d-1)-c_1\hat{h}_1(n+d-1)\\ \hat{h}_2(n+d-2)-c_1\hat{h}_1(n+d-2)\\
\vdots\\ \hat{h}_2(n-d)-c_1\hat{h}_1(n-d)\end{pmatrix}.
$$
Then as a result of \eqref{new4} and \eqref{new4-1}, we have
$$\tilde{y}_{n+1}=L^w_{E,v}(\phi_0+n\alpha)\tilde{y}_{n}+\tilde{p}_n,$$
which implies that
$$\tilde{y}_n=(L_{E,v}^w)_{n-n_0}(\phi_0+n_0\alpha)\tilde{y}_{n_0}+\sum\limits_{j=n_0+1}^n (L_{E,v}^w)_{n-j}(\phi_0+j\alpha)\tilde{p}_{j-1}.
$$

To give an estimate of $\tilde{y}_n$,  first note that  by \eqref{ee1}, \eqref{vech}, \eqref{vecu}, \eqref{errornew} and \eqref{gyz20242}, we have
\begin{eqnarray*}
 \|\tilde{p}_n]| & \leq&  e^{-\frac{1}{80}q_{n_k+1}h},\\
  \|\tilde{y}_{n_0}\|&\leq &\left\|\vec{u}_2(n_0)-\overline{c_1}\vec{u}_1(n_0)\right\|\leq e^{-\frac{1}{4}n_k}.
\end{eqnarray*}
On the other hand, we have
\begin{equation}\label{ub}
|(L_{E,v}^w)_n|_{C^0}\leq C^{|n|}, \quad \forall n\in \Z.
\end{equation}%
As a result, if $k$ is sufficiently large depending on $E,v,w$, then one can estimate
\begin{eqnarray}\label{new102}
 \nonumber \|\tilde{y}_n\|&\leq& C^{e^{-\frac{h}{10}q_{n_k}}q_{n_k+1}} e^{-\frac{1}{4}n_k}+2e^{-\frac{h}{100}q_{n_k}}q_{n_k+1}C^{e^{-\frac{h}{100}q_{n_k}}q_{n_k+1}} e^{-\frac{1}{80}q_{n_k+1}h}\\
&\leq&e^{-\frac{1}{8}n_k}.
\end{eqnarray}

By \eqref{eeee1} and \eqref{new102}, there is $K_2>0$ such that if $|k|\geq K_2$,
$$
\left|h_{l,2}-c_1 h_{l,1}\right|_{C^0}\leq e^{-\frac{h}{8}e^{-\frac{h}{2}q_{n_k}}q_{n_k+1}}.
$$
As a consequence, by \eqref{ese2}, \eqref{e33} and \eqref{gyz20242},
\begin{align*}
\left|\overrightarrow{v}-c_1\overrightarrow{u}\right|_{C^0}\leq 2l\left(\left|h_{l,2}-c_1 h_{l,1}\right|_{C^0}+ 4l^2e^{-\frac{1}{80}q_{n_k+1}h}\right)\leq  e^{-\frac{h}{16}e^{-\frac{h}{2}q_{n_k}}q_{n_k+1}}.
\end{align*}
This contradicts 
\begin{align*}
|b(x)|= |(\overrightarrow{v}(x)-c_1\overrightarrow{u}(x))^*S_w\overrightarrow{u}(x)|=1.
\end{align*}
\end{pf}

Finally by \eqref{eeee1} and \eqref{eeee22}, taking $m=[100\max\{1,h^{-1}\}n_k]$, we have
$$
|\vec{u}_j(m)|\leq |\vec{h}_j(m)|+e^{-\frac{1}{80}q_{n_k+1}h}\leq e^{-50n_k}.
$$
It follows that for $i,j=1,2$,
\begin{equation}\label{gj2}
\left|\vec{u}^*_i(n_0)S_v\vec{u}_j(n_0)\right|=\left|\vec{u}^*_i(m)S_v\vec{u}_j(m)\right|\leq e^{-50n_k}.
\end{equation}
By \eqref{ff4}, \eqref{gj1} and \eqref{gj2}, we have
\begin{equation}\label{gj3}
\left|(\vec{u}^c_i(n_0))^*S_v\vec{u}^c_j(n_0)\right|\leq Ce^{-50n_k},\ \ i,j=1,2.
\end{equation}

Recall that $(\alpha,L_{E,v}^w)$ is $PH2$ and there exist continuous invariant decompositions
$$
\C^{2d}=E_s(x)\oplus E_c(x)\oplus E_u(x),\ \ {\rm dim} E_c(x)=2.
$$
Hence there are continuous linearly independent vectors $\vec{v}_1(x), \vec{v}_2(x)\in E^c(x)$. We define
$$
\Omega(x)=\begin{pmatrix}v_1(x)& v_2(x)\end{pmatrix}\left(\begin{pmatrix}v_1(x)& v_2(x)\end{pmatrix}^*S_v \begin{pmatrix}v_1(x)& v_2(x)\end{pmatrix}\right)^{-1}\begin{pmatrix}v_1(x)& v_2(x)\end{pmatrix}^*
$$
then $\Omega(x)$ does not depend on the choice of the basis of $E^c(x)$ and 
\begin{align}\label{gjnewa}
\nonumber &\left|\det {\begin{pmatrix}v_1(x)& v_2(x)\end{pmatrix}^*
\begin{pmatrix}v_1(x)& v_2(x)\end{pmatrix}\left(\begin{pmatrix}v_1(x)& v_2(x)\end{pmatrix}^*S_v \begin{pmatrix}v_1(x)& v_2(x)\end{pmatrix}\right)^{-1}}\right|_{C^0}\\
=&\tr (\Lambda^2\Omega(x))\leq C.
\end{align}
Here we use that ${\rm Rank}(\Omega(x))=2$ so $\tr (\Lambda^2\Omega(x))$ is the product of its non-zero eigenvalues, which is equal to the determinant of the $2\times 2$ matrix defined by \eqref{gjnewa}.

On the other hand, by Proposition \ref{ff5}, $u_1^c(n_0), u_2^c(n_0)$ form a basis of $E_c(\phi_0+n_0\alpha)$, so together with \eqref{fwn} and   \eqref{gj3}, we have
\begin{align*}
&\det{\begin{pmatrix}u_1^c(n_0)& u^c_2(n_0)\end{pmatrix}^*\begin{pmatrix}u_1^c(n_0)& u^c_2(n_0)\end{pmatrix}\left(\begin{pmatrix}u_1^c(n_0)& u^c_2(n_0)\end{pmatrix}^*S_v \begin{pmatrix}u_1^c(n_0)& u^c_2(n_0)\end{pmatrix}\right)^{-1}}\\
=&\frac{\det{\begin{pmatrix}u_1^c(n_0)& u^c_2(n_0)\end{pmatrix}^*\begin{pmatrix}u_1^c(n_0)& u^c_2(n_0)\end{pmatrix}}}{\det \begin{pmatrix}u_1^c(n_0)& u^c_2(n_0)\end{pmatrix}^*S_v \begin{pmatrix}u_1^c(n_0)& u^c_2(n_0)\end{pmatrix}}\geq e^{40n_k}
\end{align*}
which contradicts  \eqref{gjnewa} if $k$ is sufficiently large. Thus \eqref{spe1} is not compatible
with $(\alpha,L_{E,v}^w)$ being $PH2$. \qed

\section{Kotani theory for {one-frequency $PH2$} operators. Proof
of Theorem \ref{L2 reducibility}}\label{kotanis1}
In this section, we assume $w\in C^\omega(\T,\R)$ and
$$
v(x)=\sum\limits_{k=-d}^d \hat{v}_k e^{2\pi ikx}.
$$
Consider the one-frequency finite-range operator $L^w_{v,\alpha,x}$  given by
\eqref{lvw} with $v,w$ given above. Recall that  the cocycle induced  is defined by the eigenequation
$L^w_{v,\alpha,x}u=Eu$ by $(\alpha,L_{E,v}^w),$ and denote 
$$
S=\begin{pmatrix}
0&-C^*\\
C&0
\end{pmatrix},\ \  C=\begin{pmatrix}
\hat{v}_d&\cdots&\hat{v}_1\\
&\ddots&\vdots\\
&&\hat{v}_d
\end{pmatrix}.
$$

For any $E\in\R$, $\delta$,  $h>0$, let 
\begin{equation}\label{ha4}
R_{\delta}(E)\times\Omega_h=\{z\in\C:\max\{|\Re z-E|,|\Im z|\}<\delta\}\times \{z\in\C:e^{-h}<|z|<e^h\}.
\end{equation} 
Let $V$ be a complex
holomorphic vector bundle over $\R_\delta(E)\times \Omega_h$ and let $\pi:V\rightarrow R_\delta(E)\times \Omega_h$ be the bundle projection. A holomorphic vector bundle $V$ of rank $r$ over $R_\delta(E)\times\Omega_h$ is called trivial if it is isomorphic to the bundle $(R_\delta(E)\times\Omega_h)\times \C^r$. This is equivalent to the existence of a global frame $v_1,\cdots,v_r$ of holomorphic sections  in $V$ over $R_\delta(E)\times\Omega_h$, such that for each $(z_1,z_2)\in R_\delta(E)\times \Omega_h$,  the elements $v_1(z_1,z_2),\cdots,v_r(z_1,z_2)\in \pi^{-1}(z_1,z_2)$ are linearly independent.
\begin{Proposition}\label{trivial}
Any holomorphic vector bundle $V$ over $R_\delta(E)\times\Omega_h$ is trivial.
\end{Proposition}
\begin{pf}
Since $R_\delta(E)$ is contractible and $\Omega_h$ is homotopy equivalent to $\mathbb{S}^1$, then $R_\delta(E)\times \Omega_h$ is homotopy equivalent to $\mathbb{S}^1$, which implies that $V$ is topologically trivial. On the other hand,  $R_\delta(E)\times \Omega_h$ is a product of Stein manifold, thus is a Stein manifold itself. It follows from Theorem 5.3.1 in \cite{stein} that $V$ is trivial.
\end{pf}
\subsection{$C^\omega$ rotations reducibility.}
We first recall an important consequence of classical Kotani theory
\cite{ak1,dcj,sim83,kot}.
Let $\alpha\in\R\backslash\Q$ and $\mathbb{H}$ be the upper half plane.
It is well known that there exists a  function
$m: \mathbb{H} \times \T
\to \mathbb{H}$ such that for the $S^w_{E}(x)$ given by \eqref{S} we have $S^w_{E}(x) \cdot m(E,x)=m(E,x+\alpha)$, thus
defining an invariant section for the Schr\"odinger cocycle
$(\alpha,S_{E}^w)$.
\begin{equation} \label {minvariantsection}
(\alpha,S_{E}^w) (x,m(E,x))=(x+\alpha,m(E,x+\alpha)).
\end{equation}
Moreover, $E \mapsto m(E,x)$ is holomorphic on $\mathbb{H}$.

Let $L(E)$ be the  Lyapunov exponent of the Schr\"odinger cocycle $(\alpha,S_E^w).$ 
\begin{Theorem}[\cite{aj}]\label{kotani}
Let $\alpha \in \R\backslash\Q$, $w\in C^\omega(\T,\R)$ and $L(E)=0$ in an open interval $J\subset\R$. Then for every $x\in\T$, the function $E\rightarrow m(E,x)$ admits a holomorphic extension to $\C\backslash(\R\backslash J)$, with values in $\mathbb{H}$. The function $m: \C\backslash(\R\backslash J)\times\T\rightarrow \mathbb{H}$ is analytic in both variables.
\end{Theorem}
Theorem \ref{kotani} played an important role in solving the almost
Mathieu ten martini problem \cite{aj}, since it implies the
$C^\omega$ rotations reducibility of the Schr\"odinger cocycle $(\alpha,S_E^w)$.  The aim
of this section is to present an analogue of such $C^\omega$ rotations reducibility
result for  finite-range operators \eqref{lvw}
which allows existence of positive Lyapunov exponents. 

Abusing the notations a little bit, we let 
$$
L_1(E)\geq\cdots\geq  L_d(E)\geq L_{d+1}(E)\geq \cdots\geq L_{2d}(E)
$$ 
be the Lyapunov exponents of the complex symplectic cocycle
$(\alpha,L_{E,v}^w)$ associated with operator $L^w_{v,\omega,x}$.  The full
matrix version of Kotani theory was proved by Kotani-Simon \cite{ks}
(see also Xu \cite{xu} for the monotonic case), assuming
$L_1(E)=\cdots=L_d(E)=0$ and all the coefficients are real.  Removing this restriction on the Lyapunov exponents was stated as
a problem in \cite{ks}, in particular for the so-called reflectionless property which is crucial to the proof of Theorem \ref{kotani}, and it has seen no serious progress until this
work. 

Here we establish $C^\omega$ rotations reducibility of the $2$-dimensional center
for analytic one-frequency  $PH2$ operators, thus solving the Kotani-Simon problem under the $PH2$
condition and give a $PH2$ analogue of Theorem \ref{kotani}. 


\begin{Theorem}\label{C0reducibility1}
Let $\alpha\in \R\backslash\Q$, $w\in C^\omega(\T,\R)$ and  $L_{v,\alpha.x}^w$ be $PH2$. Assume $L_d(E)=0$ and $\omega^{d-1}(\alpha,L_E^f)=0$ in an interval $I\subset \R.$ 
There exist   $(U_E,V_E)\in C^\omega(\T, Sp_{2d\times 2}(\C))$, $\phi_E\in C^\omega(\T,\R)$ and $R_E\in C^\omega(\T,SO(2,\R))$, depending analytically on $E$  on $I'$ for some $I'\subset I$, such that
$$
L_{E,v}^w(x)(U_E(x),V_E(x))=(U_E(x+\alpha), V_E(x+\alpha))e^{2\pi i\phi_E(x)}R_E(x).
$$
\end{Theorem}
In order to prove
Theorem \ref{C0reducibility1}, we start with some preparations.

\subsection{An extension of Johnson-Moser's theorem}\label{6.2}
For any $z\in\mathbb{H}$, the Green's function of Schr\"odinger
operator $H_{w,\alpha,x}$ is defined as
$$
g(z,x)=\langle\delta_0,(H_{w,\alpha,x}-z)^{-1}\delta_0\rangle.
$$
For Schr\"odinger operators, Johnson and Moser \cite{johonson and moser} (see also \cite{cs} for the strip case) proved the following
relation between the Lyapunov exponent and the Green's function
$$
L'_f(z)=\int_\T g(z,x)dx.
$$

Johnson-Moser's theorem plays an important role in the proof of the classical Kotani theory.  In this subsection, we extend Johnson-Moser's theorem to analytic one-frequency finite-range operators $L_{v,\alpha,x}^w$, satisfying the $PH2$ condition. Notice that the eigenvalue equations $L_{v,\alpha,x}^wu=Eu$ can be written as a second-order $2d$-dimensional difference equation by introducing the auxiliary variables
$$
\vec{u}_n=(u_{nd+d-1}\ \ \cdots\ \ u_{nd+1}\ \ u_{nd})^T \in \C^d
$$
for $n\in \Z$. By the proof of Lemma \ref{symplectic},  we have that $(\vec{u}_n)_n$ satisfies
\begin{equation}\label{ef100}
C\vec{u}_{n+1}+B(x+nd\alpha)\vec{u}_n+C^*\vec{u}_{n-1}=E\vec{u}_n,
\end{equation}
where $B(x)$ is the Hermitian matrix
$$
B(x)=\begin{pmatrix}
w(x+(d-1)\alpha)+\hat{v}_0&\hat{v}_{-1}&\cdots&\hat{v}_{-d+1}\\
\hat{v}_1&\ddots&\ddots&\vdots\\
\vdots&\ddots&w(x+\alpha)+\hat{v_0}&\hat{v}_{-1}\\
\hat{v}_{d-1}&\cdots&\hat{v}_1&w(x)+\hat{v}_0
\end{pmatrix}.
$$
Moreover, equation \eqref{ef100} is an eigenequation of the following vector-valued Schr\"odinger operator
\begin{equation}\label{ldfo}
(L_{d,v,\alpha,x}^w\vec{u})_n=C\vec{u}_{n+1}+B(x+dn\alpha)\vec{u}_n+C^*\vec{u}_{n-1},
\end{equation}
acting on $\ell^2(\Z,\C^d)$.

Let $\Sigma_{v,\alpha}^w$ be the spectrum of $L_{v,\alpha,x}^w$. By  Definition
\ref{defph2} of the 
$PH2$ property and continuity of dominated splitting \cite{bdv}, there
is $\delta_0(v,w)>0$ such that if operator $L_{v,\alpha,x}^w$ is $PH2$, then
every $E\in \mathbb{C}_{\delta_0}$ where $\mathbb{C}_{\delta_0}$ is a small
open neighborhood of $\Sigma_{v,\alpha}^w,$ is $PH2$ for $L_{v,\alpha,x}^w.$ 
It is known that for any $z\in\C_{\delta_0}\backslash\R$, the cocycle $(\alpha,L_{z,v}^w)$ is uniformly hyperbolic, thus $d$-dominated. Hence $(\alpha,L_{z,v}^w)$ is $(d-1)$, $d$, $(d+1)$-dominated. As a consequence of dominated splitting, for any $z\in\mathbb{C}_{\delta_0}\backslash\R$, there exist continuous invariant decompositions
$$
\C^{2d}=E_s(z,x)\oplus E_+(z,x)\oplus E_-(z,x)\oplus E_u(z,x),\ \ \forall x\in\T.
$$
Indeed, by Theorem 6.1 in \cite{ajs} and Proposition \ref{trivial}, there are linearly independent $\{u_z^i(x)\}_{i=1}^{d-1}\in E_s(z,x)$, $u_z^+(x)\in E_+(z,x)$, $u_z^-(x)\in E_-(z,x)$ and $\{v_z^i(x)\}_{i=1}^{d-1}\in E_u(z,x)$ depending analytically on $x$ and  $z$, such that
\begin{equation}\label{r1}
\begin{pmatrix}F_z^+(0,x)\\ F_z^+(-1,x)\end{pmatrix}:=\begin{pmatrix}u_z^1(x),\cdots, u_z^{d-1}(x),u^+_z(x)\end{pmatrix}
\end{equation}
\begin{equation}\label{r2}
\begin{pmatrix}F_z^-(0,x)\\ F_z^-(-1,x)\end{pmatrix}:=\begin{pmatrix}u^-_z(\omega), v_z^{d-1}(x), \cdots, v_z^1(x)\end{pmatrix}
\end{equation}
satisfy
$$
\sum\limits_{k=0}^\infty\|F_z^+(k,x)\|^2<\infty,\ \ \sum\limits_{k=-\infty}^{0}\|F_z^-(k,x)\|^2<\infty,
$$
where
$$
\begin{pmatrix}F_z^\pm(k,x)\\ F_z^\pm(k-1,x)\end{pmatrix}=(L_{z,v}^w)_{dk}(x)\begin{pmatrix}F_z^\pm(0,x)\\ F_z^\pm(-1,x)\end{pmatrix}.
$$
Moreover, for any $x\in\T$, we have
\begin{align}\label{eeeq1}
\limsup\limits_{k\rightarrow \infty}\frac{1}{2k}\ln\left(\|\vec{u}_z^+(k,x)\|^2+\|\vec{u}_z^+(k+1,x)\|^2\right)=dL_{d+1}(z),
\end{align}
\begin{align}\label{eeeq1'}
\limsup\limits_{k\rightarrow \infty}\frac{1}{2k}\ln\left(\|\vec{u}_z^-(k,x)\|^2+\|\vec{u}_z^-(k+1,x)\|^2\right)=dL_{d}(z).
\end{align}
where
\begin{equation}\label{u(n)}
\begin{pmatrix}\vec{u}_z^\pm(k,x)\\ \vec{u}_z^\pm(k-1,x)\end{pmatrix}=(L_{z,v}^w)_{dk}(x)u_z^\pm(x).
\end{equation}

Once we have $F_z^{\pm}(k,x)$, one can define M matrices by
\begin{equation}\label{lanair6}
M_+(z,x)=F_z^+(1,x)(F_z^{+}(0,x)^{-1},
\end{equation}
$$
M_-(z,x)=F_z^-(-1,x)(F_z^-(0,x))^{-1}.
$$
Note that $M_\pm$  satisfy  the following Ricatti equations. In the following, we set $Tx=x+\alpha$.
\begin{Lemma}For any $z\in \mathbb{C}_{\delta_0}\backslash\R$, we have
\begin{equation}\label{rica}
CM_+(z,x)+C^*M^{-1}_+(z,T^{-d}x)+(B(x)-z)=0.
\end{equation}
\begin{equation}\label{rica1}
C^*M_-(z,x)+CM^{-1}_-(z,T^{d}x)+(B(x)-z)=0.
\end{equation}
\end{Lemma}
\begin{pf}
Note that
$$
CF_z^\pm(1,x)+C^*F_z^\pm(-1,x)+(B(x)-z)F_z^\pm(0,x)=0.
$$
The results follow from the definition of $M_\pm$.
\end{pf}

For any $z\in \C_{\delta_0}\backslash\R$, we define the Green's matrix by
$$G(z,x)= \langle \vec{\delta}_0,(L_{d,v,\alpha,x}^w-z)^{-1}\vec{\delta}_0\rangle,
$$
where for $j\in\Z$,  $\vec{\delta}_j:\Z\rightarrow \C^d$ satisfies
$$
\vec{\delta}_j(n)=
\begin{cases}
0& n\neq j\\
I_d &n=j
\end{cases}.
$$
The Green's matrix $G(z,x)$ can be then expressed as:

\begin{Lemma}\label{Green_Matrix}
For any $z\in\mathbb{C}_{\delta_0}\backslash\R$, we have
\begin{align*}
G(z,x)=(CM_+(z,x)+C^*M_-(z,x)+B(x)-z)^{-1}
\end{align*}
\end{Lemma}
\begin{pf}
Let  $(L^w_{d,v,\alpha,x}\widetilde{F}_z^\pm)(k,x)=z\widetilde{F}_z^\pm(k,x)$ with $\widetilde{F}_z^+(0,x)=\widetilde{F}_z^-(0,x)=I_d$ satisfying
$$
\sum\limits_{k=0}^\infty\|\widetilde{F}_z^+(k,x)\|^2<\infty,\ \ \sum\limits_{k=-\infty}^{0}\|\widetilde{F}_z^-(k,x)\|^2<\infty,
$$ 
Since Green's matrice is unique, it follows that
\begin{align*}
&\langle \vec{\delta}_n,(L_{d,v,\alpha,x}^w-z)^{-1}\vec{\delta}_0\rangle \\
=&\begin{cases}
\widetilde{F}_z^+(n,x)(C\widetilde{F}_z^+(1,x)+C^*\widetilde{F}_z^-(-1,x)+B(x)-z)^{-1}& n\geq 0\\
\widetilde{F}_z^-(n,x)(C\widetilde{F}_z^+(1,x)+C^*\widetilde{F}_z^-(-1,x)+B(x)-z)^{-1}& n< 0
\end{cases}.
\end{align*}
The result follows from the definition of $M$ matrices.
\end{pf}

The following proposition gives the relation between $M$ matrices and the Green's matrix.
\begin{Proposition}\label{prop123} For any $z\in\mathbb{C}_{\delta_0}\backslash\R$, the following relations hold:
\begin{eqnarray}
\nonumber G(z,x)&=&(-C^*M^{-1}_+(z,T^{-d}x)+C^*M_-(z,x))^{-1}, \\
\nonumber G(z,T^{-d}x)&=& (CM_+(z,T^{-d}x)-CM^{-1}_-(z,x))^{-1},\\
\label{g3} G(z,x)C^*M^{-1}_+(z,T^{-d}x) &=& M_+(z,T^{-d}x)G(z,T^{-d}x) C - I_d.
\end{eqnarray}
\end{Proposition}
\begin{pf}
By Lemma \ref{Green_Matrix} and \eqref{rica}, one has
\begin{align*}
G(z,T^{-d}x)=(CM_+(z,T^{-d}x)-CM^{-1}_-(z,x))^{-1},
\end{align*}
\begin{align*}
G(z,x)= (-C^*M_+^{-1}(z,T^{-d}x)+C^*M_-(z,x))^{-1}.
\end{align*}
Consequently, we have the following
\begin{align*}
G(z,x)C^*M^{-1}_+(z,T^{-d}x)
=&(-I_d+M_+(z,T^{-d}x)M_-(z,x))^{-1}\\ \nonumber
=&M_-^{-1}(z,x)(-M_-^{-1}(z,x)+M_+(z.T^{-d}x))^{-1}\\ \nonumber
=&M_+(z,T^{-d}x)(-M_-^{-1}(z,x)+M_+(z.T^{-d}x))^{-1}- I_d\\ \nonumber
=&M_+(z,T^{-d}x)G(z,T^{-d}x) C - I_d.
\end{align*}
\end{pf}

We can now formulate the $PH2$ extension of the Johnson-Moser's theorem.
\begin{Theorem}\label{mg}
For any $z\in \mathbb{C}_{\delta_0}\backslash\R$, we have 
\begin{equation*}
 \frac{\partial L_{d+1}}{\partial \Im z}(z)=-\frac{1}{d}\Im\int_{\T}g(z,x) dx.
 \end{equation*}
where
$$
g(z,x):=\langle \delta_d, (F_z^+(0,x))^{-1}G(z,x)F_z^+(0,x)\delta_d\rangle.
$$
\end{Theorem}
\begin{pf}
Let $(d\alpha,L_{d,z,v}^w)$ be the cocycle corresponding to the
eigenequation $L_{d,v,\alpha,x}^wu=zu$ where $L_{d,v,\alpha,x}^w$ is defined
by \eqref{ldfo}. By invariance, there is $\tau(z,x)$ depending analytically on $x$ and  $z$ such that
$$
L_{d,z,v}^w(x)\begin{pmatrix}\vec{u}^+_z(0,x)\\ \vec{u}^+_z(-1,x)
\end{pmatrix}=\begin{pmatrix}\vec{u}^+_z(0,T^dx)\\ \vec{u}^+_z(-1,T^dx)
\end{pmatrix}\frac{1}{\tau(z,x)}
$$
where $u_z^+(n,x)$ is from \eqref{u(n)}.

By \eqref{eeeq1}, we have
\begin{align*}
 d\cdot L_{d+1}(z)=-\int_\T \ln |\tau(z,x)|dx.
\end{align*}
It suffices for us to prove
\be\label{gi}
\int_\T\frac{\partial \tau(z,x)}{\partial z}\frac{1}{\tau(z,x)}dx=-\int_{\T}g(z,x)dx,
\ee
Once we have this, then the result follows from  the Cauchy-Riemann equations.

Again by invariance and the definition of $\{F_z^+(k,x)\}_{k\in\Z}$, i.e., \eqref{r1} and \eqref{r2},  we have
\begin{align}\label{tm}
\tau(z,x)&=\langle \delta_d, (F_z^+(1,x)^{-1}F_z^+(0,T^dx)\delta_d\rangle
\end{align}
where $\{\delta_{j}\}_{j=1}^d$ is the canonical basis of $\C^{d}$.
\begin{Lemma}\label{wewe1}
We have that
\begin{align*}
\frac{\partial \tau(z,x)}{\partial z}\frac{1}{\tau(z,x)}=&\langle\delta_d,(F_z^+(0,x))^{-1}\frac{\partial  M^{-1}_+(z,x)}{\partial z}  M_+(z,x)F_z^+(0,x)\delta_d\rangle\\
&-h(z,x)+h(z,T^dx).
\end{align*}
where
$h(z,x)=\langle \delta_d, (F_z^+(0,x))^{-1}\frac{\partial  F_z^+(0,x)}{\partial z}\delta_d\rangle$.
\end{Lemma}
\begin{pf}
By invariance, \eqref{r1} and \eqref{r2},  we have  for some $U(z,x)$,
\begin{align*}
(F_z^+(1,x))^{-1}F_z^+(0,T^dx)=:\widetilde{M}^{-1}_+(z,x)={\rm diag}\{U(z,x),\tau(z,x)\}.
\end{align*}
It follows,
\begin{align*}
& (F_z^+(0,x))^{-1}\frac{\partial  M^{-1}_+(z,x)}{\partial z}  M_+(z,x)F_z^+(0,x)\\
=&(F_z^+(0,x))^{-1}\frac{\partial  F_z^+(0,x) \widetilde{M}^{-1}_+(z,x)(F_z^+(0,T^dx))^{-1}}{\partial z}  M_+(z,x)F_z^+(0,x)\\
=&(F_z^+(0,x))^{-1}\frac{\partial  F_z^+(0,x)}{\partial z} \widetilde{M}^{-1}_+(z,x)(F_z^+(0,T^dx))^{-1} M_+(z,x)F_z^+(0,x)\\
&+(F_z^+(0,x))^{-1} F_z^+(0,x) \frac{\partial\widetilde{M}^{-1}_+(z,x)}{\partial z}(F_z^+(0,T^dx))^{-1}  M_+(z,x)F_z^+(0,x)\\
&+(F_z^+(0,x))^{-1} F_z^+(0,x) \widetilde{M}^{-1}_+(z,x)\frac{\partial(F_z^+(0,T^dx))^{-1}}{\partial z} M_+(z,x)F_z^+(0,x)\\
=& \frac{\partial \widetilde M_+^{-1} (z,x)}{\partial z}\widetilde M_+(z,x) + E(z,x)
\end{align*}
where  we set
\begin{align}\label{e2}
E(z,x):=&(F_z^+(0,x))^{-1}\frac{\partial  F_z^+(0,x)}{\partial z}\\ \nonumber
&-\widetilde{M}^{-1}_+(z,x)(F_z^+(0,T^dx))^{-1}\frac{\partial F_z^+(0,T^dx)}{\partial z}\widetilde{M}_+(z,x).
\end{align}
Here we use that for any invertible matrix $A$, we have
$\frac{\partial A^{-1}}{\partial z}A=-A^{-1}\frac{\partial A}{\partial
  z}.$

The result follows since $\widetilde{M}_+(z,x)$ is a block diagonal matrix.
\end{pf}
Finally, we introduce the  auxiliary function
 $$f(z,x)=\langle \delta_d, (F_z^+(0,x)^{-1}G(z,x)\frac{\partial CM_+(z,x)}{ \partial z}F_z^+(0,x)\delta_d\rangle.
$$
\begin{Lemma}
 We have that
\begin{align}\label{ft}
&\frac{\partial \tau(z,T^{-d}x)}{\partial z}\frac{1}{\tau(z,T^{-d}x)} +g(z,x)\\ \nonumber
=&f(z,x)-f(z,T^{-d}x)+h(z,x)-h(z,T^{-d}x).
\end{align}
\end{Lemma}
\begin{pf}
By \eqref{rica}, we have
\begin{align*}
&\frac{\partial CM_+(z,x)}{ \partial z}=-\frac{\partial C^*M^{-1}_+(z,T^{-d}x)}{ \partial z}+I_d\\
=&C^*M^{-1}_+(z,T^{-d}x)\frac{\partial M_+(z,T^{-d}x)}{\partial z}M^{-1}_+(z,T^{-d}x)+I_d.
\end{align*}
Then by Proposition \ref{prop123},
\begin{align*}
&(F_z^+(0,x))^{-1}G(z,x)\frac{\partial CM_+(z,x)}{ \partial z}F_z^+(0,x)\\
=&(F_z^+(0,x))^{-1}G(z,x)\left(C^*M^{-1}_+(z,T^{-d}x)\frac{\partial M_+(z,T^{-d}x)}{\partial z}M^{-1}_+(z,T^{-d}x)+I_d\right)F_z^+(0,x)\\
=&(F_z^+(0,x))^{-1}G(z,x)C^*M^{-1}_+(z,T^{-d}x)\frac{\partial M_+(z,T^{-d}x)}{\partial z}M^{-1}_+(z,T^{-d}x)F_z^+(0,x)\\
&+(F_z^+(0,x))^{-1}G(z,x)F_z^+(0,x) \\
=&(F_z^+(0,x))^{-1}\left(M_+(z,T^{-d}x)G(z,T^{-d}x) C - I_d\right)\frac{\partial M_+(z,T^{-d}x)}{\partial z}M^{-1}_+(z,T^{-d}x)F_z^+(0,x)\\
&+(F_z^+(0,x))^{-1}G(z,x)F_z^+(0,x) \\
=&\widetilde{M}_+(z,T^{-d}x)(F_z^+(0,T^{-d}x))^{-1}G(z,T^{-d}x)\frac{\partial CM_+(z,T^{-d}x)}{ \partial z}F_z^+(0,T^{-d}x)\widetilde{M}^{-1}_+(z,T^{-d}x)\\
&+\widetilde{M}_+(z,T^{-d}x)(F_z^+(0,T^{-d}x))^{-1}\frac{\partial  M^{-1}_+(z,T^{-d}x)}{\partial z}  M_+(z,T^{-d}x)F_z^+(0,T^{-d}x)\widetilde{M}^{-1}_+(z,T^{-d}x)\\
&+(F_z^+(0,x))^{-1}G(z,x)F_z^+(0,x).
\end{align*}
By Lemma \ref{wewe1},
$$
f(z,x)-f(z,T^{-d}x)=\frac{\partial \tau(z,T^{-d}x)}{\partial z}\frac{1}{\tau(z,T^{-d}x)}+h(z,T^{-d}x)-h(z,x)+g(z,x).
$$
\end{pf}
Finally, the integral of both sides of \eqref{ft} over $\T$
\begin{align*}
\int_\T\frac{\partial \tau(z,T^{-d}x)}{\partial z}\frac{1}{\tau(z,T^{-d}x)}dx =-\int_\Omega g(z,x)dx
\end{align*}
leads to \eqref{gi} and thus we get the desired result.
\end{pf}

\subsection{Kotani theoretic estimates} By \eqref{eeeq1} and \eqref{eeeq1'}, for any $z\in
\C_{\delta_0}\backslash\R$, there are non-zero solutions to
$L_{v,\alpha,x}^wu=zu$,  $(\vec{u}_z^\pm(n,x))_{n\in\Z}$ that are $\ell^2$ at $\pm \infty$.\begin{Lemma}\label{le1}
For any $z\in \mathbb{H}_{\delta_0}=\C_{\delta_0}\cap \mathbb{H}$, we have
\begin{align}\label{L1}
&\int_\T \ln\left(1-\frac{\Im z\|\vec{u}^+_z(0,x)\|^2}{\Im \left(\vec{u}^+_z(0,x)\right)^*C\vec{u}^+_z(1,x)}\right)dx=-2dL_{d+1}(z),
\end{align}
\begin{align}\label{L2}
&\int_\T \ln\left(1+\frac{\Im z\|\vec{u}^-_{\bar{z}}(0,x)\|^2}{\Im \left(\vec{u}_{\bar{z}}^-(0,x)\right)^*C^*\vec{u}_{\bar{z}}^-(-1,x)}\right)dx=2dL_{d}(\bar{z}).
\end{align}
\end{Lemma}
\begin{pf}
We only prove \eqref{L1}, since \eqref{L2} follows similarly. Notice that $\{\vec{u}_z^+(n,x)\}_{n\in\Z}$ is the solution of the following equation
\begin{equation}\label{efsch}
C^*\vec{u}(n-1,x)+C\vec{u}(n+1,x)+B(T^{dn}x)\vec{u}(n,x)=z\vec{u}(n,x).
\end{equation}
It follows from \eqref{efsch} that
\begin{equation}\label{w1}
\left(\vec{u}_z^+(0,x)\right)^*C^*\vec{u}_z^+(-1,x)+\left(\vec{u}_z^+(0,x)\right)^*C\vec{u}_z^+(1,x)+\left(\vec{u}_z^+(0,x)\right)^*(B(x)-z)\vec{u}_z^+(0,x)=0.
\end{equation}
Let
\begin{equation}\label{w2}
c_+(x)=\left|\left(\vec{u}_z^+(0,x)\right)^*C\vec{u}_z^+(0,T^{d}x)\right|,
\end{equation}
\begin{equation}\label{ww3}
m_+(x)=-\frac{\left(\vec{u}_z^+(0,x)\right)^*C\vec{u}_z^+(1,x)}{\left|\left(\vec{u}_z^+(0,x)\right)^*C\vec{u}_z^+(0,T^{d}x)\right|},
\end{equation}
Multiplying both sides of equation \eqref{efsch} by $\vec{u}^*(n,x)$, taking the imaginary part and summing all the terms of each side, we get
\begin{equation*}
\Im \vec{u}^*(0,x)C\vec{u}(1,x)=-\Im z\sum\limits_{n=1}^{\infty}\|\vec{u}(n,x)\|^2,
\end{equation*}
Thus for $\{\vec{u}_z^{+}(k,x)\}_{k\in\Z}$, we have
\begin{equation}\label{+}
\Im \left(\vec{u}_z^{+}(0,x)\right)^*C\vec{u}_z^+(1,x)=-\Im z\sum\limits_{n=1}^{\infty}\|\vec{u}_z^+(n,x)\|^2,
\end{equation}
It follows from \eqref{ww3} and \eqref{+} that
$$
\Im m_+(x)>0.
$$

On the other hand, by the coinvariance, there is $\tau_+(x)$ such that
$$
\vec{u}_z^+(1,x)=\vec{u}_z^+(0,T^{d}x)\tau_+(x),
$$
which means that
\begin{equation}\label{w4}
m_+(x)=-\tau_+(x)\frac{\left(\vec{u}_z^+(0,x)\right)^*C\vec{u}_z^+(0,T^dx)}{\left|\left(\vec{u}_z^+(0,x)\right)^*C\vec{u}_z^+(0,T^dx)\right|},
\end{equation}
By \eqref{w1}, \eqref{w2}, \eqref{ww3} and \eqref{w4}, we have
\begin{equation*}
-c_+(T^{-d}x)m_+^{-1}(T^{-d}x)-c_+(x)m_+(x)+\left(\vec{u}_z^+(0,x)\right)^*(B(x)-z)\vec{u}_z^+(0,x)=0,
\end{equation*}
Taking the imaginary part, one has
\begin{equation}\label{equal}
-c_+(T^{-d}x)\frac{\Im m_+(T^{-d}x)}{|m_+(T^{-d}x)|^2}+c_+(x)\Im m_+(x)+\Im z\|u_z^+(0,x)\|^2=0.
\end{equation}
Thus
\begin{align*}
&\ln\left(1+\frac{\Im z \|u_z^+(0,x)\|^2}{c_+(x)\Im m_+(x)}\right)=\ln c_+(T^{-d}x)-\ln c_+(x)\\
&+\ln \Im m_+(T^{-d}x)-\ln \Im m_+(x)-2\ln |m_+(T^{-d}x)|,
\end{align*}
and it follows using \eqref{w2}. \eqref{ww3} and \eqref{w4} that
\begin{align}\label{w7}
\int_\T\ln\left(1-\frac{\Im z \|u_z^+(0,x)\|^2}{\Im \left(\vec{u}^+_z(0,x)\right)^*C\vec{u}^+_z(1,x)}\right)dx&=-2\int_\T \ln |m_+(T^{-d}x)|dx\\ \nonumber
=&-2\int_\T \ln |\tau_+(x)|dx.
\end{align}
Finally, by dominated splitting and the definition of $\{u_z^+(n,x)\}_{n\in\Z}$,  we have for almost every $x$,
$$
\lim\limits_{n\rightarrow\infty}\frac{1}{n}\ln \frac{\|u_z^+(n,x)\|}{\|u_z^+(0,x)\|}=dL_{d+1}(z),
$$
By the invariance,
$$
\ln \frac{\|u_z^+(n,x)\|}{\|u_z^+(0,x)\|}=\sum\limits_{m=0}^{n-1}\ln|\tau_+(T^{md}x)|,
$$
and hence Birkhoff's ergodic theorem implies,
\begin{equation}\label{w8}
\int_\T \ln|\tau_+(x)|dx=dL_{d+1}(z).
\end{equation}
\eqref{w7} and \eqref{w8} completes the proof.
\end{pf}
\begin{Lemma}\label{le2}
For any $z\in \mathbb{H}_{\delta_0}$, we have that
\begin{align*}
&\int_\T \frac{1}{-\Im \frac{\left(\vec{u}^+_z(0,x)\right)^*C\vec{u}^+_z(1,x)}{\|\vec{u}^+_z(0,x)\|^2}+\frac{1}{2}\Im z}dx\leq -\frac{2dL_{d+1}(z)}{\Im z},
\end{align*}
\begin{align*}
&\int_\T \frac{1}{\Im \frac{\left(\vec{u}^+_z(0,x)\right)^*C^*\vec{u}^+_z(-1,x)}{\|\vec{u}^+_z(0,x)\|^2}-\frac{1}{2}\Im z}dx\leq -\frac{2dL_{d+1}(z)}{\Im z},
\end{align*}
\begin{align*}
&\int_\T \frac{1}{\Im \frac{\left(\vec{u}^-_{\bar{z}}(0,x)\right)^*C^*\vec{u}^-_{\bar{z}}(-1,x)}{\|\vec{u}^-_{\bar{z}}(0,x)\|^2}+\frac{1}{2}\Im z}dx\leq \frac{2dL_d(\bar{z})}{\Im z},
\end{align*}
\begin{align*}
&\int_\T \frac{1}{-\Im \frac{\left(\vec{u}^-_{\bar{z}}(0,x)\right)^*C\vec{u}^-_{\bar{z}}(1,x)}{\|\vec{u}^-_{\bar{z}}(0,x)\|^2}-\frac{1}{2}\Im z}dx\leq \frac{2dL_d(\bar{z})}{\Im z}.
\end{align*}
\end{Lemma}
\begin{pf}
For $x\geq 0$, consider the function
$$
A(x)=\ln (1+x)-\frac{x}{1+\frac{x}{2}}.
$$
Clearly,
$$
A(0)=0,\ \ A'(x)=\frac{1}{1+x}-\frac{1}{1+x+\frac{x^2}{4}}\geq 0.
$$
Hence
\begin{equation}\label{ff3new}
\ln(1+x)\geq \frac{x}{1+\frac{x}{2}},\  \ \forall x\geq 0,
\end{equation}
By \eqref{ff3new} and Lemma \ref{le1}, we have
\begin{align*}
\int_\T \frac{1}{-\Im \frac{\left(\vec{u}^+_z(0,\omega)\right)^*C\vec{u}^+_z(1,\omega)}{\|\vec{u}^+_z(0,\omega)\|^2}+\frac{1}{2}\Im z}dx&=\frac{1}{\Im z}\int_\T \frac{-\frac{\Im z}{\Im \frac{\left(\vec{u}^+_z(0,\omega)\right)^*C\vec{u}^+_z(1,\omega)}{\|\vec{u}^+_z(0,\omega)\|^2}}}{1+\frac{\Im z}{-2\Im \frac{\left(\vec{u}^+_z(0,\omega)\right)^*C\vec{u}^+_z(1,\omega)}{\|\vec{u}^+_z(0,\omega)\|^2}}}dx\\
&\leq \frac{1}{\Im z}\int_\T \ln\left(1+\frac{-\Im z\|\vec{u}^+_z(0,\omega)\|^2}{\Im \left(\vec{u}^+_z(0,\omega)\right)^*C\vec{u}^+_z(1,\omega)}\right)dx\\
&=-\frac{2dL_{d+1}(z)}{\Im z}.
\end{align*}
Note that 
$$
C^*\vec{u}^\pm_z(-1,x)+C\vec{u}^\pm_z(1,x)+B(x)\vec{u}^\pm_z(0,x)=z\vec{u}^\pm_z(0,x),
$$
hence multiplying as left by $(\vec{u}^\pm_z(0,x))^*$ and taking imaginary part of both sides, we have
$$
-\Im \frac{\left(\vec{u}^+_z(0,x)\right)^*C\vec{u}^+_z(1,x)}{\|\vec{u}^+_z(0,x)\|^2}+\frac{1}{2}\Im z=\Im \frac{\left(\vec{u}^+_z(0,x)\right)^*C^*\vec{u}^+_z(-1,x)}{\|\vec{u}^+_z(0,x)\|^2}-\frac{1}{2}\Im z,
$$
$$
-\Im \frac{\left(\vec{u}^-_{\bar{z}}(0,x)\right)^*C\vec{u}^-_{\bar{z}}(1,x)}{\|\vec{u}^-_{\bar{z}}(0,x)\|^2}-\frac{1}{2}\Im z=\Im \frac{\left(\vec{u}^-_{\bar{z}}(0,x)\right)^*C\vec{u}^-_{\bar{z}}(-1,x)}{\|\vec{u}^-_{\bar{z}}(0,x)\|^2}+\frac{1}{2}\Im  z.
$$
The proof for the other inequalities  follows in exactly the same way.
\end{pf}

Recall that by the $PH2$ property, for any $z\in\C_{\delta_0}\cap \R$, there exist continuous invariant decompositions
\begin{equation}\label{eqnew1}
\C^{2d}=E_s(z,x)\oplus E_c(z,x)\oplus E_u(z,x),\ \ \forall x\in\T,
\end{equation}
where $E_c(z,x)$ is the two dimensional center and moreover again by Theorem 6.1 in \cite{ajs}, $E_s(z,x)$, $E_c(z,x)$, and $E_u(z,x)$ depend analytically on $x$ and  $z$.  

\begin{Proposition}\label{newwe'}
For any fixed $E_0\in \Sigma_{v,\alpha}^w$, there exist $0<\delta_1<\delta_0$ and two linearly independent vectors $u_E(x), v_E(x)\in E_c(x)$ for  $(E,x)\in (E_0-\delta_1,E_0+\delta_1)\times \T$, depending analytically on $E$ and  $x$, such that 
\begin{equation}\label{gjy2n}
\begin{pmatrix}
u^*_{E-i\delta}(x)\\ v^*_{E-i\delta}(x)
\end{pmatrix}S\begin{pmatrix}
u_{E+i\delta}(x)& v_{E+i\delta}(x)\end{pmatrix}=J,
\end{equation}
\begin{equation}\label{gjy3n}
\sup_{x\in\T}\left\|\begin{pmatrix}
u^*_{E+i\delta}(x)\\ v^*_{E+i\delta}(x)
\end{pmatrix}S\begin{pmatrix}
u_{E+i\delta}(x)& v_{E+i\delta}(x)\end{pmatrix}-J\right\|\leq C(|E-E_0||\delta|+\delta^2),
\end{equation}
for $|\delta|<\delta_1$, where $C$ only depends on $E_0$.
\end{Proposition}
\begin{pf}
Let $\tilde{u}_E(x)$ abd $\tilde{v}_E(x)$ be from Proposition \ref{newwe}. For any fixed $E_0\in \Sigma_{v,\alpha}^w$ and $|z-E_0|\leq \delta(E_0)<\delta_0$, let $z=E+i\delta$, we have
\begin{equation}\label{cw1}
\tilde{u}_{z}(x)=\tilde{u}_{E}(x)+\frac{\partial \tilde{u}_{E}(x)}{\partial E}i\delta+O(\delta^2),
\end{equation}
\begin{equation}\label{cw2}
\tilde{v}_{z}(x)=\tilde{v}_{E}(x)+\frac{\partial \tilde{v}_{E}(x)}{\partial E}i\delta+O(\delta^2).
\end{equation}
We denote 
$$
A_{E}(x)=\begin{pmatrix}\tilde{u}^*_{E}(x)\\ \tilde{v}^*_{E}(x)\end{pmatrix}S\begin{pmatrix}\frac{\partial \tilde{u}_{E}(x)}{\partial E}&\frac{\partial \tilde{v}_{E}(x)}{\partial E}\end{pmatrix}.
$$
By Proposition \ref{newwe}, we have
\begin{equation}\label{ha1}
\begin{pmatrix}\tilde{u}^*_{E}(x)\\ \tilde{v}^*_{E}(x)\end{pmatrix}S\begin{pmatrix}\tilde{u}_{E}(x)&\tilde{v}_{E}(x)\end{pmatrix}=J.
\end{equation}
Thus $A_{E}(x)=A_{E}^*(x)$. It follows from \eqref{cw1} and \eqref{cw2} that
\begin{align}\label{CW3}
\begin{pmatrix}
\tilde{u}^*_{z}(x)\\ \tilde{v}^*_{z}(x)
\end{pmatrix}S\begin{pmatrix}
\tilde{u}_{z}(x)& \tilde{v}_{z}(x)\end{pmatrix}=J+2A_{E}(x)i\delta+O(\delta^2).
\end{align}

We let
\begin{align}\label{CE1}
C_{z}(x):=I+JA_{E_0}(x)(z-E_0)
\end{align} 
\begin{equation}\label{CE2}
\begin{pmatrix}
u'_{z}(x)& v'_{z}(x)\end{pmatrix}=\begin{pmatrix}
\tilde{u}_{z}(x)& \tilde{v}_{z}(x)\end{pmatrix}C_{z}(x).
\end{equation}

By \eqref{ha1}, \eqref{CW3}, \eqref{CE1} and \eqref{CE2}, we have
\begin{equation}\label{ha2}
\begin{pmatrix}
u'^*_{\bar{z}}(x)\\ v'^*_{\bar{z}}(x)
\end{pmatrix}S\begin{pmatrix}
u'_{z}(x)& v'_{z}(x)\end{pmatrix}=J+A_{E_0}(x)JA_{E_0}(x)(z-E_0)^2,
\end{equation}
\begin{equation}\label{ha3}
\begin{pmatrix}
u'^*_{z}(x)\\ v'^*_{z}(x)
\end{pmatrix}S\begin{pmatrix}
u'_{z}(x)& v'_{z}(x)\end{pmatrix}=J+A_{E_0}(x)JA_{E_0}(x)(z-E_0)^2+O(|E-E_0||\delta|+\delta^2).
\end{equation}

By Lemma \ref{dominated matrix}, if  $|z-E_0|<\delta(E_0)$ is sufficiently small, there exists $D_z(x)\in GL(2,\C)$, depending analytically in $z$ and $x$, such that
\begin{equation}\label{CW33}
D_{\bar{z}}(x)^*(J+A_{E_0}(x)JA_{E_0}(x)(z-E_0)^2)D_{z}(x)=J
\end{equation}  
with 
$$
\sup_{x\in\T}\left\|\frac{\partial D_z(x)}{\partial z}\right\|\leq C_1(E_0)|z-E_0|.
$$ 
It follows that
\begin{align}\label{CE11}
\nonumber &\|D_{z}(x)^*(J+A_{E_0}(x)JA_{E_0}(x)(z-E_0)^2)D_{z}(x)-J\|\leq C_2(E_0)\sup_{x\in\T}\|D_z(x)-D_{\bar{z}}(x)\|\\
&\leq C_3(E_0)|z-E_0||\delta|\leq C_3(E_0)(|E-E_0||\delta|+\delta^2).
\end{align}

Finally, let
\begin{equation}\label{CE22}
\begin{pmatrix}
u_{z}(x)& v_{z}(x)\end{pmatrix}=\begin{pmatrix}
u'_{z}(x)& v'_{z}(x)\end{pmatrix}D_{z}(x).
\end{equation}
We have that $u_E(x), v_E(x)$ are  linearly independent, depending analytically on $E$ and $x$ on $(E_0-\delta_1,E_0+\delta_1)\times \T$ for any $0<\delta_1<\delta(E_0)$.

By \eqref{ha2}-\eqref{CE22}, we have
$$
\begin{pmatrix}
u^*_{\bar{z}}(x)\\ v^*_{\bar{z}}(x)
\end{pmatrix}S\begin{pmatrix}
u_{z}(x)& v_{z}(x)\end{pmatrix}=J,
$$
$$
\begin{pmatrix}
u^*_{z}(x)\\ v^*_{z}(x)
\end{pmatrix}S\begin{pmatrix}
u_{z}(x)& v_{z}(x)\end{pmatrix}=J+O(|E-E_0|\delta+\delta^2)
$$
for $z\in R_{\delta_1}(E_0)\times \T$ where $R_{\delta_2}(E_0)$ is defined in \eqref{ha4}.
\end{pf}

We are now ready to define the finite-range analogue of the
$m$-function.   Note that for any $z\in R_{\delta_1}(E_0)\backslash\R$, by the $PH2$ property,   $(\alpha,L_{z,v}^w)$ is $(d-1)$, $d$, $(d+1)$-dominated.  Hence, there exist continuous invariant decompositions
$$
\C^{2d}=E_s(z,x)\oplus E_+(z,x)\oplus E_-(z,x)\oplus E_u(x,x),\ \ \forall x\in\T.
$$
By Theorem 6.1 in \cite{ajs} and Proposition \ref{trivial}, there are linearly independent vectors $u_z^+(x)\in E_+(z,x)$, $u_z^-(x)\in E_-(z,x)$, depending analytically on $x$ and  $z$. It follows that there are $b^\pm_1(z,x)$ and $b^\pm_0(z,x)$, depending analytically on $x$ and  $z$ such that
\begin{align}\label{ff1}
u_z^+(x)=b^+_1(z,x)u_z(x)+b^+_0(z,x)v_z(x),
\end{align}
\begin{align}\label{ff2}
u_z^-(x)=b^-_1(z,x)u_z(x)+b^-_0(z,x)v_z(x)
\end{align}
where  $u_z(x)$ and $v_z(x)$ is a basis for $E_c(z,x)=E_+(z,x)\oplus E_-(z,x)$ satisfying \eqref{gjy2n} and \eqref{gjy3n}.

As a consequence, we have that
\begin{Lemma}\label{nonvanishing}
There exists $0<\delta_2(E_0)<\delta_1(E_0)$ such that for any $z\in R_{\delta_2}(E_0)\cap \mathbb{H}$ and $x\in\T$, we have $\Im \overline{b^+_0(z,x)}b^+_1(z,x)>0$ and $\Im \overline{b^-_0(z,x)}b^-_1(z,x)<0$.
\end{Lemma}
\begin{pf}
We only give the proof of the $+$ case. By  conjugating \eqref{+}, for any $z=E+i\delta \in R_{\delta_1}(E_0)\cap \mathbb{H}$, we have that
\begin{equation}\label{20242gjy}
\Im \left(\vec{u}_z^{+}(0,x)\right)^*C^*\vec{u}_z^+(-1,x)=\delta\sum\limits_{n=0}^{\infty}\|\vec{u}_z^+(n,x)\|^2,
\end{equation}
where 
\begin{equation*}
\begin{pmatrix}\vec{u}_z^+(k,x)\\ \vec{u}_z^+(k-1,x)\end{pmatrix}=(L_{z,v}^w)_{dk}(x)u_z^+(x).
\end{equation*}

On the other hand, by \eqref{ff1} and Proposition \ref{newwe'}, we have
\begin{align}\label{20242gjy1}
&-2i\Im \left(\vec{u}_z^{+}(0,x)\right)^*C^*\vec{u}_z^+(-1,x)= \begin{pmatrix}
\vec{u}_z^{+}(0,x)\\ \vec{u}_z^{+}(-1,x)
\end{pmatrix}^*S\begin{pmatrix}
\vec{u}_z^{+}(0,x)\\ \vec{u}_z^{+}(-1,x)
\end{pmatrix}\\ \nonumber
=&\begin{pmatrix}b_1^+(z,x)\\ b_0^+(z,x)\end{pmatrix}^*\begin{pmatrix}
u^*_{z}(x)\\ v^*_{z}(x)
\end{pmatrix}S\begin{pmatrix}
u_{z}(x)& v_{z}(x)\end{pmatrix}\begin{pmatrix}b_1^+(z,x)\\ b_0^+(z,x)\end{pmatrix}\\ \nonumber
=&-2i\Im \overline{b^+_0(z,x)}b^+_1(z,x)+O\left((|E-E_0||\delta|+\delta^2)(|b_0^+(z,x)|^2+|b_1^+(z,x)|^2)\right).
\end{align}

By \eqref{20242gjy}, we have
\begin{align}\label{202512gjy}
&\Im \left(\vec{u}_z^{+}(0,x)\right)^*C^*\vec{u}_z^+(-1,x)\geq c\delta ( \|u_z(x)\|^2|b^+_0(z,x)|^2+\|v_z(x)\|^2|b^+_1(z.x)|^2)
\end{align}
where $c$ only depend on $E_0$ \footnote{Note that $c$ depends on the uniform lower bound of the angle between $u_z(x)$ and $v_z(x)$.}.

By \eqref{20242gjy1} and \eqref{202512gjy}, there is $\delta_2(E_0)$ such that if $|z-E_0|<\delta_2$, then
$$
\left|O\left((|E-E_0||\delta|+\delta^2)(|b_0^+(z,x)|^2+|b_1^+(z,x)|^2)\right)\right|\leq  \frac{1}{100}\Im \left(\vec{u}_z^{+}(0,x)^*C^*\vec{u}_z^+(-1,x)\right)
$$
which together with \eqref{20242gjy} implies 
$$
\Im \overline{b^+_0(z,x)}b^+_1(z,x)>0.
$$
We thus get the desired result.
\end{pf}

\begin{Definition}[m functions]\label{m+def}
{\rm For any $z\in R_{\delta_2}(E_0)\backslash\R$, we define
$$
m_+(z,x)=\frac{b^+_1(z,x)}{b_0^+(z,x)},\ \ m_-(z,x)=\frac{b^-_1(z,x)}{b_0^-(z,x)}.
$$
}
\end{Definition}
By Lemma \ref{nonvanishing}, $m_\pm(z,x)$ is well-defined and depend analytically on $z$ and $x$ on $\left(R_{\delta_2}(E_0)\backslash\R\right)\times \T$.

The following $PH2$ analogue of Kotani theory is key to this work.
\begin{Theorem}[Kotani-theoretic estimates]\label{keyth}
For any fixed $E_0\in \Sigma_{v,\alpha}^w$, we have that
\begin{enumerate}
\item $m_\pm(E\pm i0,x):=\lim_{\delta\rightarrow 0^{\pm}}m_\pm(E+i\delta,x)$ exist for almost every $E\in (E_0-\delta_2,E_0+\delta_2)$ and almost every $x$.
\item For almost every $E\in (E_0-\delta_2,E_0+\delta_2)\cap \{E:L_d(E)=0\}$, we have
\begin{align*}\int_\T \frac{\left\|m_+(E+i0,x)u_E(x)+v_E(x\right\|^2}{\Im m_+(E+i0,x)}dx<\infty,
\end{align*}
\begin{align*}\int_\T \frac{\left\|m_-(E-i0,x)u_E(x)+v_E(x)\right\|^2}{\Im m_-(E-i0,x)}dx<\infty.
\end{align*}
\item For almost every $x$, we have
\begin{align*}
m_+(E+i0,x)=m_-(E-i0,x),
\end{align*}
for almost every $E\in (E_0-\delta_2,E_0+\delta_2)\cap \{E:L_d(E)=0\}$.
\end{enumerate}
\end{Theorem}
\begin{pf}
For (1) and (2), we only give the proof of the results for
$m_+(z,x)$ since the proof for  $m_-(z,x)$ is exactly the
same. By Lemma \ref{nonvanishing} and Definition \ref{m+def}, we have $m_+(\cdot,x):\in R_{\delta_2}(E_0)\cap \mathbb{H}\rightarrow \mathbb{H}$ for every $x\in\T$, thus $m_+(z,x)$ is a Herglotz function (up to a conformal map), thus $m_+(E+i0,x)$ exists for almost every $E\in (E_0-\delta_0,E_0+\delta_2)$ and almost every $x\in\T$.

For part (2), let $\vec{u'}_z^{+}(k,x)=m_+(z,x)\vec{u}_z(k,x)+\vec{v}_z(k,x)$. It follows that
\begin{align}\label{aaa1}
-2i\Im (\vec{u'}_z^{+}(0,x))^*C^*\vec{u’}_z^+(-1,x)=\left(m_+(z,x)u_z(x)+v_z(x)\right)^*S\left(m_+(z,x)u_z(x)+v_z(x)\right).
\end{align}
Hence by Proposition \ref{newwe'}, we have
\begin{equation}\label{lana1}
\lim\limits_{\Im z\rightarrow 0}\Im (\vec{u'}_z^{+}(0,x))^*C^*\vec{u’}_z^+(-1,x)=\Im m_+(E,x).
\end{equation}
Note that
$$
\lim_{\delta\rightarrow 0^+}\frac{-L_{d+1}(z)}{\Im z}=\lim_{\delta\rightarrow 0^+}-\frac{\partial L_{d+1}z)}{\partial \Im z}<\infty
$$
for almost every $E\in \{E:L_d(E)=0\}$. Thus by applying Lemma \ref{le2} to $\vec{u'}_z^{+}(x)$, \eqref{lana1} and Fatou's lemma, for almost every $E$, we have
\begin{align}\label{ff2}
&\int_\T \frac{\|m_+(E+i0,x)\vec{u}_E(0,x)+\vec{v}_E(0,x)\|^2}{\Im m_+(E+i0,x)}dx<\infty.
\end{align}
By invariance, we can write
\begin{equation}\label{eqqqq1}
\vec{u'}_E^+(0,x)=\vec{u'}_E^+(-1,T^{d}x)\tau'_+(x).
\end{equation}
By \eqref{aaa1} and \eqref{eqqqq1}, we have
\begin{equation}\label{aaa2}
\Im m_+(E+i0,x)|\tau'_+(T^{-d}x)|^2=\Im m_+(E+i0,T^{-d}x).
\end{equation}
By \eqref{eqqqq1} and \eqref{aaa2}, for almost every $E$, we have
\begin{align}\label{ff3}
&\int_\T \frac{\|m_+(E+i0,x)\vec{u}_E(-1,x)+\vec{v}_E(-1,x)\|^2}{\Im m_+(E+i0,x)}dx\\  \nonumber
=&\int_\T \frac{\|m_+(E+i0,x)\vec{u}_E(-1,x)+\vec{v}_E(-1,x)\|^2|\tau'_+(T^{-d}x)|^2}{\Im m_+(E+i0,T^{-d}x)}dx\\ \nonumber
=&\int_\T \frac{\|m_+(E+i0,T^dx)\vec{u}_E(-1,T^dx)+\vec{v}_E(-1,T^dx)\|^2|\tau'_+(x)|^2}{\Im m_+(E+i0,x)}dx\\  \nonumber
=&\int_\T \frac{\|m_+(E+i0,x)\vec{u}_E(0,x)+\vec{v}_E(0,x)\|^2}{\Im m_+(E+i0,x)}dx
\end{align}
By \eqref{ff2} and \eqref{ff3}, we have
$$\int_\T \frac{\left\|m_+(E+i0,x)u_E(x)+v_E(x)\right\|^2}{\Im m_+(E+i0,x)}dx<\infty,
$$
completing the proof of (2).

Finally, we prove (3). We need the following more precise characterization of $g(z,x)$ in Theorem \ref{mg},
\begin{Lemma}\label{final1}
For any $z\in R_{\delta_2}(E_0)\cap \mathbb{H}$,
\begin{align*}
\frac{\partial L_{d+1}(z)}{\partial\Im z}=-\frac{1}{d}\int_\T\Im \frac{\vec{u}_{\bar{z}}^-(0,x)^*\vec{u}^+_z(0,x)}{u_{\bar{z}}^-(x)^*Su^+_z(x)}dx.\end{align*}
\end{Lemma}
\begin{pf}
As before, for any $z\in R_{\delta_2}(E_0)\cap \mathbb{H}$, there exist analytic invariant decompositions
$$
\C^{2d}=E_s(z,x)\oplus E_+(z,x)\oplus E_-(z,x)\oplus E_u(z,x),\ \ \forall x\in\T,
$$
which by Proposition \ref{trivial} implies that there are $\{u_z^i(x)\}_{i=1}^{d-1}\in E_s(z,x)$, $u_z^\pm(x)\in E_\pm(z,x)$ and $\{v_z^i(x)\}_{i=1}^{d-1}\in E_u(z,x)$ depending analytically on $x$ and  $z$, such that
\begin{align}\label{r1'}
\widetilde{F}_z^+(x)=\begin{pmatrix}u_z^1(x)&\cdots& u_z^{d-1}(x)&u_z^+(x)\end{pmatrix}
\end{align}
\begin{align}\label{r2'}
\widetilde{F}_z^-(x)=\begin{pmatrix}u^-_z(x)&v_z^{d-1}(x)&\cdots&v_z^{1}(x)\end{pmatrix}
\end{align}
satisfy
$$
\sum\limits_{k=0}^\infty\|\widetilde{F}_z^+(k,x)\|^2<\infty,\ \ \sum\limits_{k=-\infty}^{0}\|\widetilde{F}_z^-(k,x\|^2<\infty,
$$
where
\begin{align*}
\widetilde{F}_z^+(k,x)=\begin{pmatrix}\vec{u}_z^1(k,x)&\cdots&\vec{u}_z^{d-1}(k,x)&\vec{u}^+_z(k,x)\\ \vec{u}_z^1(k-1,x)&\cdots&\vec{u}_z^{d-1}(k-1,x)&\vec{u}^+_z(k-1,x)\end{pmatrix}=(L_{z,v}^w)_{dk}(x)\widetilde{F}_z^+(x),
\end{align*}
\begin{align*}
\widetilde{F}_z^-(k,x)=\begin{pmatrix}\vec{u}_z^-(k,x)&\vec{v}_z^{d-1}(k,x)&\cdots&\vec{u}^1_z(k,x)\\ \vec{u}_z^-(k-1,x)&\vec{v}_z^{d-1}(k-1,x)&\cdots&\vec{v}^1_z(k-1,x)\end{pmatrix}=(L_{z,v}^w)_{dk}(x)\widetilde{F}_z^-(x).
\end{align*}
By Theorem \ref{mg}, we have
\begin{align*}
\frac{\partial L_{d+1}(z)}{\partial \Im z}=-\frac{1}{d}\int_\T \Im \langle \delta_d, (\widetilde{F}_{z}^+(0,x))^{-1}G(z,x)\widetilde{F}_{z}^+(0,x)\delta_d\rangle dx.
\end{align*}
Let
\begin{equation}\label{gjy5n}
\Phi_{z}(x)=\begin{pmatrix}
\widetilde{F}_z^+(1,x) & \widetilde{F}_z^-(1,x)\\
\widetilde{F}_z^+(0,x)& \widetilde{F}_z^-(0,x)
\end{pmatrix}.
\end{equation}
Using that $(L_{\bar{z},v}^w(x))^*S L_{z,v}^w(x)=S,$ one can check that for $i=1,2,\cdots,d-1$,
$$
\begin{pmatrix}
\vec{u}^+_{\bar{z}}(1,x) \\ \vec{u}^+_{\bar{z}}(0,x)
\end{pmatrix}^*S\begin{pmatrix}
\vec{v}^i_z(1,x) \\ \vec{v}^i_z(0,x)
\end{pmatrix}=
\begin{pmatrix}\vec{u}^+_{\bar{z}}(1,x) \\ \vec{u}^+_{\bar{z}}(0,x)
\end{pmatrix}^*
S\begin{pmatrix}
\vec{u}^i_z(1,x) \\ \vec{u}^i_z(0,x)
\end{pmatrix}=0,
$$
$$
\begin{pmatrix}
\vec{u}^-_{\bar{z}}(1,x) \\ \vec{u}^-_{\bar{z}}(0,x)
\end{pmatrix}^*S\begin{pmatrix}
\vec{v}^i_z(1,x) \\ \vec{v}^i_z(0,x)
\end{pmatrix}=
\begin{pmatrix}\vec{u}^-_{\bar{z}}(1,x) \\ \vec{u}^-_{\bar{z}}(0,x)
\end{pmatrix}^*
S\begin{pmatrix}
\vec{u}^i_z(1,x) \\ \vec{u}^i_z(0,x)
\end{pmatrix}=0,
$$
$$
\begin{pmatrix}
\vec{u}^+_{\bar{z}}(1,x) \\ \vec{u}^+_{\bar{z}}(0,x)
\end{pmatrix}^*S\begin{pmatrix}
\vec{u}^+_{z}(1,x) \\ \vec{u}^+_z(0,x)
\end{pmatrix}=\begin{pmatrix}
\vec{u}^-_{\bar{z}}(1,x) \\ \vec{u}^-_{\bar{z}}(0,x)
\end{pmatrix}^*S\begin{pmatrix}
\vec{u}^-_{z}(1,x) \\ \vec{u}^-_z(0,x)
\end{pmatrix}=0,
$$
$$
\begin{pmatrix}
\vec{u}^i_{\bar{z}}(1,x) \\ \vec{u}^i_{\bar{z}}(0,x)
\end{pmatrix}^*S\begin{pmatrix}
\vec{u}^j_{z}(1,x) \\ \vec{u}^j_z(0,x)
\end{pmatrix}=0,\ \ \forall j=1,2,\cdots,d-1.
$$
Let 
\begin{align}\label{lanair5}
\nonumber \widetilde{T}_z(x)&=\begin{pmatrix}u^+_{\bar{z}}(x)& u^-_{\bar{z}}(x)\end{pmatrix}^*S\begin{pmatrix}u^+_{z}(x)& u^-_z(x)\end{pmatrix}\\
&=\begin{pmatrix}
0& (u^+_{\bar{z}}(x))^*Su^-_{z}(x)\\
(u^-_{\bar{z}}(x))^*Su^+_{z}(x)&0
\end{pmatrix}.
\end{align}
Thus there are $C_z^\pm(x)$ such that
$$
\Phi_{\bar{z}}^*(x)S\Phi_z(x)=\begin{pmatrix}&&C_z^+(x)\\ &\widetilde{T}_z(x)&\\C_z^-(x)&&\end{pmatrix}
$$
which implies
$$
\Phi^{-1}_z(x)=\begin{pmatrix}&&(C_z^-(x))^{-1}\\ &\widetilde{T}^{-1}_z(x)&\\(C_z^+(x))^{-1}&&\end{pmatrix}\Phi_{\bar{z}}^*(x)S.
$$
On the other hand, by \eqref{gjy5n} and direct calculation, we have
\begin{align*}
\Phi^{-1}_z(x)=\begin{pmatrix}\left(\widetilde{F}_{z}^+(1,x)-\widetilde{F}_{z}^-(1,x)(\widetilde{F}_{z}^-(0,x))^{-1}\widetilde{F}_{z}^+(0,x)\right)^{-1}&*
\\ *&*
\end{pmatrix}.
\end{align*}
It follows by Lemme \ref{Green_Matrix} and \eqref{lanair6} that
\begin{align}\label{aaa}\small
&\left((\widetilde{F}_{z}^+(0,x))^{-1}G(z,x)\widetilde{F}_{z}^+(0,x)\right)(d,d)=\left(\Phi^{-1}_z(x)\begin{pmatrix}C^{-1}\widetilde{F}_{z}^+(0,x)&\\ & O_d\end{pmatrix}\right)(d,d)\\ \nonumber
=&\left(\begin{pmatrix}&&(C_z^-(x))^{-1}\\ &\widetilde{T}^{-1}_z(x)&\\(C_z^+(x))^{-1}&&\end{pmatrix}\Phi_{\bar{z}}^*(x)S\begin{pmatrix}C^{-1}\widetilde{F}_{z}^+(0,x)&\\ & O_d\end{pmatrix}\right)(d,d)
\end{align}
where $M(i,j)$ represents the $(i,j)$-th entry of the matrix M.
Hence by \eqref{r1'}, \eqref{r2'}, \eqref{lanair5} and  \eqref{aaa}, we have
\begin{align*}
\langle \delta_d,(\widetilde{F}_{z}^+(0,x))^{-1}G(z,x)\widetilde{F}_{z}^+(0,x)\delta_d\rangle
=\frac{\vec{u}_{\bar{z}}^-(0,\omega)^*\vec{u}^+_z(0,x)}{u_{\bar{z}}^-(x)^*Su^+_z(x)}.
\end{align*}
\end{pf}

We denote
\begin{align*}
f(z,x)= &\frac{1}{-\Im \frac{\left(\vec{u}^+_z(0,x)\right)^*C\vec{u}^+_z(1,x)}{\|\vec{u}^+_z(0,x)\|^2}+\frac{1}{2}\Im z}+ \frac{1}{\Im \frac{\left(\vec{u}^-_{\bar{z}}(0,x)\right)^*C^*\vec{u}^-_{\bar{z}}(-1,x)}{\|\vec{u}^-_{\bar{z}}(0,x)\|^2}+\frac{1}{2}\Im z}\\
&-4\Im\frac{\vec{u}_{\bar{z}}^-(0,x)^*\vec{u}^+_z(0,x)}{u_{\bar{z}}^-(x)^*Su^+_z(x)}.
\end{align*}

\begin{Lemma}\label{gjyle}
For every $E\in \Sigma_{v,\alpha}^w$ and arbitrary $\delta>0$ sufficiently small, there is a measurable function $g(E+i\delta,x)$ with $\int_\T g(E+i\delta,x)dx<M(E)$ for almost every $E\in \{E:L_d(E)=0\}$, such that
$$
f(E+i\delta,x)+\sqrt{\delta} g(E+i\delta,x)\geq 0.
$$
\end{Lemma}
\begin{pf}
For any fixed $E_0\in \Sigma_{v,\alpha}^w$, for any $z=E_0+i\delta$ and $\delta$ sufficiently small, there are $m_\pm(z,x)$, depending analytically on $x$ and  $z$ such that
\begin{align}\label{ff1'}
u_z^+(x)=m_+(z,x)u_z(x)+v_z(x),
\end{align}
\begin{align}\label{ff2'}
u_z^-(x)=m_-(z,x)\tilde{u}_z(x)+v_z(x).
\end{align}
satisfy $u_z^\pm(x)\in E_\pm(z,x)$ where $u_z(x), v_z(x)$ are from Proposition \ref{newwe'}.

In view of \eqref{+}, we have that
\begin{equation}\label{20241}
\Im \left(\vec{u}_z^{+}(0,x)\right)^*C\vec{u}_z^+(1,x)=-\Im z\sum\limits_{n=1}^{\infty}\|\vec{u}_z^+(n,x)\|^2,
\end{equation}
\begin{equation}\label{20242}
\Im \left(\vec{u}_z^{+}(0,x)\right)^*C^*\vec{u}_z^+(-1,x)=\Im z\sum\limits_{n=0}^{\infty}\|\vec{u}_z^+(n,x)\|^2,
\end{equation}
\begin{equation}\label{20243}
\Im \left(\vec{u}_{\bar{z}}^{-}(0,x)\right)^*C\vec{u}_{\bar{z}}^-(1,x)=-\Im z\sum\limits_{n=-\infty}^{0}\|\vec{u}_{\bar{z}}^-(n,x)\|^2,
\end{equation}
\begin{equation}\label{20244}
\Im \left(\vec{u}_{\bar{z}}^{-}(0,x)\right)^*C^*\vec{u}_{\bar{z}}^-(-1,x)=\Im z\sum\limits_{n=-\infty}^{-1}\|\vec{u}_{\bar{z}}^-(n,x)\|^2.
\end{equation}
By Proposition \ref{newwe'}, we have
\begin{align*}
&\Im \left(\vec{u}_z^{+}(0,x)\right)^*C^*\vec{u}_z^+(-1,x)=\frac{i}{2}\begin{pmatrix}
\vec{u}_z^{+}(0,x)\\ \vec{u}_z^{+}(-1,x)
\end{pmatrix}^*S\begin{pmatrix}
\vec{u}_z^{+}(0,x)\\ \vec{u}_z^{+}(-1,x)
\end{pmatrix}\\ 
=&\frac{i}{2}\begin{pmatrix}m^+(z,x)\\ 1\end{pmatrix}^*\begin{pmatrix}
u^*_{z}(x)\\ v^*_{z}(x)
\end{pmatrix}S\begin{pmatrix}
u_{z}(x)& v_{z}(x)\end{pmatrix}\begin{pmatrix}m^+(z,x)\\ 1\end{pmatrix}\\
=&\Im m_+(z,x)+O\left(\delta^2(|m_0^+(z,x)|^2+1)\right).
\end{align*}
Therefore,
\begin{equation}\label{20245}
|\Im \left(\vec{u}_z^{+}(0,x)\right)^*C^*\vec{u}_z^+(-1,x)-\Im m_+(z,x)|\leq C(|m_+(z,x)|^2+1)\delta^2.
\end{equation}
By \eqref{20242} and since $u_z(x)$ and $v_z(x)$ are uniformly transverse,
\begin{equation}\label{aaa3}
\Im \left(\vec{u}_z^{+}(0,x)\right)^*C^*\vec{u}_z^+(-1,x)>\delta(\|\vec{u}_z^+(0,x)\|^2+\|\vec{u}_z^+(1,x)\|^2)> c\delta (|m_+(z,x)|^2+1).
\end{equation}
By \eqref{20245} and \eqref{aaa3}, we have
\begin{equation}\label{20246}
|\Im \left(\vec{u}_z^{+}(0,x)\right)^*C^*\vec{u}_z^+(-1,x)-\Im m_+(z,x)|\leq C\delta  \Im \left(\vec{u}_z^{+}(0,x)\right)^*C^*\vec{u}_z^+(-1,x).
\end{equation}
Similarly, we have
\begin{equation}\label{20247}
|\Im \left(\vec{u}_{\bar{z}}^{-}(0,x)\right)^*C^*\vec{u}_{\bar{z}}^-(-1,x)-\Im m_-(\bar{z},x)|\leq  C\delta \Im \left(\vec{u}_{\bar{z}}^{-}(0,x)\right)^*C^*\vec{u}_{\bar{z}}^-(-1,x).
\end{equation}
By \eqref{20242}, \eqref{20244}, \eqref{20246} and \eqref{20247}, we have
\begin{align}\label{ffff1}
\Im m_+(z,x)-\frac{1}{2}\Im z\|\vec{u}_z^+(0,x)\|^2>\frac{1}{4}\Im \left(\vec{u}_z^{+}(0,x)\right)^*C^*\vec{u}_z^+(-1,x),\ \ \Im m_-(\bar{z},x)>0.
\end{align}

We define
\begin{equation}\label{lanair7}
\Omega_1(z)=\{x\in\T:\|\vec{u}^+_{z}(0,x)\|>\|\vec{u}^-_{\bar{z}}(0,x)\|\}
\end{equation}
For any $x\in\Omega_1$, by \eqref{ffff1} and Proposition \ref{newwe'}, we have
\begin{align}\label{gjyabc1}
&4\left|\frac{\vec{u}_{\bar{z}}^-(0,x)^*\vec{u}^+_z(0,x)}{u_{\bar{z}}^-(x)^*Su^+_z(x)}\right|\leq \frac{4\|\vec{u}^-_{\bar{z}}(0,x)\|\|\vec{u}^+_z(0,x)\|}{(\Im m_+(z,x)+\Im m_-(\bar{z},x))}\\ \nonumber
\leq &\frac{2\|\vec{u}^-_{\bar{z}}(0,x)\|\|\vec{u}^+_z(0,x)\|}{\sqrt{(\Im m_+(z,\omega)-\frac{1}{2}\Im z\|u_{\bar{z}}^-(0,x)\|^2)(\Im m_-(\bar{z},x)+\frac{1}{2}\Im z\|u_{\bar{z}}^-(0,x)\|^2)}}\\ \nonumber
\leq &\frac{2\|\vec{u}^-_{\bar{z}}(0,x)\|\|\vec{u}^+_z(0,x)\|}{\sqrt{(\Im m_+(z,\omega)-\frac{1}{2}\Im z\|u_{z}^+(0,x)\|^2)(\Im m_-(\bar{z},x)+\frac{1}{2}\Im z\|u_{\bar{z}}^-(0,x)\|^2)}}\\ \nonumber
\leq &\frac{\|\vec{u}^+_z(0,x)\|^2}{\Im m_+(z,x)-\frac{1}{2}\Im z\|u_{z}^+(0,x)\|^2}+\frac{\|\vec{u}^-_{\bar{z}}(0,x)\|^2}{\Im m_-(\bar{z},x)+\frac{1}{2}\Im z\|u_{\bar{z}}^-(0,x)\|^2}.
\end{align}
By \eqref{20246} and \eqref{20247}, we have
\begin{align}\label{gjyabc2}
& \frac{\|\vec{u}^+_z(0,x)\|^2}{\Im m_+(z,x)-\frac{1}{2}\Im z\|u_{z}^+(0,x)\|^2}(1-4\sqrt{\delta})\\ \nonumber
\leq & \frac{1}{\Im \frac{\left(\vec{u}^+_z(0,x)\right)^*C^*\vec{u}^+_z(-1,x)}{\|\vec{u}^+_z(0,x)\|^2}-\frac{1}{2}\Im z}
 \leq  \frac{\|\vec{u}^+_z(0,x)\|^2}{\Im m_+(z,x)-\frac{1}{2}\Im z\|u_{z}^+(0,x)\|^2}(1+4\sqrt{\delta}),
\end{align}
\begin{align}\label{gjyabc3}
&\frac{\|\vec{u}^-_{\bar{z}}(0,x)\|^2}{\Im m_-(\bar{z},x)+\frac{1}{2}\Im z\|u_{\bar{z}}^-(0,x)\|^2}(1-4\sqrt{\delta})\\ \nonumber
\leq & \frac{1}{\Im \frac{\left(\vec{u}^-_{\bar{z}}(0,x)\right)^*C^*\vec{u}^-_{\bar{z}}(-1,x)}{\|\vec{u}^-_{\bar{z}}(0,x)\|^2}+\frac{1}{2}\Im z}\leq  \frac{\|\vec{u}^-_{\bar{z}}(0,x)\|^2}{\Im m_-(\bar{z},x)+\frac{1}{2}\Im z\|u_{\bar{z}}^-(0,x)\|^2}(1+4\sqrt{\delta}).
\end{align}

On the other hand, by Proposition \ref{newwe'}, we can also  find two linearly independent vectors  $u'_z(x),v'_z(x)\in E_c(z,x)$ depending analytically on $x$ and $z$, satisfying 
\begin{equation}\label{gjy2n'}
\begin{pmatrix}
u'^*_{\bar{z}}(x)\\ v'^*_{\bar{z}}(x)
\end{pmatrix}S\begin{pmatrix}
u'_{z}(x)& v'_{z}(x)\end{pmatrix}=J,
\end{equation}
\begin{equation}\label{gjy3n'}
\begin{pmatrix}
((L_{z,v}^w)_d(x)u'_{z}(x))^*\\ ((L_{z,v}^w)_d(x)v'_{z}(x))^*
\end{pmatrix}S\begin{pmatrix}
(L_{z,v})_d^w(\omega)u'_{z}(x)& (L_{z,v})_d^w(x)v'_{z}(x)\end{pmatrix}=J+O(\delta^2),
\end{equation}
Similarly, there are $m'_\pm(z,x)$ depending analytically on $x$ and  $z$ such that
\begin{align}\label{ff1''}
u_z^+(x)=m'_+(z,x)u'_z(x)+v'_z(x),
\end{align}
\begin{align}\label{ff2''}
u_z^-(x)=m'_-(z,x)u'_z(x)+v'_z(x).
\end{align}
satisfy $u_z^\pm(x)\in E_\pm(z,x)$. 

By exactly the same argument as above, we have
\begin{equation}\label{20246'}
|\Im \left(\vec{u}_z^{+}(0,x)\right)^*C\vec{u}_z^+(1,x)+\Im m'_+(z,x)|\leq C\delta |\Im \left(\vec{u}_z^{+}(0,x)\right)^*C\vec{u}_z^+(1,x)|.
\end{equation}
\begin{equation}\label{20247'}
|\Im \left(\vec{u}_{\bar{z}}^{-}(0,x)\right)^*C\vec{u}_{\bar{z}}^-(1,x)+\Im m'_-(\bar{z},x)|\leq \Im C\delta |\Im \left(\vec{u}_{\bar{z}}^{-}(0,x)\right)^*C\vec{u}_{\bar{z}}^-(1,x)|,
\end{equation}
\begin{align}\label{ffff1'}
\Im m'_-(\bar{z},x)-\frac{1}{2}\Im z\|\vec{u}_{\bar{z}}^-(0,x)\|^2>-\frac{1}{4}\Im \left(\vec{u}_{\bar{z}}^{-}(0,x)\right)^*C\vec{u}_{\bar{z}}^-(1,x),\ \ \Im m'_+(\bar{z},x)>0.
\end{align}
For any $x\in\Omega^c_1$, we then have
\begin{align}\label{gjyabc4}
&4\left|\frac{\vec{u}_{\bar{z}}^-(0,x)^*\vec{u}^+_z(0,x)}{u_{\bar{z}}^-(x)^*Su^+_z(x)}\right|\leq \frac{4\|\vec{u}^-_{\bar{z}}(0,x)\|\|\vec{u}^+_z(0,x)\|}{(\Im m'_+(z,x)+\Im m'_-(\bar{z},x))}\\ \nonumber
\leq &\frac{\|\vec{u}^+_z(0,x)\|^2}{\Im m'_+(z,x)+\frac{1}{2}\Im z\|u_{z}^+(0,x)\|^2}+\frac{\|\vec{u}^-_{\bar{z}}(0,x)\|^2}{\Im m'_-(\bar{z},x)-\frac{1}{2}\Im z\|u_{\bar{z}}^-(0,x)\|^2}.
\end{align}
\begin{align}\label{gjyabc5}
&\frac{\|\vec{u}^+_z(0,x)\|^2}{\Im m'_+(z,x)+\frac{1}{2}\Im z\|u_{z}^+(0,x)\|^2}(1-4\sqrt{\delta})\leq \frac{1}{-\Im \frac{\left(\vec{u}^+_z(0,x)\right)^*C\vec{u}^+_z(1,x)}{\|\vec{u}^+_z(0,x)\|^2}+\frac{1}{2}\Im z}\\ \nonumber
 \leq &\frac{\|\vec{u}^+_z(0,x)\|^2}{\Im m'_+(z,x)+\frac{1}{2}\Im z\|u_{z}^+(0,x)\|^2}(1+4\sqrt{\delta}),
\end{align}
\begin{align}\label{gjyabc6}
&  \frac{\|\vec{u}^-_{\bar{z}}(0,x)\|^2}{\Im m'_-(\bar{z},x)-\frac{1}{2}\Im z\|u_{\bar{z}}^-(0,x)\|^2}(1-4\sqrt{\delta})\leq \frac{1}{\Im \frac{\left(\vec{u}^-_{\bar{z}}(0,x)\right)^*C\vec{u}^-_{\bar{z}}(-1,x)}{\|\vec{u}^-_{\bar{z}}(0,x)\|^2}+\frac{1}{2}\Im z}\\         \nonumber \leq & \frac{\|\vec{u}^-_{\bar{z}}(0,x)\|^2}{\Im m'_-(\bar{z},x)-\frac{1}{2}\Im z\|u_{\bar{z}}^-(0,x)\|^2}(1+4\sqrt{\delta}).
\end{align}

Finally, we let 
\begin{align}\label{gjyabc}
g(z,\omega)=16\left(\frac{1}{-\Im \frac{\left(\vec{u}^+_z(0,x)\right)^*C\vec{u}^+_z(1,x)}{\|\vec{u}^+_z(0,x)\|^2}+\frac{1}{2}\Im z}+ \frac{1}{\Im \frac{\left(\vec{u}^-_{\bar{z}}(0,x)\right)^*C^*\vec{u}^-_{\bar{z}}(-1,x)}{\|\vec{u}^-_{\bar{z}}(0,x)\|^2}+\frac{1}{2}\Im z}\right).
\end{align}
By \eqref{gjyabc1}-\eqref{gjyabc3} on $\Omega_1$ and \eqref{gjyabc4}-\eqref{gjyabc6} on $\Omega_2$, we have that 
\begin{equation}\label{lanair8}
f(z,x)\geq -\sqrt{\delta}g(z,x).
\end{equation} 
Moreover, by Lemma \ref{le2}, 
$$
\int_\T g(E_0+i\delta,x)dx\leq -64d \frac{L_{d+1}(E_0+i\delta)}{\delta}dx\leq M(E_0)
$$
for some constant $M$, for almost every $E_0\in \{E:L_d(E)=0\}$.
\end{pf}

By Lemma \ref{gjyle} and Fatou's lemma, for almost every $E\in \{E:L_d(E)=0\}$, we have
\begin{align*}
&\int_\T \liminf_{\delta\rightarrow 0}\left(f(E+i\delta,x)+\sqrt{\delta}g(E+\delta,x)\right)dx \leq \liminf_{\delta\rightarrow 0} \int_\T \left(f(E+i\delta,x)+\sqrt{\delta}g(E+\delta,x)\right)dx\\
\leq  &\liminf_{\delta\rightarrow 0} \int_\T f(E+i\delta,x)dx+\liminf_{\delta\rightarrow 0}\sqrt{\delta}\int_\T g(E+i\delta,x)dx.
\end{align*}
By Lemma \ref{le2} and Lemma \ref{final1}, 
\begin{align*}
\int_\T \liminf_{\delta\rightarrow 0}f(E+i\delta,x)dx \leq \liminf_{\delta\rightarrow 0} \int_\T f(E+i\delta,x)dx\leq 0
\end{align*}
By \eqref{ff1'}, \eqref{ff2'} and (1) of Theorem \ref{keyth} and the above inequality, for almost every $E\in (E_0-\delta_2,E_0+\delta_2)$ and $x\in\T$, we have $f(E+i0,x)$ exists. Thus by \eqref{lanair8},  $f(E+i0,x)\geq $  for almost every $E\in (E_0-\delta_2,E_0+\delta_2)$ and $x\in\T$.

We omit $x$ for simplicity, notice that by \eqref{lana1},
\begin{align}
f(E+i0)= \frac{\|\vec{u}_{E+i)}^+(0)\|^2}{\Im m_+(E+i0)}+\frac{\|\vec{u}_{E-i0}^-(0)\|^2}{\Im m_-(E-i0)}-4\Im\frac{\vec{u}_{E-i0}^-(0)^*\vec{u}^+_{E+i0}(0)}{\overline{m_-(E-i0)}-m_+(E+i0)}
\end{align}
On the other hand,
\begin{align}\label{gjyabc1'}
&\left|4\Im\frac{\vec{u}_{E-i0}^-(0)^*\vec{u}_{E+i0}(0)}{\overline{m_-(E-i0)}-m_+(E+i0)}\right|\leq 4\frac{\|\vec{u}^-_{E-i0}(0,x)\|\|\vec{u}^+_{E+i0}(0,x)\|}{\Im m_+(E+i0)+\Im m_-(E-i0)}\\ \nonumber
\leq &\frac{2\|\vec{u}^-_{E-i0}(0)\|\|\vec{u}^+_{E+i0}(0)\|}{\sqrt{\Im m_+(E+i0)\Im m_-(E-i0)}}\leq \frac{\|\vec{u}^+_{E+i0}(0)\|^2}{\Im m_+(E+i0)}+\frac{\|\vec{u}^-_{E-i0}(0)\|^2}{\Im m_-(E-i0,x)}.
\end{align}
Thus $f(E+i0)=0$ if and only all the inequalities above are equalities. Note that the second equality implies  $m_-(E-i0,x)=m_+(E+i0,x)$ for almost every $E\in (E_0-\delta_2, E_0+\delta_2)$ and almost every $x\in\T$. This finishes the proof of (3).

\end{pf}

\subsection{$L^2$-reducibility and proof of Theorems \ref{L2
    reducibility} and \ref{C0reducibility1}} 
Let $A_E\in C^\omega(\T,GL(2,\C))$ be such that
\begin{equation}\label{ff6}
L_E^f(x)\begin{pmatrix}u_E(x)& v_E(x)\end{pmatrix}=\begin{pmatrix}u_E(x+\alpha)&v_E(x+\alpha)\end{pmatrix}A_E(x).
\end{equation}
The following structural facts about the two-dimensional center are contained in higher generality in \cite{gj}. We
prove them also in Appendix~\ref{app:center-facts} for readers convenience.

\begin{Proposition}[Proposition 6.2 of \cite{gj}]\label{p1}
We have
$$
A_E(x)^*JA_E(x)=J
$$
where $J=\begin{pmatrix}0&1 \\-1&0
\end{pmatrix}$.
\end{Proposition}
\begin{pf}
This is Lemma~\ref{lem:center-J-unitary}, applied to the normalized
center frame constructed above.
\end{pf}
\begin{Corollary}[Remark 7.1 of \cite{gj}]\label{c24}
Assume in addition that
\[
\omega^{d-1}(\alpha,L^w_{E,v})=0.
\]
Then there exist
\[
\varphi_E\in C^\omega(\mathbb T,\mathbb R),
\qquad
C_E\in C^\omega(\mathbb T,SL(2,\mathbb R)),
\]
depending analytically on $E$, such that
\[
A_E(x)=e^{2\pi i\varphi_E(x)}C_E(x).
\]
\end{Corollary}

\begin{pf}
By Lemma~\ref{lem:center-det-zero}, the determinant of the
center cocycle has zero winding. The result follows from
Lemma~\ref{lem:projectively-real-factorization}.
\end{pf}

We first prove the
$L^2$-reducibility Theorem \ref{L2 reducibility} that we slightly reformulate as
\begin{Theorem}\label{L2 reducibility1}
For one-frequency $PH2$ cocycles $(\alpha,L_{E,v}^w)$, for almost every
$E\in \{E:L_d(E)=0,\,\omega^{d-1}(\alpha,L^w_{E,v})=0\}$, there exist $\phi_E\in C^\omega(\T,\R)$, $U_E,V_E\in
L^2(\T,\C^{2d})$ and $R_E(x)\in SO(2,\R) $ such that
$$
L_{E,v}^w(x)(U_E(x),V_E(x))=(U_E(x+\alpha), V_E(x+\alpha))e^{2\pi i \phi_E(x)}R_E(x),
$$
with
$$
H_E=(U_E,V_E)\in L^2(\T,Sp_{2d\times2}(\C)).
$$
\end{Theorem}
\begin{pf}
  Let $u_E(x), v_E(x)$ be given by Proposition \ref{newwe} and $m_+$ be as defined in Definition \ref{m+def}.
For almost every $E\in \{E:L_d(E)=0\}$ and almost every $x\in\T$, $m_+(E,x)$ exists with $\Im m_+(E,x)>0$,  and we set
$$
D_E(x)=\begin{pmatrix}
0&\frac{|m_+(E,x)|}{(\Im m_+(E,x))^{1/2}}\\
-\frac{(\Im m_+(E,x))^{1/2}}{|m_+(E,x)|}&\frac{\Re m_+(E,x)}{|m_+(E,x)|(\Im m_+(E,x))^{1/2}}
\end{pmatrix}.
$$
Let
$$
(U_E(x),V_E(x))=(u_E(x),v_E(x))D_E(x),
$$
then for $R_E(x)\in SL(2,\R) $ defined by
$$
L_{E,v}^w(x)(U_E(x),V_E(x))=(U_E(x+\alpha), V_E(x+\alpha))e^{2\pi i\phi_E(x)}R_E(x),
$$
we have that $R_E(x)\in SO(2,\R) $, this is because $D_E(x+\alpha)^{-1}C_E(x)D_E(x)\cdot i=i$. Moreover, we have
\begin{align*}
\|U_E\|^2_{L^2}+\|V_E\|^2_{L^2}=&\left\|\frac{|m_+(E,\cdot)|}{(\Im m_+(E,\cdot))^{1/2}}\vec{u}_E(0,\cdot)+\frac{\Re m_+(E,\cdot)}{|m_+(E,\cdot)|(\Im m_+(E,\cdot))^{1/2}}\vec{v}_E(0,\cdot)\right\|^2_{L^2}\\
&+\left\|\frac{|m_+(E,\cdot)|}{(\Im m_+(E,\cdot))^{1/2}}\vec{u}_E(-1,\cdot)+\frac{\Re m_+(E,\omega)}{|m_+(E,\cdot)|(\Im m_+(E,\cdot))^{1/2}}\vec{v}_E(-1,\cdot)\right\|^2_{L^2}\\
&+\left\|\frac{(\Im m_+(E,\cdot))^{1/2}}{|m_+(E,\cdot)|}\vec{v}_E(0,\cdot)\right\|_{L^2}^2+\left\|\frac{(\Im m_+(E,\cdot))^{1/2}}{|m_+(E,\cdot)|}\vec{v}_E(-1,\cdot)\right\|_{L^2}^2\\
\leq &C\int_\T\frac{\left\|m_+(E+i0,x)\begin{pmatrix}
\vec{u}_E(0,x)\\
\vec{u}_E(-1,x)
\end{pmatrix}+\begin{pmatrix}
\vec{v}_E(0,x)\\
\vec{v}_E(-1,x)
\end{pmatrix}\right\|^2}{\Im m_+(E+i0,x)}dx<\infty,
\end{align*}
where the last inequality follows from part (2) of Theorem
\ref{keyth}. 

Finally, by Proposition \ref{newwe},
$$
H^*_E(x)SH_E(x)=\begin{pmatrix}0&1\\ -1&0\end{pmatrix}.
$$
\end{pf}

\noindent {\bf Proof of Theorem \ref{C0reducibility1}:}
For any fixed $E_0\in I$, by Lemma \ref{nonvanishing} and Definition \ref{m+def}, there exists $\delta_2(E_0)>0$ such that for any $z\in R_{\delta_2}(E_0)\cap \mathbb{H}$, we have $\Im m_+(z,x)>0$ for any $x\in\T$. Moreover, $m_+(z,x)$ depends analytically on both $z$ and $x$.

By  part (3) of Theorem \ref{keyth}, for almost every $x$, $m_+(z,x)$ can be extended analytically to $R_{\delta_2}(E_0)\backslash(\R\backslash (E_0-\delta_2,E_0+\delta_2))$.

On the other hand,  we see that $m_+(z,x)_{x}$ is a normal
family. Thus for any compact $\mathcal K \in R_{\delta_2}(E_0)\backslash(\R\backslash (E_0-\delta_2,E_0+\delta_2))$,
$m_+(\cdot,x)$ is uniformly Lipschitz in $z$.  Namely, there exists
a constant $c=c(\mathcal K)$ that depends only on $\mathcal K$ such that
$$
|m_+(z_1,x)-m_+(z_2,x)|\leq c|z_1-z_2|,\ \ \forall  z_1, z_2\in\mathcal K.
$$
We can, for any $x\in \T,$ pick a sequence $x_n\in \T$ converging to $x$. Then we get a holomorphic function $m_+(\cdot,x)=\lim\limits_{x_n\rightarrow x}m_+(\cdot,x_n)$ on $R_{\delta_2}(E_0)\backslash(\R\backslash (E_0-\delta_2,E_0+\delta_2))$.

We need a lemma by Avila-Jitomirskaya \cite{aj}.
\begin{Lemma}\label{aj}
Let $W\subset \C$ be a domain, and let $f:W\times \R/\Z\rightarrow\C$ be a continuous function. If $z\rightarrow f(z,x)$ is holomorphic for all $x\in \R/\Z$ and $w\rightarrow f(z,x)$ is analytic for some nonpolar set $z\in W$, then $f$ is analytic.
\end{Lemma}
Note that $m_+(z,x)$ is continuous, $z\rightarrow m_+(z,x)$ is holomorphic on $R_{\delta_2}(E_0)\backslash(\R\backslash (E_0-\delta,E_0+\delta_2))$ for any $x\in\T$ and $x\rightarrow m_+(z,x)$ is analytic on $\T$ for any $z\in R_{\delta_2}(E_0)\backslash(\R\backslash (E_0-\delta_2,E_0+\delta_2))$. Thus by Lemma \ref{aj}, $m_+(E,x)$ is analytic on $(E_0-\delta_2, E_0+\delta_2)\times \T$.

Finally, let
$$
(U_E(\omega),V_E(x))=(u_E(x),v_E(x))D_E(x).
$$
Then we have  $(U_E,V_E)\in C^\omega(\T, Sp_{2d\times 2}(\C))$, and
there are $\phi_E\in C^\omega(\T,\R)$ and  $R_E\in C^\omega(\T,SO(2,\R))$, depending analytically in $E$ and $x$ on $I'\times \T$ for some $I'\subset I$, such that
$$
L_{E,v}^w(x)(U_E(x),V_E(x))=(U_E(x+\alpha), V_E(x+\alpha))e^{2\pi i\phi_E(x)}R_E(x).
$$
\qed

\section{The trigonometric potentials}\label{C}
In this section, we assume that $v$ and $w$ are trigonometric
polynomials. Let
\[
\gamma_1(E)\geq\gamma_2(E)\geq\cdots\geq\gamma_{2l}(E)
\]
denote the Lyapunov exponents of the complex symplectic cocycle
$(\alpha,L^v_{E,w})$. For real $E$ they satisfy
\[
\gamma_j(E)=-\gamma_{2l+1-j}(E),
\]
so that $\gamma_l(E)$ is the smallest nonnegative Lyapunov exponent.
We use the convention $\omega^0=0$.

\begin{Theorem}\label{th71}
Given $\alpha\in\R\backslash\Q$ and  an open interval $I\subset \Sigma_{\alpha,v}^w$ , it is impossible that for all $E\in I$,
\begin{enumerate}
\item both $(\alpha,L_{E,v}^w)$ and $(\alpha,L_{E,w}^v)$  are $PH2$ and
\item  $\gamma_l(E)=0$ and
\item
$
\omega^{l-1}(\alpha,L^v_{E,w})=0.
$
\end{enumerate}
\end{Theorem}
\begin{pf}
  The proof is via contradiction and an improvement of Corollary
  \ref{C0reducibility}. 
Assume there is an open interval $I\subset \Sigma_{\alpha,v}^w$ such that  for any $E\in I$,
\begin{enumerate}
\item both $(\alpha,L_{E,v}^w)$ and $(\alpha,L_{E,w}^v)$  are $PH2$ and
\item $\gamma_l(E)=0,\,\omega^{l-1}(\alpha,L^v_{E,w})=0. $.
\end{enumerate}

\begin{Corollary}\label{Cwreducibility}
If  $\gamma_l(E)=0, \omega^{l-1}(\alpha,L^v_{E,w})=0,$ and $(\alpha,L_{E,w}^v)$ is $PH2$ for all $E$ in an interval $I\subset \R$, then there exist   $H_E\in C^\omega(\T,Sp_{2l\times 2}(\C))$, $\phi_E(x)\in C^\omega(\T,\R)$ and $\psi_E\in C^\omega(\T,\R)$, depending analytically on $E\in I'$ for $I'\subset I$ such that
\begin{equation}\label{final}
L_{E,w}^v(x)H_E(x)=H_E(x+\alpha)e^{2\pi i\phi_E(x)}R_{\psi_E(x)}.
\end{equation}
\end{Corollary}
\begin{pf}
Direct corollary of Theorem \ref{C0reducibility1}.
\end{pf}

Notice that for $\e$ sufficiently small, \eqref{final} also holds for $E+i\e$, thus one has 
$$
\gamma_l(E+i\e)=-2\pi \Im \int_\T (\phi_{E+i\e}(x)+\psi_{E+i\e}(x))dx, 
$$
$$
\gamma_{l+1}(E+i\e)=-2\pi \Im \int_\T (\phi_{E+i\e}(x)-\psi_{E+i\e}(x))dx.
$$
On the other hand, by Lemma \ref{final1} and (3) of Theorem \ref{keyth}, for almost every $E\in I'$,
$$
\lim\limits_{\e\rightarrow 0^+}\frac{\partial \Im \int_\T \psi_{E+i\e}(x)dx}{\partial \e}\geq \frac{1}{2\pi l}\int_\T \frac{\left\|m_+(E+i0,x)\begin{pmatrix}
\vec{u}_E(0,x)\\
\vec{u}_E(-1,x)
\end{pmatrix}+\begin{pmatrix}
\vec{v}_E(0,x)\\
\vec{v}_E(-1,x)
\end{pmatrix}\right\|^2}{\Im m_+(E+i0,x)}dx>0.
$$
Thus $\int_\T\psi_E(x)dx$ is not a constant \footnote{An alternative argument can be found in \cite{gj} where it is proved that $1-2\int_\T\psi_E(x)dx=N(E)$.}, so there is $E_0\in I'$ with $\int_\T
\psi_{E_0}(x)dx=k_1\alpha+k_2,$ for some $k_1,k_2\in\Z.$ We now define
$F_{E_0}(x)=H_{E_0}(x)R_{k_1x}$, and obtain
$$
L_{E_0,w}^v(x)F_{E_0}(x)=F_{E_0}(x+\alpha)e^{2\pi i\phi_{E_0}(x)}R_{\psi_{E_0}(x)-k_1\alpha}.
$$
Since $\int_\T \left(\psi_{E_0}(x)-k_1\alpha-k_2\right)dx=0$ and $(\alpha, L_{E_0,v}^w)$ is $PH2$, this contradicts  Theorem \ref{contra2}.
\end{pf}
\subsection{Proof of Cantor spectrum for Type I operators with trigonometric polynomial potentials}

Assume  $v$ is a trigonometric polynomial and there exists an interval $I\subset \overline{\Sigma^1_{v,\alpha}}$.
Recall that the associated Schr\"odinger cocycle is denoted
$(\alpha,S_{E}^v),$  the dual cocycle by $(\alpha,L_{E,v})$, and their
non-negative Lyapunov exponents are  $L(E)$ and $\{\gamma_i(E)\}_{i=1}^d,$ respectively. We distinguish two cases:
\begin{enumerate}
\item There exists $E_0\in I\subset  \overline{\Sigma^1_{v,\alpha}}$ such that $L(E_0)>0.$ Then by continuity of
  Lyapunov exponents there is $E'_0\in I\cap\Sigma_{v,\alpha}^1$ such that $L(E'_0)>0$.  Hence there is $E_0'\in I'\subset I$ such that $L(E)>0$ on $I'$ and, by openness of type I energies, $\bar{\omega}(E)=1$ on $I'$.  Then for any $E\in
  I'$, we have
  \begin{enumerate}
\item by Theorem \ref{dominate1}, both  $(\alpha,S_{E}^v)$ and $(\alpha,L_{E,v})$  are $PH2$
\item by Theorem 1.2 in \cite{gjyz},  $\gamma_d(E)=0$ .
\item by Theorem \ref{dominate1} $\omega^{d-1}(\alpha,L_{E,v})=0$ 

\end{enumerate}
This contradicts  Theorem  \ref{th71}.

\item If $L(E)=0$ on $I'$, we apply
Theorem~\ref{th71} with the two Aubry-dual operators
interchanged. The cocycle to which
Corollary~\ref{final} is applied is then the
Schr\"odinger cocycle $(\alpha,S_E^v)$. Its center Lyapunov exponent is
$L(E)=0$, while the additional acceleration condition is automatic
since in this case $l=1$ and $\omega^0=0$. The dual cocycle
$(\alpha,L_{E,v})$ is $PH_2$ by
 Theorem~\ref{dominate1}. 
This again contradicts Theorem  \ref{th71}.\qed
\end{enumerate}

\section{From trigonometric polynomials to real analytic }\label{limit}

For a general real analytic potential, the dual operator is infinite-range
and no finite-dimensional transfer cocycle is available. We
therefore approximate $v$ by its trigonometric polynomial truncations
$v_n$. Near a fixed Type I energy, the corresponding finite-range dual
cocycles are $PH_2$ for all sufficiently large $n$. The intrinsic
symplectic theory developed in \cite{gj} shows that their two-dimensional centers admit
a natural limiting dynamics. We give here a self-contained,
parameter-dependent  construction of this limit, analytic in both energy and phase, and
establish estimates on the stable and unstable bundles that are uniform
in $n$. The parameter dependence is needed for the long-range Kotani
argument in Section~10. The uniform hyperbolic estimates are used in
Section~11 to show that the approximate dual solutions have
exponentially small stable and unstable components, while their center
projections remain quantitatively nondegenerate.The specialized two-dimensional facts from \cite{gj} used in the
construction are proved in Appendix~\ref{app:center-facts}.
where we also indicate their
counterparts in \cite{gj}.

Assume $v\in C^\omega_{h_1}(\T,\R)$. Throughout this section and
Sections~\ref{slongkotani}--\ref{slongpuig}, we use approximants of
exact degree $l(n)$. Namely, let 
\begin{equation}
v_n(x)=\sum_{m=-l(n)}^{l(n)}\hat v_m e^{2\pi i mx},\ \ l(n)=\max\{|m|\leq n:\hat{v}_m\neq 0\}.
\end{equation}
Every transfer matrix and symplectic form
associated with $v_n$ uses these coefficients, including at the leading
coefficient. For every $0\leq h<h_1$, we have
\begin{align}\label{exact-degree-convergence}
|v_n-v|_h+|v_{n+1}-v_n|_h
&\leq C_s e^{-2\pi(h_1-h)n}.
\end{align}     
Let $E_0\in\R$ with $\bar\omega(E_0)=1$. By Lemma \ref{con} and \eqref{exact-degree-convergence}, there is
$\delta_0>0$ such that $E$ is a Type I energy for $H_{v_n,\alpha,x}$
whenever $n$ is sufficiently large and $|E-E_0|\leq\delta_0$.
By Theorem~\ref{dominate1}, there is a continuous decomposition of
$\C^{2l(n)}$ corresponding to $(\alpha,L_{E,v_n})$:   
\begin{equation}\label{sub10}
\C^{2l(n)}=E_s^n(\theta)\oplus E_c^n(\theta)\oplus E_u^n(\theta).
\end{equation}
By Theorem 6.1 in \cite{ajs}, $E_s^n(\theta)$, $E_c^n(\theta)$,  $E_u^n(\theta)$ depend analytically  on both $E$ and $\theta$.

In this case, we denote the non-negative Lyapunov exponents of dual symplectic cocycle $(\alpha,L_{E,v_n})$ by $\{\gamma^{n}_j(E)\}_{j=1}^{l(n)}$ (where $0\leq \gamma_{1}^n(E)\leq \gamma^{n}_2(E)\leq \cdots \leq \gamma_n^{l(n)}(E)$) and rewrite $S$ as $S_n$. Let $L^v_\e(E)=L(\alpha,S_E^v(\cdot+i\e))$ be the Lyapunov exponent of $(\alpha,S_E^v(\cdot+i\e))$ and $\e_1(E)\geq 0$ be the first turning point of $L^v_\e(E)$ \footnote{Since $H_{v,\alpha,\theta}$ is a type I operator, such turning point always exists.}.
The following multiplicative Jensen's formula is proved  in \cite{gjyz}.
\begin{Theorem}[Theorem 2 in \cite{gjyz}]\label{1general}
For $\alpha\in\R\backslash\Q$ and $v\in C^\omega_{h_1}(\T,\R)$, there exist non-negative $\{\gamma_i(E)\}_{i=1}^m$ such that for any $E\in\R$   
$$
\gamma_i(E)=\lim\limits_{n\rightarrow \infty}\gamma^{n}_i(E), \ \ 1\leq i\leq m.
$$
Moreover, \begin{align*}\label{gne1}
L_{\e}(E)= L_0(E) -\sum_{\{i:\gamma_i(E)< 2\pi|\e|\}}  \gamma_i(E)+2\pi(\#\{i:\gamma_i(E)<2\pi|\e|\})|\e|
\end{align*} for $|\e|<h_1$.
\end{Theorem}
\begin{Remark}
{\rm Note that although Theorem \ref{1general}  is only stated in \cite{gjyz} for real $E$, but the proof holds for all complex $z\in \C$. But due to the lack of evenness, one needs to state the formula for $\e\geq 0$ and $\e\leq 0$ individually. Here for simplicity, we only state half of the formula, i.e., $\e\geq 0$ case.  There exist non-negative $\{\gamma_i\bar{z})\}_{i=1}^m$ such that    
$$
\gamma_i(\bar{z})=\lim\limits_{n\rightarrow \infty}\gamma^{n}_i(\bar{z}), \ \ 1\leq i\leq m.
$$
Moreover, \begin{align*}
L_{\e}(z)= L_0(z) -\sum_{\{i:\gamma_i(\bar{z})< 2\pi\e\}}  \gamma_i(\bar{z})+2\pi(\#\{i:\gamma_i(\bar{z})<2\pi\e\})\e
\end{align*} for $0\leq \e<h_1$.}
\end{Remark}

The following theorem which captures uniform contraction/expansion of the stable/unstable subbundles and  gives the parametrized version of uniform convergence of the center, is crucial for our argument. We denote 
$$
GL_{2l(n)\times 2}(\C)=\{ F\in M_{2l(n)\times 2}(\C): {\rm Rank}(F)=2\},\ \ I_{\delta_0}(E_0)=(E_0-\delta_0,E_0+\delta_0).
$$
\begin{Theorem} \label{theorem-main1general}
Let $\alpha\in\R\backslash\Q$, $v\in C^\omega(\T,\R)$, $\bar{\omega}(E_0)=1$, there exists $\delta_0(E_0,v)>0$, such that the following hold:

{\rm (1)}  If $\omega(E)=1$, there exist $O_n\in C^\omega(I_{\delta_0}(E_0)\times\T,GL_{2l(n)\times 2}(\C))$ and $M_n\in C^\omega(I_{\delta_0}(E_0)\times \T,GL(2,\C))$, such that 
$$
L_{E,v_n}(\theta)O_n(E,\theta)=O_n(E,\theta+\alpha)M_n(E,\theta)
$$
with
\begin{equation}\label{cenest1}
|O_n^*S_nO_n-J|_{r,h},\ \ |M_n-e^{2\pi i\phi}{\bf C}|_{r,h}\leq Ce^{-cn}
\end{equation}
for some $r,h,c,C>0$, $\phi\in C^\omega(I_{\delta_0}(E_0)\times\T,\R)$ and ${\bf C}\in C^\omega(I_{\delta_0}(E_0)\times \T,SL(2,\R))$.  If we denote
$$
O_n=\begin{pmatrix}O_n(l(n)-1)\\ O_{n}(l(n)-2)\\ \vdots\\ O_n(-l(n))\end{pmatrix},
$$ 
then for $-l(n)\leq j\leq l(n)-1$, we have 
$$
|O_n(j)- O_{n+1}(j)|_{r,h}\leq Ce^{-cn}.
$$
whee $O_n(j)$ is the $j$-th row of $O_n$.

{\rm (2)}Moreover, for any frame $W_n(E,\theta,0)$ spanning the finite truncation center, set
$$
L_n(E,\theta)=W_n(E,\theta)\left(W_n(E,\theta)^*S_nW_n(E,\theta)\right)^{-1}W_n(E,\theta)^*=(l_{ij}(E,\theta))
$$
for any $\delta>0$, there is some $C(\delta)>0$ such that 
$$
|l_{ij}(E,\theta)|\leq Ce^{2\pi(\e_1(E)+\delta)|i-j|}.
$$
In particular, if $\omega(E)=1$, we have
\begin{equation}\label{irvp}
\|O_{n}(E,\theta,0)\|^2+\left\|\frac{\partial O_{n}(E,\theta)}{\partial E}\right\|^2\leq Ce^{\delta|l(n)|},
\end{equation}

{\rm(3)} There exist an orthonormal basis $\{u_n^i(E,\theta)\}_{i=1}^{l(n)-1}$ of $E^n_s(\theta)$, an orthonormal basis $\{v_n^i(E,\theta)\}_{i=1}^{l(n)-1}$ of $E^n_u(\theta)$ and a symplectic basis $\{w_n^i(E,\theta)\}_{i=1}^2$ of $E_c^n(\theta)$ with
\begin{align*}
&\begin{pmatrix}
u_n^1(E,\theta)&\cdots&u_n^{l(n)-1}(E,\theta)
\end{pmatrix}^*S_n\begin{pmatrix}
v_n^1(E,\theta)&\cdots&v_n^{l(n)-1}(E,\theta)
\end{pmatrix}\\
=&\begin{pmatrix}
\sigma_n^1(E,\theta)\\& \ddots\\ &&& \sigma_n^{l(n)-1}(E,\theta)
\end{pmatrix},\ \ \sigma_n^i(E,\theta)>0, \ \ 1\leq i\leq l(n)-1.
\end{align*}
For any $\delta>0$, there are $C(\delta)>0$ such that for any $m \geq 0$ and $\theta\in\T$, we have 
$$
\left|u_n^i(E,\theta,m)\right|\leq C\sigma_n^i(E,\theta)e^{|m-l(n)|\delta}e^{-2\pi \e_1(E)(m-l(n))},\ \  \forall 1\leq i\leq l(n)-1,
$$
$$
\left|v_n^i(E,\theta,-m)\right|\leq C\sigma_n^i(E,\theta)e^{|m+l(n)|\delta}e^{-2\pi \e_1(E)(m+l(n))},\ \  \forall 1\leq i\leq l(n)-1.
$$
$$
\left|w_n^i(E,\theta,\pm m)\right|\leq Ce^{(2\pi \e_1(E)+\delta)(m+l(n))},\ \  \forall 1\leq i\leq 2.
$$
where 
$$
\begin{pmatrix}u_{n}^i(E,\theta,m+l(n)-1)\\ u_{n}^i(E,\theta,m+l(n)-2)\\ \vdots\\
u_{n}^i(E,\theta,m-l(n))\end{pmatrix}=(L_{E,v_n})_{m}(\theta) u_n^i(E,\theta),\ \  \forall 1\leq i\leq l(n)-1,
$$
$$
\begin{pmatrix}v_{n}^i(E,\theta,-m+l(n)-1)\\ v_{n}^i(E,\theta,-m+l(n)-2)\\ \vdots\\
v_{n}^i(E,\theta,-m-l(n))\end{pmatrix}=(L_{E,v_n})_{-m}(\theta) v_n^i(E,\theta),\ \  \forall 1\leq i\leq l(n)-1.
$$
$$
\begin{pmatrix}w_{n}^i(E,\theta,\pm m+l(n)-1)\\ w_{n}^i(E,\theta,\pm m+l(n)-2)\\ \vdots\\
w_{n}^i(E,\theta,\pm m-l(n))\end{pmatrix}=(L_{E,v_n})_{\pm m}(\theta) w_n^i(E,\theta),\ \  \forall 1\leq i\leq 2.
$$
\end{Theorem}
\begin{Remark}\label{bimsa13}
$L_n(E,\theta)$  does not depend on the choice of basis. Indeed,  suppose we have another basis $W_n'(E,\theta)$ of $E_c^n(\theta).$ Then there is some $G_n(E,\theta)\in GL(2,\C)$ such that $W_n'(E,\theta)=W_n(E,\theta)G_n(E,\theta)$. It follows that
\begin{align*}
&W'_n(E,\theta)\left(W'_n(E,\theta)^*S_nW'_n(E,\theta)\right)^{-1}W'_n(E,\theta)^*\\
=&W_n(E,\theta)G_n(E,\theta)\left(G_n(E,\theta)^*W_n(E,\theta)^*S_nW_n(E,\theta)G_n(E,\theta)\right)^{-1}G_n(E,\theta)^*W_n(E,\theta)^*\\
=&W_n(E,\theta)\left(W_n(E,\theta)^*S_nW_n(E,\theta)\right)^{-1}W_n(E,\theta)^*.
\end{align*}
\end{Remark}
\begin{pf}
Parts~(1) and~(2) are a parameter-dependent, two-dimensional
specialization of the intrinsic center construction developed in
\cite{gj}; full details are included here. Part~(3), which provides
stable/unstable estimates uniform in the truncation dimension, is completely new.
The auxiliary signature and factorization facts used below are proved
in Appendix~\ref{app:center-facts}.

For any fixed $\e>\e_1(E_0)$ and sufficiently close to $\e_1(E_0)$, we have 
$$
L^v_\e(E_0)=L_{\e_1(E_0)}(E_0)+2\pi(\e-\e_1(E_0)),
$$
thus $(\alpha,S_{E_0}^v(\cdot\pm i\e))$ is regular (see Section \ref{regge} for the definition). By the openness of regularity and the continuity of the Lyapunov exponent, there is $\delta_0>0$ such that if $|E-E_0|<\delta$, then 
$$
L^v_{\pm\e}(E)=L_{\e_1(E)}(E)+2\pi (|\e|-\e_1(E))>0,
$$
in particular,  $(\alpha,S_{E}^v(\cdot\pm i\e))$ is regular. By Theorem 6 in \cite{avila0}, $(\alpha,S_E^v(\cdot\pm i\e))$ is uniformly hyperbolic.  By Lemma 4.1 and Theorem 4.4 in \cite{gjyz}, $(H_{v(\cdot\pm i\e),\alpha,x}-E)^{-1}$ exists for all $x\in\T$ and $E\in I_{\delta_0}(E_0)$.

Assume $v\in C^\omega_{h_1}(\T,\R)$ \footnote{By our assumption, $h_1>\e_1(E)$, for any $E\in I_{\delta_0}(E_0)$.}. By the resolvent identity,  there is an analytic neighborhood $U_v$ of $v$ such that for any $v'\in U_v$ and $E\in I_{\delta_0}(E_0)$ we have
\begin{equation}\label{gjne1ten}
|(H_{v(\cdot\pm i\e),x,\alpha}-E)^{-1}-(H_{v'(\cdot\pm i\e),x,\alpha}-E)^{-1}|\leq C(v,\delta)|v-v'|_{h_1-\delta}.
\end{equation} 
for any $\delta$ sufficiently small and any  $x\in \T$.

For any fixed $\ell\in\Z$, let $\delta_\ell(x,m):=\delta_\ell(m)$ for any $x\in\T$ and $m\in\Z$. Then by \eqref{dual map},
$$
(U^{-1}\delta_\ell)(x,m)=e^{-2\pi i \ell x}\delta_0(x,m).
$$
\begin{equation}\label{keyeqten}
\left(U(H_{v'(\cdot+ i\e),\alpha}-E)^{-1}U^{-1}\delta_\ell\right)(\theta,m)=\left((L_{v'(\cdot+ i\e),\alpha}-E)^{-1}\delta_\ell\right)(\theta,m).  
\end{equation}
We denote
\begin{align}\label{zx1ten}
f_{v'(\cdot+ i\e)}(E,x,m)&=\left((H_{v'(\cdot+ i\e),\alpha}-E)^{-1}U^{-1}\delta_\ell\right)(x,m)\\ \nonumber
&=\langle \delta_m,e^{-2\pi i\ell x}(H_{v'(x+ i\e),\alpha,x}-E)^{-1}\delta_0\rangle,
\end{align}

\begin{align}\label{final3ten}
g_{v'(\cdot+i\e)}(E,\theta,m)&=\left(U(H_{v'(\cdot+ i\e),\alpha}-E)^{-1}U^{-1}\delta_\ell\right)(\theta,m)\\ \nonumber
&=\left((L_{v'(\cdot+ i\e),\alpha}-E)^{-1}\delta_\ell\right)(\theta,m)
\end{align}
By \eqref{dual map}, we have
$$
g_{v'(\cdot+i\e)}(E,\theta,m)=\sum_{p\in\Z}\int_{\T}e^{2\pi im\eta}e^{2\pi ip(m\alpha+\theta)}f_{v'(\cdot+ i\e)}(E,\eta,p)d\eta,
$$
By \eqref{gjne1ten} and \eqref{zx1ten},  for any $v'\in U_v$, $E\in  I_{\delta_0}(E_0)$ and any $\delta>0$ sufficiently small, we have
\begin{equation}\label{new15}
\sup_{\theta\in\T}|g_{v'(\cdot+ i\e)}(E,\theta,m)|\leq C(v,\delta)e^{|m-\ell|\delta}e^{2\pi  (m-\ell)(\e-\e_1(E))}.
\end{equation}
If $\omega(E)=1$, we choose a fixed $\e_0>0$ with $6\e_0<h_1$, we further have  (possibly shring the parameter neighborhoods dependin on $\e_0$)
\begin{align}\label{new15'}
|g_{v'(\cdot\pm i\e_0)}(m)-g_{v(\cdot\pm i\e_0)}(m)|_{r,h} \leq C(v,\e_0)e^{4\pi|m-\ell|\e_0}|v'-v|_{h_1-\e_0}
\end{align}
for some $r,h>0$.

Let $\e_n(E)$ be the last turning point of $L_\e^{v_n}(E)$ \footnote{Note that $v_n$ is a trigonometric polynomial, there are $l(n)$ turning points (counting multiplicities).} and $\e'>\sup_{E\in I_{\delta_0}(E_0)}\e_n(E)$. By the same argument as above, $(H_{v_n(\cdot+i\e'),\alpha,x}-E)^{-1}$ exists for all $x\in\T$ and $E\in I_{\delta_0}(E_0)$, thus we can define $g_{v_n(\cdot+i\e')}(E,\theta,m)$ as above.
Note that no uniform bound on it or on this large-deformation
resolvent is used below. Uniform estimates will use only the fixed
admissible deformation $\e$.

For any $\theta\in\T$ and $E\in I_{\delta_0}(E_0)$,  we define
\begin{equation}\label{10.7}
h^\ell_n(E,\theta,m)=-e^{-2\pi (m-\ell)\e'}g_{v_n(\cdot+i\e')}(E,\theta,m)+e^{-2\pi (m-\ell)\e}g_{v_n(\cdot+i\e)}(E,\theta,m).
\end{equation}
\begin{equation}\label{10.711}
l^\ell_n(E,\theta,m)=-e^{-2\pi (m-\ell)\e}g_{v_n(\cdot+i\e)}(E,\theta,m)+e^{2\pi (m-\ell)\e}g_{v_n(\cdot-i\e)}(E,\theta,m).
\end{equation}
The signs of these weights follow from
\[
(L_{v_n(\cdot+i\e),\alpha,\theta}-E)
\big(e^{2\pi(\cdot-\ell)\e}u\big)(m)
=e^{2\pi(m-\ell)\e}(L_{v_n,\alpha,\theta}-E)u(m).
\]  
By \eqref{final3ten},  $\{h^\ell_n(E,\theta,m)\}_{m\in\Z}$ and $\{l_n^\ell(E,\theta,m)\}_{m\in\Z}$ are solutions of $(L_{v_n,\alpha,\theta}-E)u=0$. Moreover, by \eqref{new15'} with $\e_1(E)=0$, we have $l^\ell_n(m)\in C^\omega( I_{\delta_0}(E_0)\times \T,\C)$,
\begin{equation}\label{new200}
|l^\ell_n(m)-l^\ell_{n+1}(m)|_{r,h}\leq Ce^{-\frac{\delta'}{2}n}
\end{equation}
for some $r,h,\delta'>0$ and $-n-1\leq m,\ell\leq n$. 


For any $-l(n)\leq \ell\leq l(n)-1$, we let
\begin{equation}\label{ge1}
\vec{h}^\ell_n(E,\theta,m)=\begin{pmatrix}h^\ell_n(E,\theta,m+l(n)-1)\\ h^\ell_n(E,\theta,m+l(n)-2)\\ \vdots \\h^\ell_n(E,\theta,m-l(n))\end{pmatrix},\ \ \vec{l}^\ell_n(E,\theta,m)=\begin{pmatrix}l^\ell_n(E,\theta,m+l(n)-1)\\ l^\ell_n(E,\theta,m+l(n)-2)\\ \vdots \\l^\ell_n(E,\theta,m-l(n))\end{pmatrix},
\end{equation}
\begin{equation}\label{ge2}
H_n(E,\theta,m)=\begin{pmatrix} \vec{h}_n^{l(n)-1}(E,\theta,m)&\cdots &\vec{h}_{n}^{-l(n)}(E,\theta,m)\end{pmatrix},
\end{equation}
\begin{equation}\label{ge211}
L_n(E,\theta,m)=\begin{pmatrix} \vec{l}_n^{l(n)-1}(E,\theta,m)&\cdots &\vec{l}_{n}^{-l(n)}(E,\theta,m)\end{pmatrix}.
\end{equation}

\begin{Lemma}\label{add1}
For $-l(n)\leq \ell\leq l(n)-1$, we have that 
\begin{enumerate}
\item $\vec{h}^\ell_n(E,\theta,0)\in E^n_s(\theta)$ and  ${\rm Rank} (H_n(E,\theta,0))=l(n)-1$;
\item $\vec{l}^\ell_n(E,\theta)\in E^n_c(\theta)$ and  ${\rm Rank} (L_n(E,\theta,0))=2$.
\end{enumerate}
\end{Lemma}
\begin{pf}
By \eqref{sub10}, for $n$ sufficiently large and $E\in I_{\delta_0}(E_0)$, there exists an analytic  invariant decomposition of $\C^{2l(n)}$,
$$
\C^{2l(n)}=E^n_s(\theta)\oplus E^n_c(\theta)\oplus E^n_u(\theta).
$$
By Proposition \ref{trivial}, there are linearly independent vectors $\{f^n_j(E,\theta)\}_{j=1}^{l(n)-1}\in E_s^n(\theta)$, $\{v_j^n(E,\theta)\}_{j=1}^{2}\in E^n_c(\theta)$ and $\{g_j^n(E,\theta)\}_{j=1}^{l(n)-1}\in E_u^n(\theta)$ depending analytically on both $E$ and $\theta$. Moreover, by Theorem \ref{1general}, we have
\begin{equation}\label{zx2}
\lim\limits_{n\rightarrow\infty} \gamma_1^n(E)=2\pi \e_1(E),
\end{equation}
where $\gamma_1^n(E)$ is the minimal non-negative Lyapunov exponent of $(\alpha,L_{E,v_n})$.

Let $\{f_j^n(E,\theta,\ell)\}_{j=1}^{l(n)-1}$, $\{g_j^n(E,\theta,\ell)\}_{j=1}^{l(n)-1}$, $\{v_j^n(E,\theta,\ell)\}_{j=1}^{2}$ be  $2l(n)$ linearly independent solutions of $L_{v_n,\alpha,\theta}u=Eu$ with the above corresponding initial datum. Then together with \eqref{zx2}, for any $\delta>0$ sufficiently small, we have
\begin{equation}\label{sin1}
\limsup\limits_{m\rightarrow \infty}\frac{1}{2m}\ln\sum\limits_{\ell=-l(n)}^{l(n)-1} |f_j^n(E,\theta,m+\ell)|^2<-2\pi \e_1(E)-\delta,\ \  1 \leq j\leq l(n)-1,
\end{equation}
\begin{equation}\label{sin2}
\limsup\limits_{m\rightarrow \infty}\frac{1}{2m}\ln\sum\limits_{\ell=-l(n)}^{l(n)-1} |g_j^n(E,\theta,-m+\ell)|^2<-2\pi \e_1(E)-\delta,\ \  1 \leq j\leq l(n)-1,
\end{equation}
\begin{equation}\label{sin3}
\limsup\limits_{m\rightarrow \infty}\frac{1}{2|m|}\ln\sum\limits_{\ell=-l(n)}^{l(n)-1} |v_j^n(E,\theta,m+\ell)|^2\leq 2\pi\e_1(E)+\delta/10,\ \ 1\leq j\leq 2.
\end{equation}
Note that $\{e^{ \pm2\pi \ell\e}f_j^n(E,\theta,\ell)\}_{j=1}^{l(n)-1}$, $\{e^{ \pm2\pi \ell\e}g_j^n(E,\theta,\ell)\}_{j=1}^{l(n)-1}$, $\{e^{\pm2\pi \ell\e}v_j^n(E,\theta,\ell)\}_{j=1}^{2}$ are $2l(n)$ linearly  independent solutions of $L_{v_n(\cdot\pm i\e),\alpha,\theta}u=Eu$. For $\e_1(E)+\delta/8\pi<\e<\e_1(E)+\delta/4\pi$, we have
\begin{align}\label{so1}
e^{\pm2\pi \ell\e}f_j^n(E,\theta,\ell)\in\ell^2(\Z^+), \ \ 1\leq j\leq l(n)-1,
\end{align}
\begin{align}\label{so2}
e^{\pm2\pi \ell\e}g_j^n(E,\theta,\ell)\in\ell^2(\Z^-), \ \ 1\leq j\leq l(n)-1,
\end{align}
\begin{align}\label{so3}
e^{\pm2\pi \ell\e}v_j^n(E,\theta,\ell) \in\ell^2(\Z^\mp),\ \ 1\leq j\leq 2.
\end{align}
For any $\e'>\e_n(E)$, we have
\begin{align}\label{so3'}
e^{2\pi \ell\e'}f_j^n(E,\theta,\ell), \ \ e^{2\pi \ell\e'}g_j^n(E,\theta,\ell), \ \ e^{2\pi \ell\e'}v_j^n(E,\theta,\ell) \in\ell^2(\Z^-).
\end{align}

By \eqref{final3ten}, \eqref{so1}-\eqref{so3'} and Lemma \ref{basic}, we have
\begin{align}\label{g44}
&g_{v_n(\cdot+ i\e)}(E,\theta,m)=\langle \delta_m, (L_{v_n(\cdot+ i\e),\alpha,\theta}-E)^{-1}\delta_\ell\rangle\\ \nonumber
=&\begin{cases}
\frac{e^{2\pi (m-\ell)\e}\left(\sum\limits_{j=1}^{l(n)-1}f^n_j(E,\theta,m)\Phi^n_{1,j}(E,\theta,\ell)\right)}{\hat{v}_{l(n)}\det{\Phi^n(E,\theta,\ell)}} &\text{$m\geq \ell+1$}\\
-\frac{e^{2\pi (m-\ell)\e}\left(\sum\limits_{j=1}^{2}v_j^n(E,\theta,m)\Phi^n_{1,l(n)-1+j}(E,\theta,\ell)+\sum\limits_{j=1}^{l(n)-1}g^n_j(E,\theta,m)\Phi^n_{1,l(n)+1+j}(E,\theta,\ell)\right)}{\hat{v}_{l(n)}\det{\Phi^n(E,\theta,\ell)}} &\text{$m\leq \ell$}
\end{cases},
\end{align}%
\begin{align}\label{g44ge}
&g_{v_n(\cdot- i\e)}(E,\theta,m)=\langle \delta_m, (L_{v_n(\cdot- i\e),\alpha,\theta}-E)^{-1}\delta_\ell\rangle\\ \nonumber
=&\begin{cases}
\frac{e^{-2\pi (m-\ell)\e}\left(\sum\limits_{j=1}^{l(n)-1}f^n_j(E,\theta,m)\Phi^n_{1,j}(E,\theta,\ell)+\sum\limits_{j=1}^{2}v_j^n(E,\theta,m)\Phi^n_{1,l(n)-1+j}(E,\theta,\ell)\right)}{\hat{v}_{l(n)}\det{\Phi^n(E,\theta,\ell)}} &\text{$m\geq \ell+1$}\\
-\frac{e^{-2\pi (m-\ell)\e}\left(\sum\limits_{j=1}^{l(n)-1}g^n_j(E,\theta,m)\Phi^n_{1,l(n)+1+j}(E,\theta,\ell)\right)}{\hat{v}_{l(n)}\det{\Phi^n(E,\theta,\ell)}} &\text{$m\leq \ell$}
\end{cases},
\end{align}

{\begin{align}\label{g44'}
&g_{v_n(\cdot+ i\e')}(E,\theta,m)=\langle \delta_m, (L_{v_n(\cdot+ i\e'),\alpha,\theta}-E)^{-1}\delta_\ell\rangle\\ \nonumber
=&\begin{cases}
0 &\text{$m\geq \ell+1$}\\
\frac{\sum\limits_{j=1}^{l(n)-1}\left(f^n_j(E,\theta,m)\Phi^n_{1,j}(E,\theta,\ell)+g^n_j(E,\theta,m)\Phi^n_{1,l(n)+1+j}(E,\theta,\ell)\right)+\sum\limits_{j=1}^{2}v_j^n(E,\theta,m)\Phi^n_{1,l(n)-1+j}(E,\theta,\ell)}{-e^{-2\pi (m-\ell)\e'}\hat{v}_{l(n)}\det{\Phi^n(E,\theta,\ell)}} &\text{$m\leq \ell$}
\end{cases},
\end{align}}  
where
{\tiny $$
\Phi^n(E,\theta,\ell)=\begin{pmatrix}f^n_1(E,\theta,\ell+l(n))&\cdots&v_1^n(E,\theta,\ell+l(n))&v_{2}^n(E,\theta,\ell+l(n))&\cdots &g_{l(n)-1}^n(E,\theta,\ell+l(n))\\ f^n_1(E,\theta,\ell+l(n)-1)&\cdots&v_1^n(E,\theta,\ell+l(n)-1)&v_{2}^n(E,\theta,\ell+l(n)-1)&\cdots &g_{l(n)-1}^n(E,\theta,\ell+l(n)-1)\\ \vdots& &\vdots &\vdots &&\vdots\\
f^n_1(E,\theta,\ell-l(n)+1)&\cdots&v_1^n(E,\theta,\ell-l(n)+1)&v_{2}^n(E,\theta,\ell-l(n)+1)&\cdots &g_{l(n)-1}^n(E,\theta,\ell-l(n)+1)\end{pmatrix}.
$$}

Let 
$$
R_n(E,\theta,m)=\begin{pmatrix}
v_1^n(E,\theta,m+l(n)-1)&v_2^n(E,\theta,m+l(n)-1)\\
v_1^n(E,\theta,m+l(n)-2)&v_2^n(E,\theta,m+l(n)-2)\\
\vdots&\vdots\\
v_1^n(E,\theta,m-l(n))&v_2^n(E,\theta,m-l(n))
\end{pmatrix},
$$
$$
F_n(E,\theta,m)=\begin{pmatrix}
f_1^n(E,\theta,m+l(n)-1)&f_2^n(E,\theta,m+l(n)-1)&\cdots&f_{l(n)-1}^n(E,\theta,m+l(n)-1)\\
f_1^n(E,\theta,m+l(n)-2)&f_2^n(E,\theta,m+l(n)-2)&\cdots& f_{l(n)-1}^n(E,\theta,m+l(n)-2)\\
\vdots&\vdots&&\vdots\\
f_1^n(E,\theta,m-l(n))&f_2^n(E,\theta,m-l(n))&\cdots&f_{l(n)-1}^n(E,\theta,m-l(n))
\end{pmatrix},
$$
$$
G_n(E,\theta,m)=\begin{pmatrix}
g_1^n(E,\theta, m+l(n)-1)&g_2^n(E,\theta, m+l(n)-1)&\cdots&g_{l(n)-1}^n(E,\theta, m+l(n)-1)\\
g_1^n(E,\theta, m+l(n)-2)&g_2^n(E,\theta, m+l(n)-2)&\cdots& g_{l(n)-1}^n(E,\theta, m+l(n)-2)\\
\vdots&\vdots&&\vdots\\
g_1^n(E,\theta,m-l(n))&g_2^n(E,\theta, m-l(n))&\cdots&g_{l(n)-1}^n(E,\theta, m-l(n))
\end{pmatrix}.
$$
By symplectic invariance and  \eqref{sin1}-\eqref{sin3}, we have
\begin{align*}
&F_n(E,\theta,m)^*S_nF_n(E,\theta,m)=F_n(E,\theta,0)^*S_nR_n(E,\theta,0)\\=&G_n(E,\theta,m)^*S_nG_n(E,\theta,m)=G_n(E,\theta,0)^*S_nR_n(E,\theta,0)=0.
\end{align*}
It follows that
\begin{align}\label{sin4}
\nonumber &\Phi^n(E,\theta,\ell)^*S_n\Phi^n(E,\theta,\ell)=\Phi^n(E,\theta,0)^*S_n\Phi^n(E,\theta,0)\\
=&\begin{pmatrix} &&A_n(E,\theta)\\ & \Omega_n(E,\theta)&\\  -A_n(E,\theta)^*&& \end{pmatrix}
\end{align}
where 
$$
\Omega_n(E,\theta)=R_n(E,\theta,0)^*S_nR_n(E,\theta,0),\ \ A_n(E,\theta)=F_n(E,\theta,0)^*S_nG_n(E,\theta,0).
$$

Let 
$$
F_n(E,\theta,m)=\begin{pmatrix}
f_1^n(E,\theta,m+l(n)-1)&f_2^n(E,\theta,m+l(n)-1)&\cdots&f_{l(n)-1}^n(\theta,m+l(n)-1)\\
f_1^n(E,\theta,m+l(n)-2)&f_2^n(E,\theta,m+l(n)-2)&\cdots& f_{l(n)-1}^n(E,\theta,m+l(n)-2)\\
\vdots&\vdots&&\vdots\\
f_1^n(E,\theta,m-l(n))&f_2^n(E,\theta,m-l(n))&\cdots&f_{l(n)-1}^n(E,\theta,m-l(n)),
\end{pmatrix},
$$
$$
G_n(E,\theta,m)=\begin{pmatrix}
g_1^n(E,\theta, m+l(n)-1)&g_2^n(E,\theta, m+l(n)-1)&\cdots&g_{l(n)-1}^n(E,\theta, m+l(n)-1)\\
g_1^n(E,\theta, m+l(n)-2)&g_2^n(E,\theta, m+l(n)-2)&\cdots& g_{l(n)-1}^n(E,\theta, m+l(n)-2)\\
\vdots&\vdots&&\vdots\\
g_1^n(E,\theta,m-l(n))&g_2^n(E,\theta, m-l(n))&\cdots&g_{l(n)-1}^n(E,\theta, m-l(n))
\end{pmatrix}.
$$

By \eqref{sin4} and Lemma \ref{gjygjy111}, we have that
\begin{equation}\label{g42}
\begin{pmatrix}
\Phi_{1,1}^n(E,\theta,\ell)\\
\Phi_{1,2}^n(E,\theta,\ell)\\
\vdots\\
\Phi_{1,l(n)-1}^n(E,\theta,\ell)
\end{pmatrix}=-\hat{v}_{l(n)}\det{\Phi^n(E,\theta,\ell)}A^*_n(E,\theta)^{-1}\begin{pmatrix}
\overline{g_1^n(E,\theta, \ell)}\\
\overline{g_2^n(E,\theta, \ell)}\\
\vdots\\
\overline{g_{l(n)-1}^n(E,\theta,\ell)}
\end{pmatrix},
\end{equation}
\begin{equation}\label{g42'}
\begin{pmatrix}
\Phi_{1,l(n)}^n(E,\theta,\ell)\\
\Phi_{1,l(n)+1}^n(E,\theta,\ell)
\end{pmatrix}=\hat{v}_{l(n)}\det{\Phi^n(E,\theta,\ell)}\Omega_n(E,\theta)^{-1}\begin{pmatrix}
\overline{v_1^n(E,\theta, \ell)}\\
\overline{v_2^n(E,\theta, \ell)}
\end{pmatrix}.
\end{equation}

By \eqref{10.7}, \eqref{10.711}, \eqref{ge1}, \eqref{ge2}, \eqref{ge211} and \eqref{g44}-\eqref{g42'}, we have
\begin{equation}\label{repre}
H_n(E,\theta,m)=F_n(E,\theta,m)\left(G_n(E,\theta,0)^*S_nF_n(E,\theta,0)\right)^{-1}G_n(E,\theta,0)^*,
\end{equation}
\begin{equation}\label{repre'}
L_n(E,\theta,m)=R_n(E,\theta,m)\left(R_n(E,\theta,0)^*S_nR_n(E,\theta,0)\right)^{-1}R_n(E,\theta,0)^*.
\end{equation}
It follows that $\vec{h}^\ell_n(E,\theta,0)\in E_s^n(\theta)$ and ${\rm Rank}(H_n(E,\theta,0))=l(n)-1$, $\vec{l}^\ell_n(E,\theta)\in E_c^n(\theta)$ and ${\rm Rank}(L_n(E,\theta,0))=2$.   
\end{pf}
\noindent {\bf Proof of (1):} Assume now $\omega(E)=1$, We need the following lemma
\begin{Lemma}[\cite{gj}]\label{eigen0}
For $L_{n}(E,\theta,0)$  as above, for any $\theta\in\T$ and $E\in I_{\delta_0}(E_0)$, $iL_n(E,\theta,0)$ has $1$ positive eigenvalues  $\mu^{n}(E,\theta)$ and $1$ negative eigenvalues $\kappa^n(E,\theta)$ with 
$$
|\mu^{n}(E,\theta)|, |\kappa^n(E,\theta)|\geq c_0>0
$$
where one may take $c_0=\frac{1}{2(1+|v|_0)}$ for all sufficiently
large $n$.
   
\end{Lemma}
\begin{pf}
Apply Lemma~\ref{lem:center-kernel-signature} to
\[
H=R_n(E,\theta,0),
\qquad
\Omega=R_n(E,\theta,0)^*S_nR_n(E,\theta,0).
\]
The verification of the signature hypothesis and the uniform lower
bound is given in Appendix~\ref{app:center-facts}.
\end{pf}
For any $m\leq l(n)$, we  denote 
$$
R^{m}_n(E,\theta)=\begin{pmatrix}
v_1^n(E,\theta,m-1)&v_2^n(E,\theta,m-1)\\
v_1^n(E,\theta,m-2)&v_2^n(E,\theta,m-2)\\
\vdots&\vdots\\
v_1^n(E,\theta,-m)&v_2^n(E,\theta,-m)
\end{pmatrix},$$
the matrix consisting of the middle $2m$ rows of $R_n(E,\theta).$ Let
\begin{equation}\label{irv6}
L_n^{m}(E,\theta):=R_n^{m}(E,\theta)(R_n(E,\theta,0)^*S_nR_n(E,\theta,0))^{-1}R_{n}^{m}(E,\theta,0)^*,
\end{equation}
\begin{equation}\label{defmn}
M_n^{m}(E,\theta):=R_{n}(E,\theta,0)(R_n(E,\theta,0)^*S_nR_n(E,\theta,0))^{-1}R_{n}^{m}(E,\theta)^*.
\end{equation} 
It is easily seen that $M^m_n(E,\theta)$ is a $2l(n)$ by $2m$ matrix consisting of $2m$
middle columns of $L_n(E,\theta,0),$ so 
$$
M^m_n(E,\theta)=\begin{pmatrix}
  \vec{l}_n^{m-1}(E,\theta,0)&\cdots &\vec{l}_{n}^{- m}(E,\theta,0)\end{pmatrix},
$$ 
and $L_{n}^m(E,,\theta)$ is a $2m$ by $2m$ matrix consisting $2m$ middle rows
of $M_n^m(E,\theta)$. In particular, $L_n^m(E,\theta)$ is a submatrix of
$M_n^m(E,\theta),$ and $M_n^m(E,\theta)$ is a submatrix of
$L_n(E,\theta,0)$. Let $M_n^m(E,\theta,j),j=-l(n),\ldots,l(n)-1,$ denote the $j$th
row of $M_n^m(E,\theta)$.
\begin{Lemma}\label{add2aaa}
Assume $\omega(E)=1$, there is  $n_0=l(n_0)$ sufficiently large, such that for $n>n_0$, we have that
\begin{itemize}
\item $|L_{n}^{n_0}-L_{n+1}^{n_0}|_{r,h},
  |M_{n}^{n_0}(j)-M_{n+1}^{n_0}(j)|_{r,h} \leq
  Ce^{-\frac{\delta}{4}n}, j=-l(n),\ldots,l(n)-1$;
{\item there are vectors $e_j\in C^\omega( I_{\delta_0}(E_0)\times \T,\C^{2n_0})$, $j=1,2$, such that
\begin{equation}\label{gn1'}
{\rm Rank}\left(M_{n}^{n_0}(E,\theta)(e_1(E,\theta),e_2(E,\theta)\right)=2,\ \ M_{n}^{n_0}(E,\theta)e_j(E,\theta)\in E_c^n(\theta).
\end{equation}}
\end{itemize}
\end{Lemma}
\begin{pf}
Note that by \eqref{new200} and the definition of $L_n^{m}$ and $M_n^m$, for $n$ sufficiently large, we have that
\begin{equation}\label{ext1}
|L^{m}_n-L^{m}_{n+1}|_{r,h}, \ \ |M_{n}^{m}(j)-M_{n+1}^{m}(j)|_{r,h} \leq
  Ce^{-\frac{\delta'}{4}n}, j=-l(n),\ldots,l(n)-1.
\end{equation}

Let $L^{n_0}(E,\theta)=\lim_{n\rightarrow \infty}L_n^{n_0}(E,\theta).$ We have
that $L^{n_0}(E,\theta)\in C^\omega( I_{\delta_0}(E_0)\times \T,M_{2n_0}(\C))$. Note
that by Lemma \ref{eigen0}, \eqref{ext1}, first choose $n_0=l(n_0)$ large enough, we have $|L^{n_0}(E,\theta)-L_{n_0}(E,\theta)|\leq Cn_0e^{-cn_0}<c_0/2$, thus by eigenvalue perturbation theory,
$$
{\rm Rank}(L^{n_0}(E,\theta))=2
$$ 
for $\theta\in \T$ and $E\in I_{\delta_0}(E_0)$.

Hence $L^{n_0}(E,\theta)$  is a constant rank holomorphic matrix on $I_{\delta_0}(E_0)\times \T$, It follows that ${\rm Ker} L^{n_0}$ is a holomorphic vector bundle. By Proposition \ref{trivial}, ${\rm Ker}L^{n_0}$ and $E/{\rm Ker}L^{n_0}$ where $E=(I_{\delta_0}(E_0)\times \T)\times \C^{2n_0}$ is the tangent bundle of $\C^{2n_0}$ are trivial. Thus there are globally defined linearly independent holomorphic functions $v_j(E,\theta)\in {\rm Ker}L^{n_0} (1\leq j\leq 2n_0-2)$ and $e_j(E,\theta) (1\leq j\leq 2)$ such that they form a basis of $C^{2n_0}$ for each $(E,\theta)\in I_{\delta_0}(E_0)\times \T$.

It follows that ${\rm Rank}(L^{n_0}(E,\theta)(e_1(E,\theta),e_2(E,\theta)))=2$ on $I_{\delta_0}(E_0)\times \T$. By \eqref{ext1}
$$
{\rm Rank}(L_n^{n_0}(E,\theta)(e_1(E,\theta),e_2(E,\theta)))=2
$$
for $n$ sufficiently large. Thus
 $$
 {\rm Rank}(M_n^{n_0}(E,\theta)(e_1(E,\theta),e_2(E,\theta)))\geq  {\rm Rank}(L_n^{n_0}(E,\theta)(e_1(E,\theta),e_2(E,\theta)))=2.
 $$
By Lemma \ref{add1} columns of $L_n(E,\theta)$ belong to
$E_c^n(\theta)$, by the definition of $M_n^{n_0}(E,\theta)$, so do all
its columns. Therefore we have that $M_{n}^{n_0}(E,\theta)e_j(E,\theta)\in E_c^n(\theta)$.

\end{pf}

For $\omega(E)=1$, we let
\begin{align}\label{so01}
\nonumber &O_n'(E,\theta)=R_{n}(E,\theta,0)(R_{n}(E,\theta,0)^*S_nR_{n}(E,\theta,0)^{-1}R_{n}^{n_0}(E,\theta)^*(e_1(E,\theta),e_2(E,\theta))\\
=&M_n^{n_0}(E,\theta)(e_1(E,\theta),e_2(E,\theta)).
\end{align}
By Lemma \ref{add2aaa}, the columns of $O_n'$ form a basis of $E_c^n(\theta)$. Moreover, by \eqref{ext1} and \eqref{so01},
\begin{equation}\label{final100}
|O'_n (j)-O'_{n+1}(j)|_{r,h} \leq Ce^{-\frac{\delta'}{20} n}, \ \  j=-l(n),\ldots,l(n)-1
\end{equation}
where $O'_{n+1}(j)$ is the $j$-th row of $O'_{n+1}$.

On the other hand, by  direct calculation we have
$$
O_n'(E,\theta)^*S_n O_n'(E,\theta)=-\begin{pmatrix}e_1(E,\theta)&e_2(E.\theta)\end{pmatrix}^*
L_n^{n_0}(E,\theta)\begin{pmatrix}e_1(E,\theta)&e_2(E,\theta)\end{pmatrix}$$
Note that by Lemma \ref{eigen0},  $iO_n'(E,\theta)^*S_n O_n'(E,\theta)$ has $1$ positive and $1$ negative eigenvalues for all $(E,\theta)\in I_{\delta_0}(E_0)\times \T$.

Let 
\begin{align}\label{lana16}
\nonumber \Omega'(E,\theta)&=-\begin{pmatrix}e_1(E,\theta)&e_2(E,\theta)\end{pmatrix}^*
L^{n_0}(E,\theta)\begin{pmatrix}e_1(E,\theta)&e_2(E,\theta)\end{pmatrix}\\
&=\lim_{n\rightarrow\infty}O_n'(E,\theta)^*S_n O_n'(E,\theta).
\end{align}
Again by Lemma \ref{eigen0} and eigenvalue perturbation theory,  $i\Omega'(E,\theta)$ has $1$ positive and $1$
negative eigenvalues for all $\theta\in\T$ and $E\in I_{\delta_0}(E_0)$. By the same argument as in Proposition \ref{newwe}, there is $P\in C^\omega(I_{\delta_0}(E_0)\times \T,GL(2,\C))$ such that
\begin{equation}\label{lana15}
P(E,\theta)^*\Omega'(E,\theta)P(E,\theta)=J.
\end{equation}

Let  
\begin{equation}\label{defon}
O_n(E,\theta)=O_n'(E,\theta)P(E,\theta).
\end{equation}
By \eqref{lana16} and \eqref{lana15},
\begin{equation}\label{lana18}
O_n(E,\theta)^*S_nO_n(E,\theta)\rightarrow J,
\end{equation}
Moreover,  by  \eqref{final100} and  \eqref{defon},  we have
\begin{equation}\label{lana17}
|O_n (j)-O_{n+1}(j)|_{r,h} \leq Ce^{-c n}, \ \  j=-l(n),\ldots,l(n)-1
\end{equation}
foe some $c,C>0$ where $O_n=\begin{pmatrix}O_n(l(n)-1)\\ O_n(l(n)-2)\\ \vdots\\ O_n(-l(n))\end{pmatrix}$.

By invariance of $E^n_c(\theta)$, there is $M_n(E,\theta)\in C^\omega(I_{\delta_0}(E_0)\times \T,GL(2,\C))$ such that
\begin{equation}\label{shift}
\begin{pmatrix}O_n(E,\theta,l(n))\\ O_n(E,\theta,l(n)-1)\\ \vdots\\ O_n(E,\theta,-l(n)+1)\end{pmatrix}=L_{E,v_n}(\theta)O_n(E,\theta)=O_n(E,\theta+\alpha)M_n(E,\theta).
\end{equation}
 By \eqref{lana17}, there exist  $M\in C^\omega(I_{\delta_0}(E_0)\times \T,GL(2,\C))$ such that
 \begin{equation}\label{lana19}
 |M_n- M|_{r,h}\leq Ce^{-cn}.
 \end{equation} 
By the same argument as in Proposition \ref{p1}, we have  for any $E\in I_{\delta_0}(E_0)$ and $\theta\in\T$,
$$
M_n(E,\theta)^*O_n(E,\theta+\alpha)^*S_nO_n(E,\theta+\alpha)M_n(E,\theta)=O_n(E,\theta)^*S_nO_n(E,\theta).
$$
By \eqref{lana18} and \eqref{lana19}, we have
$$
M(E,\theta)^*JM(E,\theta)=J.
$$
see also Lemma~\ref{lem:center-J-unitary}.

 For every finite-range
truncation, Theorem 4.1 gives
\[
\omega^{l(n)-1}(\alpha,L_{E,v_n})=0.
\]
By Lemma~\ref{lem:center-det-zero}, the determinant of the
corresponding center cocycle has zero winding. Since the center
cocycles converge analytically and remain invertible, zero winding
passes to the limiting center cocycle $M$. Lemma~
\ref{lem:projectively-real-factorization} therefore gives
\[
M(E,\theta)
=
e^{2\pi i\varphi(E,\theta)}{\bf C}(E,\theta),
\qquad
{\bf C}(E,\theta)\in SL(2,\mathbb R),
\]
with $\varphi$ and ${\bf C}$ analytic in $(E,\theta)$.

\noindent {\bf Proof of (2):}  
By the definition of $L_n(E,\theta,0)$, we have
$$
l_{ij}(E,\theta)=l_n^{j}(E,\theta,i).
$$
By \eqref{new15} for $\pm \e$, for any $\delta>0$, we have
$$
|l_{ij}(E,\theta)|=|l_n^j(E,\theta,i)|\leq C(v,\delta)e^{2\pi  |i-j|(\e_1(E)+\delta)}.
$$
Indeed, by the same argument as \eqref{new15}, for $z$ with $|z-E_0|<r$, we have
$$
|l_{ij}(z,\theta)|=|l_n^j(z,\theta,i)|\leq C(v,\delta)e^{2\pi  |i-j|(\max\{\e_1(z),\e_1(\bar{z})\}+\delta)}
$$
where $\e_1(z)$ is the first turning point of $L_\e^v(z)$.

By \eqref{defmn},  $M^{n_0}_n(E,\theta)$ is a $2l(n)$ by $2n_0$ matrix consisting of $2n_0$
middle columns of $L_n(E,\theta,0)$. If $\omega(E)=1$ (i.e., $\e_1(E)=0$), by  Cauchy estimate and the continuity of $\e_1(z)$, we have
\begin{equation*}
\|O_{n}(E,\theta)\|^2+\left\|\frac{\partial O_{n}(E,\theta)}{\partial E}\right\|^2\leq C(v,\delta)e^{4\pi\delta|l(n)|}.
\end{equation*}

\noindent {\bf Proof of (3):} Consider the singular value decomposition of $H_n(E,\theta,0)$, i.e., 
\begin{equation}\label{sin6}
H_n(E,\theta,0)=U_n(E,\theta,0)\Sigma_n(E,\theta,0)V_n(E,\theta,0)
\end{equation}
where $U_n(E,\theta,0) $ and $V_n(E,\theta,0)$ are unitary matrices and
$$
\Sigma_n(E,\theta,0)={\rm diag}\{\sigma_1(H_n(E,\theta,0)),\cdots, \sigma_{l(n)-1}(H_n(E,\theta,0)), 0\}.
$$
By \eqref{repre},  the first $l(n)-1$ columns of $U_n(E,\theta,0)$, denoted by $U_n^1(E,\theta,0)$ form an orthonormal basis of $E_s^n(\theta)$.

On the other hand, again by \eqref{repre},  the columns of $H_n(E,\theta,0)^*$ belong to $E_u^n(\theta)$, thus  the first $l(n)-1$ columns of $V_n(E,\theta,0)^*$, denoted by $V_n^1(E,\theta,0)^*$ form an orthonomal basis of $E_u^n(\theta)$.

Moreover, let
$$
\Sigma^+(E,\theta,0)={\rm diag}\{\sigma^{-1}_1(H_n(E,\theta,0)),\cdots, \sigma^{-1}_{l(n)-1}(H_n(E,\theta,0))\},
$$
by \eqref{repre} and \eqref{sin6}, we have
\begin{align*}
&V_n^1(E,\theta,0)^*S_nU_n^1(E,\theta,0)\\
=&\Sigma^+(E,\theta,0)U_n(E,\theta,0)^*H_n(E,\theta,0)S_n H_{n}(E,\theta,0)V_n(E,\theta,0)^*\Sigma^+(E,\theta,0)\\
=&\Sigma^+(E,\theta,0)U_n(E,\theta,0)^*H_{n}(E,\theta,0)V_n(E,\theta,0)^*\Sigma^+(E,\theta,0)=\Sigma^+(E,\theta,0).
\end{align*}
We let $u_n^i(E,\theta)$ be the $i$-th column of $U_n(E,\theta,0)$, $v_n^i(E,\theta)$ be the $i$-th column of $-V_n(E,\theta,0)^*$ and $\sigma_n^i(E,\theta)=\sigma^{-1}_i(H_n(E,\theta,0))$.

For any  $m\geq 0$, we have
$$
(L_{E,v_n})_m(\theta) U_n^1(E,\theta,0)=(L_{E,v_n})_m(\theta) H_n(E,\theta,0)V(E,\theta,0)^*\Sigma^+(E,\theta,0).
$$
Note that
\begin{align}\label{ge3}
\nonumber &(L_{E,v_n})_m(\theta) H_n(E,\theta,0)\\ \nonumber
=&(L_{E,v_n})_m(\theta)F_n(E,\theta,0)\left(G_n(E,\theta,0)^*S_nF_n(E,\theta,0)\right)^{-1}G_n(E,\theta,0)^*\\
=&F_n(E,\theta,m)\left(G_n(E,\theta,0)^*S_nF_n(E,\theta,0)\right)^{-1}G_n(E,\theta,0)^*=H_n(E,\theta,m).
\end{align}
For $m\geq 0$ and $-l(n)\leq \ell\leq l(n)-1$, by
Remark \ref{gz20}, we have  
\begin{align}\label{ggg1}
h^\ell_n(E,\theta,m)&=-e^{-2\pi (m-\ell)\e'}g_{v_n(\cdot+i\e')}(E,\theta,m)+e^{-2\pi (m-\ell)\e}g_{v_n(\cdot+i\e)}(E,\theta,m)\\
&=e^{-2\pi (m-\ell)\e}g_{v_n(\cdot+i\e)}(E,\theta,m).
\end{align}
Hence 
\begin{equation}\label{sin10}
\begin{pmatrix}u_{n}^i(E,\theta,m+l(n)-1)\\ u_{n}^i(E,\theta,m+l(n)-2)\\ \vdots\\
u_{n}^i(E,\theta,m-l(n))\end{pmatrix}=-\text{$i$-th  column of\ \ }H_n(E,\theta,m)V(E,\theta,0)^*\sigma_n^i(E,\theta).
\end{equation}

Choose $\e_1(E)<\e<\e_1(E)+\delta/2\pi$ within the first
slope-one segment. By \eqref{10.7} and \eqref{ggg1}, we have
\begin{align*}
|h_n^\ell(E,\theta,m)|&\leq Ce^{|m-\ell|\delta}e^{-2\pi\e_1(E)(m-\ell)}.
\end{align*}
Shrink the energy interval once so that its closure lies in the
Type I neighborhood.  For each fixed $\delta$, the constants are
uniform on this closure, using finitely many admissible heights.
By \eqref{ge1}, \eqref{ge2} and \eqref{sin10}, this gives, for $m\geq0$,   
\begin{equation}\label{new155}
|u_n^i(E,\theta,m)| \leq C(v,\delta)\sigma_n^i(E,\theta)e^{-2\pi(m-\ell)(\e_1(E)+\delta)}e^{|m-\ell|\delta}.
\end{equation}
the result for the unstable subbundle follows exactly the same way. 

For the remaining  part, note that $L_n(E,\theta,0)$ is skew-hermitian, thus  
\begin{equation}\label{aksaks}
L_n(E,\theta,0)=U'_n(E,\theta,0)^*\Sigma'_n(E,\theta,0)U'_n(E,\theta,0)
\end{equation}
where $U'_n(E,\theta,0) $ is unitary and
$$
\Sigma'_n(E,\theta,0)={\rm diag}\{-i\mu^n(E,\theta),-i\kappa^n(E,\theta), 0, \cdots, 0\}.
$$
where $\mu^n(E,\theta)>0$ and $\kappa^n(E,\theta)<0$ are the non-zero eiegenvalues of $iL_n(E,\theta,0)$ in Lemma \ref{eigen0}.
By \eqref{repre'},  the first two columns of $U'_n(E,\theta,0)^*$, denoted by $W_n^1(E,\theta,0)$ form an orthonormal basis of $E_c^n(\theta)$.

Let 
$$
(w_n^1(E,\theta), w_n^2(E,\theta)=W_n^1(E,\theta,0)\begin{pmatrix}{\sqrt{\mu^n(E,\theta)}}&0\\ 0& {\sqrt{-\kappa^n(E,\theta)}}\end{pmatrix}M,\ \ M=\frac{1}{\sqrt{2}}\begin{pmatrix}i&1\\ -i&1\end{pmatrix},
$$
then by \eqref{aksaks} and \eqref{repre'}, we have
$$
\begin{pmatrix}
w_n^1(E,\theta)&w_n^2(E,\theta)
\end{pmatrix}^*S_n \begin{pmatrix}
w_n^1(E,\theta)&w_n^2(E,\theta)
\end{pmatrix}=J.
$$
By Lemma \ref{eigen0} and the same argument as above, we have
$$
\left|w_n^i(E,\theta,\pm m)\right|\leq C(v)e^{|m-l(n)|\delta}(e^{-2\pi \e_1(E)(\pm m-l(n))}+e^{2\pi \e_1(E)(\pm m-l(n))}),\ \  \forall 1\leq i\leq 2.
$$
We thus finish the proof of (3).
\end{pf}


\section{Long-range Kotani theory: Proof of the (super)critical case of Theorem \ref{main12}}\label{slongkotani}
We let $L_\e^v(z)=L(\alpha,S_z^v(\cdot+i\e))$, $\omega(z)=\omega(\alpha,S_z^v)$, $\bar{\omega}(z)=\bar{\omega}(\alpha,S_z^v)$ and let $\e_1(z)$ be the first turning point of $L^v_\e(z)$, We denote
\begin{equation}\label{cos1}
\Sigma_{v,\alpha}^{1,+}=\{E\in \Sigma_{v,\alpha}:  \omega(E)=1\}.
\end{equation}

\subsection{Rotations reducibility for the center of dual long-range operator}
In this subsection, we establish Kotani theory for the dual center of a long-range operator. As a corollary, we prove $C^\omega$-rotations reducibility of the two dimensional limiting center of the dual long-range operator. 

\begin{Theorem}\label{Kotani-longrange}
Let $\alpha\in\R\backslash\Q$ and $v\in C^\omega(\T,\R)$. If there is an interval $I\subset \Sigma^{1,+}_{v,\alpha}$, then there are $B\in C^\omega(I'\times \T,SL(2,\R))$ and $\psi\in C^\omega(I'\times \T,\R)$  for some $I'\subset I$, such that   
\begin{equation}\label{cos2}
B(E,\theta+\alpha)^{-1}{\bf C}(E,\theta)B(E,\theta)=R_{\psi(E,\theta)}
\end{equation}
where ${\bf C}$ is the limiting center cocycle obtained in Theorem \ref{theorem-main1general}. 
\end{Theorem}
We shall also use the nonconstancy of the mean center rotation, proved
in Proposition~\ref{prop:mean-rotation-nonconstant} of
Appendix~\ref{app:mean-rotation}.
\begin{pf}
Since $I\subset \Sigma_{v,\alpha}^{1,+}$, by \eqref{cos1}, for any $E\in I$, we have $\omega(E)=1$.  By Theorem \ref{1general}, $\e_1(E)=L(\alpha,{\bf C}(E,\cdot))=0$. 

Now we fix $E_0\in I$. By (1) of Theorem \ref{theorem-main1general}, there exist $\delta_0>0$ with $I_{\delta_0}(E_0)\subset I$ and $O_n\in C^\omega( I_{\delta_0}(E_0)\times\T,GL_{2l(n)\times 2}(\C))$,  $M_n\in C^\omega(I_{\delta_0}(E_0)\times \T,GL(2,\C))$, such that 
$$
L_{E,v_n}(\theta)O_n(E,\theta)=O_n(E,\theta+\alpha)M_n(E,\theta)
$$
with 
\begin{equation}\label{cenestlong}
|O^*_nS_nO_n-J|_{r,h},\ \ |M_n-e^{2\pi i\phi}{\bf C}|_{r,h}\leq Ce^{-cn}
\end{equation}
for some $r,h,c,C>0$, $\phi\in C^\omega(I_{\delta_0}(E_0)\times\T,\R)$ and ${\bf C}\in C^\omega( I_{\delta_0}(E_0)\times \T,SL(2,\R))$. If we denote
$$
O_n=\begin{pmatrix}O_n(l(n)-1)\\ \vdots\\ O_{n}(0)\\ \vdots\\ O_n(-l(n))\end{pmatrix},
$$ 
then for $-l(n)\leq j\leq l(n)-1$, we have 
\begin{equation}\label{CE222con}
|O_n(j)- O_{n+1}(j)|_{r,h}\leq Ce^{-cn}.
\end{equation}
whee $O_n(j)$ is the $j$-th row of $O_n$.

In  the following, we perform the convergence argument for several key steps in Kotani-theoretic estimates.

{\bf Step 1: Convergence argument for Proposition \ref{newwe'}:}
\begin{Proposition}\label{newwe'conv}
There exist $0<\delta_1(E_0)<\min\{\delta_0,\frac{r}{8}\}$ (uniformly in $n$) and two linearly independent vectors 
$$
U_n(E,\theta), V_n(E,\theta)\in E^n_c(\theta)
$$ 
for  $(E,\theta)\in I_{\delta_1}(E_0)\times \T$,  where $E_c^n(\theta)$ is defined in \eqref{sub10}, depending analytically on $E$ and  $\theta$, such that 
\begin{equation}\label{gjy2ncon}
\left|\begin{pmatrix}
U_n& V_n
\end{pmatrix}^*S\begin{pmatrix}
U_n& V_n\end{pmatrix}-J\right|_{r',h'}\leq Ce^{-cn},
\end{equation}
\begin{align}\label{gjy3ncon}
\nonumber \sup_{E+i\delta\in R_{\delta_1}(E_0)}\sup_{\theta\in\T}&\left\|\begin{pmatrix}
U_n(E+i\delta,\theta)& V_n(E+i\delta,\theta)
\end{pmatrix}^*S\begin{pmatrix}
U_n(E+i\delta,\theta)& V_n(E+i\delta,\theta)\end{pmatrix}-J\right\|\\
&\leq C(|E-E_0||\delta|+\delta^2)+Ce^{-cn},
\end{align}
for some $0<r'<r$, $0<h'<h$ where $C$ is uniform in $n$ and only depends on $E_0$. Moreover, if we denote
$$
\begin{pmatrix}U_n&V_n\end{pmatrix}=\begin{pmatrix}U_n(l(n)-1)&V_n(l9n)-1)\\ \vdots& \vdots\\ U_{n}(0)&V_{n}(0)\\ \vdots& \vdots\\ U_n(-l(n))&V_n(-l(n))\end{pmatrix},
$$ 
then for $-l(n)\leq j\leq l(n)-1$, we have 
\begin{equation}\label{irv5}
|U_n(j)-U_{n+1}(j)|_{r',h'}+|V_n(j)-V_{n+1}(j)|_{r',h'}\leq Ce^{-cn}.
\end{equation}
\end{Proposition}
\begin{pf}
For any $|z-E_0|<r$ and $z=E+i\delta$, we have
\begin{equation}\label{cw1con}
O_n(z,\theta)=O_n(E,\theta)+\frac{\partial O_n(E,\theta)}{\partial E}i\delta+I_n(z,\theta)\delta^2,
\end{equation}
for some $I_n(z,\theta)$. We denote 
\begin{equation}\label{pku1}
A_n(E,\theta)=O_n(E,\theta)^*S_n\frac{\partial O_n(E,\theta)}{\partial E}.
\end{equation}
By \eqref{cenestlong}, \eqref{pku1} and Cauchy estimates, we have
\begin{equation}\label{bimsa1}
|A_n-A_n^*|_{\frac{r}{2},h}=|\partial_E(O_n^*S_nO_n-J)|_{\frac{r}{2},h}\leq \frac{4C}{r}e^{-cn}.
\end{equation}
Expanding $O_n$ as in \eqref{cw1con}, by \eqref{cenestlong} we have for all $\theta\in\T$, $z=E+i\delta\in R_{r}(E_0)$,
\begin{align}\label{CW3con}
\left\|O_n(z,\theta)^*S_n O_n(z,\theta)-J-(A_n(E,\theta)+A_n^*(E,\theta)\right)i\delta-K_n(z,\theta)\delta^2)\|\leq C(e^{-cn}+\delta^2)
\end{align}
where 
\begin{align}\label{defkn}
\nonumber &K_n(z,\theta)=O_n(E,\theta)^*S_n I_n(z,\theta)+I_n(z,\theta)^*S_nO_n(E,\theta)+\frac{\partial O_n(E,\theta)^*}{\partial E}S_n\frac{\partial O_n(E,\theta)}{\partial E}\\
-&i\delta\frac{\partial O_n(E,\theta)^*}{\partial E}S_nI_n(z,\theta)+i\delta I_n(z,\theta)^*S_n \frac{\partial O_n(E,\theta)}{\partial E}+I_n(z,\theta)^*S_nI_n(z,\theta)\delta^2.
\end{align}

Note that by \eqref{CE222con} and Cauchy estimates, for any $j=-l(n),\ldots,l(n)-1$,
\begin{equation}\label{final100mm'}
\left|O_n(j)-O_{n+1}(j)\right|_{\frac{r}{2},h}+\left|\frac{\partial O_n(j)}{\partial E}-\frac{\partial O_{n+1}(j)}{\partial E}\right|_{\frac{r}{2},h} \leq \frac{4C}{r}e^{-c n}.
\end{equation}

By  \eqref{irvp} in Theorem \ref{theorem-main1general} and the fact that $\e_1(E)=0$ for any $E\in I$, we have for any $\e>0$, there is $\delta_0'<\min\{\delta_0,\frac{r}{4}\}$ and $C(\e.r)$ such that
\begin{equation}\label{irv}
\|O_{n}(E,\theta)\|^2+\|\frac{\partial O_{n}(E,\theta)}{\partial E}\|^2\leq C(\e,r)e^{\e|n|},\ \ \forall (E,\theta)\in I_{\delta_0'}(E_0)\times \T.
\end{equation}

By \eqref{final100mm'}, \eqref{irv} and Lemma \ref{pkule} with $D_n=O_n$, $B_n=\frac{\partial O_n}{\partial E}$ and using $\|S_n\|\leq C$ uniformly in $n$, for any $E\in I_{\delta_0'}(E_0)$ and $\theta\in\T$, we have that
\begin{align}\label{irv1}
&\ \ \|A_n(E,\theta)-A_{n+1}(E,\theta)\|\leq C\left(\|O_{n+1}(E,\theta)\|^2+\|\frac{\partial O_{n+1}(E,\theta)}{\partial E}\|^2\right)\\ \nonumber
\cdot &\left(|v_{n+1}|+\sum_{j=-l(n)}^{l(n)-1}\|O_n(E,\theta,j)-O_{n+1}(E,\theta,j)\|+\|\frac{\partial O_{n}(E,\theta,j)}{\partial E}-\frac{\partial O_{n+1}(E,\theta,j)}{\partial E}\|\right)\\\nonumber
&\leq C(\e,r)e^{\e|n|}(e^{-cn}+\frac{C}{r}e^{-cn})\leq C(\e,r)e^{-\frac{c}{2}n}.
\end{align}
Notice that  for any $z=E+i\delta\in R_{\delta_0'}(E_0)$ and $\theta\in\T$, by  \eqref{cw1con} and \eqref{final100mm'}, we have
\begin{equation}\label{ines}
\|I_n(z,\theta,j)-I_{n+1}(z,\theta,j)\|\leq \frac{16C}{r}e^{-\frac{c}{2}n},\ \ \|I_{n}(z,\theta)\|\leq C(\e,r)e^{\e|n|}
\end{equation}
for $-l(n) \leq j\leq l(n)-1$ where $I_n(z,\theta,j)$ is the $j$-th row of $I_n(z,\theta)$.

Each term in $K_n$ is of the form $X_n^*S_nY_n$ where $X_n,Y_n\in \{O_n(E,\theta),\partial_E O_n(E,\theta),I_n(z,\theta)\}$. Applying Lemma \ref{pkule} to each such term, and using $\|S_n\|\leq C$, \eqref{final100mm'} \eqref{irv},  \eqref{ines}, we obtain for any $z=E+i\delta\in R_{\delta_0'}(E_0)$ and $\theta\in\T$, 
\begin{equation}\label{irv2}
\|K_n(z,\theta)-K_{n+1}(z,\theta)\|\leq  C(\e,r)e^{-\frac{c}{4}n}.
\end{equation}

Let
\begin{align}\label{CE1con}
C_n(z,\theta):=I+JA_n(E_0,\theta)(z-E_0),
\end{align} 
\begin{equation}\label{CE2con}
\tilde{O}_n(z,\theta)=O_n(z,\theta)C_n(z,\theta).
\end{equation}
Then
$$
C_n(\bar{z},\theta)^*JC_n(z,\theta)=J+(A_n(E_0,\theta)^*-A_n(E_0,\theta))(z-E_0)+A_n^*(E_0,\theta)JA_n(E_0,\theta)(z-E_0)^2.
$$
$$
C_n(z,\theta)^*JC_n(z,\theta)=J-(A_n(E_0,\theta)+A_n(E_0,\theta)^*)i\delta+A_n^*(E_0,\theta)JA_n(E_0,\theta)|z-E_0|^2+O(e^{-cn}).
$$
By \eqref{irv1}, $A_n(E_0,\theta)$ and so $C_n(z,\theta)$ are  uniformly bounded in $n$. Thus by  \eqref{bimsa1}, \eqref{irv2}, \eqref{CE1con} and \eqref{CE2con}, if $n$ sufficiently large, for any $z\in R_{\delta_0'}(E_0)$ and $\theta\in\T$, we have for some $C>0$ not depending on $n$,
\begin{equation}\label{ha2con}
\|\tilde{O}_n(\bar{z},\theta)^*S_n\tilde{O}_n(z,\theta)-J-A_n^*(E_0,\theta)JA_n(E_0,\theta)(z-E_0)^2\|\leq Ce^{-cn},
\end{equation}
\begin{align}\label{ha3con}
\nonumber \|\tilde{O}_n(z,\theta)^*S_n\tilde{O}_n(z,\theta)&-J-A_n^*(E_0,\theta)JA_n(E_0,\theta)(z-E_0)^2\| \\
&\leq Ce^{-cn}+C(|E-E_0||\delta|+\delta^2).
\end{align}

By Lemma \ref{dominated matrix} and \eqref{irv1}, if  $|z-E_0|<\delta_0''<\delta_0'$ is sufficiently small \footnote{The smallness only depends on the norm of $A^*_n(E_0,\theta)JA_n(E_0,\theta)$, thus is uniform in $n$.}, there exists $D_n(z,\theta)\in GL(2,\C)$, depending analytically on $z$ and $\theta$, such that
\begin{equation}\label{CW33con}
D_n(\bar{z},\theta)^*(J+A^*_n(E_0,\theta)JA_n(E_0,\theta)(z-E_0)^2)D_n(z,\theta)=J
\end{equation}  
and 
\begin{equation}\label{pku2}
\sup_{\theta\in\T}\left\|\frac{\partial D_n(z,\theta)}{\partial z}\right\|\leq C\|A_n(E_0,\theta)\|^2|z-E_0|\leq C_1(E_0)|z-E_0|,\ \ D_n(E_0,\theta)=I.
\end{equation}
It follows that
\begin{align}\label{CE11con}
\nonumber &\|D_n(z,\theta)^*(J+A_n(E_0,\theta)^*JA_n(E_0,\theta)(z-E_0)^2)D_n(z,\theta)-J\|\\
\nonumber \leq& C_2(E_0)\|D_n(z,\theta)-D_n(\bar{z},\theta)\|\\
\leq &C_3(E_0)|z-E_0||\delta|\leq C_3(E_0)(|E-E_0||\delta|+\delta^2).
\end{align}

Finally, let
\begin{equation}\label{CE22con}
(U_n(z,\theta),V_n(z,\theta))=\tilde{O}_n(z,\theta)D_n(z,\theta).
\end{equation}
We have that $U_n(E,\theta), V_n(E,\theta)$ are  linearly independent, depending analytically on $E$ and $\theta$ on $I_{\delta_1}(E_0)\times \T$ for any $\delta_1<\delta_0''$ (does not depend on $n$). They belong to $E^n_c(\theta)$ since $O_n(E,\theta)$ does.

By \eqref{ha2con}-\eqref{CE22con}, we have
\begin{equation*}
\left|\begin{pmatrix}
U_n& V_n
\end{pmatrix}^*S\begin{pmatrix}
U_n& V_n\end{pmatrix}-J\right|_{r',h'}\leq Ce^{-cn},
\end{equation*}
\begin{align*}
\nonumber \sup_{E+i\delta\in R_{\delta_1}(E_0)}\sup_{\theta\in\T}&\left\|\begin{pmatrix}
U_n(E+i\delta,\theta)& V_n(E+i\delta,\theta)
\end{pmatrix}^*S\begin{pmatrix}
U_n(E+i\delta,\theta)& V_n(E+i\delta,\theta)\end{pmatrix}-J\right\|\\
&\leq C(|E-E_0||\delta|+\delta^2)+Ce^{-cn}
\end{align*}
for some $0<r'<r$ and $0<h'<h$.

By \eqref{CE222con}, \eqref{CE2con}, \eqref{pku2} and \eqref{CE22con}, we have for $-l(n)\leq j\leq l(n)-1$, 
$$
|U_n(j)-U_{n+1}(j)|_{r',h'}+|V_n(j)-V_{n+1}(j)|_{r',h'}\leq Ce^{-cn}
$$
whee $U_n(j),V_n(j)$ is the $j$-th row of $U_n,V_n$.
\end{pf}

{\bf Step 2: Convergence argument for Lemma \ref{nonvanishing}:}  Note that by Proposition \ref{newwe'conv}, $U_n(E,\theta),V_n(E,\theta)\in E_c^n(\theta)$.  We also ahev ${\rm dim} E_c^n(\theta)=2$. Thus there exists  $M'_n\in C^\omega(I_{\delta_1}(E_0)\times \T,GL(2,\C))$, such that 
\begin{equation}\label{pku11}
L_{E,v_n}(\theta)(U_n(E,\theta), V_n(E,\theta))=(U_n(E,\theta+\alpha),V_n(E,\theta+\alpha))M'_n(E,\theta).
\end{equation}
By \eqref{irv5}, we have
$$
|M_n'-M_{n+1}'|_{r,h}\leq Ce^{-cn}.
$$
By the same argument as in (1) of Theorem \ref{theorem-main1general}, we have
\begin{equation}\label{pku10}
|M'_n-e^{2\pi i\phi}{\bf C}'|_{r,h}\leq Ce^{-cn}
\end{equation}
for some  ${\bf C}'\in C^\omega( I_{\delta_1}(E_0)\times \T,SL(2,\R))$.

On the other hand, for any $z\in R_{\delta_1}(E_0)\backslash\R$, we note that $L_\e(E)=L(E)+2\pi |\e|$ for $|\e|<\e_0$. Since $\delta_1$ is sufficiently small, by continuity and quantization of the Lyapunov exponent, $\omega(\alpha,S_z^v(\cdot\pm i\frac{\e_0}{2}))=\pm 1$. By the convexity of the Lyapunov exponent and regularity of $(\alpha,S_z^v)$, we have 
$$
L_{\e}(z)=L(z),\ \ |\e|\leq s(z)\leq \frac{\e_0}{2},\ \ s(z)>0.
$$
Thus for any compact subset $\mathcal{K}$ of $R_{\delta_2}(E_0)\backslash\R$ for some $0<\delta_2<\delta_1$, there is $\e'(K)>0$ such that
$$
L_{\e}(z)=L(z),\ \ |\e|\leq \e'.
$$
By Theorem \ref{1general}, for any $z\in \mathcal{K}$, we have
$$
L_1(\alpha, e^{2\pi i\phi(z,\cdot)}{\bf C}'(z,\cdot))>0>L_2(\alpha, e^{2\pi i\phi(z,\cdot)}{\bf C}'(z,\cdot)).
$$
Note also $(\alpha,S_z^v)$/$(\alpha,S_{\bar{z}}^v)$ is uniformly hyperbolic, thus $L(z)>0$/$L(\bar{z})>0$. Apply the nonreal agument of Theorem 5.1 in \cite{gj} \footnote{The original theorem is stated for real $E$, by the proof work for any $z\in\C$.} to the finite truncations on a   neighborhood over the compact set $\mathcal{K}$, in both phase directions, and pass to the limiting center using \eqref{pku10}, we have 
$$
 \omega(\alpha, e^{2\pi i\phi(z,\cdot)}{\bf C}'(z,\cdot))=0.
$$
Thus $(\alpha, e^{2\pi i\phi(z,\cdot)}{\bf C}'(z,\cdot))$ is 1-regular.

By \eqref{pku10}, Theorem 6.1 in \cite{ajs} and Proposition \ref{trivial}, there exist $B_n(z,\cdot)\in C^\omega(\T, GL(2,\C))$ and $\tau_n^\pm(z,\cdot)\in C^\omega(\T, \C\backslash\{0\})$, depending analytically on $z\in U$ for some open neighborhood $U$ of $\mathcal{K}$,  such that
\begin{equation}\label{pku12}
B_n(z,\theta+\alpha)^{-1}M'_n(z,\theta)B_n(z,\theta)=\begin{pmatrix}\tau_n^+(z,\theta)&0\\
0& \tau_n^-(z,\theta)\end{pmatrix}
\end{equation}
with
\begin{equation}\label{pku13}
|\tau_n^+(\theta)|<1,\ \ |\tau_n^-(\theta)|>1, \ \ \forall \theta\in\T,
\end{equation}
Moreover, by stability of the dominated splitting, for any  $z\in \mathcal{K}$ and any $\theta\in\T$,  we have
\begin{equation}\label{contib}
\|B_n(z,\theta)-B(z,\theta)\|, |\tau^\pm_n(z,\theta)-\tau^\pm(z,\theta)|
\leq C(\mathcal{K})e^{-cn}
\end{equation} 
for some  $B(z,\cdot)\in C^\omega(\T, GL(2,\C))$ and $\tau^\pm(z,\cdot)\in C^\omega(\T, \C\backslash\{0\})$, depending analytically on $z\in U$.

We donote 
$$
B_n(z,\theta)=\begin{pmatrix}
b^+_n(z,\theta,1)& b^-_n(z,\theta,1)\\ b^+_n(z,\theta,0)&b^-_n(z,\theta,0)
\end{pmatrix},\ \ B(z,\theta)=\begin{pmatrix}
b_1^+(z,\theta)& b_1^-(z,\theta)\\ b_0^+(z,\theta)&b_0^-(z,\theta)
\end{pmatrix}
$$
and define
\begin{align}\label{ff1con}
u_n^+(z,\theta)=b^+_n(z,\theta,1)U_n(z,\theta)+b^+_n(z,\theta,0)V_n(z,\theta),
\end{align}
\begin{align}\label{ff2con}
u_n^-(z,\theta)=b^-_n(z,\theta,1)U_n(z,\theta)+b^-_n(z,\theta,0)V_n(z,\theta).
\end{align}
Then by \eqref{pku11}, \eqref{pku12} and \eqref{pku13}, $u^\pm_n(z,\theta)\in E_\pm^n(z,\theta)$.

\begin{Lemma}\label{nonvanishingcon}
There exists $0<\delta_3(E_0)<\delta_2$ such that for any $z\in R_{\delta_3}(E_0)\cap \mathbb{H}$ and $\theta\in\T$, we have $\Im b_1^+(z,\theta)\overline{b_0^+(z,\theta)}>0$ and $\Im b_1^-(z,\theta)\overline{b_0^-(z,\theta)}<0$.
\end{Lemma}
\begin{pf}
By \eqref{+}, for any $z=E+i\delta \in R_{\delta_2}(E_0)\cap \mathbb{H}$, we have that
\begin{equation}\label{20242gjycon}
\Im \left(\vec{u}_n^{+}(z,\theta,0)\right)^*C_n^*\vec{u}_n^+(z,\theta,-1)=\delta\sum\limits_{m=0}^{\infty}\|\vec{u}_n^+(z,\theta,m)\|^2,
\end{equation}
where 
\begin{equation*}
\begin{pmatrix}\vec{u}_n^+(z,\theta,k)\\ \vec{u}_n^+(z,\theta,k-1)\end{pmatrix}=(L_{z,v_n})_{dk}(\theta)u_n^+(z,\theta).
\end{equation*}

On the other hand, by \eqref{ff1con}, Proposition \ref{newwe'conv} and the same argument as in \eqref{20242gjy1}, we have
\begin{align}\label{20242gjy1con}
\nonumber &\left|\Im \left(\vec{u}_n^{+}(z,\theta,0)\right)^*C_n^*\vec{u}_n^+(z,\theta,-1)-\Im \overline{b^+_n(z,\theta,0)}b^+_n(z,\theta,1))\right|\\
\leq &C(|E-E_0||\delta|+\delta^2+e^{-cn})(|b_n^+(z,\theta,0)|^2+|b_n^+(z,\theta,1)|^2).
\end{align}

Let  $m\leq l(n)$ and set,
$$
\begin{pmatrix}U^{m}_n(E,\theta)& V^{m}_n(E,\theta)\end{pmatrix}=\begin{pmatrix}
U_n(E,\theta,2m-1)&V_n(E,\theta,2m-1)\\
U_n(E,\theta,2m-2)&V_n(E,\theta,2m-2)\\
\vdots&\vdots\\
U_n(E,\theta,0)&V_n(E,\theta,0)
\end{pmatrix}.
$$
By \eqref{irv6} and the fact that $L_{n}(E,\theta+\alpha)=TL_n(E,\theta)T^*$ where $(Tu)(n+1)=u(n)$, $u\in\ell^2(\Z)$, we have
\begin{align*}
&L_n^{m}(E,\theta+m\alpha)=\begin{pmatrix}U^{m}_n(E,\theta)& V^{m}_n(E,\theta)\end{pmatrix} \\
&\cdot (\begin{pmatrix}U_n(E,\theta)& V_n(E,\theta)\end{pmatrix}^*S_n\begin{pmatrix}U_n(E,\theta)& V_n(E,\theta)\end{pmatrix})^{-1}\begin{pmatrix}U^{m}_n(E,\theta)& V^{m}_n(E,\theta)\end{pmatrix}^*.
\end{align*}
By \eqref{gjy2ncon} and \eqref{irv5}, we have that
$$
L^m(E,\theta+m\alpha)=\lim\limits_{n\rightarrow\infty}L^{m}_n(E,\theta+m\alpha)=-\begin{pmatrix}
U^{m}(E,\theta)&V^{m}(E,\theta)\end{pmatrix}J\begin{pmatrix}U^{m}(E,\theta)& V^{m}(E,\theta)\end{pmatrix}^*
$$
where 
$$
\begin{pmatrix}
U^{m}(E,\theta)&V^{m}(E,\theta)\end{pmatrix}=\lim\limits_{n\rightarrow\infty} \begin{pmatrix}
U_n^{m}(E,\theta)&V_n^{m}(E,\theta)\end{pmatrix}.
$$
By Lemma \ref{add2aaa},  for any fixed $n_0$ sufficiently large, ${\rm Rank}(L^{n_0}(E,\theta))=2$. Hence,
$$
l(E,\theta)=\frac{(U^{n_0}(E,\theta))^*V^{n_0}(E,\theta)}{\|U^{n_0}(E,\theta)\| \|V^{n_0}(E,\theta))\|}<1.
$$ 
It follows from  \eqref{irv5} that 
\begin{equation}\label{pku15}
l_n(E,\theta)=\frac{(U_n^{n_0}(E,\theta))^*V_n^{n_0}(E,\theta)}{\|U_n^{n_0}(E,\theta)\| \|V_n^{n_0}(E,\theta))\|}<1-\epsilon
\end{equation}
for some $\epsilon>0$.

By \eqref{ff1con}, \eqref{20242gjycon} and \eqref{pku15}, for all $E\in I_{\delta_2}(E_0)$ and $\theta\in\T$, we have
\begin{align}\label{202512gjycon}
&\Im \left(\vec{u}_n^{+}(z,\theta,0)\right)^*C_n^*\vec{u}_n^+(z,\theta,-1)\\ \nonumber
\geq &c\delta ( |b_n^+(z,\theta,1)|^2\|U^{n_0}_n(z,\theta)\|^2+|b_n^+(z,\theta,0)|^2\|V_n^{n_0}(z,\theta)\|^2)
\end{align}
where $c$ only depend on $E_0$ and $|\delta|\leq \delta_2$.

By \eqref{20242gjy1con} and \eqref{202512gjycon}, there is $\delta_3(E_0)>0$ such that if $|z-E_0|<\delta_3$ and $n$ sufficiently large, then
\begin{align*}
&\left|(|E-E_0||\delta|+\delta^2+e^{-cn})(|b_n^+(z,\theta,0)|^2+|b_n^+(z,\theta,1)|^2)\right|\\
&\leq  \frac{1}{100}\Im \left(\vec{u}_n^{+}(z,\theta,0)^*C^*\vec{u}_n^+(z,\theta,-1)\right)
\end{align*}
which together with \eqref{20242gjycon} and \eqref{20242gjy1con} imply 
$$
\Im b_1^+(z,\theta)\overline{b_0^+(z,\theta)}>0.
$$
We thus get the desired result.
\end{pf}

{\bf Step 3: Convergence argument for Lemma \ref{le2} and Lemma \ref{final1}:}
\begin{Definition}\label{m+defcon}
{\rm For any $z\in R_{\delta_3}(E_0)\backslash\R$, we define
$$
m^n_+(z,\theta)=\frac{b^+_n(z,\theta,1)}{b_n^+(z,\theta,0)},\ \ m^n_-(z,\theta)=\frac{b^-_n(z,\theta,1)}{b_n^-(z,\theta,0)}.
$$
$$
m_+(z,\theta)=\frac{b^+_1(z,\theta)}{b_0^+(z,\theta)},\ \ m_-(z,\theta)=\frac{b^-_1(z,\theta)}{b_0^-(z,\theta)}.
$$
}
\end{Definition}
Note that on every compact $K\subset R_{\delta_3}(E_0)\backslash\R$, by Lemma \ref{nonvanishingcon} and compactness, $\inf_{(z,\theta)\in K\times\T}|b_0^\pm(z,\theta)|>0$, so by \eqref{contib}, a lower bound holds for $|b_n^\pm(z,\theta,0)|$ for sufficiently large $n$. Hence $m^n_\pm(z,\theta)$ and $m_\pm(z,\theta)$ are well-defined and depend analytically on $z$ and $\theta$ on $\left(R_{\delta_3}(E_0)\backslash\R\right)\times \T$. 

Abusing the notations a little bit, let 
\begin{align}\label{bimsa3}
u_n^+(z,\theta)=m_+^n(z,\theta)U_n(z,\theta)+V_n(z,\theta),
\end{align}
\begin{align}\label{bimsa5}
u_n^-(z,\theta)=m_-^n(z,\theta)U_n(z,\theta)+V_n(z,\theta),
\end{align}
$$
u_n^\pm(z,\theta,0)=m^n_\pm(z,\theta)U_n(z,\theta,0)+V_n(z,\theta,0),
$$
$$
u_z^\pm(\theta,0)=m_\pm(z,\theta)U(z,\theta,0)+V(z,\theta,0).
$$
Thus $u_n^\pm(z,\theta)\in E_\pm^n(z,\theta)$. By \eqref{contib}, Lemma \ref{nonvanishingcon} and Definition \ref{m+defcon},  for any compact subset $\mathcal{K}$ of $R_{\delta_3}(E_0)\backslash\R$ and any $z\in \mathcal{K}$, $\theta\in\T$,
\begin{equation}\label{irv101}
|m_\pm(z,\theta)-m_\pm^n(z,\theta)|=\left|\frac{b^\pm_n(z,\theta,1)}{b_n^\pm(z,\theta,0)}-\frac{b^\pm_1(z,\theta)}{b_0^\pm(z,\theta)}\right|\leq C(\mathcal{K})e^{-cn}.
\end{equation}
By \eqref{irv5}, \eqref{irv101}, and Lemma~\ref{pkule}, applied as in \eqref{irv1}, the scalar symplectic pairings are exponentially Cauchy.  Define their limits by  
\begin{equation}\label{bimsa2}
y_\pm(z,\theta)=-\frac{i}{2}\lim\limits_{n\rightarrow\infty}u_n^\pm(z,\theta)^*S_nu_n^\pm(z,\theta)
\end{equation}
for any $z\in R_{\delta_3}(E_0)\backslash\R$ and $\theta\in\T$. Since $S_n^*=-S_n$, $y_\pm(z,\theta)\in\R$.
\begin{Lemma}\label{le2con}
For any $z\in R_{\delta_3}(E_0)\cap \mathbb{H}$, we have that
\begin{align*}
&\int_\T \frac{1}{-\frac{y_+(z,\theta)}{|u_z^+(\theta,0)|^2}-\frac{1}{2}\Im z}d\theta\leq \frac{4\pi\e_{1}(z)}{\Im z},
\end{align*}
\begin{align*}
&\int_\T \frac{1}{-\frac{y_-(\bar{z},\theta)}{|u_{\bar z}^-(\theta,0)|^2}+\frac{1}{2}\Im z}d\theta\leq \frac{4\pi\e_1(z)}{\Im z}.
\end{align*}
\end{Lemma}
\begin{pf}
We will only  prove the $y_+(z,\theta)$ case. 

Let  $T_n=L_{z,v_n}$. Then $T^*_nS_nT_n-S_n=2i\Im z\langle \delta_0,\cdot\rangle\delta_0$ Thus we have
\begin{equation}\label{longle}
(T_nu^+_n(z,\theta))^*S_nT_nu^+_n(z,\theta)-(u^+_n(z,\theta))^*S_nu^+_n(z,\theta)=2i\Im z |u_n^+(z,\theta,0)|^2.
\end{equation}

Since $u_n^+(z,\theta)\in E_+^n(z,\theta)$, there is $\tilde{\tau}^+_n(z,\theta)$ such that
\begin{equation}\label{longle1}
L_{z,v_n}u_n^+(z,\theta)=u_n^+(z,\theta+\alpha)\tilde{\tau}^+_n(z,\theta),
\end{equation}
and $\int_\T \ln|\tilde{\tau}_n^+(z,\theta)|d\theta=-\gamma_1^n(\bar{z})$.

By \eqref{longle} and \eqref{longle1}, we have
\begin{align*}
&(Tu^+_n(z,\theta))^*S_nTu^+_n(z,\theta)-
|\tau_n^+(z,\theta-\alpha)|^{-2}(Tu^+_n(z,\theta-\alpha))^*S_nTu^+_n(z,\theta-\alpha)\\
=&2i\Im z |u_n^+(z,\theta,0)|^2.
\end{align*}
It follows that
\begin{align}\label{L1long}
&\int_\T \ln\left(1-\frac{2i\Im z|u^+_n(z,\theta,0)|^2}{(Tu^+_n(z,\theta))^*S_n Tu^+_n(z,\theta)}\right)d\theta=2\gamma_{1}^n(\bar{z}),
\end{align}
By  the same argument as in Lemma \ref{le2} and \eqref{longle},  we have
\begin{align}\label{irv100}
&\int_\T \frac{1}{- \frac{(u^+_n(z,\theta))^*S_nu^+_n(z,\theta)}{2i|u^+_n(z,\theta,0)|^2}-\frac{1}{2}\Im z}d\theta\leq \frac{2\gamma_1^n(\bar{z})}{\Im z}.
\end{align}
Letting $n\rightarrow \infty$, by \eqref{irv5}, \eqref{bimsa2}, \eqref{irv101} and \eqref{irv100}, we get the result.
\end{pf}

\begin{Lemma}\label{final1conbimsa}
For any $z\in R_{\delta_2}(E_0)\cap \mathbb{H}$,
\begin{align*}
\frac{\partial 2\pi \e_1(z)}{\partial\Im z}=\int_\T\Im \frac{\overline{u_{\bar{z}}^-(\theta,0)}u_z^+(\theta,0)}{\overline{m_-(\bar{z},\theta)}-m_+(z,\theta)}d\theta.\end{align*}
\end{Lemma}
\begin{pf}
By \eqref{gjy2ncon}, \eqref{bimsa3} and \eqref{bimsa5}, we have
$$
u_n^-(\bar{z},\theta)S_nu_{n}^+(z,\theta)=(\overline{m^n_-(\bar{z},\theta)}-m^n_+(z,\theta))(1+O(e^{-cn}))
$$
It follows from  Lemma \ref{final1}, shift invariance of $L_{v_n,\alpha,\theta}$ and ergodicity that,
\begin{align*}
\frac{\partial \gamma_1^n(\bar{z})}{\partial\Im z}=\int_\T\Im \frac{\overline{u_n^-(\bar{z},\theta,0)}u_n^+(z,\theta,0)}{\overline{(m^n_-(\bar{z},\theta)}-m^n_+(z,\theta))(1+O(e^{-cn}))}.d\theta.\end{align*}
Letting $n\rightarrow \infty$, we get the result.
\end{pf}

{\bf Step 4: Convergence argument for Lemma \ref{gjyle}:}
We denote
\begin{align*}
f(z,\theta)= & \frac{1}{ -\frac{y_+(z,\theta)}{|u_z^+(\theta,0)|^2}-\frac{1}{2}\Im z}+ \frac{1}{-\frac{y_-(\bar{z},\theta)}{|u_{\bar z}^-(\theta,0)|^2}+\frac{1}{2}\Im z}\\
&-4\Im \frac{\overline{u_{\bar{z}}^-(\theta,0)}u_z^+(\theta,0)}{\overline{m_-(\bar{z},\theta)}-m_+(z,\theta)}.
\end{align*}
$$
g(z,\theta)=32\left(\frac{1}{-\frac{y_+(z,\theta)}{|u_{z}^+(\theta,0)|^2}-\frac{1}{2}\Im z}+ \frac{1}{-\frac{y_-(\bar{z},\theta)}{|u_{\bar{z}}^-(\theta,0)|^2}+\frac{1}{2}\Im z}\right)
$$
Note that a crucial fact is that both $f(z,\theta)$ and $g(z,\theta)$ do not depend on the choice of center basis.  
\begin{Lemma}\label{gjyleconaa}
For every $E\in I_{\delta_3}(E_0)$ and arbitrary $\delta>0$ sufficiently small (depending on $E$), we have
$$
f(E+i\delta,\theta)+\sqrt{\delta} g(E+i\delta,\theta)\geq 0.
$$
\end{Lemma}
\begin{pf}
Note that  we only need to prove this result for $E_0$ since every $E\in I$ plays the same role (for every $E$, we can choose a center basis around $E$ like what we do for $E_0$). Let $z=E_0+i\delta\in I_{\delta_3}(E_0)$, we denote
$$
\Omega_1(z)=\{\theta\in\T:|u^+_{z}(\theta,0)|>|u^-_{\bar{z}}(\theta,0)|\}.
$$
For $\theta\in \Omega_1(z)$, by \eqref{20242gjy1con}, \eqref{202512gjycon} and Definition \ref{m+defcon},   and $n$ sufficiently large (depending on $\delta$), we have
\begin{align}\label{20242gjy1con's}
\nonumber &\left|\Im \left(\vec{u}_n^{+}(z,\theta,0)\right)^*C_n^*\vec{u}_n^+(z,\theta,-1)-\Im m_+^n(z,\theta)\right|\\
\leq &C\delta \Im \left(\vec{u}_n^{+}(z,\theta,0)\right)^*C_n^*\vec{u}_n^+(z,\theta,-1).
\end{align}
Note also that
$$
 \frac{(u^+_n(z,\theta))^*S_nu^+_n(z,\theta)}{-2i}=\Im \left(\vec{u}_n^{+}(z,\theta,0)\right)^*C_n^*\vec{u}_n^+(z,\theta,-1).
$$
By \eqref{20242gjy1con's} and the same argument as in Lemma \ref{gjyle}, we have 
\begin{align}\label{ivr1bimsa}
\nonumber \frac{1}{\Im \frac{m^n_+(z,\theta)}{|u_n^+(z,\theta,0)|^2}-\frac{1}{2}\Im z}(1-4\sqrt{\delta})&\leq \frac{1}{- \frac{(u^+_n(z,\theta))^*S_nu^+_n(z,\theta)}{2i|u^+_n(z,\theta,0)|^2}-\frac{1}{2}\Im z}\\
&\leq \frac{1}{\Im \frac{m^n_+(z,\theta)}{|u_n^+(z,\theta,0)|^2}-\frac{1}{2}\Im z}(1+4\sqrt{\delta}),
\end{align}
\begin{align}\label{ivr1bimsa1}
\nonumber \frac{1}{\Im \frac{m^n_-(z,\theta)}{|u_n^-(z,\theta,0)|^2}+\frac{1}{2}\Im z}(1-4\sqrt{\delta})&\leq \frac{1}{- \frac{(u^-_n(z,\theta))^*S_nu^-_n(z,\theta)}{2i|u^-_n(z,\theta,0)|^2}+\frac{1}{2}\Im z}\\
&\leq \frac{1}{\Im \frac{m^n_-(z,\theta)}{|u_n^-(z,\theta,0)|^2}+\frac{1}{2}\Im z}(1+4\sqrt{\delta}).
\end{align}
Letting $n\rightarrow\infty$ and by the same argument as in Lemma \ref{gjyle} agian, we have  $f(z,\theta)+\sqrt{\delta}g(z,\theta)\geq 0$. For $\theta\in \Omega_1^c$, we can find another basis $\tilde{U}_n(z,\theta), \tilde{V}_n(z,\theta)$, such that Proposition \ref{newwe'conv} holds for $\tilde{U}_n(z,\theta), \tilde{V}_n(z,\theta)$ with \eqref{gjy3ncon} replaced by  
\begin{align}\label{gjy3nconreps}
\nonumber &\left\|\begin{pmatrix}
T\tilde{U}_n(E_0+i\delta,\theta)& T\tilde{V}_n(E_0+i\delta,\theta)
\end{pmatrix}^*S\begin{pmatrix}
T\tilde{U}_n(E_0+i\delta,\theta)& T\tilde{V}_n(E_0+i\delta,\theta)\end{pmatrix}-J\right\|\\
&\leq C\delta^2+Ce^{-cn},\ \ \forall E_0+i\delta\in R_{\delta_1}(E_0), \forall \theta\in\T.
\end{align}
Let $\tilde{m}^n_\pm(z,\theta)$ be the m functions defined with $U_n(z,\theta), V_n(z,\theta)$ replaced by $\tilde{U}_n(z,\theta), \tilde{V}_n(z,\theta)$. By exactly the same argument as above, we have
\begin{align}\label{ivr2}
\nonumber \frac{1}{\Im \frac{\tilde{m}^n_+(z,\theta)}{|u_n^+(z,\theta,0)|^2}+\frac{1}{2}\Im z}(1-4\sqrt{\delta})&\leq  \frac{1}{- \frac{(Tu^+_n(z,\theta))^*S_nTu^+_n(z,\theta)}{2i|u^+_n(z,\theta,0)|^2}+\frac{1}{2}\Im z}\\
&\leq \frac{1}{\Im \frac{\tilde{m}^n_+(z,\theta)}{|u_n^+(z,\theta,0)|^2}+\frac{1}{2}\Im z}(1+4\sqrt{\delta}).
\end{align}
\begin{align}\label{ivr3}
\nonumber \frac{1}{\Im \frac{\tilde{m}^n_-(\bar{z},\theta)}{|u_n^-(\bar{z},\theta,0)|^2}-\frac{1}{2}\Im z}(1-4\sqrt{\delta})&\leq \frac{1}{-\frac{(Tu^-_n(\bar{z},\theta))^*S_nTu^-_n(\bar{z},\theta)}{2i|u^-_n(\bar{z},\theta,0)|^2}-\frac{1}{2}\Im z}\\
&\leq \frac{1}{\Im \frac{\tilde{m}^n_-(\bar{z},\theta)}{|u_n^-(\bar{z},\theta,0)|^2}-\frac{1}{2}\Im z}(1+4\sqrt{\delta}).
\end{align}
Hence letting $n\rightarrow\infty$ and by the same argument as in Lemma \ref{gjyle} again, we have $f(z,\theta)+\sqrt{\delta}g(z,\theta)\geq 0$.

Moreover, by Lemma \ref{le2con}, we have
\begin{equation}\label{finite1}
\int_\T |g(z,\theta)|\leq 1000\pi \frac{\e_1(z)}{\Im z}.
\end{equation}
Let $\e_0>0$ be such that $(\alpha,S_{E_0+i\delta}^v(\cdot+i\e_0))$ is uniformly hyperbolic and $\omega(\alpha,S_{E_0+i\delta}^v(\cdot+i\e_0))=1$ for any $\delta$ sufficiently small, then 
\begin{equation}\label{finite2}
L_{\e_0}(E_0+i\delta)-L_{\e_0}(E_0)-(L_0(E_0+i\delta)-L_0(E_0))=-2\pi\e_1(E_0+i\delta).
\end{equation}
It follows that
$$
\frac{2\pi \e_1(E_0+i\delta)}{\delta}=-\frac{L_{\e_0}(E_0+i\delta)-L_{\e_0}(E_0)}
{\delta}+\frac{L_0(E_0+i\delta)-L_0(E_0)}{\delta}.
$$
The Thouless formula gives
$$
\partial_\delta L_0(E_0+i\delta)=\int_{\R}\frac{\delta}{(E_0-E')^2+\delta^2}dN(E').
$$
Its Poisson boundary limit and its averaged limit agree  for almost every $E_0$, together with the harmonicity of $L_{\e_0}(z)$ in a neighborhood of $E_0$, the above limit exist for almost every $E_0$ as $\delta\rightarrow 0$.

Hence
\begin{equation}\label{asksk11}
\sqrt{\delta}\int_\T |g(z,\theta)|\leq 1000\pi \sqrt{\delta}\frac{\e_1(E+i\delta)}{\delta}\rightarrow 0
\end{equation}
for almost every $E\in I_{\delta_3}(E_0)$.
\end{pf}

{\bf Step 5: $C^\omega$-rotations reducibility:}
By the Herglotz property of $m_\pm(z,\theta)$ and the argument aoove, for almost every $E\in I_{\delta_3}(E_0)$ and almost every $\theta\in\T$, $m_+(E+i\delta,\theta)$ and $\e_1(E+i\delta)/\delta$ has finite boundary limits. Define
$$
A_E=\{\theta:\Im m_+(E+i0,\theta)=0\}.
$$
The real M\"obius covariance of the limiting center make $A_E$ invariant modulo null sets. Irrational rotation ergodicity therefore gives $|A_E|\in \{0,1\}$. If $|A_E|=1$, by Lemma \ref{le2con} and Fatou's lemma, we have $u_{E+i0}^+(\theta,0)=0$ for almost every $\theta\in\T$. Coinvariance and countably many shifts then force every coordinate of the center solution to vanish almost everywhere, which leads to a contradiction.  
Thus for almost every $E\in I_{\delta_3}(E_0)$ and almost every $\theta\in\T$, $\lim_{\delta\rightarrow 0}\Im m_+(E+ i\delta,\theta)>0$. 
Fix an energy for which these boundary conclusions and \eqref{asksk11} hold. By Lemma \ref{gjyleconaa},
\[
F_\delta
=
f(E+i\delta,\theta)+\sqrt{\delta}g(E+i\delta,\theta)
\ge0.
\]
At almost every phase, \(F_\delta\to f(E+i0,\theta)\). Therefore Fatou's lemma, \eqref{asksk11}, andLemma \ref{le2con} and Lemma \ref{final1conbimsa} give
\[
0
\le
\int_{\T}f(E+i0,\theta)\,d\theta
\le
\liminf_{\delta\downarrow0}\int_{\T}F_\delta\,d\theta
=
\liminf_{\delta\downarrow0}\int_{\T}f(E+i\delta,\theta)\,d\theta
\le
0.
\]
Hence \(f(E+i0,\theta)=0\) for almost every phase.

By \eqref{gjy3ncon}, \eqref{bimsa3}, \eqref{bimsa5} \eqref{irv101} and \eqref{bimsa2}, we have
$$
y_\pm(z,\theta)=-\Im m_\pm(z,\theta)+O(|E-E_0|\delta+\delta^2)(1+|m_\pm(z,\theta)|^2).
$$
Thus
\begin{equation}\label{lana1aa}
y_\pm(E\pm i0,\theta)=-\Im m_\pm(E\pm i0,\theta).
\end{equation}
We omit $\theta$ for simplicity, thus
\begin{align}
f(E+i0)= \frac{|u_{E+i0)}^+(0)|^2}{\Im m_+(E+i0)}+\frac{|u_{E-i0}^-(0)|^2}{\Im m_-(E-i0)}-4\Im\frac{\overline{u_{E-i0}^-(0)}u^+_{E+i0}(0)}{\overline{m_-(E-i0)}-m_+(E+i0)}
\end{align}
On the other hand,
\begin{align}\label{gjyabc1'}
&\left|4\Im\frac{\overline{u_{E-i0}^-(0)}u_{E+i0}(0)}{\overline{m_-(E-i0)}-m_+(E+i0)}\right|\leq 4\frac{|u^-_{E-i0}(0,x)||\vec{u}^+_{E+i0}(0,x)|}{\Im m_+(E+i0)+\Im m_-(E-i0)}\\ \nonumber
\leq &\frac{2|u^-_{E-i0}(0)||u^+_{E+i0}(0)|}{\sqrt{\Im m_+(E+i0)\Im m_-(E-i0)}}\leq \frac{|u^+_{E+i0}(0)|^2}{\Im m_+(E+i0)}+\frac{|u^-_{E-i0}(0)|^2}{\Im m_-(E-i0,x)}.
\end{align}
Thus $f(E+i0)=0$ if and only if all the inequalities above are equalities. Fix an energy $E\in I_{\delta_3}(E_0)$ for which the preceding conclusions hold. The set of phases where $u^+_{E+i0}(\theta,0)\ne0$ has positive measure; otherwise covariance would force the entire center solution to vanish. On this set, equality throughout (10.56) gives $m_-(E-i0,\theta)=m_+(E+i0,\theta)$. The set where these two projective lines coincide is invariant modulo null sets under $\theta\mapsto\theta+\alpha$. Ergodicity therefore gives this equality for almost every $\theta$.  

Then by the same argument as in the previous section, $m_+(E,\theta)$ is analytic on $I_{\delta_3}(E_0)\times \T$.

Finally, we  define 
$$
D(E,\theta)=\begin{pmatrix}
0&\frac{|m_+(E,\theta)|}{(\Im m_+(E,\theta))^{1/2}}\\
-\frac{(\Im m_+(E,\theta))^{1/2}}{|m_+(E,\theta)|}&\frac{\Re m_+(E,\theta)}{|m_+(E,\theta)|(\Im m_+(E,\theta))^{1/2}}
\end{pmatrix}.
$$
Then $D(E,\theta+\alpha)^{-1}{\bf C}'(E,\theta)D(E,\theta)\in SO(2,\R)$, hence there exists  $I'\subset I$  and $B\in C^\omega(I'\times \T,SL(2,\R))$ such that  $B(E,\theta+\alpha)^{-1}{\bf C}(E,\theta)B(E,\theta)\in SO(2,\R)$,
thus we complete the proof.
\end{pf}

We immediately have the following
\begin{Corollary}\label{long-range-kotani1}
Let $\alpha\in\R\backslash\Q$ and $v\in C^\omega(\T,\R)$. If there is an interval $I\subset \Sigma^{1,+}_{v,\alpha}$, there exist  $O^j_E\in C^\omega(\T,\C^{\Z})$ for $j=1,2$, $\psi_E,\phi_E\in C^\omega(\T,\R)$, depending analytically on $E\in I'\subset I,$ such that
$$
L_{v,\alpha,\theta}O_E^j(E)=EO_E^j(\theta),
$$
$$
T(O^1_E(\theta),O^2_E(\theta))=(O^1_E(\theta+\alpha),O^2_E(\theta+\alpha))e^{2\pi i\phi_E(\theta)}R_{\psi_E(\theta)}.
$$
\end{Corollary}
\begin{pf}
By (1) of Theorem \ref{theorem-main1general}, for $j=1,2$, there are $O_E^j\in C^\omega(\T,\C^\Z)$ such that
$$
L_{v,\alpha,\theta}\widetilde{O}_E^j(E)=E\widetilde{O}_E^j(\theta),
$$
$$
T(\widetilde{O}^1_E(\theta),\widetilde{O}^2_E(\theta))=(\widetilde{O}^1_E(\theta+\alpha),\widetilde{O}^2_E(\theta+\alpha))e^{2\pi i\phi_E(\theta)}{\bf C}(\theta).
$$
Let $(O^1_E(\theta),O^2_E(\theta))=(\widetilde{O}^1_E(\theta),\widetilde{O}^2_E(\theta))B(E,\theta)$ where $B(E,\theta)$ is from Theorem \ref{Kotani-longrange}, we obtain the result.
\end{pf}

\subsection{All-frequency Puig's argument: the limiting case}
We present the limiting case of the all-frequency Puig's argument. The main strategy is similar to that  in Section \ref{puig2}. We will address the dependence on the dimension.
\begin{Theorem}\label{contra2-general}
Let $\alpha\in\R\backslash\Q$ and $v\in C^\omega(\T,\R)$. There does not exist $B_E\in C^\omega(\T,SL(2,\R))$ and $\psi_E\in C^\omega(\T,\R)$ with $\int_\T \psi_E(\theta)d\theta=0$, such that
$$
B_E^{-1}(\theta+\alpha){\bf C}(E,\theta)B_E(\theta)=R_{\psi_E(\theta)}
$$
where ${\bf C}(E,\theta)$ is the limiting center obtained in Theorem \ref{theorem-main1general}. 
\end{Theorem}
\begin{pf}
We prove by contradiction. We omit $E$ for simplicity. Assume there exist $B\in C^\omega(\T,SL(2,\R))$ and $\psi\in C^\omega(\T,\R)$ with $\int_\T \psi(\theta)d\theta=0$, such that
\begin{equation}\label{bimsa6}
B^{-1}(\theta+\alpha){\bf C}(\theta)B(\theta)=R_{\psi(\theta)}.
\end{equation}
First by Theorem \ref{theorem-main1general},  there exist $F_n\in C_h^\omega(\T,GL_{2l(n)\times 2}(\C))$,  $M_n\in C^\omega(\T,GL(2,\C))$ such that
\begin{equation}\label{bimsa7}
L_{E,v_n}(\theta)F_n(\theta)=F_n(\theta+\alpha)M_n(\theta),
\end{equation}
with
\begin{equation}\label{conve2}
|F_n(\theta)^*S_nF_n(\theta)-J|_h,\ \ |M_n-e^{2\pi i\phi}{\bf C}|_h\leq Ce^{-c|n|},
\end{equation}
for some $h,C,c>0$, $\phi\in {\bf C}^\omega(\T,\R)$ and ${\bf C}\in C^\omega(\T,SL(2,\R))$.  Moreover, if we denote
$$
F_n(\theta)=\begin{pmatrix}f_{1n}(\theta,l(n)-1)&f_{2n}(\theta,l(n)-1)\\ f_{1n}(\theta,l(n)-2)&f_{2n}(\theta,l(n)-2)\\\ \vdots\\ f_{1n}(\theta,-l(n))&f_{2n}(\theta,-l(n))\end{pmatrix},
$$ 
then there are $f, g\in C^\omega(\T,\C)$ such that
\begin{equation}\label{ge6}
|f_{1n}(0)-f|_h,\ \ |f_{2n}(0)-g|_h\leq Ce^{-c n}.
\end{equation}
We let $\widetilde{F}_n(\theta)=F_n(\theta)B(\theta)\in C^\omega(\T,GL_{2l(n)\times2}(\C))$ and $\tilde{c}_n\in C^\omega(\T,gl(2,\C))$ be such that 
\begin{equation}\label{cos3}
e^{c_n(\theta)}=I+e^{-2\pi i\phi(\theta)}R_{-\psi(\theta)}B^{-1}(\theta+\alpha)(M_n(\theta)-e^{2\pi i\phi(\theta)}{\bf C}(\theta))B(\theta),
\end{equation}
then we have by \eqref{bimsa6}, \eqref{bimsa7} and \eqref{cos3},
\begin{equation}\label{spe}
L_{E,v_n}(\theta)\widetilde{F}_n(\theta)=\widetilde{F}_n(\theta+\alpha)e^{2\pi i\phi(\theta)}R_{\psi(\theta)}e^{c_n(\theta)}.
\end{equation}
Moreover, by \eqref{conve2} and \eqref{cos3}, we have
\begin{equation}\label{ge15}
|c_n|_h\leq C(e^{2\pi(|\psi|_h+|\phi|_h)}|B|^2_h|M_n-e^{2\pi i\phi}{\bf C}|_h)\leq Ce^{-cn}.
\end{equation}

If we denote
$$
\widetilde{F}_n(\theta)=\begin{pmatrix}\tilde{f}_{1n}(\theta,l(n)-1)&\tilde{f}_{2n}(\theta,l(n)-1)\\ \tilde{f}_{1n}(\theta,l(n)-2)&\tilde{f}_{2n}(\theta,l(n)-2)\\\ \vdots\\ \tilde{f}_{1n}(\theta,-l(n))&\tilde{f}_{2n}(\theta,-l(n))\end{pmatrix},
$$ 
then by \eqref{ge6} and the definition of $\widetilde{F}_n$, there are $\tilde{f}, \tilde{g}\in C^\omega(\T,\C)$ such that
\begin{equation}\label{gjy1}
|\tilde{f}_{1n}(0)-\tilde{f}(\cdot)|_h,\ \ |\tilde{f}_{2n}(0)-\tilde{g}(\cdot)|_h\leq Ce^{-c n}.
\end{equation}

We divide the proof into two cases depending on the arithmetic of the frequency.

Case I: $\beta(\alpha)<h$. Define
$$
\Psi(\theta)=\sum\limits_{k\in\Z\backslash\{0\}}\frac{1}{e^{2\pi ik\alpha}-1}\hat{\psi}(k)e^{2\pi i k\theta},
$$
$$
\Phi(\theta)=\sum\limits_{k\in\Z\backslash\{0\}}\frac{1}{e^{2\pi ik\alpha}-1}\hat{\phi}(k)e^{2\pi i k\theta}.
$$
We have that $\Phi,\Psi\in C^\omega(\T,\R)$ and
$$
\Psi(\theta+\alpha)-\Psi(\theta)=\psi(x),\ \ \Phi(\theta+\alpha)-\Phi(\theta)=\phi(\theta)-\hat{\phi}(0).
$$
For
\begin{equation}\label{cos5}
H_n(\theta)=e^{2\pi i\Phi(\theta)}\widetilde{F}_n(\theta)R_{\Psi(\theta)},
\end{equation}
since $\widetilde{F}_n\in C^\omega(\T,GL_{2l(n)\times 2}(\C))$, we have that  $H_n\in C^\omega(\T,GL_{2l(n)\times 2}(\C))$, and by \eqref{spe},
\begin{equation}\label{ge8}
L_{E,v_n}(\theta)H_n(\theta)=e^{2\pi i\phi_0}H_n(\theta+\alpha)e^{c'_n(\theta)},
\end{equation}
where $\phi_0=\int_\T \phi(\theta)d\theta$ and $e^{c'_n(\theta)}=R_{-\Psi(\theta)}e^{c_n(\theta)}R_{\Psi(\theta)}$.

Let
$$
H_n(\theta)=\begin{pmatrix}
h_{1,1}(\theta)&h_{1,2}(\theta)\\
h_{2,1}(\theta)&h_{2,2}(\theta)\\
\vdots&\vdots\\
h_{2l(n),1}(\theta)&h_{2l(n),2}(\theta)
\end{pmatrix}.
$$
Involving the form of $L_{E,v_n}(\theta)$ and \eqref{ge8}, one has for $j=1,2$,
\begin{equation}\label{e22ge}
-\frac{1}{\hat{v}_{l(n)}}\left(\sum\limits_{m=1}^{2l(n)} \hat{v}_{l(n)-m}h_{m,j}(x)+(2\cos2\pi(\theta)-E)h_{l(n),j}(\theta)\right)-e^{2\pi i\phi_0}h_{1,j}(\theta+\alpha)=e^{2\pi i\phi_0}g_{1,j}(\theta),
\end{equation} 
\begin{equation}\label{e33ge}
h_{m,j}(\theta)=e^{2\pi i\phi_0}(h_{m+1,j}(\theta+\alpha)+g_{m+1,j}(\theta)), \ \ \forall 1\leq m\leq 2l(n)-1,
\end{equation}
where
\begin{equation}\label{ge10}
\begin{pmatrix}
g_{m,1}(\theta)&
g_{m,2}(\theta)
\end{pmatrix}=\begin{pmatrix}
h_{m,1}(\theta+\alpha)&
h_{m,2}(\theta+\alpha)
\end{pmatrix}\left(e^{c'_n(\theta)}-id\right)
\end{equation} 
It follows from \eqref{e22ge} and \eqref{e33ge} that for $j=1,2$ we have

\begin{equation}\label{add1ge}
\sum\limits_{k=-l(n)}^{l(n)} \hat{v}_{k}e^{2\pi ik\phi_0}h_{l(n),j}(\theta+k\alpha)+(2\cos2\pi(\theta)-E)h_{l(n),j}(\theta)=e_{n,j}(\theta).
\end{equation} 
By \eqref{ge10}, \eqref{cos5} and \eqref{gjy1}, we have that
\begin{align*}
|e_{n,j}|_{h-\beta}&\leq C(v)n^2|h_{l(n),j}|_{h-\beta}|c'_n|_h\leq C(v)n^2(|\tilde{f}_E|_{h}+|\tilde{g}_E|_{h})(|\Phi|_{h-\beta}+|\Psi|_{h-\beta})|c'_n|_h\\
&\leq C(v)n^2e^{-cn}\leq e^{-\frac{c}{2}n}.
\end{align*}

Notice that by \eqref{gjy1} and \eqref{cos5}, there are $h_j\in C^\omega(\T,\C)$ such that $h_j(\theta)=\lim\limits_{n\rightarrow\infty}h_{l(n),j}(\theta)$ and letting $n\rightarrow\infty$ in \eqref{add1ge}, we have
\begin{equation}\label{add2}
\sum\limits_{k\in\Z} \hat{v}_{k}e^{2\pi ik\phi_0}h_{j}(\theta+k\alpha)+(2\cos2\pi(\theta)-E)h_{j}(\theta)=0,
\end{equation}
Let $h_{j}(\theta)=\sum_k\hat{h}_j(k)e^{2\pi i k\theta}$ be the Fourier
expansion. By \eqref{add2}, $\left\{\hat{h}_1(n)\right\}_{n\in\Z}$ and $\left\{\hat{h}_2(n)\right\}_{n\in\Z}$ are two eigenfunctions of $H_{v,\alpha,\phi_0}$ corresponding to the same eigenvalue $E$, which contradicts the simplicity of point spectrum. 

Case II: $\beta(\alpha)\geq h$.  Let $p_n/q_n$ be the approximants of the continued fraction expansion of $\alpha$. By the definition  of $\beta(\alpha)$, for any $0<\e<\frac{\beta}{100}$, there is a subsequence $q_{n_k}$ of $q_n$ such that
\begin{equation}\label{bimsa8}
q_{n_{k}+1}>e^{(\beta-\e)q_{n_k}}. \footnote{If $\beta=\infty$, just choose a subsequence such that $q_{n_{k}+1}>e^{100h q_{n_k}}$. }
\end{equation}
Again by Lemma \ref{small}, there are  $\Phi^k,\Psi^k\in C_{h/2}^\omega(\T,\R)$ such that
\begin{equation}\label{sd1ge}
|\Psi^k|_{\frac{h}{2}}+|\Phi^k|_{\frac{h}{2}}\leq 8(q_{n_k}+e^{-\frac{h}{2}q_{n_k}}q_{n_{k}+1})(|\psi|_h+|\phi|_h),
\end{equation}
\begin{equation}\label{sd2ge}
|\psi-(\Psi^k(\cdot+\alpha)-\Psi^k(\cdot))|_{\frac{h}{2}}\leq e^{-\frac{1}{20}q_{n_k+1}h}|\psi|_h,
\end{equation}
\begin{equation}\label{sd2'ge}
|\phi-\phi_0-(\Phi^k(\cdot+\alpha)-\Phi^k(\cdot))|_{\frac{h}{2}}\leq e^{-\frac{1}{20}q_{n_k+1}h}|\phi|_h.
\end{equation}
Let $n=[\frac{1}{20c}q_{n_k+1}h]$ and define   
\begin{equation}\label{cos8}
H_n^k(\theta)=e^{2\pi i\Phi^k(\theta)}\widetilde{F}_n(\theta)R_{\Psi^k(\theta)}\in GL_{2l(n)\times 2}(\C).
\end{equation}
By \eqref{spe}, we have
\begin{equation}\label{allge}
L_{E,v_n}(\theta)H_n^k(\theta)=e^{2\pi i\phi_0}H_n^k(\theta+\alpha)e^{2\pi i\e_{\phi}^k(\theta)}R_{\e_{\psi}^k(\theta)}e^{c^k_n(\theta)},
\end{equation}
where $\e_{\phi}^k(\theta)=\phi(\theta)-\phi_0-(\Phi^k(\theta+\alpha)-\Phi^k(\theta))$, $\e_{\psi}^k(\theta)=\psi(\theta)-(\Psi^k(\theta+\alpha)-\Psi^k(\theta))$ and 
\begin{equation}\label{bimsa9}
e^{c^k_n(\theta)}=R_{-\Psi^k(\theta)}e^{c_n(\theta)}R_{\Psi^k(\theta)}.
\end{equation}
Let
$$
H^k_n(\theta)=:\begin{pmatrix}
h^k_{1,1}(\theta)&h^k_{1,2}(\theta)\\
h^k_{2,1}(\theta)&h^k_{2,2}(\theta)\\
\vdots&\vdots\\
h^k_{2l(n),1}(\theta)&h^k_{2l(n),2}(\theta)
\end{pmatrix}.
$$Let $r_n(\theta):=\begin{pmatrix}
\tilde{f}_{1,n}(\theta,0)& \tilde{f}_{2,n}(\theta,0)\end{pmatrix}$ be the $l(n)$-th row of $\widetilde{F}_n(\theta)$. Then by \eqref{cos8},
$$
\begin{pmatrix}
h_{l(n),1}^k(\theta)& h_{l(n),2}^k(\theta)\end{pmatrix}=e^{2\pi i\Phi^k(\theta)}r_n(\theta)R_{\Psi_k(\theta)}.
$$
By \eqref{gjy1}, 
$$
|r_n|_{h/2}\leq |\tilde{f}|_h+|\tilde{g}|_h+Ce^{-cn}.
$$
Tehn by \eqref{sd1ge}-\eqref{sd2'ge}, we have
\begin{align}\label{ge16}
\nonumber |h_{l(n),j}^k|_{\frac{h}{2}}&\leq C(|\tilde{f}|_h+|\tilde{g}|_h+e^{-cn})e^{C (|\Phi^k|_{h/2}+|\Psi^k|_{h/2})}\\
&\leq C(|\tilde{f}|_h+|\tilde{g}|_h)e^{C(q_{n_k}+e^{-\frac{h}{2}q_{n_k}}q_{n_{k}+1})(|\phi|_h+|\psi|_h)},\ \ j=1,2,
\end{align}
$$
|\e_{\phi}^k|_\frac{h}{2},\ \ |\e_{\psi}^k|_{\frac{h}{2}}\leq e^{-\frac{1}{20}q_{n_k+1}h}(|\phi|_h+|\psi|_h),
$$

Again involving the form of $L_{E,v_n}(\theta)$ and \eqref{ge8}, one has for $j=1,2$,
\begin{equation}\label{e22ge'}
-\frac{1}{\hat{v}_{l(n)}}\left(\sum\limits_{m=1}^{2l(n)} \hat{v}_{l(n)-m}h^k_{m,j}(x)+(2\cos2\pi(\theta)-E)h^k_{l(n),j}(\theta)\right)-e^{2\pi i\phi_0}h^k_{1,j}(\theta+\alpha)=e^{2\pi i\phi_0}g^k_{1,j}(\theta),
\end{equation}  
\begin{equation}\label{e33ge'}
h^k_{m,j}(\theta)=e^{2\pi i\phi_0}(h^k_{m+1,j}(\theta+\alpha)+g^k_{m+1,j}(\theta)), \ \ \forall 1\leq m\leq 2l(n)-1,
\end{equation}
where%
\begin{equation}\label{ge10'}
\begin{pmatrix}
g^k_{m,1}(\theta)&
g^k_{m,2}(\theta)
\end{pmatrix}=\begin{pmatrix}
h^k_{m,1}(\theta+\alpha)&
h^k_{m,2}(\theta+\alpha)
\end{pmatrix}\left(e^{2\pi i\e_{\phi}^k(\theta)}R_{\e_{\psi}^k(\theta)}e^{c^k_n(\theta)}-id\right)
\end{equation}
It follows from \eqref{e22ge'} and \eqref{e33ge'} that for $j=1,2$ we have

\begin{equation}\label{add1ge'}
\sum\limits_{m=-l(n)}^{l(n)} \hat{v}_{m}e^{2\pi im\phi_0}h^k_{l(n),j}(\theta+m\alpha)+(2\cos2\pi(\theta)-E)h^k_{l(n),j}(\theta):=e^k_{n,j}(\theta).
\end{equation} 
By \eqref{ge16} and
\eqref{ge10'}, the same estimate gives
\begin{align}\label{cos7}
&|e^k_{n,j}|_{h/2}\leq C(v)n^2|h^k_{l(n),j}|_{h/2}(|\e_{\phi}^k|_\frac{h}{2}|+|\e_{\psi}^k|_{\frac{h}{2}}+|c^k_n|_{\frac{h}{2}}+e^{-cn})\\ \nonumber
\leq &C(v)n^2(|\tilde{f}|_{h}+|\tilde{g}|_{h})e^{C(|\Phi^k|_{\frac{h}{2}}+|\Psi^k|_{\frac{h}{2}})}(|\e_{\phi}^k|_\frac{h}{2}|+|\e_{\psi}^k|_{\frac{h}{2}}+|c^k_n|_{\frac{h}{2}}+e^{-cn})\leq e^{-\frac{1}{40}q_{n_k+1}h}.
\end{align}

Let $h^k_{l(n),j}(\theta)=\sum_m\hat{h}_j(m)e^{2\pi i m\theta}$ be the Fourier
expansion. By Aubry duality, $\left\{\hat{h}_1(m)\right\}_{m\in\Z}$ and $\left\{\hat{h}_2(m)\right\}_{m\in\Z}$ are two approximate solutions of $H_{v_n,\alpha,\phi_0}u=Eu$, i.e. they satisfy
\begin{equation}\label{new4ge}
\left((H_{v_n,\alpha,\phi_0}-E)\hat{h}_1\right)(m)=\hat{e}_1(m),
\end{equation}
\begin{equation}\label{new4-1ge}
\left((H_{v_n,\alpha,\phi_0}-E)\hat{h}_2\right)(m)=\hat{e}_2(m),
\end{equation}   
where $\{\hat{e}_j(m)\}_{n\in\Z}$ are the Fourier coefficients of $e^k_{n,j}(\theta)$.

We denote
$$
I_{k}=\left[-e^{-\frac{h}{8}q_{n_k}}q_{n_k+1},e^{-\frac{h}{8}q_{n_k}}q_{n_{k}+1}\right].
$$
Then by \eqref{bimsa8} and \eqref{ge16}, for $m\in\Z$ and $j=1,2$, we have
\begin{align}\label{eeee1ge}
\nonumber |\hat{h}_{j}(m)|&\leq (|\tilde{f}|_h+|\tilde{g}|_h)e^{8(q_{n_k}+e^{-\frac{h}{2}q_{n_k}}q_{n_{k}+1})(|\phi|_h+|\psi|_h)} e^{-\frac{h}{4}|m|}\\
&\leq Ce^{-\frac{h}{4}|m-C(q_{n_k}+e^{-\frac{h}{2}q_{n_k}}q_{n_{k}+1})|}.
\end{align}
\begin{Lemma}\label{initialge}
There exists $m_0\in I_k$, such that
\begin{equation}\label{z1-estimate-21ge}
|\hat{h}_{1}(m_0)|\geq e^{\frac{h}{16}q_{n_k}}q^{-8}_{n_k+1}.
\end{equation}
\end{Lemma}
\begin{pf}
Denote
$$
\overrightarrow{u}(\theta)=\begin{pmatrix}
h^k_{1,1}(\theta)\\
h^k_{2,1}(\theta)\\
\vdots\\
h^k_{2l(n),1}(\theta)
\end{pmatrix},\ \ \overrightarrow{v}(\theta)=\begin{pmatrix}
h^k_{1,2}(\theta)\\
h^k_{2,2}(\theta)\\
\vdots\\
h^k_{2l(n),2}(\theta)
\end{pmatrix}
$$
By \eqref{conve2}, \eqref{cos8} and the definition of $\widetilde{F}_n(\theta)$, we have that
\begin{align}\label{cos9}
\left\|\begin{pmatrix}\overrightarrow{u}(\theta)&\overrightarrow{v}(\theta)\end{pmatrix}^*S_{n}\begin{pmatrix} \overrightarrow{u}(\theta)&\overrightarrow{v}(\theta)\end{pmatrix}-J\right\|\leq C|B|_0^2e^{-cn}\leq e^{-\frac{1}{40}q_{n_k+1}h}.
\end{align}
Thus
$$
\|\overrightarrow{u}\|_{L^2}\|S_n\overrightarrow{v}\|_{L^2}\geq \frac{1}{2}
$$
which implies that
$$\|\overrightarrow{u}\|_{L^2}\geq \frac{1}{\|S_n\overrightarrow{v}\|_{L^2}}>\frac{1}{C(v)n\|H_n^k\|_{C^0}}.$$
By \eqref{ge15}, \eqref{bimsa9} and \eqref{ge10'}, one has the long-range version of \eqref{b111},
\begin{align}\label{b111ge}
2n\|\hat{h}_{1}\|_{\ell^2}\geq \|\overrightarrow{u}\|_{L^2}-Cn^2\sup_{n,j}\|g^k_{n,j}\|_{L^2}\geq (2Cn\|H_n^k\|_{C^0})^{-1} \geq cn^{-2}\geq cq_{n_{k}+1}^{-2}.
\end{align}
By \eqref{eeee1ge}, we have that
\begin{align*}
\sum\limits_{m\notin I_k}|\hat{h}_{1}(m)|^2
\leq 2e^{-\frac{h}{4}e^{-\frac{h}{4}q_{n_k}}q_{n_k+1}}.
\end{align*}
By \eqref{b111ge} and the fact that $|I_k|\leq 2e^{-\frac{h}{8}q_{n_k}}q_{n_k+1}$, it follows that there exists $m_0\in I_k$, such that
\begin{align*}
2e^{-\frac{h}{8}q_{n_k}}q_{n_k+1}|\hat{h}_{1}(m_0)|^2 \geq \sum\limits_{m\in I_k}|\hat{h}_{1}(m)|^2 &=\|\hat{h}_{1}\|_{\ell^2}^2-\sum\limits_{n\notin I_k}|\hat{h}_{1}(n)|^2\\
 &\geq  cq_{n_k+1}^{-6}-2e^{-\frac{h}{4}e^{-\frac{h}{4}q_{n_k}}q_{n_k+1}} .
\end{align*}
Hence  there exists $K_0>0$, such that
$$
|\hat{h}_{1}(m_0)|\geq  e^{\frac{h}{16}q_{n_k}}q^{-8}_{n_k+1},
$$
provided $k>K_0$.
\end{pf}

For $j=1,2$, we define
$$
\vec{h}_j(m)=\begin{pmatrix}
\hat{h}_{j}(m)\\
\hat{h}_{j}(m-1)
\end{pmatrix},\ \ \vec{p}_j(m)=\begin{pmatrix}\hat{e}_j(m)\\0\end{pmatrix}.
$$   
Notice that by \eqref{new4ge} and \eqref{new4-1ge}, we have
\begin{align}\label{vechge}
\vec{h}_j(m)=(S_E^{v_n})_{m-m_0}(\phi_0+m_0\alpha)\vec{h}_j(m_0)+\sum\limits_{k=m_0+1}^m (S_E^{v_n})_{m-k}(\phi_0+k\alpha)\vec{p}_j(k-1).
\end{align}
Let $\vec{u}_j(m)$ be formal solutions of $H_{v_n,\alpha,\phi_0}u=Eu$ with initial data $\vec{h}_j(m_0)$,
\begin{align}\label{vecuge'}
\vec{u}_j(m)=(S_E^{v_n})_{m-m_0}(\phi_0+m_0\alpha)\vec{h}_j(m_0),
\end{align}   
By \eqref{cos7}, we have
\begin{equation}\label{cos10}
|\vec{p}_j(m)|\leq e^{-\frac{1}{40}q_{n_k+1}h}e^{-\frac{h}{4}|m|}.
\end{equation}
By \eqref{vechge} and \eqref{cos10}, for $|m|\leq e^{-q^{\frac{1}{4}}_{n_k}}q_{n_k+1}$, we have
\begin{equation}\label{cos12}
|\vec{u}_j(m)-\vec{h}_j(m)|\leq e^{-q^{\frac{1}{4}}_{n_k}}q_{n_k+1}C^{e^{-q^{\frac{1}{4}}_{n_k}}q_{n_k+1}}e^{-\frac{1}{40}q_{n_k+1}h}\leq e^{-\frac{1}{80}q_{n_k+1}h}.
\end{equation}
\begin{Proposition}\label{ff5ge}
There exist $K_1(E,v,\alpha)$ such that if $k>K_1$, then
\begin{equation}\label{fwnge}
\left|\det\begin{pmatrix}
\vec{u}_{1}(m_0)&\vec{u}_2(m_0)
\end{pmatrix}\right| \geq e^{-e^{-q^{\frac{1}{2}}_{n_k}}q_{n_k+1}}.
\end{equation} 
\end{Proposition}
\begin{pf}
Note that
\begin{equation}\label{gjnew2ge}
|\det{\begin{pmatrix}
\vec{u}_1(m_0)&\vec{u}_2(m_0)
\end{pmatrix}}|=\|\vec{u}_1(m_0)\|\cdot \left\|\vec{u}_2(m_0)-\frac{\langle \vec{u}_2(m_0),\vec{u}_1(m_0)\rangle}{\|\vec{u}_1(m_0)\|^2}\vec{u}_1(m_0)\right\|
\end{equation}

We now prove \eqref{fwnge} by contradiction. If
\begin{equation}\label{ge17}
\left|\det\begin{pmatrix}
\vec{u}_{1}(m_0)&\vec{u}_2(m_0)
\end{pmatrix}\right| < e^{-e^{-q^{\frac{1}{2}}_{n_k}}q_{n_k+1}}.
\end{equation}
Using Lemma \ref{initialge}. one has
\begin{equation}\label{ff4ge}
 |\vec{u}_{1}(m_0)|\geq e^{\frac{h}{16}q_{n_k}}q^{-8}_{n_k+1}.
\end{equation}

By \eqref{ge17} and \eqref{ff4ge}, we have that
\begin{equation}\label{errornewge}
\left\|\vec{u}_2(m_0)-\frac{\langle \vec{u}_2(m_0),\vec{u}_1(m_0)\rangle}{\|\vec{u}_1(m_0)\|^2}\vec{u}_1(m_0)\right\|\leq e^{-\frac{1}{10}e^{-q^{\frac{1}{2}}_{n_k}}q_{n_k+1}}
\end{equation}
which means the orthogonal projection of $\vec{u}_2(m_0)$ to the vector  $\vec{u}_1(m_0)$ is large.  

In the following, we consider
\begin{align*}
b(\theta)=\overrightarrow{u}^*(\theta)S_n\left(\overrightarrow{v}(\theta)- c_1\overrightarrow{u}(\theta)\right)\end{align*}
where $c_1=\frac{\langle \vec{u}_2(m_0), \vec{u}_1(m_0)\rangle}{\|\vec{u}_1(m_0)\|^2}$.  By \eqref{ge16} and \eqref{z1-estimate-21ge}, we have
\begin{equation}\label{cos11}
|c_1|\leq C(|\tilde{f}|_h+|\tilde{g}|_h)^2e^{16(q_{n_k}+e^{-\frac{h}{2}q_{n_k}}q_{n_{k}+1})(|\phi|_h+|\psi|_h)}e^{-\frac{h}{16}q_{n_k}}q^{8}_{n_k+1}\leq e^{o(q_{n_k+1})},
\end{equation}

We aim to estimate $b(\theta)$. By \eqref{cos9}, we have $|b(\theta)|\geq \frac{1}{2}$.   Then, we will show as a consequence of
\eqref{errornewge}, that
$
\overrightarrow{v}(\theta)-c_1\overrightarrow{u}(\theta)
$
is small. To this end, we only need to estimate its
Fourier coefficients
$$
\hat{h}_2(m)-c_1\hat{h}_1(m)=\int_\T (h_{n,2}^k(\theta)-c_1h_{n,1}^k(\theta))e^{-2\pi im\theta}d\theta.
$$  
We distinguish two cases:\\

\textbf{Case I:} If $|m|\geq e^{-\frac{h}{100}q_{n_k}}q_{n_k+1}$, then
by \eqref{eeee1ge} and \eqref{cos11}, we have
\begin{align}\label{new1010ge}
\sum_{|m|>e^{-\frac{h}{100}q_{n_k}}q_{n_k+1}}\left|\hat{h}_2(m)-c_1 \hat{h}_1(m)\right|\leq e^{-\frac{h}{4}e^{-\frac{h}{2}q_{n_k}}q_{n_k+1}}
\end{align}
provided $k$ is sufficiently large.\\

\textbf{Case II:} If $|m|\leq e^{-\frac{h}{100}q_{n_k}}q_{n_k+1}$, we set
$$
\tilde{p}_m=\begin{pmatrix}\hat{e}_2(m)-c_1\hat{e}_1(m)\\0\end{pmatrix},
$$  
$$
\tilde{y}_m=\begin{pmatrix}\hat{h}_2(m)-c_1\hat{h}_1(m)\\ \hat{h}_2(m-1)-c_1\hat{h}_1(m-1)\end{pmatrix}.
$$
Then as a result of \eqref{new4ge} and \eqref{new4-1ge}, we have
$$\tilde{y}_{m+1}=S_E^{v_n}(\phi_0+m\alpha)\tilde{y}_{m}+\tilde{p}_m,$$   
which implies that%
$$\tilde{y}_m=(S_E^{v_n})_{m-m_0}(\phi_0+m_0\alpha)\tilde{y}_{m_0}+\sum\limits_{j=m_0+1}^m (S_E^{v_n})_{m-j}(\phi_0+j\alpha)\tilde{p}_{j-1}.
$$

To give an estimate of $\tilde{y}_m$,  first note that by \eqref{cos7}, \eqref{errornewge} and \eqref{cos11}, we have
\begin{eqnarray*}
 |\tilde{p}_m| & \leq&  e^{-\frac{1}{80}q_{n_k+1}h},\\
  |\tilde{y}_{m_0}|&\leq &\left\|\vec{u}_2(m_0)-c_1\vec{u}_1(m_0)\right\|\leq e^{-\frac{1}{10}e^{-q^{\frac{1}{2}}_{n_k}}q_{n_k+1}}.
\end{eqnarray*}
On the other hand, we have
\begin{equation}\label{ubge}
\|(S_E^{v_n})_m\|_{C^0}\leq C^{|m|}, \quad \forall m\in \Z.
\end{equation}  
As a result, if $k$ is sufficiently large depending on $v$, then one can estimate
\begin{eqnarray}\label{new102ge}
 \nonumber |\tilde{y}_m|&\leq& C^{e^{-\frac{h}{100}q_{n_k}}q_{n_k+1}} e^{-\frac{1}{10}e^{-q^{\frac{1}{2}}_{n_k}}q_{n_k+1}}+2e^{-\frac{h}{100}q_{n_k}}q_{n_k+1}C^{e^{-\frac{h}{100}q_{n_k}}q_{n_k+1}} e^{-\frac{1}{40}q_{n_k+1}h}\\
&\leq&e^{-\frac{1}{100}e^{-q^{\frac{1}{2}}_{n_k}}q_{n_k+1}}.
\end{eqnarray}

Consequently,  by \eqref{new1010ge} and \eqref{new102ge} summed up over $|m|\leq e^{-\frac{h}{100}q_{n_k}}q_{n_k+1}$, there is $K_1>0$ such that if $|k|\geq K_1$,
$$
\left\|h_{n,2}-c_1 h_{n,1}\right\|_{C^0}\leq e^{-\frac{h}{8}e^{-\frac{h}{2}q_{n_k}}q_{n_k+1}}.
$$
As a consequence, by \eqref{e33ge} and \eqref{cos11},
\begin{align*}
|b(\theta)|\leq 2n\left(\left\|h_{n,2}-c_1 h_{n,1}\right\|_{C^0}+ 4n^2e^{-\frac{1}{40}q_{n_k+1}h}\right)\leq  e^{-\frac{h}{16}e^{-\frac{h}{2}q_{n_k}}q_{n_k+1}}.
\end{align*}
This contradicts 
$
|b(\theta)|\geq \frac{1}{2}.
$
\end{pf}

Finally by \eqref{eeee1ge}, taking $m_1=[e^{-q^{\frac{1}{4}}_{n_k}}q_{n_k+1}]$, by \eqref{cos12}, we have
$$
|\vec{u}_j(m_1)|\leq |\vec{h}_j(m_1)|+e^{-\frac{1}{80}q_{n_k+1}h}\leq e^{-\frac{h}{8}e^{-q^{\frac{1}{4}}_{n_k}}q_{n_k+1}}.
$$
It follows
\begin{equation}\label{gj2ge}
\det\begin{pmatrix}
\vec{u}_{1}(m_0)&\vec{u}_2(m_0)
\end{pmatrix}=\det\begin{pmatrix}
\vec{u}_{1}(m_1)&\vec{u}_2(m_1)
\end{pmatrix}\leq Ce^{-\frac{h}{8}qe^{-^{\frac{1}{4}}_{n_k}}q_{n_k+1}}.
\end{equation}
which contradicts  \eqref{fwnge} if $k$ is sufficiently large.

\end{pf}

\subsection{Proof of the (super)critical case of Theorem \ref{main12}}
We prove by contradiction. Assume there is $I\subset\overline{\Sigma_{v,\alpha}^{1,+}}$. We distinguish two cases.

Case I: there is $E_0\in I$ such that $L(E_0)>0$, then by continuity of
  Lyapunov exponents and openness of type I energies, there is $I'\subset I$ such that $L(E'_0)>0$ and $\omega(E_0')=1$ for any $E_0'\in I'$.  

Case II: $L(E)=0$ and $\omega(E)=1$ for all $E\in I$, then again by openness of type I energies, there is $I'\subset I$ such that $\omega(E)=1$ for all $E\in I'$.  

By Theorem \ref{Kotani-longrange},  there are $B\in C^\omega(I''\times \T,SL(2,\R))$ and $\psi\in C^\omega(I''\times \T,\R)$  for some $I''\subset I'$, such that
\begin{equation*}
B(E,\theta+\alpha)^{-1}{\bf C}(E,\theta)B(E,\theta)=R_{\psi(E,\theta)}
\end{equation*}
where ${\bf C}$ is the limiting center obtained in Theorem \ref{theorem-main1general}. 

 By Proposition~\ref{prop:mean-rotation-nonconstant},
the function
\[
E\longmapsto
\int_{\mathbb T}\psi_E(\theta)\,d\theta
\]
is nonconstant. Hence there exists $E_0\in I'''$ such that 
 $\int_\T
\psi_{E_0}(x)dx=k_1\alpha+k_2,$ for some $k_1,k_2\in\Z.$ We now define
$\widetilde{B}_{E_0}(x)=B_{E_0}(x)R_{k_1x}$, and obtain
$$
\widetilde{B}_{E_0}^{-1}(\theta+\alpha){\bf C}(E_0,\theta)\widetilde{B}_{E_0}(\theta)=R_{\psi_{E_0}(x)-k_1\alpha}.
$$
Then $\int_\T \left(\psi_{E_0}(x)-k_1\alpha-k_2\right)dx=0$ contradicting Theorem \ref{contra2-general}. \qed

\section{Long-range Puig's agument: proof of the subcritical case of Theorem \ref{main11}}\label{slongpuig}
In this subsection, we give the  proof of the remaining part of Theorem \ref{main12}. We denote
$$
\Sigma^{1,0}_{v,\alpha}=\{E\in
\Sigma_{v,\alpha}:L_0^v(E)=\omega(E)=0, \bar{\omega}(E)=1\}.
$$
\subsection{The classical Kotani theory}
\begin{Theorem}[\cite{aj}]\label{Kotani-st}
Let $\alpha\in\R\backslash\Q$ and $v\in C^\omega(\T,\R)$, if there is an interval $I\subset\Sigma_{v,\alpha}^{1,0}$, then there are $B_E\in C^\omega(\T,SL(2,\R))$ and $\psi_E\in C^\omega(\T,\R)$ depending analytically on $E\in I$, such that
$$
B_E^{-1}(x+\alpha)S^v_E(x)B_E(x)=R_{\psi_E(x)}.
$$
\end{Theorem}

\subsection{All-frequency Puig's argument for the dual long-range case}
\begin{Theorem}\label{subkey}
 Let $\alpha\in\R\backslash\Q$, $v\in C^\omega(\T,\R)$ and $E\in \Sigma_{v,\alpha}^{1,0}$. There does not exist $B_E\in C^\omega(\T,SL(2,\R))$ and $\psi_E\in C^\omega(\T,\R)$ with $\int_\T \psi_E(x)dx=0$, such that
$$
B_E^{-1}(x+\alpha)S^v_E(x)B_E(x)=R_{\psi_E(x)}.
$$
\end{Theorem}
\begin{pf}
We argue by contradiction. Let $h$ be the analytic radius of $B_E$ and $p_l/q_l$ be the approximants in the continued fraction expansion of $\alpha$. Let $q_{s_k}$ be a subsequence of $q_l$ such that
\begin{equation}\label{tan1}
q_{s_k+1}>e^{hq_{s_k}}.
\end{equation}
If such a sequence is finite, let $q_{s_{k_0}}$ be the last one satisfying \eqref{tan1}. If no such $k_0$ exists, set $k_0=0$ and $q_{s_0}=q_{s_0+1}=1$. Then we set  ${s_k}={s_{k+1}}=\infty$ for any $k>k_0$. 

We define  $N_k=[\frac{q_{s_k+1}}{6}]$ and let $\phi_k(x)$ be an approximate  solution of the cohomological equation,
$$
\phi_k(x)=\sum\limits_{j=-N_k}^{-1}\frac{\hat{\psi}_E(j)}{e^{2\pi i j\alpha}-1}e^{2\pi ijx}+\sum\limits_{j=1}^{N_k}\frac{\hat{\psi}_E(j)}{e^{2\pi i j\alpha}-1}e^{2\pi ijx}. \footnote{If $k_0=0$, just define $\phi_0(x)=\sum_{j\in\Z, j\neq 0}\frac{\hat{\psi}_E(j)}{e^{2\pi i j\alpha}-1}e^{2\pi ijx}$, then $|\phi_0|_\frac{h}{2}\leq C(\alpha,h)|\psi_E|_h$ and we use this upper bound for \eqref{tan2} and \eqref{tan8}.  Note also for every $k>k_0$ in the finite-resonance case, $\phi_k$ is the full Fourier solution.}
$$

Recall that $R_n=\{lq_n: l\in\Z\}.$ In view of \eqref{tan1} and Proposition \ref{naaa1}, we distinguish two cases:

Case 1: $0< |j|< \frac{q_{s_k+1}}{6}$, $j\notin R_{s_k}$, then $|e^{2\pi i j\alpha}-1|\geq \frac{1}{4q_{s_k}}$,

Case 2: $q_{s_k}\leq |j|< \frac{q_{s_k+1}}{6}$, $j\in R_{s_k}$, then $|e^{2\pi i j\alpha}-1|\geq \frac{1}{2q_{s_k+1}}$.

It follows that
\begin{equation}\label{tan2}
|\phi_k|_{\frac{h}{2}}\leq \begin{cases}
8(q_{s_k}+e^{-\frac{h}{2}q_{s_k}}q_{{s_k}+1})|\psi_E|_h & \text{$s_k<\infty$}\\
16(q_{s_{k_0}}+e^{-\frac{h}{2}q_{s_{k_0}}}q_{{s_{k_0}}+1})|\psi_E|_h &\text{$s_k=\infty$}
\end{cases}.
\end{equation}
Moreover, we have
$$
\psi_E(x)-(\phi_k(x+\alpha)-\phi_k(x))=\sum\limits_{|j|>[q_{s_k+1}/6]}\hat{\psi_E}(j)e^{2\pi ijx},
$$
which implies that
\begin{equation}\label{tan3}
|\psi_E(\cdot)-(\phi_k(\cdot+\alpha)-\phi_k(\cdot))|_{\frac{h}{2}}\leq e^{-\frac{1}{20}q_{s_k+1}h}|\psi_E|_h.
\end{equation}

Let $B^k(x)=B_E(x)R_{\phi_k(x)}$. Then
\begin{equation}\label{alage}
B^k(x+\alpha)^{-1}S_E^{v}(x)B^k(x)=R_{\e_k(x)},
\end{equation}
where   $\e_k\in C_{h/2}^\omega(\T,\R)$  and
\begin{equation}\label{sinnew1}
\e_k(x)=\psi_E(x)-(\phi_k(x+\alpha)-\phi_k(x)).
\end{equation}
By \eqref{tan2}-\eqref{sinnew1}, we have
$$
|B^k|_0\leq C(v,\alpha),
$$
\begin{equation}\label{tan8}
|B^k|_{\frac{h}{2}}\leq 
\begin{cases}
C(v,\alpha)e^{16\pi(q_{s_k}+e^{-\frac{h}{2}q_{s_k}}q_{{s_k}+1})|\psi_E|_h} & \text{$s_k<\infty$}\\
C(v,\alpha)e^{32\pi(q_{s_{k_0}}+e^{-\frac{h}{2}q_{s_{k_0}}}q_{{s_{k_0}}+1})|\psi_E|_h} &\text{$s_k=\infty$ }
\end{cases},
\end{equation}
\begin{equation}\label{tan9}
|\e_{\psi_E}^k|_{\frac{h}{2}}\leq e^{-\frac{1}{20}q_{s_k+1}h}|\psi_E|_h.
\end{equation}
We denote
$$
B^k(x)=\begin{pmatrix}
b^k_{11}(x)&b^k_{12}(x)\\
b^k_{21}(x)&b^k_{22}(x)
\end{pmatrix}.
$$
By  \eqref{alage}, for $j=1,2$ we have
\begin{equation}\label{tan4}
(E-v(x))b^k_{1j}(x)-b_{2j}^k(x)=b_{1j}^k(x+\alpha)+e^k_{1j}(x),
\end{equation}
\begin{equation}\label{tan5}
b^k_{1j}(x)=b^k_{2j}(x+\alpha)+e_{2j}^k(x).
\end{equation}
where 
\begin{equation}\label{tan6}
\begin{pmatrix}
e^k_{11}(x)&e^k_{12}(x)\\
e^k_{21}(x)&e^k_{22}(x)
\end{pmatrix}=\begin{pmatrix}
b^k_{11}(x+\alpha)&b^k_{12}(x+\alpha)\\
b^k_{21}(x+\alpha)&b^k_{22}(x+\alpha)
\end{pmatrix}\left(R_{\e_k(x)}-id\right)
\end{equation}
Set $e_j^k(x)=e_{2j}^k(x-\alpha)-e_{1j}^k(x)$. Eliminating
$b_{2j}^k$ from \eqref{tan4}--\eqref{tan5} gives
\begin{equation}\label{subge}
b^k_{1j}(x+\alpha)+b^k_{1j}(x-\alpha)+(v(x)-E)b^k_{1j}(x)=e_j^k(x).
\end{equation}   
 Moreover, by  \eqref{tan8}, \eqref{tan9} and \eqref{tan6}, we have
\begin{equation}\label{tan7}
|e^k_j|_{\frac{h}{2}}\leq C|B^k|_{\frac{h}{2}}e^{-\frac{1}{20}q_{s_k+1}h}|\psi_E|_h\leq e^{-\frac{1}{40}q_{s_k+1}h},
\end{equation}   

Let  $b^k_{1j}(x)=\sum_m\hat{b}^k_j(m)e^{-2\pi i mx}$ be the Fourier
expansion. By \eqref{subge}, $\left\{\hat{b}^k_1(m)\right\}_{m\in\Z}$ and $\left\{\hat{b}^k_2(m)\right\}_{m\in\Z}$ are two approximate solutions of $L^{2\cos}_{v,\alpha,0}u=Eu$, i.e. they satisfy  
\begin{equation}\label{new4age}
\left((L^{2\cos}_{v,\alpha,0}-E)\hat{b}^k_1\right)(m)=\hat{e}^k_1(m),
\end{equation}
\begin{equation}\label{new4-1age}
\left((L^{2\cos}_{v,\alpha,0}-E)\hat{b}^k_2\right)(m)=\hat{e}^k_2(m),
\end{equation}
where $\{\hat{e}^k_j(m)\}_{m\in\Z}$ are defined by $e_j^k(x)=\sum_m\hat{e}_j^k(m)e^{-2\pi imx}$.

Recall that $v\in C_{h_1}^\omega(\T,\R)$. If $s_k$ is an infinite sequence,  we let 
$$
n_k=[\frac{h}{100\pi h_1}q_{s_k+1}] ,
$$
$$
r_k=[32(q_{s_k}+e^{-\frac{h}{2}q_{s_k}}q_{{s_k}+1})(|\psi|_h+1)], \ \ R_k=[\frac{1}{80}hq_{s_k+1}]
$$
and
$$
I_{k}=\left[-\frac{[r_k]}{h},\frac{[r_k]}{h}\right].
$$
If $q_{s_k}$ is a finite subsequence, let $q_{s_{k_0}}$ be the last denominator in the continue fractional expansion of $\alpha$ satisfying \eqref{tan1}. If $k>k_0$, we set $r_0=100$ and pick $R_k=\infty$ and 
\begin{equation}\label{sequence}
r_k=r_{k_0}k^2,\ \ n_k=r_{k_0}k^{100}.
\end{equation}
In this way, $r_k,R_k,n_k\rightarrow\infty$. We always have $n_k\gg r_k$ fot large $k$.

Then by \eqref{tan8}, for $m\notin I_k$ and $j=1,2$, we have
\begin{align}\label{eeee1age}
\sum_{m\notin I_k}|\hat{b}^k_{j}(m)|\leq C(v,\alpha)e^{-r_k/2}.
\end{align}
\begin{Lemma}\label{gejito}
For $j=1,2$, there exists $m_j\in I_k$, such that
\begin{equation}\label{z1-estimate-21age}
|\hat{b}^k_{j}(m_j)|\geq r_k^{-1}.
\end{equation}
\end{Lemma}
\begin{pf}
Note that $B^k\in SL(2,\R)$, thus
\begin{equation}\label{sinnew2}
b_{11}^k(x)b_{22}^k(x)-b_{12}^k(x)b_{21}^k(x)=1
\end{equation}
Thus by \eqref{tan8}, \eqref{tan9}, \eqref{tan5} and \eqref{sinnew2},
\begin{equation}\label{bimsa10}
2|\hat{b}_{j}^k|_{\ell^2}=2|b_{1j}^k|_{L^2}\geq |b_{1j}^k|_{L^2}+|b_{2j}^k|_{L^2}-|g_{2j}^k|_{L^2}\geq |B^k|_0^{-1}-\|B^k\|_0e^{-\frac{1}{20}q_{k+1}h}|\psi_E|_0\geq c.
\end{equation}
By \eqref{eeee1age} and \eqref{bimsa10},  there exist $m_j\in I_k$ such that
$$
|\hat{b}_j^k(m_j)|\geq \sqrt{r_k^{-1}(c-Ce^{-r_k/2})}\geq r_k^{-1}.
$$ 
\end{pf}

By \eqref{new4age} and \eqref{new4-1age}, for $j=1,2$, we have
\begin{equation}\label{cot3}
\left((L^{2\cos}_{v_{n_k},\alpha,0}-E)\hat{b}^k_j\right)(m)=\hat{e}^k_j(m)-\sum_{l=n_k+1}^\infty \hat{v}_l \hat{b}_j^k(m+l)-\sum_{l=-\infty}^{-n_k-1} \hat{v}_l \hat{b}_j^k(m+l)=:\hat{f}_j^k(m).
\end{equation}
By \eqref{tan8}, \eqref{tan7} and \eqref{cot3}, we have
\begin{align}\label{cot4}
\nonumber |\hat{f}_j^k(m)|&\leq |\hat{e}_j^k(m)|+\sum_{l=n_k+1}^\infty e^{-2\pi|l|h_1}|B^k|_0+\sum_{l=-\infty}^{-n_k-1}e^{-2\pi |l|h_1}|B^k|_0\\
&\leq e^{-\frac{1}{40}q_{s_k+1}h}+Ce^{-2\pi(n_k+1)h_1}\leq Ce^{-2\pi n_k (h_1-\delta)}.
\end{align}
for any fixed $\delta>0$ and large $k$.

For $j=1,2$ and $i=1,2,\cdots,l(n_k)-1$, we define
\begin{equation}\label{cot1}
\begin{pmatrix}\vec{b}_j(m)\\\vec{b}_j(m-1)
\end{pmatrix}=\begin{pmatrix}
\hat{b}^k_{j}(m_1+(m+1)\cdot l(n_k)-1)\\
\hat{b}^k_{j}(m_1+(m+1)\cdot l(n_k)-2)\\
\vdots\\
\hat{b}^k_{j}(m_1+(m-1)\cdot l(n_k))
\end{pmatrix},\ \ 
\begin{pmatrix}\vec{f}_j(m)\\\vec{f}_j(m-1)
\end{pmatrix}=\begin{pmatrix}
\hat{f}^k_{j}(m_1+(m+1)\cdot l(n_k)-1)\\
\hat{f}^k_{j}(m_1+(m+1)\cdot l(n_k)-2)\\
\vdots\\
\hat{f}^k_{j}(m_1+(m-1)\cdot l(n_k))
\end{pmatrix},
\end{equation}
\begin{equation}\label{cot2}
\begin{pmatrix}\vec{u}_i(0)\\ \vec{u}_i(-1)
\end{pmatrix}=\begin{pmatrix}
u_{n_k}^{i}(m_1\alpha,l(n_k)-1)\\
u_{n_k}^{i}(m_1\alpha,l(n_k)-2)\\
\vdots\\
u_{n_k}^{i}(m_1\alpha,-l(n_k))
\end{pmatrix},\ \ \begin{pmatrix}\vec{v}_i(0)\\ \vec{v}_i(-1)
\end{pmatrix}=\begin{pmatrix}
v_{n_k}^{i}(m_1\alpha,l(n_k)-1)\\
v_{n_k}^{i}(m_1\alpha,l(n_k)-2)\\
\vdots\\
v_{n_k}^{i}(m_1\alpha,-l(n_k))
\end{pmatrix}
\end{equation}
where $u_{n_k}^i(m_1\alpha,m)$ and $v_{n_k}^i(m_1\alpha,m)$ are solutions of $L^{2\cos}_{v_{n_k},\alpha,m_1\alpha}u=Eu$  from (3) of Theorem \ref{theorem-main1general} (here we omit $E$ from all notations for simplicity). Let
$$
S=\begin{pmatrix}
0&-C^*\\
C&0
\end{pmatrix},\ \  C=\begin{pmatrix}
\hat{v}_{l(n_k)}&\cdots&\hat{v}_1\\
&\ddots&\vdots\\
&&\hat{v}_{l(n_k)}
\end{pmatrix}.
$$
Let  $\vec{u}_i(m)$ and $\vec{v}_i(m)$ be solutions of $L^{2\cos}_{v_{n_k},\alpha,m_1\alpha}u=Eu$ with initial data $\begin{pmatrix}\vec{u}_i(0)\\ \vec{u}_i(-1)
\end{pmatrix}$, $\begin{pmatrix}\vec{v}_i(0)\\ \vec{v}_i(-1)
\end{pmatrix}$ repectively.   
\begin{Lemma}\label{sub8}
There are $C,\delta>0$ such that for  $j=1,2$,  we have that
\begin{align}\label{cot10}
\left|\begin{pmatrix}
\vec{b}_j(0)\\ \vec{b}_j(-1)
\end{pmatrix}^*S \begin{pmatrix}
\vec{u}_i(0)\\ \vec{u}_i(-1)
\end{pmatrix}\right| \leq C\sigma_{n_k}^i(m_1\alpha)e^{-n_k \delta}, \ \ i=1,\cdots,l(n_k)-1,
\end{align}
\begin{align}\label{cot11}
\left|\begin{pmatrix}
\vec{b}_j(0)\\ \vec{b}_j(-1)
\end{pmatrix}^*S \begin{pmatrix}
\vec{v}_i(0)\\ \vec{v}_i(-1)
\end{pmatrix}\right| \leq C\sigma_{n_k}^i(m_1\alpha)e^{-n_k \delta}, \ \ i=1,\cdots,l(n_k)-1,
\end{align}
\begin{align}\label{cot12}
\left|\begin{pmatrix}
\vec{b}_j(0)\\ \vec{b}_j(-1)
\end{pmatrix}^*S \begin{pmatrix}
\vec{b}_i(0)\\ \vec{b}_i(-1)
\end{pmatrix}\right| \leq  e^{-n_k h_1/2},\ \ i=1,2,
\end{align}
where $\sigma_n^i(m_1\alpha)$ are from Theorem \ref{theorem-main1general} and $h_1$ is the analytic radius of $v$.
\end{Lemma}
\begin{pf}
By \eqref{cot3}, \eqref{cot1} and \eqref{cot2}, for $j=1,2$ and $i=1,2,\cdots,l(n_k)-1$, we have
\begin{equation}\label{sub1}
C^*\vec{b}_j(m-1)+C\vec{b}_j(m+1)+(B(m_1\alpha+ml(n_k)\alpha)-EI_{l(n_k)})\vec{b}_j(m)=\vec{f}_j(m)
\end{equation}
\begin{equation}\label{sub2}
C^*\vec{u}_i(m-1)+C\vec{u}_i(m+1)+(B(m_1\alpha+ml(n_k)\alpha)-EI_{l(n_k)})\vec{u}_i(m)=0.
\end{equation}
By \eqref{cot4}, we have
\begin{equation}\label{sub3}
|\vec{f}_j(m)| \leq Cn_ke^{-2\pi n_k (h_1-\delta)}.
\end{equation}  
Multiplying both sides  of \eqref{sub1} by $\vec{b}_i(m)$, we have
\begin{align}\label{cot5}
\nonumber &\vec{b}_i(m)^*C^*\vec{b}_j(m-1)+\vec{b}_i(m)^*C\vec{b}_j(m+1)\\
+&\vec{b}_i(m)^*(B(m_1\alpha+ml(n_k)\alpha)-EI_{l(n_k)})\vec{b}_j(m)=\vec{b}_i(m)^*\vec{f}_j(m).
\end{align}
Multiplying both sides  of \eqref{sub1} by $\vec{u}_i(m)$, we have
\begin{align}\label{cot8}
\nonumber &\vec{u}_i(m)^*C^*\vec{b}_j(m-1)+\vec{u}_i(m)^*C\vec{b}_j(m+1)\\
+&\vec{u}_i(m)^*(B(m_1\alpha+ml(n_k)\alpha)-EI_{l(n_k)})\vec{b}_j(m)=\vec{u}_i(m)^*\vec{f}_j(m).
\end{align}
Multiplying both sides  of \eqref{sub2} by $\vec{b}_j(m)$, we have
\begin{align}\label{cot9}
\nonumber &\vec{b}_j(m)^*C^*\vec{u}_i(m-1)+\vec{b}_j(m)^*C\vec{u}_i(m+1)\\
+&\vec{b}_j(m)^*(B(m_1\alpha+ml(n_k)\alpha)-EI_{l(n_k)})\vec{u}_i(m)=0.
\end{align}
Let 
$$
W(m)=\begin{pmatrix}
\vec{b}_j(m)\\ \vec{b}_j(m-1)
\end{pmatrix}^*S \begin{pmatrix}
\vec{u}_i(m)\\ \vec{u}_i(m-1)
\end{pmatrix}.
$$
By \eqref{cot8} and \eqref{cot9}, we have
\begin{equation}\label{sub4}
\left|W(m+1)-W(m)\right|=\|\vec{u}_i(m)^*\vec{f}_j(m)\|\leq \|\vec{u}_i(m)\|\|\vec{f}_j(m)\|
\end{equation} 
Note that $h_1>\e_1(E)>0$ \footnote{Since $E$ is subcritical.}. Choose $\delta$ sufficiently small such that $3\delta<h_1-\e_1(E)$. For $k$ sufficiently large depending on $\delta$, by (3) of Theorem \ref{theorem-main1general}, \eqref{sub3} and \eqref{sub4}, we have that for any $m\geq 1$,
\begin{align*}
&|W(m)-W(0)|\leq \left(\sum_{l=0}^{m-1} \|\vec{u}_i(l)\|\right)Cn_ke^{-2\pi n_k (h_1-\delta)}\\
\leq &Cl(n_k)\sigma_{n_k}^i(m_1\alpha)e^{2\pi l(n_k)(\e_1(E)+\delta)}Ce^{-2\pi n_k (h_1-\delta)}\leq C\sigma_{n_k}^i(m_1\alpha)e^{-n_k \delta}.
\end{align*}
Note that $W(m)\rightarrow 0$ as $m\rightarrow\infty$ by (3) of Theorem \ref{theorem-main1general}  and \eqref{eeee1age}. Hence 
$$
|W(0)|\leq C\sigma_{n_k}^i(m_1\alpha)e^{-n_k \delta}.
$$
This completes  the proof of \eqref{cot10}.  \eqref{cot11} follows in  the same way.

For \eqref{cot12}, we let
$$
W'(m)=\begin{pmatrix}
\vec{b}_j(m)\\ \vec{b}_j(m-1)
\end{pmatrix}^*S \begin{pmatrix}
\vec{b}_i(m)\\ \vec{b}_i(m-1)
\end{pmatrix}.
$$%
Subtracting the two forced equations gives the exact identity
\[
W'(m+1)-W'(m)
=\vec{b}_j(m)^*\vec{f}_i(m)-\vec{f}_j(m)^*\vec{b}_i(m).
\]
Consequently, by \eqref{tan8}, \eqref{sub3} and the Fourier decay,
\begin{align*}
|W'(m)-W'(0)|
&\leq Cn_ke^{-2\pi n_k(h_1-\delta)}
\sum_{l=0}^{m-1}\bigl(\|\vec{b}_i(l)\|+\|\vec{b}_j(l)\|\bigr)\\
&\leq Cn_ke^{-2\pi n_k(h_1-\delta)+Cr_k}
\leq e^{-n_kh_1/2}.
\end{align*}
Here the full $\ell^1$ norms of the Fourier coefficients are at most
$Ce^{Cr_k}$, with $C$ independent of $k$, and $r_k=o(n_k)$.
Since $W'(m)\rightarrow0$ as $m\rightarrow\infty$, we have 
\begin{align*}
|W'(0)|\leq e^{-n_k h_1/2}.
\end{align*}
\end{pf}

\begin{Lemma}\label{key1}
For $j=1,2$, let
$$
\begin{pmatrix}\vec{b}^s_j(0)\\ \vec{b}^s_j(-1)\end{pmatrix}\in E_s^{n_k}(m_1\alpha),\ \ \begin{pmatrix}\vec{b}^c_j(0)\\ \vec{b}^c_j(-1)\end{pmatrix}\in E_c^{n_k}(m_1\alpha),\ \ \begin{pmatrix}\vec{b}^u_j(0)\\ \vec{b}^u_j(-1)\end{pmatrix}\in E_u^{n_k}(m_1\alpha)
$$
be such that
\begin{align*}
\begin{pmatrix}\vec{b}_j(0)\\ \vec{b}_j(-1)\end{pmatrix}=\begin{pmatrix}\vec{b}^s_j(0)\\ \vec{b}^s_j(-1)\end{pmatrix}+\begin{pmatrix}\vec{b}^c_j(0)\\ \vec{b}^c_j(-1)\end{pmatrix}+\begin{pmatrix}\vec{b}^u_j(0)\\ \vec{b}^u_j(-1)\end{pmatrix}.
\end{align*}
Then  there is $C>0$, such that
\begin{equation}\label{erraa}
\left\|\begin{pmatrix}\vec{b}^s_j(0)\\ \vec{b}^s_j(-1)\end{pmatrix}\right\|,\ \ \left\|\begin{pmatrix}\vec{b}^u_j(0)\\ \vec{b}^u_j(-1)\end{pmatrix}\right\|\leq Ce^{-\delta n_k/2}.
\end{equation}
\end{Lemma}
\begin{pf}
By  (3) of Theorem \ref{theorem-main1general}, 
$$
 \begin{pmatrix}
\vec{u}_i(0)\\ \vec{u}_i(-1)
\end{pmatrix},  \ \ i=1,2,\cdots, l(n_k)-1
$$ 
form an orthonormal basis of $E_s^{n_k}(m_1\alpha)$,
$$
 \begin{pmatrix}
\vec{v}_i(0)\\ \vec{v}_i(-1)
\end{pmatrix},  \ \ i=1,2,\cdots, l(n_k)-1
$$ 
form an orthonormal basis of $E_u^{n_k}(m_1\alpha)$ and
\begin{align}\label{cot20}
\nonumber & \begin{pmatrix}
\vec{u}_1(0)&\cdots&\vec{u}_{l(n_k)-1}(0)\\ \vec{u}_1(-1)&\cdots&\vec{u}_{l(n_k)-1}(-1)
\end{pmatrix}^*S\begin{pmatrix}\vec{v}_1(0)&\cdots&\vec{v}_{l(n_k)-1}(0)\\ \vec{v}_1(-1)&\cdots&\vec{v}_{l(n_k)-1}(-1)
\end{pmatrix}\\
=&\begin{pmatrix}\sigma_{n_k}^1(m_1\alpha)\\ &\ddots\\&& \sigma_{n_k}^{l(n_k)-1}(m_1\alpha)\end{pmatrix}.
\end{align}

Assume
$$
\begin{pmatrix}\vec{b}^s_j(0)\\ \vec{b}^s_j(-1)\end{pmatrix}=c_{1,j} \begin{pmatrix}
\vec{u}_1(0)\\ \vec{u}_1(-1)
\end{pmatrix}+\cdots+c_{l(n_k)-1,j} \begin{pmatrix}
\vec{u}_{l(n_k)-1}(0)\\ \vec{u}_{l(n_k)-1}(-1)
\end{pmatrix}.
$$
By Lemma \ref{sub8} and \eqref{cot20}, we have
$$
|c_{i,j}\sigma_{n_{k}}^i(m_1\alpha)|=\left|\begin{pmatrix}
\vec{b}_j(0)\\ \vec{b}_j(-1)
\end{pmatrix}^*S \begin{pmatrix}
\vec{v}_i(0)\\ \vec{v}_i(-1)
\end{pmatrix}\right| \leq C\sigma_{n_k}^i(m_1\alpha)e^{-n_k \delta},
$$
Hence 
$$
\left\|\begin{pmatrix}\vec{b}^s_j(0)\\ \vec{b}^s_j(-1)\end{pmatrix}\right\|^2=\sum_{j=1}^{l(n_k)-1}|c_{i,j}|^2\leq Cn_ke^{-2n_k \delta}.
$$
Similarly, we have the result for $\left\|\begin{pmatrix}\vec{b}^u_j(0)\\ \vec{b}^u_j(-1)\end{pmatrix}\right\|$.
\end{pf}
By Lemma \ref{gejito} and the fact that $\|B^k\|_{C^0}\leq \|B\|_{C^0}\leq C$, for $j=1,2$, one has
\begin{equation}\label{ff4aa}
\left\|\begin{pmatrix}\vec{b}_j(0)\\ \vec{b}_j(-1)\end{pmatrix}\right\| \leq C,\ \  \left\|\begin{pmatrix}\vec{b}_1(0)\\ \vec{b}_1(-1)\end{pmatrix}\right\|\geq r_k^{-1}.
\end{equation}
By Lemma \ref{key1}, we have
\begin{equation}\label{ff4aage}
 \left\|\begin{pmatrix}\vec{b}^c_j(0)\\ \vec{b}^c_j(-1)\end{pmatrix}\right\| \leq 2C,\ \ \left\|\begin{pmatrix}\vec{b}^c_1(0)\\ \vec{b}^c_1(-1)\end{pmatrix}\right\|\geq \frac{1}{2}r_k^{-1}.
\end{equation}
Let $C=10^8(\frac{200\pi}{h}+2\pi h)(2+4\pi\e_1(E))+1000$.
\begin{Proposition}\label{sub9}
There exist $K_1(E,v,\alpha)$, such that if $k>K_1$, then
\begin{equation}\label{fwnaa}
\det{\begin{pmatrix}\vec{b}_1^c(0)&\vec{b}_2^c(0)\\ \vec{b}_1^c(-1)&\vec{b}_2^c(-1)\end{pmatrix}^*\begin{pmatrix}\vec{b}_1^c(0)&\vec{b}_2^c(0)\\ \vec{b}_1^c(-1)&\vec{b}_2^c(-1)\end{pmatrix}} \geq e^{-Cr_k}.
\end{equation}
\end{Proposition}
\begin{pf}
Note that ${\rm dim}E_c^{n_k}(m_\alpha)=2$. Thus there is a unit vector $\begin{pmatrix}\vec{v}_c(0)\\ \vec{v}_c(-1)\end{pmatrix}\in E^{n_k}_c(m_1\alpha)$ which is orthogonal to $\begin{pmatrix}\vec{b}^c_1(0)\\ \vec{b}^c_1(-1)\end{pmatrix}$, and we write
\begin{align}\label{gjnew2aa}
\nonumber \begin{pmatrix}\vec{b}^c_2(0)\\ \vec{b}^c_2(-1)\end{pmatrix}=&\frac{\left\langle\begin{pmatrix}\vec{b}^c_2(0)\\ \vec{b}^c_2(-1)\end{pmatrix},\begin{pmatrix}\vec{b}^c_1(0)\\ \vec{b}^c_1(-1)\end{pmatrix}\right\rangle}{\left\|\begin{pmatrix}\vec{b}^c_1(0)\\ \vec{b}^c_1(-1)\end{pmatrix}\right\|^2}\begin{pmatrix}\vec{b}^c_1(0)\\ \vec{b}^c_1(-1)\end{pmatrix}\\
&+\left\langle \begin{pmatrix}\vec{b}^c_2(0)\\ \vec{b}^c_2(-1)\end{pmatrix},\begin{pmatrix}\vec{v}_c(0)\\ \vec{v}_c(-1)\end{pmatrix}\right\rangle\begin{pmatrix}\vec{v}_c(0)\\ \vec{v}_c(-1)\end{pmatrix}.
\end{align}
By direct calculation, we have
\begin{align}\label{gjnew1aa}
\nonumber &\det{\begin{pmatrix}\vec{b}_1^c(0)&\vec{b}_2^c(0)\\ \vec{b}_1^c(-1)&\vec{b}_2^c(-1)\end{pmatrix}^*\begin{pmatrix}\vec{b}_1^c(0)&\vec{b}_2^c(0)\\ \vec{b}_1^c(-1)&\vec{b}_2^c(-1)\end{pmatrix}}\\
=&\left|\left\langle \begin{pmatrix}\vec{b}^c_2(0)\\ \vec{b}^c_2(-1)\end{pmatrix},\begin{pmatrix}\vec{v}_c(0)\\ \vec{v}_c(-1)\end{pmatrix}\right\rangle\right|^2 \left\|\begin{pmatrix}\vec{b}^c_1(0)\\ \vec{b}^c_1(-1)\end{pmatrix}\right\|^2.
\end{align}

We now prove \eqref{fwnaa} by contradiction. If
$$
\det{\begin{pmatrix}\vec{b}_1^c(0)&\vec{b}_2^c(0)\\ \vec{b}_1^c(-1)&\vec{b}_2^c(-1)\end{pmatrix}^*\begin{pmatrix}\vec{b}_1^c(0)&\vec{b}_2^c(0)\\ \vec{b}_1^c(-1)&\vec{b}_2^c(-1)\end{pmatrix}}<e^{-Cr_k},
$$
then by \eqref{gjnew1aa} and  \eqref{ff4aage}, we further have for $k$ sufficiently large,
\begin{equation}\label{erroraa}
\left|\left\langle \begin{pmatrix}\vec{b}^c_2(0)\\ \vec{b}^c_2(-1)\end{pmatrix},\begin{pmatrix}\vec{v}_c(0)\\ \vec{v}_c(-1)\end{pmatrix}\right\rangle\right|\leq  2r_ke^{-Cr_k/2}\leq  e^{-Cr_k/3}.
\end{equation}
By \eqref{gjnew2aa} and \eqref{erroraa}, we have that
\begin{equation}\label{errornewaage}
\left\|\begin{pmatrix}\vec{b}^c_2(0)\\ \vec{b}^c_2(-1)\end{pmatrix}-\frac{\left\langle \begin{pmatrix}\vec{b}^c_2(0)\\ \vec{b}^c_2(-1)\end{pmatrix},\begin{pmatrix}\vec{b}^c_1(0)\\ \vec{b}^c_1(-1)\end{pmatrix}\right\rangle}{\left\|\begin{pmatrix}\vec{b}^c_1(0)\\ \vec{b}^c_1(-1)\end{pmatrix}\right\|^2}\begin{pmatrix}\vec{b}^c_1(0)\\ \vec{b}^c_1(-1)\end{pmatrix}\right\|\leq e^{-Cr_k/4}.
\end{equation}
By \eqref{erraa}, we have
\begin{equation}\label{gj1aa}
\left\|\begin{pmatrix}\vec{b}_j(0)\\ \vec{b}_j(-1)\end{pmatrix}-\begin{pmatrix}\vec{b}^c_j(0)\\ \vec{b}^c_j(-1)\end{pmatrix}\right\|\leq Ce^{-\delta n_k/2}.
\end{equation} 
By  \eqref{ff4aa}, \eqref{ff4aage} and \eqref{gj1aa}, we have that
\begin{equation}\label{errornewaa}
\left\|\begin{pmatrix}\vec{b}_2(0)\\ \vec{b}_2(-1)\end{pmatrix}-\frac{\left\langle \begin{pmatrix}\vec{b}_2(0)\\ \vec{b}_2(-1)\end{pmatrix},\begin{pmatrix}\vec{b}_1(0)\\ \vec{b}_1(-1)\end{pmatrix}\right\rangle}{\left\|\begin{pmatrix}\vec{b}_1(0)\\ \vec{b}_1(-1)\end{pmatrix}\right\|^2}\begin{pmatrix}\vec{b}_1(0)\\ \vec{b}_1(-1)\end{pmatrix}\right\|\leq Ce^{-Cr_k/5}.
\end{equation}
which means the orthogonal projection of $\begin{pmatrix}\vec{b}_2(0)\\ \vec{b}_2(-1)\end{pmatrix}$ to the vector  $\begin{pmatrix}\vec{b}_1(0)\\ \vec{b}_1(-1)\end{pmatrix}$ is large.  

Now we consider
\begin{align*}
b(x):=\det{\begin{pmatrix}b^k_{11}(x)&b^k_{12}(x)\\ b^k_{21}(x)&b^k_{22}(x)\end{pmatrix}}=\det{\begin{pmatrix}b^k_{11}(x)&b^k_{12}(x)-c_1b^k_{11}(x)\\ b^k_{21}(x)&b^k_{22}(x)-c_1b^k_{21}(x)\end{pmatrix}}
\end{align*}
where 
$$
c_1=\frac{\left\langle \begin{pmatrix}\vec{b}_2(0)\\ \vec{b}_2(-1)\end{pmatrix},\begin{pmatrix}\vec{b}_1(0)\\ \vec{b}_1(-1)\end{pmatrix}\right\rangle}{\left\|\begin{pmatrix}\vec{b}_1(0)\\ \vec{b}_1(-1)\end{pmatrix}\right\|^2}.
$$
By \eqref{ff4aa}, we have
\begin{equation}\label{sinnew3}
|c_1|\leq Cr_k^2.
\end{equation} We aim to estimate $b(x)$ to arrive at a contradiction. Since $\|B^k\|_0\leq C$, it is enough to  estimate its
Fourier coefficients
$$
\hat{b}^k_{2}(m)-c_1\hat{b}^k_{1}(m)=\int_\T (b^k_{12}(x)-c_1b^k_{11}(x))e^{2\pi imx}dx.
$$
We distinguish two cases:\\

\textbf{Case I:} If $|m|\geq \frac{100r_k}{h}$, then
by \eqref{tan8} and \eqref{sinnew3}, we have
\begin{align}\label{new1010aa}
\left|\hat{b}^k_2(m)-c_1 \hat{b}^k_{1}(m) \right|\leq Ce^{-|m|h/2}e^{r_k}\leq e^{-40r_k}.
\end{align}
provided $k$ is sufficiently large.\\

\textbf{Case II:} If $|m|\leq \frac{100r_k}{h}$, if $l(n_k)>4\frac{100r_k}{h}$, then by \eqref{errornewaa}, we have
\begin{equation}\label{cot174}
\left|\hat{b}^k_2(m)-c_1 \hat{b}^k_1(m)\right|\leq Ce^{-Cr_k/10}.
\end{equation}
Otherwise, use
$$
p=[\frac{m-m_1}{l(n_k)}],\ \ m-m_1\in [(p-1)l(n_k),(p+1)l(n_k)-1],
$$
we have $|p|\leq 2([\frac{100r_k}{hl(n_k)}]+1)$. We denote 
\begin{equation}\label{cot2'}
\begin{pmatrix}\vec{u}'_i(0)\\ \vec{u}'_i(-1)
\end{pmatrix}=\begin{pmatrix}
u_{n_k}^{i}((m_1+pl(n_k))\alpha,l(n_k)-1)\\
u_{n_k}^{i}((m_1+pl(n_k))\alpha,l(n_k)-2)\\
\vdots\\
u_{n_k}^{i}((m_1+pl(n_k))\alpha,-l(n_k))
\end{pmatrix},\ \ \begin{pmatrix}\vec{v}'_i(0)\\ \vec{v}'_i(-1)
\end{pmatrix}=\begin{pmatrix}
v_{n_k}^{i}((m_1+pl(n_k))\alpha,l(n_k)-1)\\
v_{n_k}^{i}((m_1+pl(n_k))\alpha,l(n_k)-2)\\
\vdots\\
v_{n_k}^{i}((m_1+pl(n_k))\alpha,-l(n_k))
\end{pmatrix}
\end{equation}
\begin{equation}\label{cot2''}
\begin{pmatrix}\vec{w}'_i(0)\\ \vec{w}'_i(-1)
\end{pmatrix}=\begin{pmatrix}
w_{n_k}^{i}((m_1+pl(n_k))\alpha,l(n_k)-1)\\
w_{n_k}^{i}((m_1+pl(n_k))\alpha,l(n_k)-2)\\
\vdots\\
w_{n_k}^{i}((m_1+pl(n_k))\alpha,-l(n_k))
\end{pmatrix},
\end{equation}
where $u_{n_k}^i((m_1+pl(n_k))\alpha,m)$, $v_{n_k}^i((m_1+pl(n_k))\alpha,m)$ and $w_{n_k}^i((m_1+pl(n_k))\alpha,m)$ are solutions of $L^{2\cos}_{v_{n_k},\alpha,(m_1+pl(n_k))\alpha}u=Eu$  from (3) of Theorem \ref{theorem-main1general} (here we also omit $E$ from all notations for simplicity).

By exactly the same argument as in Lemma \ref{sub8}, weh have
\begin{align}\label{cot10'}
\left|\begin{pmatrix}
\vec{b}_j(p)\\ \vec{b}_j(p-1)
\end{pmatrix}^*S \begin{pmatrix}
\vec{u}'_i(0)\\ \vec{u}'_i(-1)
\end{pmatrix}\right| \leq C\sigma_{n_k}^i((m_1+pl(n_k))\alpha)e^{-n_k \delta}, \ \ i=1,\cdots,l(n_k)-1,
\end{align}
\begin{align}\label{cot11'}
\left|\begin{pmatrix}
\vec{b}_j(p)\\ \vec{b}_j(p-1)
\end{pmatrix}^*S \begin{pmatrix}
\vec{v}'_i(0)\\ \vec{v}'_i(-1)
\end{pmatrix}\right| \leq C\sigma_{n_k}^i((m_1+pl(n_k))\alpha)e^{-n_k \delta}, \ \ i=1,\cdots,l(n_k)-1.
\end{align}

Moreover, if we let $\vec{w}'_i(m)$  be solutions of $L^{2\cos}_{v_{n_k},\alpha,(m_1+pl(n_k))\alpha}u=Eu$ with initial data $\begin{pmatrix}\vec{w}'_i(0)\\ \vec{w}'_i(-1)
\end{pmatrix}$, $i=1,2$, repectively. We further have
\begin{align}\label{cot8'}
\nonumber &\vec{w}'_i(m)^*C^*\vec{b}_j(m+p-1)+\vec{w}'_i(m)^*C\vec{b}_j(m+p+1)\\
+&\vec{w}'_i(m)^*(B(m_1\alpha+(m+p)l(n_k)\alpha)-EI_{l(n_k)})\vec{b}_j(m+p)=\vec{w}'_i(m)^*\vec{f}_j(m+p).
\end{align}
\begin{align}\label{cot9'}
\nonumber &\vec{b}_j(m+p)^*C^*\vec{w}'_i(m-1)+\vec{b}_j(m+p)^*C\vec{w}'_i(m+1)\\
+&\vec{b}_j(m+p)^*(B(m_1\alpha+(m+p)l(n_k)\alpha)-EI_{l(n_k)})\vec{w}'_i(m)=0.
\end{align}
Let 
$$
W''(m)=\begin{pmatrix}
\vec{b}_2(m+p)-c_1\vec{b}_1(m+p)\\ \vec{b}_2(m+p-1)-c_1\vec{b}_1(m+p-1)
\end{pmatrix}^*S \begin{pmatrix}
\vec{w}'_i(m)\\ \vec{w}'_i(m-1)
\end{pmatrix}.
$$
By \eqref{cot8'} and \eqref{cot9'}, we have
\begin{equation}\label{sub4'}
\left|W''(m+1)-W''(m)\right|\leq \|\vec{w}'_i(m)\|\|\vec{f}_2(m+p)-c_1\vec{f}_1(m+p)\|
\end{equation} 
By (3) of Theorem \ref{theorem-main1general} and \eqref{sub4'}, we have 
\begin{align}\label{cot100'}
&|W''(0)-W''(-p)|\leq \left(\sum_{l=\min\{-p,0\}}^{\max\{-p,0\}-1} \|\vec{w}'_i(l)\|\right)Cn_ke^{-2\pi n_k (h_1-\delta)}\\
\leq &C(\frac{100r_k}{h}+1)e^{8\pi([\frac{100r_k}{hl(n_k)}]+2)l(n_k)(\e_1(E)+\delta)}e^{-2\pi n_k (h_1-\delta)}\leq Ce^{-n_k \delta}.
\end{align}
By \eqref{errornewaa} and (3) of Theorem \ref{theorem-main1general}, we have
\begin{equation}\label{cot172}
|W''(-p)|\leq Ce^{-Cr_k/5}C(\frac{100r_k}{h}+1)e^{8\pi([\frac{100r_k}{hl(n_k)}]+2)l(n_k)(\e_1(E)+\delta)}\leq Ce^{-Cr_k/8}.
\end{equation}
By \eqref{cot100'} and \eqref{cot172}, we have
\begin{equation}\label{cot173}
|W''(0)|\leq |W''(-p)|+ Ce^{-n_k \delta}\leq   2Ce^{-Cr_k/8}.
\end{equation}
Again by (3) of Theorem \ref{theorem-main1general}, we have
\begin{align*}
&\begin{pmatrix}
\vec{b}_2(p)-c_1\vec{b}_1(p)\\ \vec{b}_2(p-1)-c_1\vec{b}_1(p-1)
\end{pmatrix}=c'_{1} \begin{pmatrix}
\vec{u}'_1(0)\\ \vec{u}'_1(-1)
\end{pmatrix}+\cdots+c'_{l(n_k)-1} \begin{pmatrix}
\vec{u}'_{l(n_k)-1}(0)\\ \vec{u}'_{l(n_k)-1}(-1)
\end{pmatrix}\\
&+c'_{l(n_k)} \begin{pmatrix}
\vec{w}'_1(0)\\ \vec{w}'_1(-1)
\end{pmatrix}+c'_{l(n_k)+1} \begin{pmatrix}
\vec{w}'_2(0)\\ \vec{w}'_2(-1)
\end{pmatrix}+c'_{l(n_k)+2} \begin{pmatrix}
\vec{v}'_1(0)\\ \vec{v}'_1(-1)
\end{pmatrix}+\cdots+c'_{2l(n_k)} \begin{pmatrix}
\vec{v}'_{l(n_k)-1}(0)\\ \vec{v}'_{l(n_k)-1}(-1)
\end{pmatrix}
\end{align*}
By \eqref{cot10'}, \eqref{cot11'}, \eqref{cot173} and the same argument as in Lemma \ref{key1}, we have
\begin{equation}\label{askd1}
|c'_{j}|\leq 2Ce^{-Cr_k/8},\ \ 1\leq j\leq 2l(n_k).
\end{equation}
By \eqref{askd1} and (3) of Theorem \ref{theorem-main1general}, we have
\begin{align*}
&\left\|\begin{pmatrix}
\vec{b}_2(p)-c_1\vec{b}_1(p)\\ \vec{b}_2(p-1)-c_1\vec{b}_1(p-1)
\end{pmatrix}\right\|\\
&\leq \sum_{j=1}^{l(n_k)-1}|c_j'|+|c_{l(n_k)}'|\left|\begin{pmatrix}w_1'(0)\\ w_1'(-1)\end{pmatrix}\right|+|c_{l(n_k)+1}'|\left|\begin{pmatrix}w_2'(0)\\ w_2'(-1)\end{pmatrix}\right|+\sum_{j=l(n_k)+2}^{2l(n_k)}|c_j'|\\
&\leq C(\frac{100r_k}{h}+1)e^{4\pi([\frac{100r_k}{hl(n_k)}]+2)l(n_k)(\e_1(E)+\delta)}e^{-Cr_k/8}\leq Ce^{-Cr_k/10}.
\end{align*}
Hence
\begin{equation}\label{sinnew4}
\left|\hat{b}^k_2(m)-c_1 \hat{b}^k_1(m)\right|\leq Ce^{-Cr_k/10}.
\end{equation}

By \eqref{new1010aa}, \eqref{cot174} and \eqref{sinnew4}, we have
\begin{equation}\label{sinnew5}
\left\|b_{12}-c_1 b_{11}\right\|_0\leq 100\frac{r_k}{h}Ce^{-r_k/2}+Ce^{-Cr_k/20}\leq e^{-r_k/4}.
\end{equation}
Finally, by \eqref{tan5} and \eqref{tan6}, we have
\begin{align*}
|b(x)|\leq C\left(\left\|b_{12}-c_1 b_{11}\right\|_{C^0}+\left\|b_{22}-c_1 b_{21}\right\|_{C^0}\right)\leq  Ce^{-r_k/4}.
\end{align*}
This contradicts 
\begin{align*}
|b(x)|=1.
\end{align*}
\end{pf}

Combined with quantitative almost reducibility, Lemma \ref{sub8} and Proposition \ref{sub9} are the key technical inputs that yield the subcritical dry Ten Martini statement announced in
Remark 1.3.

Finally, by Lemma \ref{sub8}, for $i,j=1,2$, we have
\begin{align}\label{gj2aa}
\left|\begin{pmatrix}
\vec{b}_j(0)\\ \vec{b}_j(-1)
\end{pmatrix}^*S \begin{pmatrix}
\vec{b}_i(0)\\ \vec{b}_i(-1)
\end{pmatrix}\right| \leq  e^{-n_k h_1/2}.
\end{align}
By Lemma \ref{key1} and \eqref{ff4aa}, we have
\begin{equation}\label{gj3geaa}
\left|\begin{pmatrix}
\vec{b}^c_j(0)\\ \vec{b}^c_j(-1)
\end{pmatrix}^*S \begin{pmatrix}
\vec{b}^c_i(0)\\ \vec{b}^c_i(-1)
\end{pmatrix}\right| \leq    Ce^{-n_k h_1/2}+Ce^{-\delta n_k}\leq e^{-\delta n_k/2}.
\end{equation}

Assume $Q_k=[16\pi C(1+h^{-1})r_k]$ and $R_k=\min\{l(n_k),Q_k\}$. Let
\begin{equation}\label{sinnew8}
\begin{pmatrix}
\vec{b}^c_1(0)&\vec{b}^c_2(0)\\ \vec{b}^c_1(-1)&\vec{b}^c_2(-1)
\end{pmatrix}_{2R_k}=\begin{pmatrix}
0_{l(n_k)-R_k}\\ & I_{2R_k}\\ && 0_{l(n_k)-R_k}
\end{pmatrix}\begin{pmatrix}\vec{b}^c_1(0)&\vec{b}^c_2(0)\\ \vec{b}^c_1(-1)&\vec{b}^c_2(-1)
\end{pmatrix},
\end{equation}
\begin{equation}\label{sinnew6}
A=\begin{pmatrix}\vec{b}_1^c(0)&\vec{b}_2^c(0)\\ \vec{b}_1^c(-1)&\vec{b}_2^c(-1)\end{pmatrix}^*\begin{pmatrix}\vec{b}_1^c(0)&\vec{b}_2^c(0)\\ \vec{b}_1^c(-1)&\vec{b}_2^c(-1)\end{pmatrix}.
\end{equation}
\begin{equation}\label{sinnew7}
A_{2R_k}=\left(\begin{pmatrix}\vec{b}_1^c(0)&\vec{b}_2^c(0)\\ \vec{b}_1^c(-1)&\vec{b}_2^c(-1)\end{pmatrix}_{2R_k}\right)^*\begin{pmatrix}\vec{b}_1^c(0)&\vec{b}_2^c(0)\\ \vec{b}_1^c(-1)&\vec{b}_2^c(-1)\end{pmatrix}_{2R_k}.
\end{equation}
By  \eqref{tan8} and \eqref{gj1aa}, we have
\begin{equation}\label{sinnew9}
\left\|\begin{pmatrix}\vec{b}^c_1(0)&\vec{b}^c_2(0)\\ \vec{b}^c_1(-1)&\vec{b}^c_2(-1)
\end{pmatrix}-\begin{pmatrix}\vec{b}^c_1(0)&\vec{b}^c_2(0)\\ \vec{b}^c_1(-1)&\vec{b}^c_2(-1)
\end{pmatrix}_{2R_k}\right\|  \leq 2e^{-7Cr_k}.
\end{equation}
By \eqref{ff4aa}, \eqref{ff4aage}, \eqref{sinnew6}, \eqref{sinnew7} and \eqref{sinnew9}, we have
\begin{equation}\label{bimsa11}
\|A-A_{2R_k}\|\leq C e^{-7Cr_k}
\end{equation}
Note also that 
\begin{equation}\label{bimsa12}
\|A\|\leq 4\|B^k\|_0^2\leq C,
\end{equation}
$$
A_{2R_k}=A+A_{2R_k}-A=A(I+A^{-1}(A_{2R_k}-A)).
$$
Hence by \eqref{bimsa11} and \eqref{bimsa12}, $|\det{(I+A^{-1}(A_{2R_k}-A))}-1|\leq Ce^{Cr_k}e^{-7Cr_k}$, hence
$$
|\det A_{2R_k}-\det A|\leq e^{-5Cr_k}.
$$
which by  Proposition \ref{sub9} implies that 
\begin{equation}\label{sincos1}
|\det A_{2R_k}|\geq \frac{1}{2}e^{-Cr_k}.
\end{equation}

Recall that $(\alpha,L_{E,v_{n_k}})$ is $PH2$ and there exists a continuous invariant decomposition
$$
\C^{2l(n_k)}=E_s^{n_k}(\theta)\oplus E^{n_k}_c(\theta)\oplus E^{n_k}_u(\theta),\ \ {\rm dim} E^{n_k}_c(\theta)=2.
$$
By Proposition \ref{sub9}, $\begin{pmatrix}
\vec{b}^c_1(0)\\ \vec{b}^c_1(-1)
\end{pmatrix}, \begin{pmatrix}
\vec{b}^c_2(0)\\ \vec{b}^c_2(-1)
\end{pmatrix}$ form a basis of $E_c^{n_k}(m_1\alpha)$. Together with  \eqref{sincos1} and \eqref{gj3geaa}, we have
\begin{align}\label{sincos2}
\Omega=\nonumber &\det{A_{2R_k}\left(\begin{pmatrix}\vec{b}_1^c(0)&\vec{b}_2^c(0)\\ \vec{b}_1^c(-1)&\vec{b}_2^c(-1)\end{pmatrix}^*S\begin{pmatrix}\vec{b}_1^c(0)&\vec{b}_2^c(0)\\ \vec{b}_1^c(-1)&\vec{b}_2^c(-1)\end{pmatrix}\right)^{-1}}\\
=&\frac{\det{A_{2R_k}}}{\det \begin{pmatrix}\vec{b}_1^c(0)&\vec{b}_2^c(0)\\ \vec{b}_1^c(-1)&\vec{b}_2^c(-1)\end{pmatrix}^*S \begin{pmatrix}\vec{b}_1^c(0)&\vec{b}_2^c(0)\\ \vec{b}_1^c(-1)&\vec{b}_2^c(-1)\end{pmatrix}}\geq e^{\delta n_k/4}.
\end{align}

Notice  that by Remark \ref{bimsa13}, $L_{n_k}(\theta)$ defined in Theorem \ref{theorem-main1general} does not depend on the choice of basis. Since $\begin{pmatrix}
\vec{b}^c_1(0)\\ \vec{b}^c_1(-1)
\end{pmatrix}, \begin{pmatrix}
\vec{b}^c_2(0)\\ \vec{b}^c_2(-1)
\end{pmatrix}$ form a basis of $E_c^{n_k}(m_1\alpha)$,  we have
\begin{align*}
L_{n_k}(m_1\alpha)=&\begin{pmatrix}\vec{b}_1^c(0)&\vec{b}_2^c(0)\\ \vec{b}_1^c(-1)&\vec{b}_2^c(-1)\end{pmatrix}\left(\begin{pmatrix}\vec{b}_1^c(0)&\vec{b}_2^c(0)\\ \vec{b}_1^c(-1)&\vec{b}_2^c(-1)\end{pmatrix}^*S\begin{pmatrix}\vec{b}_1^c(0)&\vec{b}_2^c(0)\\ \vec{b}_1^c(-1)&\vec{b}_2^c(-1)\end{pmatrix}\right)^{-1}\\
&\cdot\begin{pmatrix}\vec{b}_1^c(0)&\vec{b}_2^c(0)\\ \vec{b}_1^c(-1)&\vec{b}_2^c(-1)\end{pmatrix}^*
\end{align*}

For any $n\times 2$ matrix $M$ and $2\times n$ matrix $N$, we have 
\begin{equation}\label{bimsa15}
\det{NM}=\tr\Lambda^2 MN.
\end{equation}

By \eqref{sinnew7}, \eqref{sincos2} and applying \eqref{bimsa15} with $n=2l(n_k)$, $M=\begin{pmatrix}\vec{b}^c_1(0)&\vec{b}^c_2(0)\\ \vec{b}^c_1(-1)&\vec{b}^c_2(-1)
\end{pmatrix}_{2R_k}$ and $N=\left(\begin{pmatrix}\vec{b}_1^c(0)&\vec{b}_2^c(0)\\ \vec{b}_1^c(-1)&\vec{b}_2^c(-1)\end{pmatrix}^*S\begin{pmatrix}\vec{b}_1^c(0)&\vec{b}_2^c(0)\\ \vec{b}_1^c(-1)&\vec{b}_2^c(-1)\end{pmatrix}\right)^{-1}\begin{pmatrix}\vec{b}^c_1(0)&\vec{b}^c_2(0)\\ \vec{b}^c_1(-1)&\vec{b}^c_2(-1)
\end{pmatrix}_{2R_k}^*$,   we have
$$
\Omega={\tr}\Lambda^2\left(\begin{pmatrix}
0_{l(n_k)-R_k}\\ & I_{2R_k}\\ && 0_{l(n_k)-R_k}
\end{pmatrix}L_{n_k}(m_1\alpha) \begin{pmatrix}
0_{l(n_k)-R_k}\\ & I_{2R_k}\\ && 0_{l(n_k)-R_k}
\end{pmatrix}\right).
$$
By (2) of Theorem \ref{theorem-main1general}, 
\begin{align}\label{gjnewage}
|\Omega|\leq CR_k^2e^{128C\pi (\e_1(E)+\delta)R_k}
\end{align}
which contradicts  \eqref{sincos2} if $k$ is sufficiently large. 
\end{pf}
\subsection{Proof of the subcritical case of Theorem \ref{main12}}
We prove by contradiction. Assume there is a nonempty  $I\subset\overline{\Sigma_{v,\alpha}^{1,0}}$. We distinguish two case:

Case I: If there is a subcritical energy $E\in I$, then there is $E'\in I\cap \Sigma_{v,\alpha}^{1,0}$ with $E'$ subcritical, By openness of subcriticality, there is an interval $I'\subset I$, such that any $E_0'\subset I'$ is subcritical. Thus by Theorem \ref{Kotani-st}, then there are $B_E\in C^\omega(\T,SL(2,\R))$ and $\psi_E\in C^\omega(\T,\R)$ depending analytically on $E\in I$, such that
$$
B_E^{-1}(x+\alpha)S^v_E(x)B_E(x)=R_{\psi_E(x)}.
$$
Note that by the integrated density of states consideration, $\int_\T\psi_E(x)dx$ is not a constant in $E$, so there is $E_0\in I$ with $\int_\T
\psi_{E_0}(x)dx=k_1\alpha+k_2,$ for some $k_1,k_2\in\Z.$ We now define
$\widetilde{B}_{E_0}(x)=B_{E_0}(x)R_{k_1x}$, and obtain
$$
\widetilde{B}_{E_0}^{-1}(\theta+\alpha)S^v_{E_0}(\theta)\widetilde{B}_{E_0}(\theta)=R_{\psi_{E_0}(x)-k_1\alpha}.
$$
Note that $\int_\T \left(\psi_{E_0}(x)-k_1\alpha-k_2\right)dx=0$ contradicting Theorem \ref{subkey}. \qed

 Case II: All $E\in I$ is critical, this follows from case II in the proof of the (super)critical part of Theorem \ref{main12}.
\section{Appendix}
\subsection{Specialized facts on the two-dimensional center}
\label{app:center-facts}
Let $O(z,x)$ be a symplectic holomorphic frame of an invariant two-dimensional
bundle such that
\[
 O(\bar{z},\bar{x})^* S O(z,x)=J,
 \qquad
 \mathcal L(z,x)O(z,x)=O(z,x+\alpha)M(z,x).
\]
\begin{Lemma}[Symplecticity of the center dynamics]
\label{lem:center-J-unitary}
For real $E\in\R$ and $x\in\mathbb T$,
\[
        M(E,x)^*JM(E,x)=J.
\]
\end{Lemma}

\begin{pf}
Using invariance of the frame and preservation of $S$, we obtain
\[
\begin{aligned}
J
 &=O(\bar{z},\bar{x})^*T\mathcal L(\bar{z},\bar{x})^*
       S\mathcal L(z,x)O(z,x)\\
 &=M(\bar{z},\bar{x})^* O(\bar{z},\bar{x}+\alpha)^*
       S O(z,x+\alpha)M(z,x)\\
 &=M(\bar{z},\bar{x})^*JM(z,x).
\end{aligned}
\]
\end{pf}
\begin{Lemma}
\label{lem:projectively-real-factorization}

Let \(M\in C_h^\omega(\mathbb T,GL(2,\mathbb C))\)
satisfy \(M(x)^*JM(x)=J,\) and \(\deg\det M=0.\)
Then, after possibly decreasing $h$, there exist
\[
\varphi\in C^\omega(\mathbb T,\mathbb R),
\qquad
C\in C^\omega(\mathbb T,SL(2,\mathbb R)),
\]
such that
\[
M(x)=e^{2\pi i\varphi(x)}C(x).
\]
If $M$ depends analytically on an additional parameter in an interval,
then $\varphi$ and $C$ may be chosen jointly analytic after shrinking
the interval.

\end{Lemma}

\begin{pf}
Set $d(z)=\det M(z)$.  Since $d$ is nonvanishing and periodic, its winding
number
\[
 \nu(d)=\frac{1}{2\pi i}\int_0^1\frac{d'(x)}{d(x)}\,dx
\]
is an integer.  If
$F(\eta)=\int_{\mathbb T}\log|d(x+i\eta)|\,dx$, then
\[
        F'(\eta)=-2\pi\nu(d).
\]
Thus  $\nu(d)=0$.  Hence $d$ admits a
periodic holomorphic logarithm $g$ on $|\Im z|<h_0$.  Taking determinants in $M(x)^*JM(x)=J$ gives $|d(x)|=1$ on the real
axis.  Hence the real part of any holomorphic logarithm of $d$ vanishes
there, and $g(x)$ is purely imaginary for real $x$.

Let
\[
        s=e^{g/2},\qquad
        \varphi=\frac{g}{4\pi i},\qquad
        C=s^{-1}M.
\]
Then $s=e^{2\pi i\varphi}$, $\varphi$ is real on the real axis, and
$\det C=1$.  Moreover, $C^*JC=J$ on the real axis.  Since every
$2\times2$ matrix of determinant one satisfies $C^TJC=J$, comparison gives
\[
        \bar C^{\,T}JC=C^TJC,
\]
and hence $C=\bar C$.  Thus $C(x)\in SL(2,\mathbb R)$ for real $x$.
The parameter-dependent statement follows by choosing the logarithm at one
base point and shrinking the energy interval if necessary.
\end{pf}

\begin{Lemma}
\label{lem:center-det-zero}
Let $(\alpha,A)$ be an analytic complex symplectic $PH_2$ transfer
cocycle of dimension $2d$, and let $M$ be the cocycle induced on its
two-dimensional center in an analytic frame. Assume that $\det A$ is
constant of modulus one. If
\[
\omega^{d-1}(\alpha,A)=0,
\]
then, for all sufficiently small $|\varepsilon|$,
\[
\int_{\mathbb T}
\log|\det M(x+i\varepsilon)|\,dx=0.
\]
In particular,
\[
\deg\det M=0.
\]
\end{Lemma}

\begin{pf}
Let
\[
L_1(\varepsilon)\geq\cdots\geq L_{2d}(\varepsilon)
\]
be the Lyapunov exponents of
$(\alpha,A(\cdot+i\varepsilon))$, and write
\[
L^{d-1}(\varepsilon)
=
\sum_{j=1}^{d-1}L_j(\varepsilon).
\]
Since $(\alpha,A)$ is $(d-1)$-dominated,
$L^{d-1}(\varepsilon)$ is affine near $\varepsilon=0$. The assumption
$\omega^{d-1}(\alpha,A)=0$ therefore implies
\[
L^{d-1}(\varepsilon)=L^{d-1}(0)
\]
for all sufficiently small $|\varepsilon|$.

The complex symplectic symmetry gives
\[
L_j(\varepsilon)
=
-L_{2d+1-j}(-\varepsilon),
\]
while the assumption on $\det A$ gives
\[
\sum_{j=1}^{2d}L_j(\varepsilon)=0.
\]
The two Lyapunov exponents of the center cocycle
$M(\cdot+i\varepsilon)$ are $L_d(\varepsilon)$ and
$L_{d+1}(\varepsilon)$. Hence
\[
\begin{aligned}
\int_{\mathbb T}
\log|\det M(x+i\varepsilon)|\,dx
&=
L_d(\varepsilon)+L_{d+1}(\varepsilon)\\
&=
L^{d-1}(-\varepsilon)-L^{d-1}(\varepsilon)
=0.
\end{aligned}
\]
Finally, for any nonvanishing periodic analytic function $f$,
\[
\frac{d}{d\varepsilon}
\int_{\mathbb T}\log|f(x+i\varepsilon)|\,dx
=
-2\pi\deg f.
\]
Applying this to $f=\det M$ gives $\deg\det M=0$.
\end{pf}

\subsection{Signature and nondegeneracy of the center kernel}

Let $S^*=-S$, let $H$ be a $2n\times2$ matrix of rank two, and assume that
\[
        \Omega=H^*SH
\]
is nondegenerate and that $i\Omega$ has signature $(1,1)$.  Set
\[
        L=H\Omega^{-1}H^*.
\]
\begin{Lemma}
\label{lem:center-kernel-signature}
The Hermitian matrix $iL$ has one positive and one negative eigenvalue and
all its remaining eigenvalues are zero.  If $\lambda$ is either nonzero
eigenvalue of $iL$, then
\[
        |\lambda|\ge \|S\|^{-1}.
\]%
\end{Lemma}
\begin{pf}
Put $P=H^*H>0$.  By simultaneous congruence of the positive form $P$
and the Hermitian form $i\Omega$, there is $U\in GL(2,\mathbb C)$ such
that
\[
 U^*PU=I,
 \qquad
 U^*(i\Omega)U=\begin{pmatrix}g&0\\0&h\end{pmatrix},
\]
where $g>0>h$.  The columns of $W=HU$ are orthonormal, and a direct
calculation gives
\[
 iL=-W\begin{pmatrix}g^{-1}&0\\0&h^{-1}\end{pmatrix}W^*.
\]
Thus $iL$ has one positive and one negative nonzero eigenvalue.  If $w$
is either column of $W$, then
\[
        |g|\ \text{or}\ |h|=|iw^*Sw|\le \|S\|,
\]
which proves the bound.  
\end{pf}
In the application of Section~9, take $H=R_n$ and
$\Omega=R_n^*S_nR_n$.  The full solution matrix $\Phi_n$ is invertible and
$ i\Phi_n^*S_n\Phi_n$ has the same inertia as $iS_n$, namely $(l(n),l(n))$.
The stable and unstable blocks contribute $(l(n)-1,l(n)-1)$, because they are
paired nondegenerately by the symplectic form.  Hence $i\Omega$ has
signature $(1,1)$, and Lemma~\ref{lem:center-kernel-signature} gives
exactly the signature and uniform lower bound asserted in Lemma~9.2.

\subsubsection{The cofactor identity}

We give the linear-algebra calculation used in Section~9.  To avoid a
collision between the ambient matrix and its blocks, denote the ambient
matrix by $\mathcal A=(a_{ij})$.  Let
\[
 T_1=\begin{pmatrix}
 t_l&\cdots&t_1\\
   &\ddots&\vdots\\
   &&t_l
 \end{pmatrix},
 \qquad
 T=\begin{pmatrix}0&-T_1^*\\T_1&0\end{pmatrix},
\]
and suppose, with respect to the block decomposition
$l_1+(2l-2l_1)+l_1$, that
\begin{equation}
 \mathcal A^*T\mathcal A=
 \begin{pmatrix}
  0&0&Q_1\\
  0&Q_2&0\\
  -Q_1^*&0&0
 \end{pmatrix}.                                      \label{eq:cofactor-block}
\end{equation}
Let $\alpha_1,\alpha_2,\alpha_3$ be the three corresponding blocks of the
$(l+1)$-st row of $\mathcal A$, and let $\mathcal A_{ij}$ denote the
$(i,j)$-cofactor of $\mathcal A$.

\begin{Lemma}[\cite{gj}]
\label{gjygjy111}
Under the assumptions above,
\[
 \begin{pmatrix}
  \mathcal A_{1,1}\\ \mathcal A_{1,2}\\ \vdots\\
  \mathcal A_{1,2l}
 \end{pmatrix}
 =t_l\det\mathcal A
 \begin{pmatrix}
  -(Q_1^*)^{-1}\alpha_3^*\\
  Q_2^{-1}\alpha_2^*\\
  Q_1^{-1}\alpha_1^*
 \end{pmatrix}.
\]
\end{Lemma}

\subsubsection{Nonconstancy of the mean center rotation}
\label{app:mean-rotation}

We prove the only consequence of the rotation theory needed in the
(super)critical argument.  Let $I$ be an interval as in
the long-range Kotani theorem.  Write the limiting center dynamics as
\[
        M(z,\theta)=e^{2\pi i\varphi(z,\theta)}C(z,\theta),
\]
and suppose that, after shrinking $I$ if necessary,
\begin{equation}
 B(z,\theta+\alpha)^{-1}C(z,\theta)B(z,\theta)
   =R_{\psi(z,\theta)}.                              \label{eq:center-rotation}
\end{equation}
Here all quantities are holomorphic in $z$ near $I$, and $B,C,\psi$ are
real on the real axis.  Set
\[
        \widehat\psi(z)=\int_{\mathbb T}\psi(z,\theta)\,d\theta.
\]

\begin{Proposition}
\label{prop:mean-rotation-nonconstant}
The function $E\mapsto\widehat\psi(E)$ is not constant on any nonempty
subinterval of $I$.  More precisely, after fixing the orientation of the
rotation coordinate, for almost every $E\in I'$ for some $I'\subset I$ \footnote{$I'$ Is the interval such that $m(E,\theta)$ is analytic.},
\begin{equation}
 \left|\widehat\psi'(E)\right|
 \geq \frac{1}{8\pi}\int_{\mathbb T}
   \frac{|u_E(\theta,0)|^2}{\Im m(E,\theta)}\,d\theta>0,             \label{eq:mean-rotation-derivative}
\end{equation}
where
\[
        u_E(\theta,0)=m(E,\theta)U(E,\theta,0)+V(E,\theta,0).
\]
\end{Proposition}

\begin{pf}
Let $Q$ be a constant matrix satisfying
\[
 Q^{-1}R_tQ=\operatorname{diag}(e^{2\pi it},e^{-2\pi it}).
\]
By \eqref{eq:center-rotation}, the two invariant line multipliers of the
complexified center dynamics are
\[
 e^{2\pi i(\varphi+\psi)}
 \quad\text{and}\quad
 e^{2\pi i(\varphi-\psi)}.
\]
By the construction of the limiting center, their Lyapunov exponents are
$-\gamma_1(\bar z)$ and $\gamma_1(z)$, respectively.  Interchanging the two
lines if necessary therefore gives
\begin{equation}
 4\pi\,\Im\widehat\psi(z)
   =\gamma_1(z)+\gamma_1(\bar z).                    \label{eq:psi-gamma}
\end{equation}
The choice of sign is immaterial below.

Let $z=E+i\eta$, $\eta>0$.  In the notation of Section~10, the corrected
long-range derivative identity is
\[
 \frac{d}{d\eta}\gamma_1(E-i\eta)
 =\int_{\mathbb T}\Im
   \frac{\overline{u^-_{E-i\eta}(\theta,0)}
         u^+_{E+i\eta}(\theta,0)}
        {\overline{m_-(E-i\eta,\theta)}-m_+(E+i\eta,\theta)}\,d\theta.
\]
Together with the reflectionless boundary relation proved in Step~5 of
Theorem~10.1, this gives for almost every $E\in I'$ where $I'$ is the subinterval of $I$ such that $m(E,\theta)$ is analytic, $\Im m(E+i0,\theta)\geq c_E>0$ for any  $\theta\in\T$ \footnote{By Step 5 in Theorem 10.1, $\Im m(E+i0,\theta)>0$ for almost every $\theta$, by continuity, minimality of irrational rotation and coinvariance, we have $c_E=\min_{\theta\in\T}\Im m(E,\theta)>0$.}, the holomorphic boundary extension then gives uniform convergence of the numerator and keep $\overline{m_-(E-i\eta,\theta)}-m_+(E+i\eta,\theta)$ away from zero, uniformly on $\theta$, for small $\eta$. Hence 
\begin{equation}
 \left.\frac{d}{d\eta}\gamma_1(E-i\eta)\right|_{\eta=0^+}
 =\frac12\int_{\mathbb T}
   \frac{|u_E(\theta,0)|^2}{\Im m(E,\theta)}\,d\theta.              \label{eq:gamma-boundary-derivative}
\end{equation}
Indeed, at the boundary the denominator in that identity is
$\overline{m(E,\theta)}-m(E,\theta)=-2i\Im m(E,\theta)$.

Finally, $\widehat\psi$ is holomorphic and real on $I$, so the
Cauchy--Riemann equations give
\[
 \widehat\psi'(E)
 =\left.\partial_\eta\Im\widehat\psi(E+i\eta)\right|_{\eta=0^+}\geq \frac{1}{8\pi}\int_{\mathbb T}
   \frac{|u_E(\theta,0)|^2}{\Im m(E,\theta)}\,d\theta>0.
\]
Combining this identity with \eqref{eq:psi-gamma} and
\eqref{eq:gamma-boundary-derivative} proves
\eqref{eq:mean-rotation-derivative}, up to the choice of orientation.
The integral is finite by the Kotani estimate.  It is strictly positive:
if $u_E(\theta,0)$ vanished identically in $\theta$, covariance would force
the corresponding center solution to vanish identically.
\end{pf}

\begin{Remark}
The proof uses only the limiting center construction, the long-range Kotani
derivative identity, and the reflectionless boundary relation established in
the main text.  It does not use the general rotation--IDS correspondence proved in
\cite[Theorem~8.1]{gj}.
\end{Remark}
\subsection{Other techinical lemmas}
Assume $(\Omega,T)$ is minimal and uniquely ergodic, $f:\Omega\rightarrow \R$ is continuous
and $L_E^f$ is $PH2$. Partial hyperbolicity means that  there exist a continuous invariant decomposition
$$
\C^{2d}=E^s(\omega)\oplus E^c(\omega)\oplus E^u(\omega).
$$
Moreover, there are $C(E),\delta(E)>\delta'(E)>0$, such that for any $\omega\in\Omega$ and $n\geq 1$, we have
\begin{align}\label{eq10a}
\left\|(L_{E}^{f})_{-n}(\omega)v\right\|>C^{-1}e^{\delta n},\ \  \forall v\in E^s(\omega)\backslash\{0\},\ \ \|v\|=1,
\end{align}
\begin{align}\label{eq11a}
\left\|(L_{E}^{f})_{n}(\omega)u\right\|>C^{-1}e^{\delta n},\ \  \forall u\in E^u(\omega)\backslash\{0\},\ \ \|u\|=1.
\end{align}
\begin{align}\label{eq13a}
\left\|(L_{E}^{f})_{\pm n}(\omega)w\right\|<Ce^{\delta' n},\ \  \forall w\in E^c(\omega)\backslash\{0\},\ \ \|w\|=1.
\end{align}
\begin{align}\label{eq12a}
{\rm dim} E^c(\omega)=2.
\end{align}

For any $\begin{pmatrix}
u(d-1)\\
u(d-2)\\
\vdots\\
u(-d)
\end{pmatrix}\in \C^{2d}$, there exist
$$
\begin{pmatrix}
u^s(d-1)\\
u^s(d-2)\\
\vdots\\
u^s(-d)
\end{pmatrix}\in E^s(\omega), \ \ \begin{pmatrix}
u^c(d-1)\\
u^c(d-2)\\
\vdots\\
u^c(-d)
\end{pmatrix}\in E^c(\omega), \ \
\begin{pmatrix}
u^u(d-1)\\
u^u(d-2)\\
\vdots\\
u^u(-d)
\end{pmatrix}\in E^u(\omega)
$$
such that
$$
\begin{pmatrix}
u(d-1)\\
u(d-2)\\
\vdots\\
u(-d)
\end{pmatrix}=\begin{pmatrix}
u^s(d-1)\\
u^s(d-2)\\
\vdots\\
u^s(-d)
\end{pmatrix}+\begin{pmatrix}
u^c(d-1)\\
u^c(d-2)\\
\vdots\\
u^c(-d)
\end{pmatrix}+\begin{pmatrix}
u^u(d-1)\\
u^u(d-2)\\
\vdots\\
u^u(-d)
\end{pmatrix}.
$$

\begin{Lemma}\label{lee2}
 For $n$ sufficiently large, if $\left\|\begin{pmatrix}
u(\pm n+d-1)\\
u(\pm n+d-2)\\
\vdots\\
u(\pm n-d)
\end{pmatrix}\right\|\leq Ce^{\delta'|n|}$ ,  then  we have
$$
\left\|\begin{pmatrix}
u^s(d-1)\\
u^s(d-2)\\
\vdots\\
u^s(-d)
\end{pmatrix}\right\|,\ \ \left\|\begin{pmatrix}
u^u(d-1)\\
u^u(d-2)\\
\vdots\\
u^u(-d)
\end{pmatrix}\right\|\leq Ce^{-(\delta-\delta') |n|}.
$$
\end{Lemma}
\begin{pf}
Let $\Pi_s(\omega)$ and $\Pi_u(\omega)$ be the projections (associated with the invariant direct sum) onto
$E^s(\omega)$ and $E^u(\omega)$. By continuity of the splitting and compactness
of $\Omega$,
$$
\sup_{\omega\in\Omega}
\max\{\|\Pi_s(\omega)\|,\|\Pi_u(\omega)\|\}\le C.
$$
By invariance and (12.6),
$$
\begin{aligned}
\left\|
\begin{pmatrix}
u^s(d-1)\\
u^s(d-2)\\
\vdots\\
u^s(-d)
\end{pmatrix}
\right\|
&\le Ce^{-\delta n}
\left\|
\Pi_s(T^{-n}\omega)
\begin{pmatrix}
u(-n+d-1)\\
u(-n+d-2)\\
\vdots\\
u(-n-d)
\end{pmatrix}
\right\|\\
&\le Ce^{-(\delta-\delta')n}.
\end{aligned}
$$
Similarly, by invariance and (12.7),
$$
\begin{aligned}
\left\|
\begin{pmatrix}
u^u(d-1)\\
u^u(d-2)\\
\vdots\\
u^u(-d)
\end{pmatrix}
\right\|
&\le Ce^{-\delta n}
\left\|
\Pi_u(T^n\omega)
\begin{pmatrix}
u(n+d-1)\\
u(n+d-2)\\
\vdots\\
u(n-d)
\end{pmatrix}
\right\|\\
&\le Ce^{-(\delta-\delta')n}.
\end{aligned}
$$
\end{pf}

\begin{Proposition}\label{newwe}
There exist two linearly independent $u_E(x), v_E(x)\in E_c(x)$  depending analytically on $E$ and  $x$ on $\C_{\delta}\cap \R\times \T$, such that 
\begin{equation}\label{gjy1n}
\begin{pmatrix}
u^*_{E}(x)\\ v^*_{E}(x)
\end{pmatrix}S\begin{pmatrix}
u_{E}(x)& v_{E}(x)\end{pmatrix}=\begin{pmatrix}
0& 1\\
-1&0\end{pmatrix}=:J.
\end{equation}
\end{Proposition}
\begin{pf}
This is a two-dimensional, parameter-dependent variant of
\cite[Lemma~6.1]{gj}; we include the proof for completeness. Note that it follows from \eqref{eqnew1} and Proposition \ref{trivial}, there are global holomorphic frames $\{f_E^j(x)\}_{j=1}^{d-1}\in E_s(x)$,  $\{\tilde{u}^j_E(x)\}_{j=1}^{2}\in E_c(x)$  and $\{g_E^j(x)\}_{j=1}^{d-1}\in E_u(x)$ respectively.  

Let
 $$
 \tilde{O}_E(\theta)=\begin{pmatrix}\tilde{u}_E^1(\theta)&\tilde{u}_E^{2}(\theta)\end{pmatrix},
 $$
 $$
 P_E(x)=\begin{pmatrix}f_E^1(x)&\cdots&f_E^{d-1}(x)&\tilde{u}_E^1(x)&\tilde{u}_E^{2}(x)&g_E^1(x)&\cdots&g^{d-1}_E(x)\end{pmatrix},
 $$
 \begin{equation}\label{eiga}
\tilde{\Omega}_E(x)=\tilde{O}_E(x)^*S\tilde{O}_E(x), \ \ \Lambda_E(x)=P_E(x)^*SP_E(x)=\begin{pmatrix}&&A_E(x)\\
&\tilde{\Omega}_E(x)&\\ -A_E(x)^*&&\end{pmatrix} \footnote{$\Lambda_E(\theta)$ has this form because of the symplectic orthogonality.}
\end{equation}
for some $A_E(x)$. 
Since $iS$ has $d$ positive and $d$ negative eigenvalues, by \eqref{eiga},  we have $i\tilde{\Omega}_E(x)$ has $1$-positive and $1$-negative eigenvalues for all $(E,x)\in(\C_\delta\cap \R)\times \T$. Let $-i\mu_E^1(x)$, $i\mu_E^2(x)$ be the eigenvalues of $\tilde{\Omega}_E(x)$ ($\mu^1_E(x), \mu_E^2(x)>0$).  They depend analytically on $E$ and  $x$, and thus are bounded away from $0$. This further implies that ${\rm dim}\Ker{(\tilde{\Omega}_E(x)-(-1)^{j}i\mu_E^j(x))}=1$ for $j=1,2$ and all $(E,x)\in(\C_\delta\cap \R)\times \T$. Hence, by Proposition \ref{trivial}, there exists a $2\times 2$ matrix $U_E(x)$,  depending analytically on $E$ and  $x$, such that
$$
\tilde{\Omega}_E(x)=U_E(x)^*\begin{pmatrix} -i\mu_E^1(x)&0\\0& i\mu_E^2(x)\end{pmatrix}U_E(x),\ \ U_E(x)^*U_E(x)=I.
$$

Finally, we define
$$
\begin{pmatrix} u_E(x)& v_E(x)\end{pmatrix}=\frac{1}{\sqrt{2}}\widetilde{O}_E(x)U_E(x)^*\begin{pmatrix} \sqrt{\frac{1}{\mu_E^1(x}}&0\\0&\sqrt{\frac{1}{\mu_E^2(x)}}\end{pmatrix}\begin{pmatrix}-i&1\\ i&1\end{pmatrix},
$$
then we have that $u_E(x)$ and $v_E(x)$ depend analytically on $E$ and $x$. Moreover, we have 
$$
\begin{pmatrix} u_E^*(x)\\ v_E^*(x)\end{pmatrix}S\begin{pmatrix}u_E(x)& v_E(x)\end{pmatrix}=\begin{pmatrix}
0&1\\ -1&0\end{pmatrix}.
$$
\end{pf}
\begin{Remark}
The proof of the above proposition is simpler than that of Lemma 6.1 in \cite{gj}  because we do not need a fixed strip of analyticity without any loss as $d\rightarrow\infty$. Here the aim is to prove  the symplectic basis is analytic in both $E$ and $x$.
\end{Remark}
\begin{Lemma}\label{dominated matrix}
Assume $A_z(x)=J+P(x)z^2$ where $P\in C^\omega(\T,M(2,\C))$ with $P^*(x)=-P(x)$, there is $\delta(P)>0$ such that if $|z|<\delta$, there exist $B_z(x)$ depending analytically on both $z$ and $x$ such thst
$$
B^*_{\bar{z}}(x)A_z(x)B_z(x)=J.
$$
Moreover, $\sup_{x\in\T}\left|\frac{\partial B_z(x)}{\partial z}\right|\leq C|P(x)|_0|z|$ where $C$ is an absolute constant.
\end{Lemma}
\begin{pf}
Choose $h>0$ such that $P$ extends holomorphically to
$|\Im x|<h$, and choose $\delta>0$ sufficiently small that
$\delta^2|P|_h<1/2$. Define
$$
B_z(x)=\bigl(I-JP(x)z^2\bigr)^{-1/2},
$$
using the convergent binomial series at the identity.
This defines an invertible matrix depending holomorphically
on $z$ and $x$.

For real $x$, the identities $P(x)^*=-P(x)$,
$J^*=-J$, and $J^2=-I$ imply
$$
(JP(x))^*J=J(JP(x)).
$$
Since the binomial series has real coefficients, it follows that
$$
B_{\bar z}(x)^*J=JB_z(x).
$$
Moreover,
$$
A_z(x)=J\bigl(I-JP(x)z^2\bigr).
$$
Consequently,
$$
\begin{aligned}
B_{\bar z}(x)^*A_z(x)B_z(x)
&=JB_z(x)\bigl(I-JP(x)z^2\bigr)B_z(x)\\
&=J.
\end{aligned}
$$
Finally, differentiating the power series gives
$$
\frac{\partial B_z(x)}{\partial z}
=zJP(x)\bigl(I-JP(x)z^2\bigr)^{-3/2}.
$$
Hence
$$
\sup_{x\in\mathbb T}
\left\|\frac{\partial B_z(x)}{\partial z}\right\|
\le C|P|_0|z|,
$$
where $C$ is an absolute constant.
\end{pf}

\begin{Lemma}[\cite{gjyz}]\label{basic}
Consider the following $2d$ order difference operator,
$$
(Lu)(n)=\sum\limits_{k=-d}^da_ku(n+k)+V(n)u(n).
$$
If the eigenequation $Lu=Eu$ has $2d$ linearly independent solutions $\{\phi_i\}_{i=1}^{2d}$ satisfying
$$
\phi_i\in \ell^2(\Z^+),\ \ (i=1,\cdots,m),
$$
$$
\phi_i\in\ell^2(\Z^-),\ \  (i=m+1,\cdots,2d),
$$
and $L-EI$ is invertible. Then,
$$
\langle\delta_p,(L-EI)^{-1}\delta_q\rangle=\begin{cases}
\frac{\sum\limits_{i=1}^m\phi_i(p)\Phi_{1,i}(q)}{a_d\det{\Phi(q)}} &\text{$p\geq q+1$}\\
-\frac{\sum\limits_{i=m+1}^{2d}\phi_i(p)\Phi_{1,i}(q)}{a_d\det{\Phi(q)}} &\text{$p\leq q$}
\end{cases},
$$
where
$$
\Phi(q)=\begin{pmatrix}\phi_1(q+d)&\phi_2(q+d)&\cdots&\phi_{2d}(q+d)\\ \phi_1(q+d-1)&\phi_2(q+d-1)&\cdots&\phi_{2d}(q+d-1)\\ \vdots&\vdots& &\vdots\\
\phi_1(q-d+1)&\phi_2(q-d+1)&\cdots&\phi_{2d}(q-d+1)\end{pmatrix}
$$
and $\Phi_{i,j}(q)$ is the $(i,j)$-th cofactor of $\Phi(q)$.
\end{Lemma}
\begin{pf}
This lemma is essentially proved in \cite{gjyz}, however the conditions added there is  not enough.  The proof remains unchanged. 
\end{pf}
\begin{Remark}\label{gz20}
Note that $\sum_{i=1}^{2d}\phi_i(p)\Phi_{1,i}(q)=0$ for any $q-d+1\leq p\leq q+d-1$, thus we have for any $q-d+1\leq p\leq q+d-1$, we have
$$
-\sum\limits_{i=m+1}^{2d}\phi_i(p)\Phi_{1,i}(q)=\sum\limits_{i=1}^{m}\phi_i(p)\Phi_{1,i}(q).
$$
\end{Remark}

Let \begin{equation*}
C_n=\begin{pmatrix}
a_{l(n)}&\cdots&a_1\\
0&\ddots&\vdots\\
0&0&a_{l(n)}
\end{pmatrix},\ \ S_n=\begin{pmatrix}0&-C_n^*\\
C_n&0\end{pmatrix},
\end{equation*}
and let $D_n, B_n$ be sequences of $2l(n)\times m$ matrices. We denote
$$
D_n=\begin{pmatrix}
D_n(l(n))\\ \vdots \\ D_n(-l(n)+1)
\end{pmatrix}, \tilde{D}_n=\begin{pmatrix}
D_{n+1}(l(n))\\ \vdots \\ D_{n+1}(-l(n)+1)
\end{pmatrix}, 
$$
$$
B_n=\begin{pmatrix}
B_n(l(n))\\ \vdots\\ B_n(-l(n)+1)
\end{pmatrix}, \tilde{B}_n=\begin{pmatrix}
B_{n+1}(l(n))\\ \vdots\\ B_{n+1}(-l(n)+1)
\end{pmatrix}.
$$
\begin{Lemma}\label{pkule}
We have that 
\begin{align*}
&\|D_{n+1}^*S_{n+1} B_{n+1}-D_n^*S_n B_n\|\\
\leq &|a_{n+1}|\|D_{n+1}\|\|B_{n+1}\|+\left(\|\tilde{D}_n\|+\|D_n\|+\|\tilde{B}_n\|+\|B_n\|\right)
\|S_n\|\left(\|\tilde{D}_n-D_n\|+\|\tilde{B}_n-B_n\|\right)
\end{align*}
\end{Lemma}
\begin{pf}
By a direct calculation,
\begin{align*}
&\|D_{n+1}^*S_{n+1} B_{n+1}-D_n^*S_n B_n\|\\
\leq &|a_{n+1}|\|D_{n+1}\|\|B_{n+1}\|+\|D_n-\tilde{D}_n\|\|S_n\|\|B_n\|+\|\tilde{D}_n\|\|S_n\|\|B_n-\tilde{B}_n\|
\end{align*}
The result follows direcly.  
\end{pf}
\section{AI-use disclosure} The main results and proofs of this work were obtained without the use of generative AI: the original even trigonometric version appeared in 2023, and the present general-analytic version was largely completed in 2025. ChatGPT 5.6 Sol, and later Astra were used during the final 2026 revision to stress-test arguments, identify missing edge cases or points requiring clarification, and improve exposition. All resulting changes were independently checked by the authors, who take full responsibility for the contents of the paper.
\section*{Acknowledgements}
L. Ge was partially supported by NSFC grant (12371185) and the Fundamental Research Funds for the Central Universities (the start-up fund), Peking University. J. You   was partially supported by NSFC grant (12531006,12526201) and
Nankai Zhide Foundation. SJ's work was supported by NSF DMS-2052899,
DMS-2155211, and Simons 681675. She is also grateful to School of
Mathematics at Georgia Institute of Technology where a part of this
work was done.


\begin{thebibliography}{99}

\bibitem{AA} F. Argentieri and A. Avila. In preparation.

\bibitem{oaag} S. Aubry and G. Andr\'e. Analyticity breaking and Anderson localisation in incommensurate lattices. {\it Ann. Israel Phys. Soc.} {\bf 3} (1980), 133–164.

\bibitem{a1} A. Avila. The absolutely continuous spectrum of the almost Mathieu operator. preprint, \url{https://arxiv.org/abs/0810.2965}.

\bibitem{avilajams} A. Avila. Density of positive Lyapunov exponents for SL(2,R)-cocycles. {\it J. Amer. Math. Soc.} {\bf 24(4)} (2011), 999-1014.

\bibitem{avila0} A. Avila. Global theory of one-frequency Schr\"{o}dinger operators. {\it Acta Math.} {\bf 215} (2015), 1-54.

\bibitem{arc1} A. Avila. Almost reducibility and absolute continuity. preprint. http://w3.impa.br/~avila/ (2704,2711).

\bibitem{arc2} A. Avila. KAM, Lyapunov exponents, and the Spectral Dichotomy for typical one-frequency Schrodinger operators. arXiv:2307.11071.

\bibitem{afk} A. Avila, B. Fayad and R. Krikorian. A KAM scheme for $SL(2,\R)$ cocycles with Liouvillean frequencies. {\it Geom. Funct. Anal.} {\bf 21} (2011), 1001-1019.

\bibitem{solving} A. Avila and S. Jitomirskaya. Solving the ten martini problem. {\it Lecture Notes in Phys.} {\bf 690} (2006), 5–16.

\bibitem{aj} A. Avila and S. Jitomirskaya. The Ten Martini Problem. {\it Ann. of Math.} {\bf 170} (2009), 303-342.

\bibitem{aj1} A. Avila and S. Jitomirskaya. Almost localization and almost reducibility. {\it J. Eur. Math. Soc} {\bf 12} (2010), 93-131.

\bibitem{ajs} A. Avila, S. Jitomirskaya and C. Sadel. Complex one-frequency cocycles. {\it J. Eur. Math. Soc.} {\bf 16} (2014), 1915-1935.

\bibitem{ak} A. Avila and R. Krikorian. Reducibility or non-uniform hyperbolicity for quasiperiodic Schr\"odinger cocycles. {\it Ann. of Math.} {\bf 164} (2006), 911-940.

\bibitem{ak1} A. Avila and R. Krikorian. Monotonic cocycles. {\it Invent. Math.} {\bf 202(1)} (2015), 271-331.

\bibitem{alsz} A. Avila, Y. Last, M. Shamis and Q. Zhou. On the abominable properties of the almost Mathieu operator with well-approximated frequencies. {\it Duke Math. J.} {\bf 173(4)} (2024), 603--672.

\bibitem{ayz} A. Avila, J. You and Q. Zhou. Sharp phase transitions for the almost Mathieu operator. {\it Duke Math. J.} {\bf 166(14)} (2017), 2697-2718.

\bibitem{dryAYZ} A. Avila, J. You and Q. Zhou. Dry Ten martini problem for noncritical almost Mathieu operator. arXiv:2306.16254.

\bibitem{aos} J.E. Avron, D. Osadchy and R. Seiler. A topological look at the quantum Hall effect. \textit{Physics today.} (2003), 38-42.

\bibitem{as} J. Avron and B. Simon. Almost periodic operators, II. The density of states. {\it Duke Math J.} {\bf 50} (1983), 369–391.

\bibitem{bs} J. Bellissard and B. Simon. Cantor spectrum for the almost Mathieu equation. {\it J. Funct. Anal.} {\bf 48} (1982), 408-419.

\bibitem{bdv} C. Bonatti, L. Diaz and M. Viana. Dynamics Beyond Uniform Hyperbolicity. A Global Geometric and Probabilistic Perspective (Encyclopaedia of Mathematical Sciences, 102. Mathematical Physics, III). Springer, Berlin, 2005.

\bibitem{bbook} J. Bourgain. Green's function estimates for lattice Schr\"odinger operators and applications. {\it Annals of Mathematics Studies. Princeton University Press, Princeton, NJ.} {\bf 158} (2005).

\bibitem{bourgainjanal} J. Bourgain. Positivity and continuity of the Lyapunov exponent for shifts on ${\mathbb T}^d$ with arbitrary frequency vector and real analytic potential. {\it J. Anal. Math.} {\bf 96} (2005), 313–355.

\bibitem{bj} J. Bourgain and S. Jitomirskaya. Continuity of the Lyapunov exponent for quasiperiodic operators with analytic potential. {\it J. Stat. Phys.} {\bf 108} (2002), 1203-1218.

\bibitem{cey} M. Choi, G. Elliott and N. Yui. Gauss polynomials and the rotation algebra. {\it Invent. Math.} {\bf 99} (1990), 225-246.

\bibitem{dcj} C. Concini and R. Johnson. The algebraic-geometric AKNS potentials. {\it Ergodic Theory Dynam. Systems} {\bf 7(1)} (1987), 1-24.

\bibitem{cs} W. Craig and B. Simon. Log H\"older continuity of the integrated density of states for stochastic Jacobi matrices. {\it Comm. Math. Phys.} {\bf 90(2)} (1983), 207-218.

\bibitem{damanik1} D. Damanik. Lyapunov exponents and spectral analysis of ergodic Schr\"odinger operators: a survey of Kotani theory and its applications. {\it Proc. Sympos. Pure Math.} {\bf 76} Part 2 (2007), 539–563.

\bibitem{DF1} D. Damanik and J. Fillman. One-Dimensional Ergodic Schr\"odinger Operators. Graduate Studies in Mathematics, Vol.~221. American Mathematical Society, Providence, RI, 2022.

\bibitem{DF2} D. Damanik and J. Fillman. One-Dimensional Ergodic Schr\"odinger Operators: II. Specific Classes. Graduate Studies in Mathematics, Vol.~222. American Mathematical Society, Providence, RI, 2022.

\bibitem{ds} E. Dinaburg and Ya. Sinai. The one dimensional Schr\"odinger equation with a quasi-periodic potential. {\it Funct. Anal. Appl.} {\bf 9} (1975), 279-289.

\bibitem{Eli92} L. H. Eliasson. Floquet solutions for the 1-dimensional quasi-periodic Schr\"{o}dinger equation. {\it Commun. Math. Phys.} {\bf 146} (1992), 447-482.

\bibitem{Eli97} L. Eliasson. Discrete one-dimensional quasi-periodic Schr\"{o}dinger operators with pure point spectrum. {\it Acta Math.} {\bf 179} (1997), 153-196.

\bibitem{stein} F. Forstneric. Stein Manifolds and holomorphic mappings. {\it Ergeb. Math. Grenzgeb. (3)}, {\bf 56} Springer, Heidelberg, 2011, xii+489 pp.

\bibitem{gps} S. Ganeshan, J. H. Pixley and S. Das Sarma. Nearest neighbor tight binding models with an exact mobility edge in one dimension. {\it Physical Review Letters} {\bf 114(14)} (2015), 146601.

\bibitem{ge} L. Ge. On the almost reducibility conjecture. {\it Geom. Funct. Anal.} {\bf 34} (2024), 23-59.

\bibitem{gx} L. Ge and D. Xu. In preparation.

\bibitem{gj} L. Ge and S. Jitomirskaya. Intrinsic symplectic structure and sharp arithmetic universality. arXiv:2407.08866.

\bibitem{gjy2} L. Ge, S. Jitomirskaya and J. You. In preparation.

\bibitem{gjy}
L. Ge, S. Jitomirskaya and J. You. Kotani theory, Puig's argument, and stability of the Ten Martini Problem.
arXiv:2308.09321.

\bibitem{gjyz} L. Ge, S. Jitomirskaya, J. You and Q. Zhou. Multiplicative Jensen's formula and quantitative global theory of one-frequency Schr\"odinger operators. Forum of Mathematics, Pi {\bf 14} (2026), Paper No. e1.

\bibitem{gjz} L. Ge, S. Jitomirskaya and X. Zhao. Stability of the Non-Critical Spectral Properties I: Arithmetic Absolute Continuity of the Integrated Density of States. {\it Comm. Math. Phys.} {\bf 402(1)} (2023), 213–229.

\bibitem{gwx} L. Ge, Y. Wang and J. Xu. The dry ten martini problem for $C^2$ cos-type quasiperiodic Schr\"odinger operators. arXiv:2503.06918.

\bibitem{gs1} M. Goldstein and W. Schlag. On resonances and the formation of gaps in the spectrum of quasi-periodic Schr\"odinger equations. Ann. of Math. {\bf 173} (2011), 337-475.

\bibitem{gs2} M. Goldstein and W. Schlag. Fine properties of the integrated density of states and a quantitative separation property of the Dirichlet eigenvalues. {\it Geom. Funct. Anal.} {\bf 18} (2008), 755-869.

\bibitem{gjls} A.Y. Gordon, S. Jitomirskaya, Y. Last and B. Simon. Duality and singular continuous spectrum in the almost Mathieu equation. {\it Acta Math.} {\bf 178} (1997), 169-183.

\bibitem{gk} A. Gorodetski and V. Kleptsyn. Generalized Aubry-André formula and continuity of the intersection spectrum of the Almost Mathieu operator. arXiv:2604.23852.

\bibitem{fsw} J. Fr\"{o}hlich, T. Spencer and P. Wittwer. Localization for a class of one dimensional quasi-periodic Schr\"{o}dinger operators. {\it Commun. Math. Phys.} {\bf 132} (1990), 5-25.

\bibitem{H} M. Herman. Une m\'ethode pour minorer les exposants de Lyapounov et quelques exemples montrant le caract\`ere local d'un th\'eor\`eme d'Arnol'd et de Moser sur le tore de dimension $2$. {\it Comment. Math. Helv.} {\bf 58(3)} (1983), 453-502.

\bibitem{h} B.I. Halperin. Quantized Hall conductance, current-carrying edge states, and the existence of extended states in a two-dimensional disordered potential. Phys. Rev. B {\bf 25} (1982), 2185.

\bibitem{haro} A. Haro and J. Puig. A Thouless formula and Aubry duality for long-range Schr\"odinger skew-products. {\it Nonlinearity} {\bf 26(5)} (2013), 1163-1187.

\bibitem{hm} R. Han and C. Marx. Large coupling asymptotics for the Lyapunov exponent of quasiperiodic Schr\"odinger operators with analytic potentials. {\it Ann. Henri Poincar\'e} {\bf 19} (2018), 249-265.

\bibitem{harper} P.G. Harper. Single band motion of conduction electrons in a uniform magnetic field. {\it Pro. Phys. Soc.} {\bf 68(10)} (1955), 874-878.

\bibitem{hs} B. Helffer and J. Sj\"ostrand. Semiclassical analysis for Harper's equation. III. Cantor structure of the spectrum. {\it Mem. Soc. Math. France} {\bf 39} (1989), 1-124.
\bibitem{hs1}
R. Han and W. Schlag. Nonperturbative localization on the strip and Avila's almost reducibility conjecture. Forum Math. Pi 14 (2026), Paper No. e24.
\bibitem{hk} H. Hiramoto and M. Kohmoto. Scaling analysis of quasiperiodic systems: Generalized harper model. {\it Physical Review B} {\bf 40(12)} (1989), 8225.

\bibitem{hof} D.R. Hofstadter. Energy levels and wave functions of Bloch electrons in rational and irrational magnetic fields. {\it Physical Review B} {\bf 14} (1976), 2239-2249.

\bibitem{hy} X. Hou and J. You. Almost reducibility and non-perturbative reducibility of quasiperiodic linear systems. {\it Invent. Math.} {\bf 190} (2012), 209-260.

\bibitem{hhsy} Jiawei He, Xuanji Hou, Yuan Shan, and Jiangong You. Explicit construction of quasi-periodic analytic Schr\"odinger operators with Cantor spectrum. Math. Ann., 391(1):179--225, 2025.

\bibitem{hsy} Xuanji Hou, Yuan Shan, and Jiangong You. Construction of quasiperiodic Schr\"odinger operators with Cantor spectrum. Ann. Henri Poincar\'e, 20(11):3563--3601, 2019.

\bibitem{hz} Xuanji Hou and Li Zhang. Explicit construction of {Gevrey} {Quasi}-{Periodic} {Discrete} {Schr\"odinger} operators with cantor spectrum. Discrete and Continuous Dynamical Systems-Series B, 33:364--397, 2026.

\bibitem{j} S. Jitomirskaya. Metal-Insulator transition for the almost Mathieu operator. {\it Ann. of Math.} {\bf 150} (1999), 1159-1175.

\bibitem{lxz}
X. Li, D. Xu and Q. Zhou. Monotonicity, global symplectification and the stability of Dry Ten Martini Problem. arXiv:2601.02222.

\bibitem{jitcongr} S. Jitomirskaya. One-dimensional quasiperiodic operators: global theory, duality, and sharp analysis of small denominators. {\it Proceedings of ICM} 2022.

\bibitem{jks} S. Jitomirskaya, D. A. Koslover and M. S. Schulteis. Localization for a Family of One-dimensional Quasiperiodic Operators of Magnetic Origin. {\it Ann. Henri Poincar\'e} {\bf 6(1)} (2005), 103--124. DOI: \texttt{10.1007/s00023-005-0200-5}.

\bibitem{jliu1} S. Jitomirskaya and W. Liu. Universal hierarchical structure of quasiperiodic eigenfunctions. {\it Ann. of Math.} {\bf 187(3)} (2018), 721-776.

\bibitem{jliu2} S. Jitomirskaya and W. Liu. Universal reflective-hierarchical structure of quasiperiodic eigenfunctions and sharp spectral transition in phase. To appear in J. Eur. Math. Soc.

\bibitem{johonson and moser} R. Johnson and J. Moser. The rotation number for almost periodic potentials. {\it Commun. Math. Phys.} {\bf 84} (1982), 403-438.

\bibitem{kac} M. Kac. Public commun, at 1981 AMS Annual Meeting.

\bibitem{kot} S. Kotani. Generalized Floquet theory for stationary Schr\"odinger operators in one dimension. {\it Chaos Solitons Fractals} {\bf 8(11)} (1997), 1817-1854.

\bibitem{ks} S. Kotani and B. Simon. Stochastic Schr\"odinger Operators and Jacobi Matrices on the strip. {\it Comm. Math. Phys} {\bf 119} (1988), 403-429.

\bibitem{last} Y. Last. Zero measure spectrum for the almost Mathieu operator. {\it Comm. Math. Phys.} {\bf 164} (1994), 421-432.

\bibitem{LYZZ} M. Leguil, J. You, Z. Zhao and Q. Zhou. Asymptotics of spectral gaps of quasi-periodic Schr\"odinger operators. \url{https://arxiv.org/pdf/1712.04700.pdf}.

\bibitem{liu1} W. Liu. Distributions of resonances of supercritical quasi-periodic operators. Int. Math. Res. Not. IMRN {\bf 2024} (2024), no. 1, 197--233.

\bibitem{liuresonant} W. Liu. Small denominators and large numerators of quasiperiodic Schr\"odinger operators. Peking Math. J. {\bf 8} (2025), no. 3, 503--532.

\bibitem{Liuresonant1} W. Liu. A doubled {Gordon} threshold for palindromic quasiperiodic {Schr\"odinger} operators. arXiv preprint arXiv:2607.24188, 2026.

\bibitem{ly} W. Liu and X. Yuan. Spectral gaps of almost Mathieu operators in the exponential regime. {\it J. Fractal Geom.} {\bf 2} (2015), no. 1, 1–51.

\bibitem{ntw} Q. Niu, D.J. Thouless and Y.S. Wu. Quantized Hall conductance as a topological invariant. {\it Phys. Rev. B} {\bf 31} (1985), 3372.

\bibitem{nichen} X. Ni, K. Chen, M. Weiner, D. J. Apigo, C. Prodan, A. Alù, E. Prodan, A. B. Khanikaev. Observation of Hofstadter butterfly and topological edge states in reconfigurable quasi-periodic acoustic crystals. {\it Communications Physics} {\bf 2} 55 (2019).

\bibitem{oliviera} F. V. Oliveira and S. L. Carvalho. Kotani theory for ergodic block Jacobi operators. Oper. Matrices {\bf 16} (2022), no. 3, 827--857.

\bibitem{Peierls} R. Peierls. Zur theorie des diamagnetismus von leitungselektronen. {\it Zeitschrift f\"ur Physik} {\bf 80(11-12)} (1933), 763-791.

\bibitem{Puig} J. Puig. Cantor spectrum for the almost Mathieu operator. {\it Comm. Math. Phys.} {\bf 244} (2004), 297-309.

\bibitem{puig} J. Puig. A nonperturbative Eliasson's reducibility theorem. {\it Nonlinearity} {\bf 19(2)} (2006), 355-376.

\bibitem{R} A. Rauh. Degeneracy of Landau levels in crystals. \textit{Phys. Status Solidi B} {\bf 65} (1974), 131-135.

\bibitem{barry} B. Simon. Almost periodic Schrödinger operators: A review. {\it Adv. Appl. Math.} {\bf 3} (1982), 463-490.

\bibitem{sim83} B. Simon. Kotani theory for one-dimensional stochastic Jacobi matrices. {\it Comm. Math. Phys.} {\bf 89(2)} (1983), 227–234.

\bibitem{sim15} B. Simon. Fifteen problems in mathematical physics. {\it Perspectives in Mathematics: Anniversary of Oberwolfach} (1984), 423–454.

\bibitem{simXXI} B. Simon. Schr\"odinger operators in the twenty-first century. {\it Mathematical Physics} Imp. Coll. Press, London (2000), 283–288.

\bibitem{sin} Ya. G. Sinai. Anderson localization for one-dimensional difference Schr\"{o}dinger operator with quasi-periodic potential. {\it J. Stat. Phys.} {\bf 46} (1987), 861-909.

\bibitem{se} C. Soukoulis and E. Economou. Localization in one-dimensional lattices in the presence of incommensurate potentials. {\it Physical Review Letters} {\bf 48(15)} (1982), 1043.

\bibitem{ss} E. Sorets and T. Spencer. Positive Lyapunov exponents for Schr\"odinger operators with quasi-periodic potentials. {\it Comm. Math. Phys.} {\bf 142} (1991), 543–566.

\bibitem{wz} Y. Wang and Z. Zhang. Uniform positivity and continuity of Lyapunov exponents for a class of $C^2$ quasiperiodic Schr\"odinger cocycles. {\it J. Funct. Anal.} {\bf 268} (2015), 2525-2585.

\bibitem{wwxyz} Y. Wang, X. Xia, J. You, Z. Zheng and Q. Zhou. Exact Mobility Edges for 1D Quasiperiodic Models. {\it Comm. Math. Phys.} {\bf 401} (2023), no. 3, 2521–2567.

\bibitem{xu} D. Xu. Density of positive Lyapunov exponents for symplectic cocycles. {\it J. Eur. Math. Soc. (JEMS)} {\bf 21(10)} (2019), 3143-3190.


\bibitem{sy} S. Yamashita. Some remarks on analytic continuations. {\it Tohoku Math. J.} (2) {\bf 21} (1969), 328-335.

\bibitem{youcongr} J. You. Quantitative Almost Reducibility and Its Applications. {\it Proceedings of ICM} 2018.

\bibitem{yicmp} J. You. Some problems in quasiperiodic Schr{\"o}dinger operators. {\it Journal of Mathematical Physics} {\bf 67(6)} (2026), 062704.

\end{thebibliography}
\end{document}